\documentclass[11pt,a4paper,reprint,superscriptaddress,aps,prx,longbibliography]{revtex4-2}

\usepackage{tensor}
\usepackage{amsmath}
\usepackage{amsthm}
\usepackage{graphicx}
\usepackage{comment}
\usepackage{amsfonts}
\usepackage{bm}
\usepackage{braket}
\usepackage{xcolor}

\usepackage{braket}
\usepackage{enumerate}
\usepackage{mathtools}
\usepackage{bbm}
\usepackage{mathrsfs}
\usepackage{amssymb}
\usepackage{cases}

\usepackage[colorlinks=true,linkcolor=cyan,urlcolor=teal,citecolor=cyan]{hyperref}

\newcommand{\temp}[1]{\textcolor{blue}{TEMP: #1}}

\newcommand{\ii}{\mathrm{i}}
\newcommand{\dd}{\mathrm{d}}

\newcommand{\gconv}{\xrightarrow{G\text{-cov.}}}

\newcommand{\rap}{R}

\newcommand{\cov}{\mathrm{Cov}}

\newcommand{\dist}{\mathrm{d}}
\newcommand{\cost}{\mathrm{c}}

\DeclarePairedDelimiter\ceil{\lceil}{\rceil}
\DeclarePairedDelimiter\floor{\lfloor}{\rfloor}

\newtheorem{thm}{Theorem}
\newtheorem{lem}[thm]{Lemma}
\newtheorem*{lem*}{Lemma}
\newtheorem{cor}[thm]{Corollary}

\newtheorem{prop}[thm]{Proposition}
\theoremstyle{remark}

\begin{document}

\title{Quantum geometric tensor determines the pure-state i.i.d. conversion rate in the resource theory of asymmetry for any compact Lie group}

\makeatletter
\let\inserttitle\@title
\makeatother

\author{Koji Yamaguchi}
\affiliation{Department of Communication Engineering and Informatics, University of Electro-Communications, 1-5-1 Chofugaoka, Chofu,
Tokyo, 182-8585, Japan}

\author{Yosuke Mitsuhashi}
\affiliation{Department of Basic Science, University of Tokyo, 3-8-1 Komaba, Meguro-ku, Tokyo 153-8902, Japan}

\author{Tomohiro Shitara}
\affiliation{NTT Computer and Data Science Laboratories, NTT Corporation,
3-9-11 Midori-cho, Musashino-shi, Tokyo 180-8585, Japan}

\author{Hiroyasu Tajima}
\affiliation{Department of Communication Engineering and Informatics, University of Electro-Communications, 1-5-1 Chofugaoka, Chofu,
Tokyo, 182-8585, Japan}
\affiliation{JST, FOREST, 4-1-8 Honcho, Kawaguchi, Saitama, 332-0012, Japan}

\begin{abstract}
Quantifying physical concepts in terms of the ultimate performance of a given task
has been central to theoretical progress, as illustrated by thermodynamic entropy and entanglement entropy, which respectively quantify irreversibility and quantum correlations. Symmetry breaking is equally universal, yet lacks such an operational quantification. While an operational characterization of symmetry breaking through asymptotic state-conversion efficiency is a central goal of the resource theory of asymmetry (RTA), such a characterization has so far been completed only for the $U(1)$ group among continuous symmetries. Here, we identify the complete measure of symmetry breaking for a general continuous symmetry described by any compact Lie group. Specifically, we show that the asymptotic conversion rate between many copies of pure states in RTA is determined by the quantum geometric tensor, thereby establishing it as the complete measure of symmetry breaking.
As an immediate consequence of our conversion rate formula, we also resolve the Marvian-Spekkens conjecture on conditions for reversible conversion in RTA, which has remained unproven for over a decade.
Leveraging the connection between symmetry breaking and the theory of quantum reference frames, we also systematically introduce a standardized reference state for frameness based on our asymptotic conversion theory.
In addition, by applying our analysis to a standard quantum-thermodynamic scenario, we show that asymptotic state conversion in contact with heat baths generally requires macroscopic coherence in the thermodynamic limit.
\end{abstract}

\maketitle

\section{Introduction}
Quantifying fundamental concepts through the ultimate performance of a given task has driven theoretical advances in physics. Thermodynamic entropy, emerging from the study of the ultimate limits of state transformation under adiabatic operations, serves as a measure of irreversibility and has led to a modern formulation of the second law of thermodynamics \cite{lieb_physics_1999}.
Likewise, entanglement entropy was discovered as a measure of quantum correlations that fully characterizes the state convertibility under local operations and classical communication (LOCC) \cite{bennett_concentrating_1996}.  Its relevance now extends beyond quantum information science, playing a pivotal role in developments across various fields of physics, including condensed matter physics \cite{vidal_entanglement_2003,calabrese_entanglement_2004} and high-energy physics~\cite{ryu_holographic_2006}.

However, symmetry and its breaking, equally universal to irreversibility and quantum correlations, still lack such an operational characterization.
They are cornerstones of physics, which play a vital role in characterizing natural phenomena across almost every modern field.
Considering its ubiquity, quantifying symmetry breaking under symmetry constraints in terms of operational limits will be a foundation to yield profound theoretical insights and a far-reaching impact, paralleling the influential role of entanglement entropy in modern physics.

The operational quantification of fundamental concepts constitutes a central challenge in resource theories \cite{chitambar_quantum_2019}, which serve as versatile frameworks generalizing entanglement theory and thermodynamics. 
Their major goal is to identify the quantity that fully specifies the optimal asymptotic conversion rate between many identical copies of quantum states by given allowed operations, which is called a complete measure~\cite{sagawa_entropy_2022,datta_is_2023}. 
Entanglement entropy, von Neumann entropy, and Helmholtz free energy fulfill this role for LOCC operations \cite{bennett_concentrating_1996}, adiabatic operations \cite{horodecki_reversible_2003}, and isothermal operations \cite{brandao_resource_2013}, respectively. Similarly, identifying an analogous measure for symmetry breaking necessitates the development of a suitable resource theory.

The resource theory of asymmetry (RTA) \cite{bartlett_reference_2007,gour_resource_2008,gour_measuring_2009,marvian_mashhad_symmetry_2012,korzekwa_resource_2013,marvian_asymmetry_2014}, previously referred to as the resource theory of quantum reference frames \cite{bartlett_reference_2007,gour_resource_2008,gour_measuring_2009}, provides a rigorous framework for quantifying symmetry breaking. However, despite active research in RTA for over a decade \cite{bartlett_reference_2007,gour_resource_2008,gour_measuring_2009, marvian_quantum_2016, marvian_no-broadcasting_2019,lostaglio_coherence_2019,lostaglio_quantum_2015, Faist2015Gibbs-preserving, tajima_GibbsPreservingOperationsRequiringInfiniteAmount_2025,marvian_coherence_2020,marvian_operational_2022,kubica_using_2021,zhou_new_2021,yang_optimal_2022,tajima_UniversalLimitationQuantuminformationrecovery_2021,liu_QuantumErrorCorrectionmeetscontinuous_2021,tajima_universal_2022,liu_approximate_2023,tajima_uncertainty_2018,tajima_coherence_2020,tajima_universal_2022, WAY_RTA1, marvian_InformationtheoreticAccountWignerArakiYanasetheorem_2012,korzekwa_resource_2013,tajima_CoherencevarianceUncertaintyRelationcoherencecost_2019a, tajima_universal_2022,marvian_mashhad_symmetry_2012,marvian_asymmetry_2014,takagi_skew_2019}, progress toward identifying a complete measure for symmetry breaking remains limited. In particular, for continuous symmetries, a complete measure has been established only for the simplest $U(1)$ group \cite{gour_resource_2008,marvian_coherence_2020,marvian_operational_2022}, while nothing is currently known about any other continuous groups beyond a few highly specific examples of states \cite{gour_resource_2008,yang_units_2017}. 
Consequently, while the RTA for the $U(1)$ group has led to applications in a wide range of fields, including quantum thermodynamics \cite{lostaglio_quantum_2015,Faist2015Gibbs-preserving,marvian_coherence_2020,tajima_GibbsPreservingOperationsRequiringInfiniteAmount_2025}, measurements \cite{WAY_RTA1,marvian_InformationtheoreticAccountWignerArakiYanasetheorem_2012,korzekwa_resource_2013,tajima_CoherencevarianceUncertaintyRelationcoherencecost_2019a, tajima_universal_2022,ET2023}, quantum computing \cite{tajima_uncertainty_2018,tajima_coherence_2020,tajima_UniversalLimitationQuantuminformationrecovery_2021,tajima_universal_2022}, error-correcting codes \cite{kubica_using_2021,zhou_new_2021,yang_optimal_2022,tajima_UniversalLimitationQuantuminformationrecovery_2021,liu_QuantumErrorCorrectionmeetscontinuous_2021,tajima_universal_2022,liu_approximate_2023}, and black hole physics \cite{tajima_UniversalLimitationQuantuminformationrecovery_2021,tajima_universal_2022}, it has a limited impact on research domains dictated by non-Abelian Lie groups, such as $SU(2)$ spin‑rotation symmetry in quantum magnets \cite{haldane_nonlinear_1983,affleck_rigorous_1987} and cold‑atom Hubbard models \cite{jaksch_cold_2005}, and the $SU(4)$ symmetry realized in multicomponent fractional quantum Hall systems in graphene \cite{dean_multicomponent_2011}.
Therefore, identifying fundamental measures for symmetry breaking for non-Abelian Lie groups is an urgent task to open new avenues in these areas.

In this paper, we identify the complete measure of symmetry breaking for general continuous symmetries described by any compact Lie groups. Specifically, we show that the quantum geometric tensor (QGT) \cite{provost_riemannian_1980,berry_quantum_1989}, consisting of the quantum Fisher information matrix and the Berry curvature, fully determines the asymptotic conversion rate between independent and identically distributed (i.i.d.) pure states in RTA. 
Interestingly, the QGT has been studied in a different context of topological phases of matter~\cite{thouless_quantized_1982,zhao_singularities_2009,gu_fidelity_2010,xiao_berry_2010}.
Our result reveals its operational interpretation by establishing it as the complete measure of symmetry breaking in RTA, and thus provides information-theoretic foundations for the use of QGT, much like entanglement entropy in entanglement theory.

Our result also provides a fundamental formula for the exploitation of symmetry breaking as a resource in quantum technologies. 
Concretely, it presents the ultimate limits of distillation~\cite{nielsen_quantum_2010}, i.e., converting many low-quality resource states into a few high-quality ones, which has played a critical role in the field of entanglement, as illustrated in applications such as quantum cryptography \cite{ekert_quantum_1991}, distributed quantum computing \cite{cirac_distributed_1999}, and quantum internet \cite{kimble_quantum_2008}.
Moreover, since our formula is expressed through a single matrix inequality, it yields a computable prediction of the fundamental efficiency bound of the distillation processes.
As an immediate consequence of the formula, we also derive a necessary and sufficient condition for reversible conversion, which resolves a long-standing open problem known as the Marvian-Spekkens conjecture \cite{marvian_mashhad_symmetry_2012,marvian_asymmetry_2014}. 

The significance of quantifying symmetry breaking also extends to the investigation of quantum reference frames.
Any physical operation is fundamentally defined relative to a reference structure, such as a time origin or a Cartesian frame. If such a structure is absent or misaligned, operations cannot, in general, be perfectly implemented. This effective lack of reference gives rise to superselection rules~\cite{bartlett_reference_2007}, which restrict the set of feasible operations. 
A symmetry-breaking quantum state therefore serves as a valuable resource for overcoming these restrictions.
For example, a state that is not invariant under time translation can serve as a quantum clock, specifying the time origin. Similarly, a state that is not invariant under spatial rotations can function as a quantum gyroscope, defining the orientation of a Cartesian frame. Building on this perspective, the RTA has in fact been developed as a theory of quantum reference frames in earlier works~\cite{bartlett_reference_2007,gour_resource_2008,gour_measuring_2009}.

When symmetry breaking is regarded as a resource for correcting misalignment, the ability to manipulate such resources becomes essential. In particular, efficient communication requires systematic conversions of these resources. For instance, when quantum communication channels are limited, it is essential to convert existing resources into high-quality states, thereby maximizing the reference-frame information transmitted per channel use. Conversely, if the supply of symmetry-breaking resources is restricted but the goal is to distribute them among many parties, the priority shifts to producing as many resource states as possible, even at the expense of their individual quality. The conversion theory within the RTA directly applies to such scenarios, providing explicit protocols for optimal conversions. As a demonstration, we show that the conversion theory enables a systematic definition of standardized reference states for continuous symmetries associated with semisimple compact Lie groups.

The development of tools for quantifying symmetry breaking in RTA enables applications across diverse scenarios where dynamics are constrained by conservation laws. A particularly important case is energy conservation, which is equivalently expressed as time-translation symmetry. The breaking of this symmetry---namely, coherent superpositions of energy eigenstates with different eigenvalues---represents a fundamental resource, commonly referred to as energy coherence. Quantum thermodynamics~\cite{brandao_resource_2013,lostaglio_quantum_2015,lostaglio_description_2015,weilenmann_axiomatic_2016,gour_resource_2015,faist_macroscopic_2019,sagawa_asymptotic_2021} offers a natural setting for studying this resource, as its central objective is to determine which state transformations are achievable in a system in contact with thermal baths under total energy conservation. 

In the restricted setting where only a single heat bath at a fixed temperature is available, the associated transformations are termed thermal operations, which are central to the resource theory of athermality~\cite{brandao_resource_2013,lostaglio_quantum_2015,lostaglio_description_2015,weilenmann_axiomatic_2016,gour_resource_2015,faist_macroscopic_2019,sagawa_asymptotic_2021} (see also~\cite{lipka-bartosik_AllStatesAreUniversalCatalysts_2021,woods_AutonomousQuantumDevicesWhenAre_2023,kondra_CoherenceManipulationAsymmetryThermodynamics_2024,tajima_GibbsPreservingOperationsRequiringInfiniteAmount_2025,shiraishi_QuantumThermodynamicsCoherenceCovariantGibbsPreserving_2025,zambon_QuantumProcessesThermodynamicResourcesRole_2025} for recent works in this area). 
For states without any energy coherence, called quasiclassical states, the non-equilibrium free energy fully characterizes state convertibility under thermal operations~\cite{brandao_resource_2013}. However, the situation for general, non-quasiclassical states remains unsettled~\cite{faist_macroscopic_2019,sagawa_asymptotic_2021}. The difficulty arises from the fact that thermal operations cannot generate energy coherence, which must therefore be regarded as an independent resource, distinct from non-equilibrium free energy. Existing approaches~\cite{brandao_resource_2013,faist_macroscopic_2019,sagawa_asymptotic_2021} have often quantified the required coherence in terms of the operator norm of Hamiltonian of an external system that provides energy coherence, suggesting that only a small source of coherence is needed in the thermodynamic limit. However, this perspective does not yield a refined characterization of state-specific resource requirements, as it reflects only the external system's specifications. 

Building on the establishment of the QGT as a complete measure of asymmetry for a broad class of continuous symmetries, we revisit this problem in a more general scenario involving contact with multiple baths at different temperatures under energy conservation. We demonstrate that, contrary to the conventional view, achieving certain state conversions via thermal contact require a macroscopic amount of energy coherence that scales extensively with system size when quantified with the QGT, even when baths at different temperatures are allowed. This reveals that macroscopic coherence can be necessary in quantum thermodynamics not only in the single-bath setting, but also in multi-bath scenarios relevant, for example, to heat-engine setups. This finding highlights that the choice of asymmetry measure is a decisive factor for advancing our understanding of resource requirements in thermodynamic processes. Moreover, we show that the same observation extends to settings with additional conserved charges, such as particle number or angular momentum, implying that state transformations in thermodynamic processes generally require an extensive amount of asymmetry.

\section*{Outline}
The outline of the remainder of the paper with a summary of the main results and key ideas are as follows.

In the Preliminaries (Sec.~\ref{sec:preliminary}), we review and introduce key notions required to state our main results. Specifically,
\begin{itemize}
\item In Sec.~\ref{sec:preliminary_RTA}, we review the standard framework of the RTA. We also introduce the conversion rate between states in the RTA, which is the primary quantity we study in this work.
\item In Sec.~\ref{sec:preliminaries_QGT}, we briefly review the QGT, which quantifies the sensitivity of a quantum state under an infinitesimal change of parameters. In our study of symmetries associated with a compact Lie group $G$, the QGT is a $\dim G \times \dim G$ positive semidefinite matrix, which can be expressed as a non-symmetrized covariance matrix of the generators of infinitesimal group transformations (Eq.~\eqref{eq:qgt_covariance_matrix}).
\end{itemize}

In Sec.~\ref{sec:main_result}, we first present the main theorem of this paper (Theorem~\ref{thm:conversion_rate_projective_finite_number}). Equation~\eqref{eq:rate_formula_compact_connected} gives an explicit formula for the conversion rate between i.i.d.\ pure states, expressed as a single matrix inequality between the QGTs of the input and output states. This establishes the QGT as the complete measure of asymmetry for continuous symmetries described by a compact Lie group. We further rewrite the formula in terms of the quantum max-relative entropy between the QGTs, which makes its computability via semidefinite programming transparent. Finally, we illustrate the formula for various symmetry groups and clarify its relation to prior studies. Specifically,
\begin{itemize}
\item In Sec.~\ref{sec:U(1)_same_periods}, we apply the formula to the $U(1)$ symmetry, which reproduces the result of \cite{marvian_operational_2022}.
\item In Sec.~\ref{sec:reversible_conversion}, as a corollary of the main result, we prove a necessary and sufficient condition for asymptotically reversible conversion for a compact Lie group (Corollary~\ref{cor:reversible_conversion}). We then show that this result proves the Marvian--Spekkens conjecture \cite{marvian_asymmetry_2014,marvian_mashhad_symmetry_2012}, which has remained open for over a decade. This also yields reversible conversion rates for specific examples that were calculated in earlier, separate studies \cite{gour_resource_2008,yang_units_2017,marvian_operational_2022}.
\item In Sec.~\ref{sec:irreversiblity}, we study the irreversibility of asymptotic conversion. As an explicit example, we compute the conversion rate for the $SU(2)$ group, obtaining a result consistent with a prior study for a restricted class of pure states \cite{gour_resource_2008}.
\item In Sec.~\ref{sec:finite_groups}, we show that our formula also applies to finite groups. Consequently, we find that the conversion rate is either infinite or zero, consistent with \cite{shitara_iid_2024}.
\end{itemize}

In Sec.~\ref{sec:proof_sketch}, we explain the key ideas used to prove the main theorem. The proof consists of two parts: the converse part, which establishes an upper bound on the conversion rate, and the direct part, which proves the achievability of this bound (i.e., optimality). Specifically,
\begin{itemize}
    \item In the converse part (Sec.~\ref{sec:monotonicity_QGT_sketch}), we first prove the monotonicity of the QGT. To this end, we first establish a connection between the QGT and Petz's monotone metrics \cite{morozova_markov_1989,petz_monotone_1996}. Using the monotonicity of Petz's metrics, we show that the QGT is non-increasing under exact conversion, in the sense of a matrix inequality, i.e., the L\"{o}wner order. This implies that the QGT is a valid asymmetry monotone (Eq.~\eqref{eq:monotonicity_QGT_exact}).

    Since the QGT is additive, it is natural to expect that its linear rate (i.e., the QGT per copy) provides a useful monotone in the asymptotic regime (analogous to densities of extensive quantities in the thermodynamic limit). Proving monotonicity of this QGT rate for asymptotic conversion with vanishing error, however, is technically delicate. The key difficulty is that the QGT can, in general, scale quadratically with system size, so its linear rate can change drastically even under small conversion errors. This phenomenon is known as asymptotic discontinuity~\cite{donald_uniqueness_2002,plenio_introduction_2007,gour_measuring_2009,marvian_coherence_2020,marvian_operational_2022,yamaguchi_smooth_2023}. Nonetheless, by carefully analyzing the asymptotic behavior of Petz's monotone metrics, we establish monotonicity of the QGT rate (Eq.~\eqref{eq:monotonicity_QGT_converse_part}). This monotonicity in turn yields an upper bound on the conversion rate (Eq.~\eqref{eq:converse_part_statement}). The technical details of the converse part are given in Appendix~\ref{app:section_for_converse_part}.
    
    \item In the direct part (Sec.~\ref{sec:optimality_sketch}), we prove the achievability of the above bound (and hence optimality). To this end, we explicitly construct conversion channels. A key observation is that convertibility between states in the RTA is directly related to convertibility between quantum statistical models (Lemma~\ref{lem:gcov_and_cptp}). The asymptotic behavior of these statistical models is analyzed using quantum local asymptotic normality \cite{guta_local_2006,kahn_local_2009,girotti_optimal_2024,lahiry_minimax_2024}, which identifies a limit model that approximates the behavior of i.i.d.\ statistical models in the asymptotic regime. Importantly, we find that this limit model is characterized by the QGT associated with the underlying statistical model. Building on this observation, we construct the conversion channels via what we call an ``estimate-and-convert'' strategy. The technical details of the direct part are given in Appendix~\ref{app:section_for_direct_part}.
\end{itemize}

Having established the pure-state theory of asymptotic conversion, we next extend our analysis to mixed states in Sec.~\ref{sec:mixed_state_asymmetry}. We consider two fundamental scenarios: asymmetry distillation and asymmetry dilution. Although these results do not provide a complete characterization of mixed-state conversion rates, they serve as a foundation for future work. Specifically,
\begin{itemize}
    \item In Sec.~\ref{sec:distillation_of_asymmetry}, we study the conversion rate from i.i.d.\ mixed states to i.i.d.\ pure states in the RTA, namely, the distillable asymmetry. To this end, we introduce an extension of the QGT to mixed states and show that it is a valid asymmetry monotone. Consequently, we derive an upper bound on the distillable asymmetry and obtain a simple condition under which the distillable asymmetry vanishes (Corollary~\ref{cor:projector_distillablle_asymmetry}). When specialized to time-translation symmetry, this condition reduces to the result of \cite{marvian_coherence_2020}.
    \item In Sec.~\ref{sec:dilution_of_asymmetry}, we study the opposite scenario, namely, the conversion rate from i.i.d.\ pure states to i.i.d.\ mixed states in the RTA, i.e., the asymmetry cost. By combining the standard typical-sequence argument~\cite{hayden_asymptotic_2001,marvian_operational_2022} with our estimate-and-convert strategy, we prove an upper bound on the asymmetry cost by explicitly constructing conversion channels (Proposition~\ref{prop:asymmetry_of_formation}). This bound is tight for $U(1)$ symmetry, reproducing the result of \cite{marvian_operational_2022}. While we leave open whether the bound is tight for a general Lie group, we clarify why the $U(1)$ argument does not directly extend to the general case.
\end{itemize}

Building on the results above, in Sec.~\ref{sec:quantum_reference_frame} we apply our analysis to quantum reference frames. We begin with a careful review of the relation between asymmetry and quantum reference frames. 
In general, any physical operation is defined relative to some reference structure, such as a clock or a spatial direction. In the absence of an absolute (i.e., shared) reference frame, asymmetry becomes the resource required to implement such operations.
In this sense, asymmetry can also be regarded as a resource of frameness. Specifically,
\begin{itemize}
    \item We show that our pure-state conversion theory enables us to introduce a standardized reference state that serves as a benchmark for comparing frameness, analogous to the role of an \textit{ebit} in entanglement theory. So far, such a standardized reference state has been proposed only for the $SU(2)$ group~\cite{yang_units_2017}. Our conversion theory provides a systematic way to define standardized reference states for any semisimple compact Lie group. Concretely, we define a standardized reference state in Eq.~\eqref{eq:standardized_state_def}. This state is a universal frameness resource, in the sense that it can be converted into any pure state at a nonzero rate. As a consequence of the conversion-rate formula, we further show that the rate is given by the min-entropy of the QGT (Eq.~\eqref{eq:conversion_from_standardized_state}).
    \item We then clarify the distinction between differential-geometric quantities (such as the QGT and the quantum Fisher information) and commonly used entropic quantities (such as $G$-asymmetry) from the viewpoint of quantum reference frames. For instance, for the rotation group, the former quantifies both the direction and the magnitude of asymmetry, whereas the latter is insensitive to direction. Therefore, if we are concerned not only with the amount of asymmetry but also with the direction in which a state breaks the symmetry, geometric quantities provide more refined information.
\end{itemize}

In addition, we apply our framework to thermodynamics by considering interactions with thermal baths under conservation laws. The central physical setting is described in Sec.~\ref{subsec:qt-operational-setting}: a system couples to heat baths initially prepared in Gibbs states via a global unitary that conserves the total energy. When all baths are at the same fixed temperature, the resulting dynamics reduces to thermal operations. Here, however, we allow the baths to be at different temperatures, so that our framework applies to a broader class of thermodynamic scenarios. We further extend this setting to include additional conserved quantities, such as particle number and angular momentum, associated with a Lie group. 

Using the tools developed in this work, we aim to demonstrate nontrivial constraints on state transformations under such thermal contact and to analyze the resources required to circumvent them. To this end, we relate thermal contact to covariant operations in Sec.~\ref{subsec:qt-resource-perspective}. A key theoretical ingredient is Theorem~\ref{thm:noniid_monotonicity_q}, which establishes monotonicity of the QGT in the asymptotic distillation of i.i.d. pure states. Importantly, this theorem applies to general initial states of the system (including non-i.i.d. and mixed states), enabling analysis of a wide range of scenarios.

In particular, we study the following two classes of scenarios:
\begin{itemize}
    \item In Sec.~\ref{subsec:qt-no-go-iid}, we consider the asymptotic distillation of i.i.d. pure states from general i.i.d. states. Using the monotonicity of the QGT, we establish a no-go theorem for distillation (Corollary~\ref{cor:projector_distillablle_athermality}). In particular, we show that for a typical full-rank state, no asymmetric pure state can be distilled under thermal contact at a linear rate. We illustrate this result with the example of $U(1)$ symmetry.
    \item In Sec.~\ref{subsec:qt-resource-requirement}, we consider a generalized asymptotic distillation setting that allows for an external resource system assisting the conversion. By applying the monotonicity of the QGT (Theorem~\ref{thm:noniid_monotonicity_q}), we show that circumventing the above no-go theorem requires the external resource system to possess a macroscopic amount of asymmetry. In the case of time-translation symmetry, this implies that asymptotic state transformation under thermal contact in general requires macroscopic energy coherence (Eq.~\eqref{eq:variance_rate_bound}). 
\end{itemize}

Finally, we present our conclusions in Sec.~\ref{sec:conclusions}.

\section{Preliminaries}\label{sec:preliminary}

We review the resource theory of asymmetry (RTA) and the quantum geometric tensor (QGT), and introduce notations used in this paper.

\subsection{Resource theory of asymmetry}\label{sec:preliminary_RTA}
In this paper, we follow the standard setup of RTA \cite{bartlett_reference_2007,gour_resource_2008,gour_measuring_2009,marvian_mashhad_symmetry_2012,korzekwa_resource_2013,marvian_asymmetry_2014}. 
We study a symmetry described by a group $G$ that is realized by a \textit{projective unitary representation} $U$, which maps $g\in G$ to a unitary operator $U(g)$ on a quantum system. From the consistency of successive application of symmetry transformations, $U$ must satisfy $U(g_1)U(g_2)=\omega(g_1,g_2)U(g_1g_2)$ for any $g_1,g_2\in G$, where $\omega$ is a complex-valued function such that $|\omega(g_1,g_2)|=1$. In a special case where $\omega(g_1,g_2)=1$ for all $g_1,g_2\in G$, $U$ is called a \textit{non-projective unitary representation} or 
a \textit{unitary representation} for short. 

Just like in other resource theories, including the entanglement theory, RTA is defined by specifying the \textit{free states} and \textit{free operations} that are considered to be freely prepared and implemented. In RTA, free states are symmetric states, which are invariant under any symmetry transformations. That is, a state $\rho$ is $G$-symmetric iff $U(g)\rho U(g)^\dag =\rho$ for all $g\in G$. Free operations are $G$-covariant channels. Here, a channel $\mathcal{E}$ is $G$-covariant iff it satisfies $\mathcal{E}\circ \mathcal{U}_g=\mathcal{U}_g'\circ \mathcal{E}$ for all $g\in G$, where $\mathcal{U}_g$ and $\mathcal{U}_g'$ are defined by $\mathcal{U}_g(\cdot)\coloneqq U(g)(\cdot)U(g)^\dag $ and $\mathcal{U}_g'(\cdot)\coloneqq U'(g)(\cdot)U'(g)^\dag $ using projective unitary representations $U$ and $U'$ of $G$ on the input and output systems, respectively. Throughout this paper, we use the prime symbol to represent a quantity related to the output system. 

The physical implication of $G$-covariance is clarified by the covariant Stinespring dilation theorem~\cite{keyl_optimal_1999,marvian_mashhad_symmetry_2012}, which states that a $G$-covariant channel admits a representation as the dynamics of an open system interacting with an environment under conservation laws. More precisely, a quantum channel $\mathcal{E}$ is $G$-covariant if and only if there exist ancillary systems $E$ and $E'$, carrying projective unitary representations $U_E$ and $U_{E'}$ of $G$, such that
\begin{align}
\mathcal{E}(\cdot)=\mathrm{Tr}_{E'}\left(V\left(\cdot\otimes \eta_{E}\right)V^\dagger\right),
\end{align}
where $\eta_E$ is a $G$-symmetric state on $E$, and $V$ is a unitary operator satisfying 
\begin{align}
V(U(g)\otimes U_E(g))=(U'(g)\otimes U_{E'}(g)) V,\quad \forall g\in G.
\label{eq:Stinespring_unitary_conservation}
\end{align}
Since Eq.~\eqref{eq:Stinespring_unitary_conservation} implies that the joint unitary evolution $V$ respects the symmetry described by $G$, a $G$-covariant channel characterizes the most general open-system dynamics consistent with conservation laws, without supplying asymmetric states from the environment (see also Appendix~\ref{app:Stinespring_general}).

As a simple yet important characterization of convertibility among states, the symmetry subgroup is known, which is the set of group elements that leave the state invariant, i.e., 
\begin{align}
    \mathrm{Sym}_G(\rho)\coloneqq \{g\in G \mid \mathcal{U}_g(\rho) =\rho\}
\end{align}
for the input state $\rho$, and similarly for the output state $\sigma$. 
It is shown in \cite{marvian_mashhad_symmetry_2012} that if a state $\rho$ can be converted to a state $\sigma$ via some $G$-covariant operation without error, the symmetry subgroups must satisfy $\mathrm{Sym}_G(\rho)\subset\mathrm{Sym}_G(\sigma)$.

We analyze the i.i.d. setup in RTA where identical copies of a state are converted to identical copies of another state. 
In this setup, we use the tensor product of a representation, which corresponds to the conservation of global quantities such as total particle number in the $U(1)$ case and total spins in the $SU(2)$ case. 
Concretely, we say that a channel $\mathcal{E}$ from $N$ copies of input state to $M$ copies of output state is $G$-covariant if 
$\mathcal{E}$ satisfies $\mathcal{E}\circ \mathcal{U}_{g}^{\otimes N}=\mathcal{U}_{g}^{\prime \otimes M}\circ \mathcal{E}$ for all $g\in G$.

We finally define the asymptotic conversion rate. 
We say that a state $\rho$ is asymptotically convertible to another state $\sigma$ with a conversion rate $r$ iff there exists a sequence $\{\mathcal{E}_N\}_N$ of $G$-covariant channels such that $\lim_{N\to\infty}T\left(\mathcal{E}_N(\rho^{\otimes N}),\sigma^{\otimes \floor{rN}}\right)=0$. Here $T(\rho,\sigma)\coloneqq \frac{1}{2}\|\rho-\sigma\|_1$ is the trace distance, which has operational significance in state distinguishability \cite{helstrom_quantum_1969,holevo_statistical_1973}. 
We denote this conversion by 
\begin{align}
    \{\rho^{\otimes N}\}_N\gconv\{\sigma^{\otimes \floor{rN}}\}_N, 
\end{align}
and we define the asymptotic conversion rate $\rap(\rho\to \sigma)$ by the supremum of achievable conversion rates $r$. 
In the following, we simply call $\rap(\rho\to \sigma)$ the conversion rate. 
In this paper, we establish a formula for calculating this conversion rate among pure states, which is applicable to any compact Lie group $G$.

\subsection{Quantum geometric tensor}\label{sec:preliminaries_QGT}

\subsubsection{Definition and basic properties}
Here, we introduce the QGT \cite{provost_riemannian_1980,berry_quantum_1989}, which quantifies the sensitivity of a quantum state under an infinitesimal change of parameters. Let $\ket{\xi(\lambda)}$ be a family of pure states, parameterized by $m$ real parameters $\lambda=(\lambda^1,\ldots,\lambda^m)^\top\in\mathbb{R}^m$, where $\top$ denotes the transpose. The QGT evaluated at $\ket{\psi}\coloneqq \ket{\xi(0)}$, is an $m\times m$ Hermitian matrix $\mathcal{Q}^{\psi, \xi}$ whose matrix elements are defined by 
\begin{align}
    \mathcal{Q}_{\mu\nu}^{\psi, \xi}\coloneqq \braket{\partial_\mu\psi|(I-\ket{\psi}\bra{\psi})|\partial_\nu \psi},\label{eq:QGT_definition}
\end{align}
where $\ket{\partial_\mu\psi}\coloneqq \frac{\partial}{\partial\lambda^\mu}\ket{\xi(\lambda)}|_{\lambda =0}$. Here, the superscripts $\psi$ and $\xi$ intend to emphasize that the QGT depends on both the chosen state $\ket{\psi}$ (i.e., the point in the parameterized family of states) and the parameterized family $\ket{\xi(\lambda)}$ itself. 

We remark that the QGT transforms as a $(0,2)$-tensor under smooth reparametrizations of the parameters $\lambda^\mu$. The indices $\mu$ and $\nu$ label tensor components: a superscript index indicates a contravariant component, whereas a subscript index indicates a covariant component under such reparametrizations. In this paper, we use tensor notation to clarify the transformation properties, even though we do not explicitly change coordinates. Hereafter, we employ the Einstein summation convention: whenever an index appears once as an upper index and once as a lower index, summation over that index is implied.

We remark that the QGT is positive semi-definite since $\gamma^{\mu*} \mathcal{Q}_{\mu\nu}^{\psi, \xi}\gamma^\nu=\|(I-\ket{\psi}\bra{\psi})\gamma^\nu\ket{\partial_\nu \psi}\|^2\geq 0$ holds for any complex vector $\gamma\in\mathbb{C}^m$. 
By regarding the QGT as a matrix, we also write this inequality as $\gamma^\dag \mathcal{Q}^{\psi, \xi}\gamma\geq 0$. 

Since the QGT is Hermitian, i.e., $\mathcal{Q}_{\mu\nu}^{\psi,\xi}=(\mathcal{Q}_{\nu\mu}^{\psi,\xi})^*$, its real and imaginary parts define symmetric and antisymmetric tensors, respectively. Specifically, the real part $\mathrm{Re}(\mathcal{Q}_{\mu\nu}^{\psi,\xi})$ is symmetric and corresponds to the quantum Fisher information matrix, while twice the imaginary part, $ 2\, \mathrm{Im}(\mathcal{Q}_{\mu\nu}^{\psi,\xi})$, is anti-symmetric and corresponds to the Berry curvature (for further details, see, e.g., \cite{provost_riemannian_1980,berry_quantum_1989,cheng_quantum_2013,liu_quantum_2019}). 

\subsubsection{QGT in the RTA}
Having introduced the general definitions above, we now proceed to the RTA setup. In RTA for a Lie group, we can define the QGT using a natural parametrized family of pure states defined with a projective unitary representation. Let $G$ be a Lie group, and $\dim G$ denote its dimension as a smooth manifold. Elements in the neighborhood of the identity $e\in G$ can be parametrized by $\lambda\in\mathbb{R}^{\dim G}$ as $g(\lambda)=e^{\ii \lambda^\mu A_\mu}$ by using a basis $\{A_\mu\}_{\mu=1}^{\dim G}$ of the Lie algebra $\mathfrak{g}$. Here, the convention in physics for the definition of Lie algebra is adopted, which differs from that in mathematics by a factor of the imaginary unit. Given a projective unitary representation $U$, we define $\ket{\xi(\lambda)}\coloneqq U(g(\lambda))\ket{\psi}$ for a pure state $\ket{\psi}$. For simplicity, we assume that the map $U$ is differentiable, which follows from continuity in the case of unitary representations. We treat the cases where $U$ is continuous but not differentiable in Appendix~\ref{app:differentiability_of_reprensetation}.
For a given $U$, we introduce Hermitian operators 
\begin{align}
    X_\mu\coloneqq  -\ii \frac{\partial}{\partial \lambda^\mu} U(g(\lambda))\biggl|_{\lambda=0}\quad \label{eq:hermitian_operators_Lie_alg}
\end{align}
for $\mu=1,\cdots,\dim G$. In other words, by introducing the derived representation $L$ of the Lie algebra associated with $U$ by $L(A)\coloneqq -\ii \left.\frac{\dd}{\dd t}U(e^{\ii t A})\right|_{t=0}$, the operator $X_\mu$ is equal to $L(A_\mu)$. 

As mentioned earlier, the QGT $\mathcal{Q}^{\psi, \xi}$ depends both on $\ket{\psi}$ and the parameterized family $\ket{\xi(\lambda)}$ in general. In the rest of this paper, however, we always consider the family $\ket{\xi(\lambda)}=U(g(\lambda))\ket{\psi}$ defined via a projective unitary representation $U$. Accordingly, we simply write the QGT as $\mathcal{Q}^\psi$ without explicitly indicating the dependence on $\ket{\xi(\lambda)}$. Since $\ket{\partial_\mu\psi}=\ii X_\mu\ket{\psi}$ holds with the operator $X_\mu$ defined in Eq.~\eqref{eq:hermitian_operators_Lie_alg}, we find the QGT is equal to the non-symmetrized covariance matrix for $\{X_\mu\}_{\mu=1}^{\dim G}$:
\begin{align}
    \mathcal{Q}_{\mu\nu}^{\psi}&=\braket{\psi|X_\mu(I-\ket{\psi}\bra{\psi})X_\nu|\psi}\nonumber\\
    &=
    \braket{\psi|X_\mu X_\nu|\psi}-\braket{\psi|X_\mu|\psi}\braket{\psi|X_\nu|\psi}\label{eq:qgt_covariance_matrix}.
\end{align} 
Note that in this case, the size of the QGT is determined only by the dimension of the Lie group, $\dim G$, independent of the dimensions of Hilbert spaces or the representations of the group. Since QGTs are Hermitian matrices, we can introduce a partial order $\mathcal{Q}^\psi\geq \mathcal{Q}^\phi$, which means that $\mathcal{Q}^\psi-\mathcal{Q}^\phi$ is positive semi-definite. We remark that this ordering is independent of the parametrization of the Lie group $G$ since both $\mathcal{Q}^\psi$ and $\mathcal{Q}^\phi$ transform as tensors under a coordinate transformation on the group.

So far, we have represented a pure state by a vector $\ket{\psi}$. In what follows, following a standard convention, we also represent the same pure state by the rank-one projector $\psi\coloneqq\ket{\psi}\bra{\psi}$, using the same Greek symbol when no confusion arises. Accordingly, the QGT evaluated at a pure state $\mathcal{U}_g(\psi)$ is denoted by $\mathcal{Q}^{\mathcal{U}_g(\psi)}$. We also denote the set of all the rank-one projectors on a Hilbert space $\mathcal{H}$ by $\mathcal{P}(\mathcal{H})$.

\section{Main results}\label{sec:main_result}

The main theorem of this paper is the following formula for the asymptotic conversion rate:
\begin{thm}
\label{thm:conversion_rate_projective_finite_number}
    Let $U$ and $U'$ be projective unitary representations of a compact Lie group $G$ on finite-dimensional Hilbert spaces $\mathcal{H}$ and $\mathcal{H}'$. The conversion rate from a pure state $\psi\in\mathcal{P}(\mathcal{H})$ to another pure state $\phi\in\mathcal{P}(\mathcal{H}')$ is given by 
    \begin{align}
        &\rap(\psi\to \phi)= \sup\{r\geq0 \mid\mathcal{Q}^{\psi}\geq r\mathcal{Q}^{\phi}\}\label{eq:rate_formula_compact_connected}
    \end{align}
    if $\mathrm{Sym}_{G}(\psi)\subset \mathrm{Sym}_G(\phi)$, and $\rap(\psi\to \phi)=0$ otherwise.
\end{thm}
See Fig.~\ref{fig:standard} for a schematic picture of the setup of Theorem~\ref{thm:conversion_rate_projective_finite_number}.

\begin{figure}[htb]
    \centering
    \includegraphics[width=8.3cm]{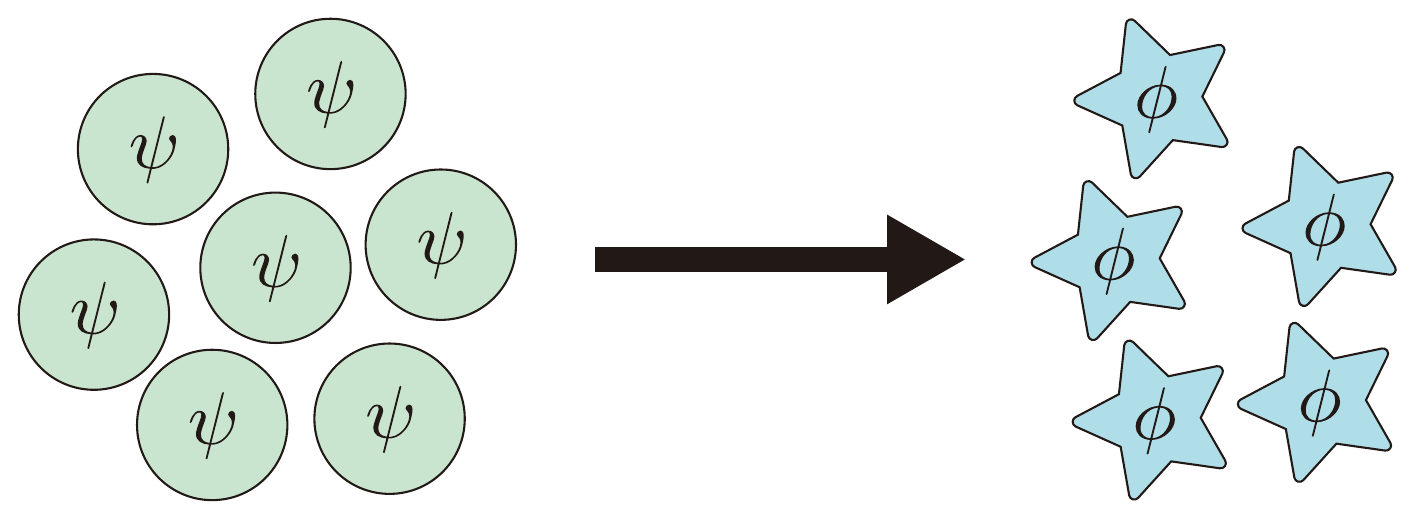}
    \caption{Schematic picture of the setup of Theorem~\ref{thm:conversion_rate_projective_finite_number}, where i.i.d. copies of a pure state $\psi$ are converted into i.i.d. copies of another pure state $\phi$ with an error that vanishes asymptotically. }
    \label{fig:standard}
\end{figure}

Although Theorem~\ref{thm:conversion_rate_projective_finite_number} is of practical significance, the following formula is also theoretically useful, where all group elements are treated equally:
\begin{widetext}
        \begin{numcases}{\rap(\psi \to \phi)=}
            \sup\{r\geq0 \mid\forall g\in G,\,\mathcal{Q}^{\mathcal{U}_g(\psi)}\geq r\mathcal{Q}^{\mathcal{U}'_g(\phi)}\}&(if $\mathrm{Sym}_{G}(\psi)\subset \mathrm{Sym}_G(\phi)$)\label{eq:conversion_rate_formula_forall_g}\\
            0&(otherwise)\label{eq:conversion_rate_formula_subgrop}
        \end{numcases}
\end{widetext}

The equivalence of Eq.~\eqref{eq:rate_formula_compact_connected} in Theorem~\ref{thm:conversion_rate_projective_finite_number} to Eq.~\eqref{eq:conversion_rate_formula_forall_g} follows from the fact that if the matrix inequality holds at a point 
in a compact Lie group, then it also holds at any other point in the group. This is because the QGTs at different group elements are interrelated by a congruence transformation independent of representation, as detailed in Appendix~\ref{app:reduction_to_a_finite_number_of_inequalities}. We note that the same argument for quantum Fisher information matrices instead of QGT can be found in \cite{gao_sufficient_2024}. 

As we show in Sec.~\ref{sec:monotonicity_QGT_sketch}, the QGT is an asymmetry measure: it is non-increasing under any $G$-covariant channel with respect to the partial order defined by matrix inequalities, i.e., the L\"{o}wner order. Equation~\eqref{eq:rate_formula_compact_connected}, therefore, identifies the QGT as the complete measure of symmetry breaking.

This result carries an important conceptual consequence. Resource measures are typically formulated as real-valued scalar functions, while operator-valued resource measures have recently been considered in \cite{kudo_fisher_2023,gao_sufficient_2024}. Our formula shows that an operator, namely the QGT matrix, governs the asymptotic conversion problem in the RTA. This establishes that the operator-valued formulation is not merely a generalization of the scalar formulation, but a necessary extension for resolving this central problem.

Within this operator-valued paradigm, scalar-valued measures are recovered as the special case of $1 \times 1$ matrices. As we shall see explicitly, this reduction occurs, for example, for $U(1)$ symmetry.

Interestingly, the conversion rate formula can also be expressed by using the quantum max-relative entropy. For positive semi-definite operators $\rho$ and $\sigma$, the quantum max-relative entropy~\cite{datta_min-_2009,tomamichel_quantum_2016} of $\rho$ with respect to $\sigma$ is defined by
\begin{align}
    D_{\max}(\rho\|\sigma)\coloneqq \inf \{\lambda\in\mathbb{R} \mid \rho \leq 2^\lambda\sigma\}.\label{eq:D_max_QGT}
\end{align}
Accordingly, the formula for a compact Lie group in Eq.~\eqref{eq:rate_formula_compact_connected} can be rewritten as
\begin{align}
    R(\psi\to\phi)=2^{- D_{\max}(\mathcal{Q}^{\phi}\|\mathcal{Q}^{\psi})}.\label{eq:rate_formula_D_max_QGT}
\end{align}
It is worth emphasizing that, despite the widespread use of the quantum max-relative entropy in various contexts---including entanglement theory~\cite{datta_min-_2009,datta_max-relative_2009}, quantum asymptotic equipartition property~\cite{tomamichel_fully_2009}, quantum cryptography~\cite{konig_operational_2009}, entropic uncertainty relations~\cite{tomamichel_uncertainty_2011} and quantum coherence~\cite{bu_maximum_2017}---it has been used exclusively as a divergence measure between (sub)normalized states in these applications. In contrast, to the best of our knowledge, this study is the first to apply the max-relative entropy between resource measures themselves and to connect it directly to the operationally meaningful conversion rate in the RTA. In this sense, our study extends its significance beyond established uses. 

Note that the optimization in Eq.~\eqref{eq:rate_formula_compact_connected} can be written explicitly as
\begin{align}
    \begin{aligned}
        \text{maximize}\quad & r \\
        \text{subject to}\quad 
        & \mathcal{Q}^\psi - r\,\mathcal{Q}^\phi \geq  0, \\
        & r \geq 0 .
    \end{aligned}
\end{align}
Since the constraint $\mathcal{Q}^\psi - r\,\mathcal{Q}^\phi \geq 0$ is a linear matrix inequality in the scalar variable $r$, this problem is a semidefinite program. Moreover, by using the Moore--Penrose inverse $\mathcal{Q}^{\psi+}$ of the QGT $\mathcal{Q}^\psi$, we obtain the following closed-form expression:
\begin{align}
    &\rap(\psi\to\phi)\nonumber\\
    &\quad =
    \begin{cases}
    +\infty & (\text{if } \mathcal{Q}^\phi=0)\\
    0 & (\text{if } \ker(\mathcal{Q}^\psi)\not\subset \ker(\mathcal{Q}^\phi))\\
     \displaystyle\frac{1}{\bigl\|
    (\mathcal{Q}^{\psi+})^{1/2}
    \mathcal{Q}^\phi
    (\mathcal{Q}^{\psi+})^{1/2}
    \bigr\|_\infty}
    & (\text{otherwise})
    \end{cases}
\end{align}
where $\|A\|_\infty$ denotes the operator norm of $A$. 

Key ideas for proving Eq.~\eqref{eq:conversion_rate_formula_forall_g}, under the assumption that $U$ and $U'$ are (non-projective) unitary representations of a compact Lie group $G$, are explained in Sec.~\ref{sec:proof_sketch}, while a rigorous proof is provided in Appendices~\ref{app:section_for_converse_part} and \ref{app:section_for_direct_part}. The extension of the formula to any projective unitary representations is provided in Appendix~\ref{app:formula_projective_unitary_rep} by using the method in \cite{shitara_iid_2024} that relates the conversion rate for projective unitary representations to (non-projective) unitary representations.

We remark that for any states $\rho$ and $\sigma$, Proposition~5 of \cite{marvian_mashhad_symmetry_2012} shows that if $\mathrm{Sym}_{G}(\rho)\not \subset \mathrm{Sym}_G(\sigma)$, then an exact conversion from $\rho$ to $\sigma$ via a $G$-covariant channel is impossible.
As a generalization of this observation to the asymptotic setting with vanishing error, we prove Eq.~\eqref{eq:conversion_rate_formula_subgrop} in Appendix~\ref{sec:symmetry_subgroup_zero_rate}. In Appendix~\ref{sec:sym_subgroup_conversion_rate}, we propose a way to circumvent this limitation by using additional resource states.

Theorem~\ref{thm:conversion_rate_projective_finite_number} is valid for any compact Lie group. As a demonstration, we here apply our formula for several groups, which provides a unified understanding of prior studies and proves an unsolved conjecture on reversible transformations. 

\subsection{U(1) group}\label{sec:U(1)_same_periods}
In a study \cite{marvian_operational_2022} on RTA for time-translation symmetry, the asymptotic conversion rate among states having the same finite period was calculated, which generalizes prior studies \cite{schuch_nonlocal_2004,schuch_quantum_2004,gour_resource_2008}. When the Hamiltonians of the input and output systems are $H$ and $H'$, the time-translation unitary operators are given by $e^{-\ii Ht}$ and $e^{-\ii H't}$. For pure states $\psi$ and $\phi$, the periods are defined as $\tau\coloneqq \inf \{t>0\mid e^{-\ii Ht}\psi e^{\ii Ht}=\psi\}$ and $\tau'\coloneqq \inf \{t>0\mid e^{-\ii H't}\phi e^{\ii H' t}=\phi\}$. In Theorem~1 in \cite{marvian_operational_2022}, it is proven that if $\tau=\tau'$, the conversion rate is given by the ratio of the variances, i.e., 
\begin{align}
     \rap(\psi\to\phi)=\frac{V(\psi,H)}{V(\phi,H')}\label{eq:rate_U1},
\end{align}
where we defined $V(\psi,H)\coloneqq\braket{\psi|H^2|\psi}-\braket{\psi|H|\psi}^2$. 
We note that, when investigating the conversion among states with the same period, we can assume without loss of generality that the Hamiltonians have integer eigenvalues after appropriately redefining them \cite{marvian_operational_2022,yamaguchi_beyond_2023}. Therefore, this result corresponds to the case of $G=U(1)$. 

In the notation of the present paper, we consider $G=U(1)$ and its representations $U(e^{\ii \theta})=e^{\ii H \theta}$ and $U'(e^{\ii \theta})=e^{\ii H' \theta }$ for $\theta\in [0,2\pi)$, where $H$ and $H'$ are Hermitian operators whose eigenvalues are integers. Since $\dim G=1$ for $G=U(1)$, the QGT is a scalar and given by $\mathcal{Q}^{\psi}=V(\psi,H)$ and $\mathcal{Q}^{\phi}=V(\phi,H')$. When the pure states $\psi$ and $\phi$ has the same period, i.e., $\mathrm{Sym}_{G}(\psi)=\mathrm{Sym}_{G}(\phi)$, we find Eq.~\eqref{eq:rate_U1} immediately follows from Eq.~\eqref{eq:rate_formula_compact_connected}.

\subsection{Reversible asymptotic conversion: Proof of the Marvian-Spekkens conjecture}\label{sec:reversible_conversion}
For pure states $\psi$ and $\phi$, we say that they are asymptotically reversibly convertible if and only if the conversion rates satisfy $\rap(\psi\to\phi)\rap(\phi\to\psi)=1$. The conversion between pure states with equal periods is an example of reversible conversion, as Eq.~\eqref{eq:rate_U1} shows. Prior to the establishment of the conversion theory in RTA for $U(1)$ group \cite{marvian_operational_2022}, Marvian and Spekkens proposed a conjecture for a necessary and sufficient condition for reversible conversion in RTA for connected compact Lie groups \cite{marvian_asymmetry_2014,marvian_mashhad_symmetry_2012}. 
We prove a statement equivalent to the Marvian-Spekkens conjecture here, where the equivalence is shown later in this subsection. 
\begin{cor}\label{cor:reversible_conversion}
    For a compact Lie group $G$, pure states $\psi$ and $\phi$ are asymptotically reversibly convertible if and only if both of the following conditions (A) and (B) are satisfied:
    \begin{enumerate}[(A)]
        \item $\mathrm{Sym}_G(\psi)=\mathrm{Sym}_G(\phi)$.
        \item There exists a unique $r>0$ such that $\mathcal{Q}^{\psi}=r\mathcal{Q}^{\phi}$.
    \end{enumerate}
    Note that the proportional constant $r$ in condition (B) provides the conversion rate $\rap(\psi\to\phi)$. 
\end{cor}
Here, in condition~(B), the requirement that $r$ be unique excludes the case where $\mathcal{Q}^{\psi}=\mathcal{Q}^{\phi}=0$, for which the proportionality constant is not unique. In this case, one has $R(\psi\to\phi)=R(\phi\to\psi)=\infty$ since $\psi$ and $\phi$ are symmetric, and hence these states formally fail to satisfy the reversibility condition $\rap(\psi\to\phi)\rap(\phi\to\psi)=1$.
\begin{proof}[Proof of Corollary~\ref{cor:reversible_conversion}]
 Suppose that $\psi$ and $\phi$ are asymptotically reversibly convertible. Then, $\rap(\psi\to\phi)$ and $\rap(\phi\to\psi) $
 must be non-vanishing, which requires that the condition (A) must hold. In this case, Theorem~\ref{thm:conversion_rate_projective_finite_number} implies
    \begin{align}
        \rap(\psi\to\phi)&=\sup\left\{r\geq 0\mid \mathcal{Q}^{\psi}\geq r  \mathcal{Q}^{\phi}\right\},\label{eq:reversible_psi_phi}\\
        \rap(\phi\to\psi)&=\sup\left\{r\geq 0\mid \mathcal{Q}^{\phi}\geq r  \mathcal{Q}^{\psi}\right\}\label{eq:reversible_phi_psi}.
    \end{align}
    These equations imply that $\rap(\psi\to\phi)\rap(\phi\to\psi)=1$ holds only if the condition (B) holds. Conversely, if conditions (A) and (B) are satisfied, then Eqs.~\eqref{eq:reversible_psi_phi} and \eqref{eq:reversible_phi_psi} imply $\rap(\psi\to\phi)\rap(\phi\to\psi)=1$ holds. 
\end{proof}

Let us now review the statement of the Marvian-Spekkens conjecture and prove that it is equivalent to Corollary~\ref{cor:reversible_conversion}. Let $\{X_\mu\}_\mu$ be the derived representation of a basis of the Lie algebra $\mathfrak{g}$. We define the \textit{symmetirized} covariance matrix as
\begin{align}
    \left(C_{\mathfrak{g}}(\psi)\right)_{\mu\nu}&\coloneqq \frac{1}{2}\braket{\psi|(X_\mu X_\nu +X_\nu X_\mu)|\psi}\nonumber\\
    &\quad-\braket{\psi|X_\mu|\psi}\braket{\psi|X_\nu|\psi}.
\end{align}
The commutator subalgebra $\ii[\mathfrak{g},\mathfrak{g}]$ is defined as the subalgebra spanned by $\ii[L_1,L_2]$ for all $L_1,L_2\in\mathfrak{g}$. In general, the input and output Hilbert spaces $\mathcal{H}_{\mathrm{in}}$ and $\mathcal{H}_{\mathrm{out}}$, to which $\psi$ and $\phi$ belong respectively, are different. However, as shown in Appendix~B in \cite{marvian_theory_2013}, by considering a larger Hilbert space $\mathcal{H}\coloneqq \mathcal{H}_{\mathrm{in}}\oplus \mathcal{H}_{\mathrm{out}}$ and a representation on it, it suffices to consider the case where the input and output Hilbert spaces are the same. The Marvian–Spekkens conjecture \cite{marvian_asymmetry_2014,marvian_mashhad_symmetry_2012} states that pure states $\psi$ and $\phi$ are asymptotically reversibly convertible in RTA for a connected compact Lie group $G$ if and only if all the following three conditions are satisfied:
\begin{enumerate}[(i)]
            \item $\mathrm{Sym}_G(\psi)=\mathrm{Sym}_G(\phi)$. 
            \item $C_{\mathfrak{g}}(\psi)=\rap(\psi\to \phi)C_{\mathfrak{g}}(\phi)$.
            \item $\braket{\psi|L|\psi}=\rap(\psi\to \phi)\braket{\phi|L|\phi}$ for any element $L$ in the representation of the commutator subalgebra $\ii[\mathfrak{g},\mathfrak{g}]$. 
\end{enumerate}
Condition (i) is the same as condition (A) in Corollary~\ref{cor:reversible_conversion}. Conditions (ii) and (iii) correspond respectively to the symmetric and anti-symmetric parts of condition (B) in Corollary~\ref{cor:reversible_conversion}.

We remark that other results on reversible conversion rates \cite{gour_resource_2008,yang_units_2017,marvian_operational_2022} also follow since the Marvian–Spekkens conjecture has now been proven. It is also worth emphasizing that, while the Marvian--Spekkens conjecture was originally formulated for connected compact Lie groups, Corollary~\ref{cor:reversible_conversion} holds for non-connected compact Lie groups as well.

\subsection{Irreversibility in asymptotic conversion}\label{sec:irreversiblity}
The result in the previous subsection shows that conversion is asymptotically irreversible if and only if at least one of the conditions (A) and (B) in Corollary~\ref{cor:reversible_conversion} is not satisfied. 

Condition (A) can be violated even when $G=U(1)$. Indeed, if the periods $\tau$ and $\tau'$ of $\psi$ and $\phi$ satisfy $\tau=k\tau'$ for some integer $k>1$, then $\mathrm{Sym}_G(\psi) \subsetneqq \mathrm{Sym}_G(\phi)$ holds, implying that $\rap(\psi\to \phi)=V(\psi,H)/V(\phi,H')$, while $\rap(\phi\to \psi)=0$. 

Condition (B) is easily violated when $\dim G>1$. As a simple illustration, let us analyze $G=SU(2)$ and its unitary representation $e^{\ii\sum_{i=x,y,z}\theta^i J_i}$ with spin operators $J_i$. We adopt the $z$-axis as a quantization axis and denote simultaneous eigenstates of $J^2\coloneqq J_x^2+J_y^2+J_z^2$ and $J_z$ by $\ket{j,m}$. To simplify the argument, we consider pure states $\ket{\psi} $ and $\ket{\phi}$ in a subspace spanned by the highest-weight states $\{\ket{j,j}\}_{j=0,\frac{1}{2},1,\frac{3}{2},\cdots,K}$, where $K$ is an integer or a half-integer introduced to make the Hilbert space finite-dimensional. Following the notation in \cite{gour_resource_2008}, we define
\begin{align}
    \mathcal{J}&\coloneqq \textstyle\sum_{j=0,\frac{1}{2},1,\frac{3}{2},\cdots, K}j\ket{j,j}\bra{j,j},\\
        \mathcal{M}(\psi)&\coloneqq 2\braket{\psi|\mathcal{J}|\psi},\\
        \mathcal{V}(\psi)&\coloneqq 4(\braket{\psi|\mathcal{J}^2|\psi}-\braket{\psi|\mathcal{J}|\psi}^2).
\end{align}
Since the QGT is calculated for $\psi$ as
\begin{align}
    \mathcal{Q}^\psi=\frac{1}{4}
    \begin{pmatrix}
        \mathcal{M}(\psi)&\frac{1}{\ii}\mathcal{M}(\psi)&0\\
        -\frac{1}{\ii}\mathcal{M}(\psi)&\mathcal{M}(\psi)&0\\
         0&0&\mathcal{V}(\psi)
    \end{pmatrix},
\end{align}
and similarly for $\phi$, condition (B) holds only if $\mathcal{M}(\psi)/\mathcal{M}(\phi)= \mathcal{V}(\psi)/\mathcal{V}(\phi)$. This can also be explicitly confirmed from the conversion rate $\rap(\psi\to\phi)=\min\left\{\frac{\mathcal{M}(\psi)}{\mathcal{M}(\phi)},\frac{\mathcal{V}(\psi)}{\mathcal{V}(\phi)}\right\}$ calculated from Eq.~\eqref{eq:rate_formula_compact_connected} for pure states satisfying $\mathrm{Sym}_{G}(\psi)\subset \mathrm{Sym}_G(\phi)$. We remark that this conversion rate is consistent with the prior result on $SU(2)$ in Theorem~24 in \cite{gour_resource_2008}, where the conversion rate is studied for a restricted set of pure states. 

\subsection{Finite groups}\label{sec:finite_groups}
Here, we analyze the asymptotic conversion theory in RTA for a finite group. As proven in \cite{shitara_iid_2024}, the conversion rate diverges for a finite group. We show that this fact can also be derived from our formula in Eq.~\eqref{eq:rate_formula_compact_connected} by lifting a finite group to a compact Lie group by appending a trivial phase, though a finite group itself is not a Lie group \footnote{Formally, a finite group can be regarded as a Lie group since it can be viewed as a zero-dimensional smooth manifold. However, in our terminology, we consider Lie groups to have a dimension greater than zero.}. 

Let $G$ be a finite group given by $G=\{g_i\mid i=0,\cdots, k\}$. Let $U$ and $U'$ be projective unitary representations of $G$ on the input and output Hilbert spaces. Let us introduce unitary representations $\tilde{U}$ and $\tilde{U}'$ of $\tilde{G}\coloneqq G\times U(1)$ such that $\tilde{U}(g,e^{\ii\theta})\coloneqq U(g)e^{\ii \theta}$ and $\tilde{U}'(g,e^{\ii\theta})\coloneqq U'(g)e^{\ii \theta}$ for $g\in G$ and $\theta\in [0,2\pi)$. Since the phase $\theta$ does not affect the state, a channel is $G$-covariant if and only if $\tilde{G}$-covariant. Therefore, appending the trivial phase leaves the conversion rate invariant. Since $\tilde{G}$ is a compact Lie group, we can apply Eq.~\eqref{eq:rate_formula_compact_connected}. The QGT vanishes for any pure state for $\tilde{U}$ and $\tilde{U}'$, implying that $\mathcal{Q}^{\psi}\geq r\mathcal{Q}^{\phi}$ holds for any $r$. In addition,  $\mathrm{Sym}_G(\psi)\subset\mathrm{Sym}_G(\phi)$ holds iff $ \mathrm{Sym}_{\tilde{G}}(\psi)\subset \mathrm{Sym}_{\tilde{G}}(\phi)$. Therefore, we get
\begin{align}
    \rap(\psi\to\phi)=
    \begin{cases}
        \infty &(\text{if }\mathrm{Sym}_{G}(\psi)\subset \mathrm{Sym}_G(\phi))\\
        0  &(\text{otherwise})
    \end{cases},
\end{align}
reproducing the result on the asymptotic conversion rate in \cite{shitara_iid_2024}. We remark that our construction of the conversion channels provides an intuitive explanation of the reason why the conversion rate for finite groups diverges, as will be explained at the end in Sec.~\ref{sec:achievability_proof}.

\section{Proof sketch of Theorem~\ref{thm:conversion_rate_projective_finite_number}}\label{sec:proof_sketch}
In this section, we provide a proof sketch of Theorem~\ref{thm:conversion_rate_projective_finite_number}, and explain the key ideas. A fully detailed, rigorous proof can be found in Appendices~\ref{app:section_for_converse_part} and \ref{app:section_for_direct_part}. 

\subsection{QGT as asymmetry monotone}\label{sec:monotonicity_QGT_sketch}
A key property of the QGT is its monotonicity as an asymmetry measure, in the form of a matrix inequality.
We remark that resource monotones in quantum resource theories are usually given by real-valued functions on states that do not increase under free operations, thereby providing necessary conditions for convertibility. 
Recent studies~\cite{kudo_fisher_2023,gao_sufficient_2024} have extended this viewpoint to matrix-valued monotones, which yield necessary conditions for convertibility formulated via partial orders induced by matrix inequalities. Our study aligns with these works in that the QGT is an asymmetry monotone expressed through a matrix inequality.

To prove the monotonicity of the QGT, we relate it to a family of Petz's monotone metrics~\cite{morozova_markov_1989,petz_monotone_1996}, which are Riemannian metrics on the state space that contract under information processing. To introduce Petz's monotone metrics, we begin with the notion of operator monotone functions. A function $f:[0,\infty)\to[0,\infty)$ is called operator monotone if and only if $0\leq A\leq B$ implies $f(A)\leq f(B)$. Following~\cite{petz_introduction_2011}, Petz's monotone metric associated with $f$ is defined by
\begin{align}
    \braket{A,B}_{f,\rho}\coloneqq
    \sum_{k,l:\, m_f(p_k,p_l)>0}
    \frac{\braket{k|A^\dagger|l}\braket{l|B|k}}{m_f(p_k,p_l)},\label{eq:monotone_metric_definition}
\end{align}
where $m_f(x,y)\coloneqq y f(x/y)$ and $\rho=\sum_k p_k\ket{k}\bra{k}$ is the eigenvalue decomposition. For a linear operator $A$, we define the corresponding norm by
$\|A\|_{f,\rho}\coloneqq \sqrt{\braket{A,A}_{f,\rho}}$. 

We remark that an additional symmetry condition $f(t)=t f(t^{-1})$ is often imposed on operator monotone functions~\cite{petz_monotone_1996, hansen_metric_2008}. However, as discussed in Appendix~\ref{app:symmetric_monotone_function}, the monotone metric associated with a symmetric operator monotone function does not capture the anti-symmetric part of the QGT. We therefore do not impose this symmetry condition.

A central object in our study is a one-parameter family of operator monotone functions $f_q$, defined by
\begin{align}
    f_q(x)\coloneqq (1-q)+ qx,\quad q\in (0,1).\label{eq:definition_of_f_q}
\end{align}
Since $m_{f_q}(p_k,p_l)=(1-q)p_l+qp_k$, for any pure state $\psi$ and any linear operator $O$, we obtain
\begin{align}
    &\|\ii[\psi,O]\|_{f_q,\psi}^2=\frac{1}{1-q}V(\psi,O)+\frac{1}{q}V(\psi,O^\dag)\label{eq:norm_formula_general_r},
\end{align}
where we have defined the generalized variance $V$ by $V(\psi,O)\coloneqq \braket{\psi|O (I-\psi)O^\dag|\psi}$.
Note that when $O$ is Hermitian, $V$ coincides with the ordinary variance. 

To relate the QGT to the above norm, let $O=\gamma^\dag X\coloneqq\gamma^{\mu*} X_\mu$, where $\{X_\mu\}_{\mu=1}^{\dim G}$ are generators defined in Eq.~\eqref{eq:hermitian_operators_Lie_alg} and $\gamma\in \mathbb{C}^{\dim G}$. From Eq.~\eqref{eq:norm_formula_general_r}, we obtain
\begin{align}
    &\|\ii[\psi,O]\|_{f_q,\psi}^2=\gamma^\dag \left(\frac{1}{1-q}\mathcal{Q}^{\psi}+\frac{1}{q}\left(\mathcal{Q}^{\psi}\right)^*\right)\gamma,\label{eq:norm_formula_general_r_QGT}
\end{align}
which implies
\begin{align}
    \lim_{q\to 1^-} f_q(0)\|\ii[\psi,O]\|_{f_q,\psi}^2
    =\gamma^\dag\mathcal{Q}^{\psi}\gamma\label{eq:norm_limit_QGT}.
\end{align}

As detailed in Appendix~\ref{app:monotonicity_fr_proof}, the norm $\|\cdot\|_{f,\rho}$ is monotonic under a quantum channel in the following sense:
\begin{lem}\label{lem:monotonicity_fr}
    For $f_q(x)\coloneqq (1-q)+ qx$ with $q\in (0,1)$, any state $\rho$, any linear operator $O$, and any quantum channel $\mathcal{E}$, it holds
    \begin{align}
        \|\ii[\rho,O]\|_{f_q,\rho}^2\geq \|\mathcal{E}(\ii[\rho,O])\|_{f_q,\mathcal{E}(\rho)}^2
    \end{align}
    if there is an operator $O'$ satisfying $\mathcal{E}(\ii[\rho,O])=\ii [\mathcal{E}(\rho),O']$.  
\end{lem}

The monotonicity of the QGT in RTA is proven by using Lemma~\ref{lem:monotonicity_fr}. Suppose that a state $\rho$ is convertible without error to another state $\sigma$ via a $G$-covariant channel $\mathcal{E}$. Since the $G$-covariance of the channel implies $\mathcal{E}(\mathcal{U}_g(\rho))=\mathcal{U}_g'(\sigma)$, its derivative with respect to the parameter $\lambda\in\mathbb{R}^{\dim G}$ for $g\in G$ yields $\mathcal{E}\left(\ii [\rho,X_\mu]\right)=\ii [\sigma,X_\mu']$, where $X_\mu$ is given in Eq.~\eqref{eq:hermitian_operators_Lie_alg}, and $X_\mu'$ is defined similarly by differentiating the projective unitary representation $U'$ of $G$ on the output system instead of $U$. 

By using the linearity of $\mathcal{E}$, we get $\mathcal{E}(\ii[\rho,O])=\ii[\mathcal{E}(\rho),O']$ for $O\coloneqq \gamma^{\dag} X$ and $O'\coloneqq \gamma^\dag X'$ for any $\gamma\in\mathbb{C}^{\dim G}$. The monotonicity in Lemma~\ref{lem:monotonicity_fr} implies
\begin{align}
    \|\ii[\rho,O]\|_{f_q,\rho}^2\geq\|\ii[\sigma,O']\|_{f_q,\sigma}^2.\label{eq:monotonicity_fr_RTA}
\end{align}

Let us now apply this inequality to the case where $\rho$ and $\sigma$ are pure states, which we denote by $\psi$ and $\phi$, respectively. Multiplying Eq.~\eqref{eq:monotonicity_fr_RTA} by $f_q(0)=1-q>0$ and taking the limit of $q\to 1^-$, 
Eq.~\eqref{eq:norm_limit_QGT} implies, $\gamma^\dag \mathcal{Q}^{\psi}\gamma\geq \gamma^\dag \mathcal{Q}^{\phi}\gamma$ for any $\gamma\in\mathbb{C}^{\dim G}$, or equivalently, 
\begin{align}
    \mathcal{Q}^{\psi}\geq \mathcal{Q}^{\phi},\label{eq:monotonicity_QGT_exact}
\end{align}
completing the proof of the monotonicity of QGT under exact conversion among pure states in RTA. 

We now turn to the i.i.d. case. For $N$ i.i.d. copies of a system, the QGT is additive: $\mathcal Q^{\psi^{\otimes N}} = N \mathcal Q^\psi$. When $\{\psi^{\otimes N}\}_N \gconv \{\phi^{\otimes \lfloor rN\rfloor}\}_N$, monotonicity suggests that the asymptotic QGT rates decrease, yielding
\begin{align}
    \mathcal Q^\psi \geq r\,\mathcal Q^\phi .\label{eq:monotonicity_QGT_asymptotic}
\end{align}
However, proving this requires care because the QGT has asymptotic discontinuity~\cite {donald_uniqueness_2002,plenio_introduction_2007,gour_measuring_2009,marvian_coherence_2020,marvian_operational_2022,yamaguchi_smooth_2023}: an asymptotically vanishing conversion error can change the QGT by $O(1)$ per copy.

This is already seen for $G=U(1)$ with representation $e^{iHt}$, where the QGT equals the energy variance~\cite{gour_measuring_2009,marvian_coherence_2020,marvian_operational_2022,yamaguchi_smooth_2023}. For $N$ copies, the operator norm of the Hamiltonian $H_N$ grows linearly with $N$. Since the variance is quadratic in $H_N$, a conversion error $\epsilon_N$ can change the variance by $O(\epsilon_N N^2)$, i.e., $O(\epsilon_N N)$ per copy, which need not vanish even if $\epsilon_N\to 0$. The same issue arises for general Lie groups, hence a more detailed analysis is needed.

We now prove the monotonicity of the QGT under asymptotic conversion. Suppose that $\{\psi^{\otimes N}\}_N\gconv \{\phi^{\otimes \floor{rN}}\}_N$, that is, there exists a sequence of quantum channels $\{\mathcal{E}_N\}_N$ such that $\lim_{N\to\infty}T\left(\mathcal{E}_N(\psi^{\otimes N}),\phi^{\otimes \floor{rN}}\right)=0$. For notational convenience, we define $M(N)\coloneqq \floor{rN}$ and $\sigma_{M(N)}\coloneqq \mathcal{E}_N(\psi^{\otimes N})$.
In the following, we omit the dependence of $M(N)$ on $N$ and simply write $M$. We define $O_N$ and $O_M'$ by $O_N\coloneqq \sum_{n=1}^NI^{\otimes n-1}\otimes O\otimes I^{\otimes N-n}$ and $O'_M\coloneqq \sum_{n=1}^MI^{\otimes n-1}\otimes O'\otimes I^{\otimes M-n}$ with $O =\gamma^\dag X$ and $O'=\gamma^\dag X'$. Applying Eq.~\eqref{eq:monotonicity_fr_RTA} to $ \psi^{\otimes N}$ and $\sigma_M$ instead of $\rho$ and $\sigma$, we find
\begin{align}
    &\|\ii[ \psi^{\otimes N},O_N]\|_{f_q,\psi^{\otimes N}}^2\geq \|\ii[\sigma_M,O_M']\|_{f_q,\sigma_M}^2.\label{eq:monotonicity_fr_RTA_iid}
\end{align}
By additivity of the monotone metric, the left-hand side equals $N\|\ii[ \psi,O]\|_{f_q,\psi}^2$. On the other hand, the right-hand side requires careful treatment due to the asymptotic discontinuity. Nevertheless, assuming that $\lim_{M\to\infty}T\left(\sigma_M,\phi^{\otimes M}\right)=0$, we show that there exists a real-valued function $h$, independent of $M$, with $\lim_{\epsilon\to 0}h(\epsilon)=0$, such that for any sufficiently small $\epsilon>0$,
\begin{align}
   f_q(0) \|\ii[\sigma_M,O_M']\|_{f_q,\sigma_M}^2\geq MV(\phi,O)-Mh(\epsilon)+o(M)\label{eq:asymptotics_norm_asymmetric}
\end{align}
for all sufficiently large $M$. The proof is provided in Appendix~\ref{sec:asymptotics_norm_asymmetric}.

From Eqs.~\eqref{eq:monotonicity_fr_RTA_iid} and \eqref{eq:asymptotics_norm_asymmetric} and $V(\phi, O)=\gamma^\dagger Q^\phi\gamma$, we get
\begin{align}
    &f_q(0)\|\ii[\psi,O]\|_{f_q,\psi}^2\geq  \frac{M}{N}\gamma^\dag \mathcal{Q}^\phi\gamma-\frac{M}{N}h(\epsilon)+\frac{o\left(M\right)}{N},\label{eq:fq_norm_monotonicity}
\end{align}
where we used $f_q(0)>0$. Taking the limit $N\to\infty$ yields
\begin{align}
    f_q(0)\|\ii[\psi,O]\|_{f_q,\psi}^2\geq r\gamma^\dag \mathcal{Q}^\phi\gamma-rh(\epsilon).
\end{align}
Since this inequality holds for any sufficiently small $\epsilon>0$, we get 
$f_q(0)\|\ii[\psi,O]\|_{f_q,\psi}^2\geq r\, \gamma^\dag \mathcal{Q}^\phi\gamma$. Taking the limit of $q\to 1^-$ and using Eq.~\eqref{eq:norm_limit_QGT}, we finally obtain Eq.~\eqref{eq:monotonicity_QGT_asymptotic}. Since $\{\psi^{\otimes N}\}_N\gconv \{\phi^{\otimes \floor{rN}}\}_N$ is equivalent to $\{\mathcal{U}_g(\psi)^{\otimes N}\}_N\gconv \{\mathcal{U}_g'(\phi)^{\otimes \floor{rN}}\}_N$, by repeating the above argument for $\mathcal{U}_g(\psi)$ and $\mathcal{U}_g'(\phi)$ instead of $\psi$ and $\phi$, we get
\begin{align}
    \forall g\in G,\quad \mathcal{Q}^{\mathcal{U}_g(\psi)}\geq r\mathcal{Q}^{\mathcal{U}_g'(\phi)},\label{eq:monotonicity_QGT_converse_part}
\end{align}
which completes the proof of the monotonicity of QGT in asymptotic conversion. Consequently, we obtain 
\begin{align}
        &\rap(\psi \to \phi)\leq \sup\{r\geq0 \mid\forall g\in G,\,\mathcal{Q}^{\mathcal{U}_g(\psi)}\geq r\mathcal{Q}^{\mathcal{U}'_g(\phi)}\}.\label{eq:converse_part_statement}
\end{align}
Equation~\eqref{eq:converse_part_statement} gives the converse part of Eq.~\eqref{eq:conversion_rate_formula_forall_g}, i.e., it rules out any conversion rate exceeding the right-hand side of Eq.~\eqref{eq:conversion_rate_formula_forall_g}. In the next subsection, we show that this bound is achievable and hence optimal. 

Equation~\eqref{eq:converse_part_statement} remains valid for non-compact Lie groups as long as the representation space is finite-dimensional, since the proof relies only on local properties of states. It also applies regardless of whether $\mathrm{Sym}_{G}(\psi)\subset \mathrm{Sym}_G(\phi)$ holds.

\subsection{Optimality}\label{sec:optimality_sketch}
Here we explain key ideas to prove the optimality of the bound in Eq.~\eqref{eq:converse_part_statement}. Essentially, we constructively show that if $r>0$ satisfies $\mathcal{Q}^{\mathcal{U}_g(\psi)}\geq r\,\mathcal{Q}^{\mathcal{U}'_g(\phi)}$ for all $g\in G$, then any conversion rate smaller than $r$ is achievable.

The key step in proving optimality is the following lemma, which provides an alternative yet equivalent characterization of convertibility in RTA in terms of quantum channels that are not assumed to be $G$-covariant:
\begin{lem}\label{lem:gcov_and_cptp}
    Let $U$ and $U'$ be projective unitary representations of a compact Lie group $G$ on finite-dimensional Hilbert spaces $\mathcal{H}$ and $\mathcal{H}'$. Let $\rho$ and $\sigma$ be arbitrary states on $\mathcal{H}$ and $\mathcal{H}'$. For any $\epsilon\geq 0$, the following are equivalent:
    \begin{enumerate}[(i)]
        \item There exists a quantum channel $\mathcal{E}$ such that $T(\mathcal{E}(\mathcal{U}_g(\rho)),\mathcal{U}'_g(\sigma))\leq\epsilon$ for any $g\in G$.
        \item There exists a $G$-covariant channel $\mathcal{E}'$ such that 
        $T(\mathcal{E}'(\rho),\sigma)\leq \epsilon$.
    \end{enumerate}
\end{lem}
The proof is given in Appendix~\ref{app:gcov_and_cptp}. In particular, we show that if $\mathcal{E}$ satisfies condition~(i), then $\mathcal{E}'\coloneqq \int \dd \mu_G(g)\,\mathcal{U}_{g^{-1}}'\circ \mathcal{E}\circ \mathcal{U}_g$ is a $G$-covariant channel and satisfies condition~(ii), where $\mu_G$ denotes the normalized Haar measure on $G$. 

Lemma~\ref{lem:gcov_and_cptp} implies that $\{\psi^{\otimes N}\}_N\gconv \{\phi^{\otimes \floor{rN}}\}_N$ if and only if there exists a sequence of quantum channels $\{\mathcal{E}_N\}_N$ such that
\begin{align}
    \lim_{N\to\infty}\sup_{g\in G}T\left(\mathcal{E}_N\left(\mathcal{U}_g(\psi)^{\otimes N}\right),\mathcal{U}'_g(\phi)^{\otimes \floor{rN}}\right)=0.\label{eq:trace_distance_sup}
\end{align}
Importantly, if we regard $g\in G$ as an unknown parameter, this sequence of quantum channels $\{\mathcal{E}_N\}_N$ achieves a conversion of a statistical model $\{\mathcal{U}_g(\psi)^{\otimes N}\}_{g\in G}$ to another statistical model $\{\mathcal{U}'_g(\phi)^{\otimes \floor{rN}}\}_{g\in G}$ in the asymptotic limit. 

We show that such an asymptotic transformation between statistical models is achieved by a two-step process: estimation and subsequent conversion. In the rest of this subsection, we assume that the unitary representations of $G$ are non-projective for simplicity. 

\noindent\textbf{Step 1 (Estimation):}
In this step, we obtain a rough estimation $\hat{g}$ of $g\in G$ using only a sublinear number of copies. More precisely, we apply an estimation procedure to $n=N^{1-\epsilon}$ copies of $\mathcal{U}_g(\psi)$ (for a fixed small $\epsilon$), thereby producing an estimate $\hat{g}$. The purpose of this step is not to learn $g$ precisely, but to \emph{localize} it: with success probability tending to one as $N\to\infty$, the true state $\mathcal{U}_g(\psi)$ lies in a local neighborhood of $\mathcal{U}_{\hat g}(\psi)$, in the sense that, on the success event, there exists a small parameter
$\theta\in\mathbb{R}^{\dim G}$ with $\|\theta\|=O(N^{-1/2+\epsilon})$ such that $\mathcal{U}_g(\psi)=e^{\ii \theta^\mu X_\mu }\mathcal{U}_{\hat g}(\psi) e^{-\ii \theta^\nu X_\nu }$. 

After localization, we may treat the remaining copies (of order $N$) as a \emph{local i.i.d.\ model} parameterized at the $N^{-1/2}$ scale. This is precisely the regime where quantum local asymptotic normality (QLAN) applies \cite{guta_local_2006,kahn_local_2009,girotti_optimal_2024,lahiry_minimax_2024}, which approximates the properties of local i.i.d.\ models by Gaussian shift models. 
This will be exploited in the conversion step.

Notably, since this estimation step consumes only $n=N^{1-\epsilon}=o(N)$ copies, it incurs a sublinear overhead and hence does not change the achievable conversion rate in the $N\to\infty$ limit.

\noindent\textbf{Step 2 (Conversion):}
In this step, conditioned on the estimate $\hat{g}$ obtained in Step~1, we convert the remaining $N'=N-n$ copies of $\mathcal{U}_g(\psi)$ into copies of $\mathcal{U}'_g(\phi)$. Since the success probability of Step~1 tends to one as $N\to\infty$, it suffices to describe the conversion on the success event, where the problem is reduced to a local i.i.d.\ model.

The key observation is that, in the local regime, QLAN translates pure-state unitary i.i.d.\ models by Gaussian shift models that are essentially characterized by the QGT. Importantly, from the additivity of the QGT, the Gaussian shift model associated with $\{\mathcal{U}_g(\phi)^{\otimes \floor{rN}}\}_{g\in G}$ is governed by the scaled QGT, $r\,\mathcal{Q}^{\mathcal{U}_g(\phi)}$. Since the QGT quantifies the sensitivity of the state under infinitesimal transformations, the inequality
\begin{align}
    \mathcal{Q}^{\mathcal{U}_g(\psi)} \geq r\,\mathcal{Q}^{\mathcal{U}'_g(\phi)} \qquad (\forall g\in G)\label{eq:QGT_inequality_for_all_g}
\end{align}
can be interpreted as the condition that the Gaussian shift model associated with the input state is sufficiently ``sensitive'' to reproduce that for the output state at rate $r$ (uniformly in $g\in G$). This suggests that the conversion at a rate $r$ is possible if Eq.~\eqref{eq:QGT_inequality_for_all_g} is satisfied. This intuition is made rigorous in Appendix~\ref{app:section_for_direct_part} by explicitly constructing conversion channels.

Here, the condition $\mathrm{Sym}_G(\psi)\subset \mathrm{Sym}_G(\phi)$ is essential: estimation can only identify $g$ up to symmetries that leave $\psi$ invariant, and the inclusion guarantees that such residual ambiguity is also irrelevant for $\phi$. Hence, the conversion conditioned on $\hat{g}$ produces the correct target state $\mathcal{U}'_g(\phi)$ whenever the localization in Step~1 succeeds.

Through the above two-step procedure, we construct a sequence of quantum channels $\{\mathcal{E}_N\}_N$ that achieves Eq.~\eqref{eq:trace_distance_sup}. Using Lemma~\ref{lem:gcov_and_cptp}, this establishes the opposite inequality of Eq.~\eqref{eq:converse_part_statement}, i.e., the direct part of Eq.~\eqref{eq:conversion_rate_formula_forall_g}: 
\begin{align}
    \rap(\psi \to \phi)\geq \sup\{r\geq0 \mid\forall g\in G,\,\mathcal{Q}^{\mathcal{U}_g(\psi)}\geq r\mathcal{Q}^{\mathcal{U}'_g(\phi)}\}.\label{eq:direct_part_rates}
\end{align}

As detailed in Appendix~\ref{app:section_for_direct_part}, our proof is constructive: we explicitly construct a sequence of $G$-covariant channels that achieves the asymptotic conversion. Also, we show that the conversion error decays polynomially in $N$.

\section{Mixed-state conversion}\label{sec:mixed_state_asymmetry}

In the i.i.d. setting, the pure-state asymmetry conversion theory has been thoroughly developed for a compact Lie group in the preceding sections. This section partially extends the theory to mixed states. Although the theoretical framework remains incomplete, it is expected to serve as a foundation for further investigations.

Concretely, we investigate the distillable asymmetry and asymmetry cost, defined as the optimal rates for converting a mixed state $\rho$ to and from a pure reference state $\phi$. A similar scenario has already been explored for the time-translation symmetry in \cite{marvian_coherence_2020,marvian_operational_2022}, but not for other symmetries. See Fig.~\ref{fig:mixed} for a schematic picture of the setups.
\begin{figure}[htb]
    \centering
    \includegraphics[width=8.3cm]{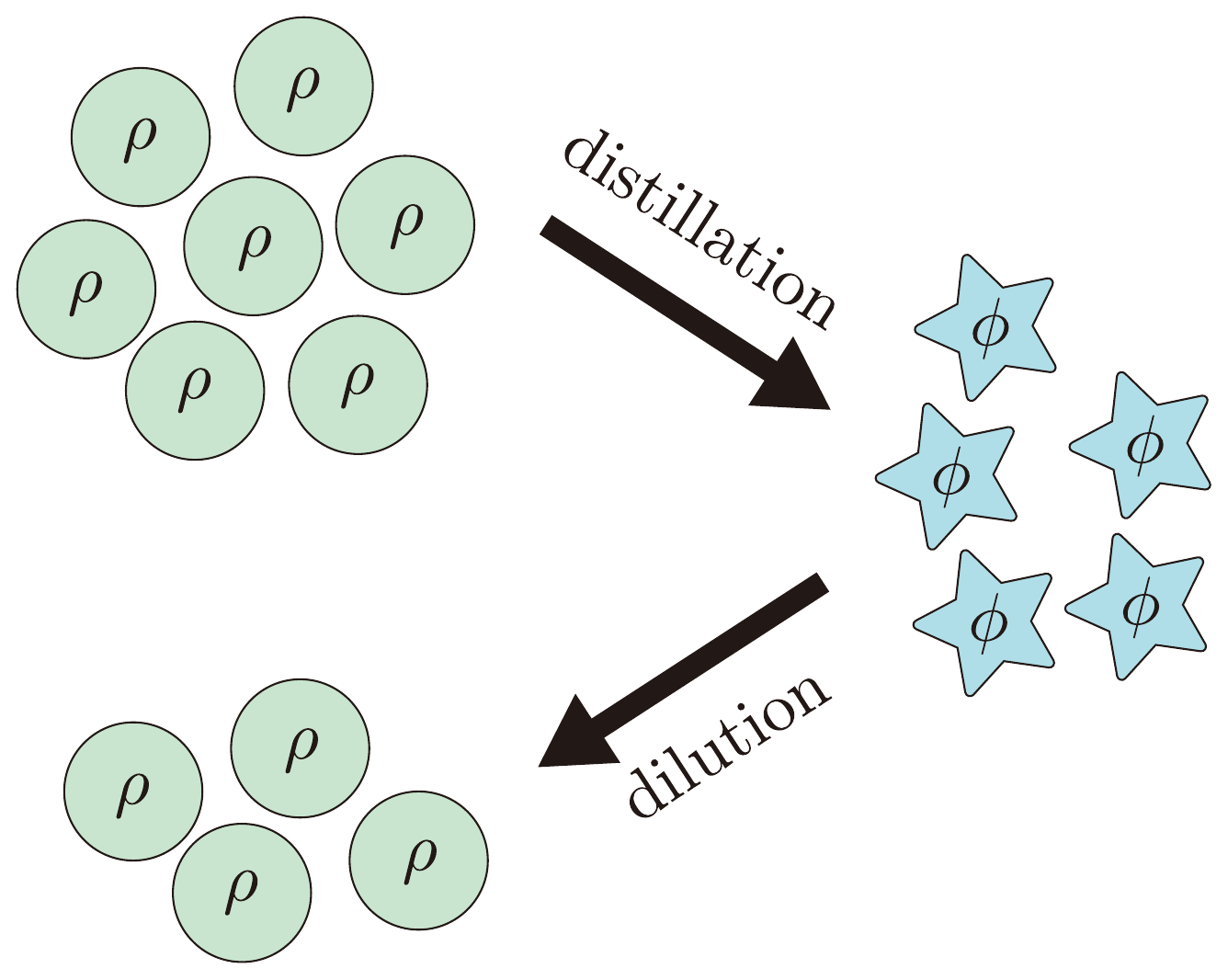}
    \caption{Schematic picture of distillation and dilution setups, in which i.i.d. copies of a mixed state $\rho$ are converted to and from i.i.d. copies of a pure reference state $\phi$.}
    \label{fig:mixed}
\end{figure}

\subsection{Distillation of asymmetry}\label{sec:distillation_of_asymmetry}
Here, we investigate the scenario where we convert a general state $\rho$ to a fixed pure state $\phi$ that serves as a reference. We define the distillable asymmetry of a general state $\rho$ as $\mathcal{A}_{\dist}(\rho)\coloneqq R(\rho\to \phi)$, i.e., 
\begin{align}
    \mathcal{A}_{\dist}(\rho)=\sup\{r\geq 0\mid\{\rho^{\otimes N}\}_N\gconv \{\phi^{\otimes \floor{rN}}\}_N \}\label{eq:definition_distillation_of_asymmetry}.
\end{align}
As a central quantity in our analysis, we introduce a $\dim G\times \dim G$ Hermitian matrix $\mathcal{Q}^\rho$ whose elements are given by
\begin{align}
    \left(\mathcal{Q}^\rho\right)_{\mu\nu}\coloneqq \mathrm{Tr}\left(\rho X_\mu (I-\Pi_\rho)X_\nu\right),\label{eq:definition_S_matrix}
\end{align}
where $\Pi_\rho$ denotes the projector to the support of $\rho$. 
This quantity is an extension of QGT, originally defined for pure states in Eq.~\eqref{eq:qgt_covariance_matrix}, to mixed states. As we shall show below, the monotonicity of $\mathcal{Q}^\rho$ under $G$-covariant channels yields an upper bound of $\mathcal{A}_{\dist}(\rho)$, from which we can derive a sufficient condition for the rate to vanish. 

We start the analysis by relating $\mathcal{Q}^\rho$ with the monotone metric with an operator monotone function $f_q(x)$ given in Eq.~\eqref{eq:definition_of_f_q}. By using the eigenvalue decomposition $\rho=\sum_{k}p_k\ket{k}\bra{k}$ of the state $\rho$, we have
\begin{align}
    &\|\ii [O,\rho]\|_{f_q,\rho}^2
    =\sum_{\underset{(1-q)p_l+qp_k>0}{k,l;}}\frac{(p_l-p_k)^2}{(1-q)p_l+qp_k}|\braket{l|O|k}|^2.
\end{align}
When $q\to 1^-$, the above sum contains divergent terms proportional to $1/(1-q)$ for $(k,l)$ such that $p_l>0$ and $p_k=0$. Therefore, we find $\lim_{q\to 1^-}f_q(0)\|\ii [O,\rho]\|_{f_q,\rho}^2=\sum_{l;p_l>0}\sum_{k;p_k=0}p_l\braket{l|O|k}\braket{k|O^\dag |l}$, i.e., 
\begin{align}
    \lim_{q\to 1^-}f_q(0)\|\ii [O,\rho]\|_{f_q,\rho}^2=\mathrm{Tr}\left(\rho O (I-\Pi_\rho)O^\dag\right). 
\end{align}
In particular, when $O=\gamma^\dag X$ with $\gamma\in\mathbb{C}^{\dim G}$, we get
\begin{align}
    \lim_{q\to 1^-}f_q(0)\|\ii [\gamma^\dag X,\rho]\|_{f_q,\rho}^2=\gamma^\dag \mathcal{Q}^\rho \gamma.\label{eq:relation_between_S_and_f_q}
\end{align}

We now prove the monotonicity of $\mathcal{Q}^{\rho}$ under $G$-covariant operation. Suppose a state $\rho$ is convertible to $\sigma$ via a $G$-covariant channel without error. 
Multiplying Eq.~\eqref{eq:monotonicity_fr_RTA} by $f_q(0)=1-q>0$ yields
\begin{align}
    f_q(0)\|\ii [\gamma^\dag X,\rho]\|_{f_q,\rho}^2\geq f_q(0) \|\ii [\gamma^\dag X',\sigma]\|_{f_q,\sigma}^2.\label{eq:monotonicity_norm_general_fq}
\end{align}
Taking the limit of $q\to 1^-$ in Eq.~\eqref{eq:monotonicity_norm_general_fq}, Eq.~\eqref{eq:relation_between_S_and_f_q} implies $\gamma^\dag \mathcal{Q}^\rho \gamma\geq \gamma^\dag \mathcal{Q}^\sigma \gamma$. Since this inequality holds for any $\gamma\in \mathbb{C}^{\dim G}$, we get
\begin{align}
     \mathcal{Q}^\rho \geq \mathcal{Q}^\sigma, \label{eq:monotonicity_S}
\end{align}
meaning that $\mathcal{Q}^\rho$ is a valid asymmetry monotone for any state $\rho$. 

The asymptotic discontinuity of QGT is inherited by $\mathcal{Q}^{\rho}$, as they coincide for pure states, implying that extending its monotonicity to an asymptotic setting requires careful consideration. Nevertheless, as an extension of Eq.~\eqref{eq:converse_part_statement}, an upper bound of the distillable asymmetry can be proven by using Eq.~\eqref{eq:asymptotics_norm_asymmetric}. 
\begin{thm}[Upper bound on distillable asymmetry]
    Let $U,U'$ be projective unitary representations of a Lie group $G$ on finite-dimensional Hilbert spaces $\mathcal{H}$ and $\mathcal{H}'$. For any state $\rho$ on $\mathcal{H}$ and pure state $\phi$ on $\mathcal{H}'$, it holds
    \begin{align}
        \sup\{r\geq0 \mid\forall g\in G,\,\mathcal{Q}^{\mathcal{U}_g(\rho)}\geq r\mathcal{Q}^{\mathcal{U}'_g(\phi)}\}\geq \mathcal{A}_{\dist}(\rho).\label{extension_of_converse_part}
    \end{align}
\end{thm}
\begin{proof}
    Suppose $\{\rho^{\otimes N}\}_N\gconv \{\phi^{\otimes \floor{rN}}\}_N$ for $r\geq 0$. From Eqs.~\eqref{eq:monotonicity_fr_RTA_iid} and~\eqref{eq:asymptotics_norm_asymmetric}, as an extension of Eq.~\eqref{eq:fq_norm_monotonicity}, we get
\begin{align}
    &f_q(0)\|\ii [\gamma^\dag X,\rho]\|_{f_q,\rho}^2\geq \frac{M}{N}\gamma^\dag \mathcal{Q}^\phi\gamma-\frac{M}{N}h(\epsilon)+\frac{1}{N}o(M),\label{eq:monotonicity_iid_asymptotic}
\end{align}
where $M\coloneqq \floor{rN}$. 
In the limit of $N\to\infty$, we have
\begin{align}
    f_q(0)\|\ii [\gamma^\dag X,\rho]\|_{f_q,\rho}^2\geq r\gamma^\dag \mathcal{Q}^\phi\gamma-r h(\epsilon). 
\end{align}
Since this inequality holds for all sufficiently small $\epsilon>0$, we get $f_q(0)\|\ii [\gamma^\dag X,\rho]\|_{f_q,\rho}\geq r\gamma^\dag \mathcal{Q}^\phi\gamma$. In the limit of $q\to 1^-$, we obtain $\mathcal{Q}^\rho\geq r\mathcal{Q}^\phi$. Repeating this argument for $\mathcal{U}_g(\rho)$ and $\mathcal{U}'_g(\phi)$ instead of $\rho$ and $\phi$, we get
\begin{align}
    \forall g\in G,\quad \mathcal{Q}^{\mathcal{U}_g(\rho)}\geq r\mathcal{Q}^{\mathcal{U}'_g(\phi)},
\end{align}
which concludes the proof of Eq.~\eqref{extension_of_converse_part}. 
\end{proof}
This result first establishes an upper bound on the distillable asymmetry by using a geometric quantity, $\mathcal{Q}$, i.e., the QGT and its extension to mixed states, which applies to any state $\rho$ and any Lie group $G$. From Theorem~\ref{thm:conversion_rate_projective_finite_number}, the inequality~\eqref{extension_of_converse_part} becomes tight when $\rho$ is pure, and $G$ is compact.

We remark that in a special case where $[\Pi_\rho,\gamma^\dag X]=0$ holds for some $\gamma\in\mathbb{C}^{\dim G}$, we have $\gamma^\dag \mathcal{Q}^\rho\gamma=0$. As an immediate corollary of the above theorem, we therefore get the following:
\begin{cor}[Sufficient condition for vanishing distillable asymmetry]\label{cor:projector_distillablle_asymmetry}
    If there exists a vector $\gamma\in\mathbb{C}^{\dim G}$ satisfying $[\Pi_\rho,\gamma^\dag X]=0$ and $\gamma^\dag \mathcal{Q}^\phi\gamma\neq 0$, then $\mathcal{A}_{\dist}(\rho)=0$.
\end{cor}
This corollary implies that for a typical state with full rank, where $\Pi_\rho= I$, $\mathcal{A}_{\dist}(\rho)=0$ holds for any asymmetric pure state $\phi$. This result is consistent with a result in prior research \cite{marvian_coherence_2020} on time-translation asymmetry, where the distillable coherence, defined as the optimal conversion rate from a state $\rho$ to a pure state $\phi$ given by a superposition of the ground and excited state of a Hamiltonian $H$, is shown to be zero when $[\Pi_\rho, H]=0$. 
We remark that although a state $\rho$ satisfying the condition in Corollary~\ref{cor:projector_distillablle_asymmetry} has zero distillable asymmetry, it does not necessarily implies that $\rho$ is useless in asymptotic conversion. For example, while a full-rank state has zero distillable asymmetry, a sublinear number of such states may be used to enhance the conversion rate by removing the restriction imposed by symmetry subgroups in asymptotic conversion among pure states, as shown in Appendix~\ref{sec:sym_subgroup_conversion_rate}.

\subsection{Dilution of asymmetry}\label{sec:dilution_of_asymmetry}
Next, we investigate the opposite scenario, where we create a general state $\rho$ by consuming the reference pure state $\phi$. The asymmetry cost is defined as
\begin{align}
    \mathcal{A}_{\cost}(\rho)\coloneqq \frac{1}{R(\phi\to\rho)},
\end{align}
where $1/\infty$ and $1/0$ are formally regarded as $0$ and $\infty$, respectively. From the definition of conversion rate $R(\phi\to\rho)$, we also obtain an equivalent expression
\begin{align}
    \mathcal{A}_{\cost}(\rho)=\inf \{r\geq 0\mid \{\phi^{\otimes \ceil{rN}}\}_N \gconv  \{\rho^{\otimes N}\}_N\}.
\end{align}
We remark that Eq.~\eqref{eq:conversion_rate_formula_subgrop} implies $\mathcal{A}_{\cost}(\rho)<\infty$ only if $\mathrm{Sym}_G(\phi)\subset\mathrm{Sym}_G(\rho)$. Thus, we only investigate $\mathcal{A}_{\cost}(\rho)$ for $\rho$ satisfying $\mathrm{Sym}_G(\phi)\subset\mathrm{Sym}_G(\rho)$.

Following the standard argument using typical sequence \cite{hayden_asymptotic_2001,marvian_operational_2022}, it is proven that
\begin{align}
    \sum_{i=1}^kp_i\mathcal{A}_{\cost}(\rho_i)\geq \mathcal{A}_{\cost}(\rho),\label{eq:cost_ensemble}
\end{align}
where $\{\rho_i\}_{i=1}^k$ is a set of states and $\{p_i\}_{i=1}^k$ is a probability distribution such that $p_i>0$, $\sum_{i=1}^kp_i=1$, and $\sum_{i=1}^kp_i \rho_i=\rho$. An intuition behind this inequality is that, instead of preparing $\rho$ directly, we can adopt a strategy where $\{\rho_i\}_{i=1}^k$ is probabilistically generated according to the probability distribution $\{p_i\}_{i=1}^k$, although this strategy is not necessarily optimal. See Appendix~\ref{app:cost_ensemble} for the proof of Eq.~\eqref{eq:cost_ensemble}. We remark that the left-hand side of Eq.~\eqref{eq:cost_ensemble} is finite only if $\mathrm{Sym}_G(\phi)\subset\mathrm{Sym}_G(\rho_i)$ for all $i=1,\cdots,k$. 

In light of the pure-state conversion theory established in Theorem~\ref{thm:conversion_rate_projective_finite_number}, a promising approach to obtaining a computable upper bound of $\mathcal{A}_{\cost}(\rho)$ is to decompose a state $\rho$ into an ensemble of pure states. Suppose that a state $\rho$ can be decomposed as $\rho=\sum_{i=1}^{k-1}p_i\psi_i+p_{\mathrm{s}}\rho_{\mathrm{s}}$, where $\rho_{\mathrm{s}}$ is a symmetric (possibly mixed \footnote{We remark that not all symmetric mixed states can be decomposed as an ensemble of symmetric pure states. Indeed, when a unitary representation of a group is decomposed into irreducible components, Schur’s lemma implies that only a one-dimensional irreducible representation contains a symmetric pure state. In other words, no symmetric pure state exists in any higher-dimensional irreducible representation.}) state, and $\{\psi_i\}_{i=1}^{k-1}$ are pure states satisfying $\mathrm{Sym}_G(\phi)\subset \mathrm{Sym}_G(\psi_i)$. Then, Eq.~\eqref{eq:conversion_rate_formula_forall_g} implies
\begin{align}
    \sum_{i=1}^{k-1}p_ir_i\geq \mathcal{A}_{\cost}(\rho),\label{eq:cost_upper_bound}
\end{align}
where
\begin{align}
    r_i\coloneqq \inf\left\{r\geq 0\,\middle|\, \forall g\in G,\, r\mathcal{Q}^{\mathcal{U}_g(\phi)} \geq  \mathcal{Q}^{\mathcal{U}_g'(\psi_i)}\right\}.\label{eq:rate_each_pure_state}
\end{align} 
Similarly to Eq.~\eqref{eq:rate_formula_D_max_QGT}, this quantity can be expressed by using the quantum max-relative entropy as
\begin{align}
    r_i=2^{D_{\max}\left(\mathcal{Q}^{\psi_i}\middle\|\mathcal{Q}^{\phi}\right)}.
\end{align}

When a state $\rho$ can be obtained as an ensemble of pure states and symmetric state, minimizing the left-hand side of Eq.~\eqref{eq:cost_upper_bound} over such an ensemble leads to a tighter bound on the asymmetry cost $\mathcal{A}_{\cost}(\rho)$. In the case of $G=U(1)$, it is known that such a minimization yields a tight bound~\cite{marvian_operational_2022}. Let $U(\theta)=e^{\ii H\theta}$ and $U'(\theta)=e^{\ii H'\theta}$ be unitary representations of $G=U(1)$ on finite-dimensional Hilbert spaces with Hamiltonians $H=\sum_{n=0}^{d-1}n\ket{n}\bra{n}$ and $H'=\sum_{n=0}^{d'-1}n\ket{n}\bra{n}$. We fix the reference state $\phi$ to be the so-called coherence bit, given by $\ket{\phi}=(\ket{0}+\ket{1})/\sqrt{2}$ and assume that $\mathrm{Sym}_G(\rho)=\mathrm{Sym}_G(\phi)$. Since $r_i=V(\psi_i,H')/V(\phi,H)=4V(\psi_i,H')$, Eq.~\eqref{eq:cost_upper_bound} implies $\min_{\{p_i,\psi_i\}_i} 4\sum_{i}p_iV(\psi_i,H')\geq \mathcal{A}_{\cost}(\rho)$, where the minimization is taken over the set of all pure-state ensembles $\{p_i,\psi_i\}_i$ satisfying $\mathrm{Sym}_G(\rho)\subset \mathrm{Sym}_G(\psi_i)$ and $\rho=\sum_ip_i\psi_i$. Note that it suffices to consider pure-state ensembles for $G=U(1)$ since any symmetric state is given as a probabilistic mixture of symmetric pure states. It was conjectured in \cite{toth_extremal_2013_rev} and later proven in \cite{yu_quantum_2013} that $\min_{\{p_i,\psi_i\}_i} 4\sum_{i}p_iV(\psi_i,H')=\mathcal{F}_{H'}(\rho)$, where $\mathcal{F}_{H'}(\rho)$ denotes the SLD quantum Fisher information, given by $\mathcal{F}_{H'}(\rho)\coloneqq2\sum_{i,j}\frac{(\lambda_i-\lambda_j)^2}{\lambda_i+\lambda_j}|\braket{i|H|j}|^2$ using the eigenvalue decomposition $\rho=\sum_i\lambda_i\ket{i}\bra{i}$. Furthermore, in \cite{marvian_operational_2022}, it was shown that the minimum is attained by a set of pure states $\{\psi_i\}_i$ such that $\mathrm{Sym}_G(\rho)\subset \mathrm{Sym}_G(\psi_i)$, which implies $\mathcal{F}_{H'}(\rho)\geq \mathcal{A}_{\cost}(\rho)$. In \cite{marvian_operational_2022}, the opposite inequality was also shown by proving the monotonicity of the SLD quantum Fisher information rate under asymptotic conversion, which completed the proof of $\mathcal{F}_{H'}(\rho)=\mathcal{A}_{\cost}(\rho)$ for $G=U(1)$.

In contrast to the $U(1)$ case, a decomposition into $\rho=\sum_{i=1}^{k-1}p_i\psi_i+p_{\mathrm{s}}\rho_{\mathrm{s}}$ with symmetric $\rho_{\mathrm{s}}$ and pure states $\{\psi_i\}_i$ satisfying $\mathrm{Sym}_G(\phi)\subset \mathrm{Sym}_G(\psi_i)$ does not always exist for a general Lie group $G$. See Appendix~\ref{app:cost_ensemble} for a concrete example for $G=U(4)$. In such cases, Eq.~\eqref{eq:cost_upper_bound} does not provide a meaningful upper bound on the asymmetry cost. However, by carefully refining the estimate-and-convert approach introduced in Sec.~\ref{sec:optimality_sketch}, we find that this limitation due to the symmetry subgroups can be circumvented. Specifically, in Appendix~\ref{app:asymmetry_of_formation}, we show the following:
\begin{prop}[Upper bound on the asymmetry cost]\label{prop:asymmetry_of_formation}
    For a reference pure state $\phi$, consider a state $\rho$ satisfying $\mathrm{Sym}_G(\phi)\subset \mathrm{Sym}_G(\rho)$.
    The asymmetry cost of a state $\rho$ is upper bounded as follows:
    \begin{align}
        &\min \sum_i p_i r_i\geq \mathcal{A}_{\mathrm{c}}(\rho),
    \end{align}
    where $r_i$ is defined in Eq.~\eqref{eq:rate_each_pure_state}, and the minimization is taken over all the decompositions of $\rho$ into ensemble $\rho=\sum_i p_i \psi_i +p_{\mathrm{s}}\rho_{\mathrm{s}}$, such that $\{\psi_i\}_i$ are pure states, $\rho_{\mathrm{s}}$ is a symmetric state, and $\{p_i\}_i, p_{\mathrm{s}}$ denote a probability distribution satisfying $p_i,p_{\mathrm{s}}\geq 0$ and $\sum_ip_i +p_{\mathrm{s}}=1$. 
\end{prop}

Given the success in the case of $G=U(1)$, where the property of the SLD quantum Fisher information $\mathcal{F}_H(\rho)=\min_{\{p_i,\psi_i\}_i} 4\sum_{i}p_iV(\psi_i,H)$ is utilized, a natural extension is to consider the convex roof of QGT, namely $\min_{\{p_i,\psi_i\}_i} \sum_{i}p_i\mathcal{Q}^{\psi_i}$, as a possible asymmetry quantifier in RTA for a general compact Lie group. However, as we demonstrate in Appendix~\ref{app:convex_roof_QGT}, $\min_{\{p_i,\psi_i\}_i} \sum_{i}p_i\mathcal{Q}^{\psi_i}$ does not always exist. Analogous to another property of the SLD quantum Fisher information, $\mathcal{F}_H(\rho)=\min_{\Psi_\rho,H_A} 4V(\Psi_\rho,H+H_A)$, where the minimization is taken over any purification $\Psi_\rho$ of $\rho$ and Hermitian operators $H_A$ on the ancillary system \cite{marvian_operational_2022}, we also examine the minimization of the QGT of purifications of a state $\rho$. We show that such a minimum does not also exist in general in Appendix~\ref{app:QGT_of_purification}. The non-existence of these extensions of the SLD quantum Fisher information stems from the fact that a partial order on matrices is not necessarily a total order, highlighting the exceptional case of the $U(1)$ group.

\section{Application to quantum reference frames}\label{sec:quantum_reference_frame}
We now examine a broader perspective on the physical significance of asymmetry, focusing on its role as a quantum reference frame. In physics, operations are always defined relative to a reference structure, such as a clock, a phase origin, or a spatial direction. While such reference frames are often taken for granted in standard quantum protocols, their availability is neither universal nor guaranteed. 
From a foundational viewpoint, the reference structures should themselves be treated as quantum systems, i.e., quantum reference frames. This perspective has attracted attention across diverse contexts, including superselection rules and their relation to reference frames~\cite{aharonov_ChargeSuperselectionRule_1967, bartlett_reference_2007,dowling_ObservingCoherentSuperpositionatommolecule_2006}, quantum information processing~\cite{bartlett_reference_2007,gour_resource_2008,gour_measuring_2009,marvian_no-broadcasting_2019}, and relational and covariant formulations motivated by relativistic and quantum gravitational considerations~\cite{giacomini_quantum_2019_rev,vanrietvelde_change_2020,hohn_trinity_2021,chandrasekaran_AlgebraObservablesSitterspace_2023, fewster_QuantumReferenceFramesMeasurementSchemes_2024,devuyst_CrossedProductsQuantumreferenceframes_2025}. 

In this context, a quantum system in a symmetry-breaking state serves as a resource by encoding information about the corresponding reference frame, thereby enabling operations that would otherwise be impossible. For this reason, the RTA was historically also referred to as the resource theory of quantum reference frames in earlier studies~\cite{bartlett_reference_2007,gour_resource_2008,gour_measuring_2009}.
Here, we revisit this viewpoint and carefully review how the (a)symmetry of quantum states and the $G$-covariance of operations relate to physical scenarios where access to a shared reference frame is entirely absent~\cite{bartlett_reference_2007,marvian_mashhad_symmetry_2012}.

We begin by exploring a communication scenario~\cite{marvian_mashhad_symmetry_2012} in which two parties, Alice and Bob, do not share a common reference frame associated with a group $G$. 
For example, one can imagine the case in which they are located in separate spaceships without a shared Cartesian reference frame~\cite{bartlett_reference_2007}, implying that the relative orientation between them is unknown to Alice and Bob. In this case, the misalignment is characterized by a group element $g\in G$, where $G = SO(3)$. 
Suppose Alice wishes to delegate a quantum operation $\mathcal{F}$ on her qubit to Bob. She transmits the qubit to Bob, he applies an operation $\mathcal{E}$ in his own frame, and sends it back. Since their frames are misaligned by $g \in G$, the effective operation in Alice's frame
is $\mathcal{U}_{g^{-1}} \circ \mathcal{E} \circ \mathcal{U}_g$, which in general depends on the parameter $g\in G$, unknown to her. Therefore, the intended operation $\mathcal{F}$ is only correctly implemented if it can be expressed as $\mathcal{F} =\mathcal{U}_{g^{-1}} \circ \mathcal{E} \circ \mathcal{U}_g$ for all $g\in G$. Averaging over the Haar measure $\mu_G$, we obtain 
\begin{align}
    \mathcal{F}=\int_{g\in G}\dd \mu_G(g)\, \mathcal{U}_{g^{-1}}\circ \mathcal{E}\circ \mathcal{U}_g,
\end{align}
which is equivalent to the condition that $\mathcal{F}$ is $G$-covariant. In essence, only $G$-covariant operations can be correctly implemented when Alice delegates them to Bob in the absence of a shared reference frame associated with the group $G$.

The above observation raises a critical question: \emph{What additional information must be provided in order to implement quantum operations that are not $G$-covariant?} To gain further insight, let us revisit the example in which Alice and Bob have no shared Cartesian frames. In this case, performing non-$SO(3)$-covariant operations requires knowledge of a group element $g \in SO(3)$ that characterizes the relative orientation between Alice's and Bob's Cartesian frames. Operationally, this means that Alice must convey information about her local reference frame, i.e., directional information, to Bob. However, this cannot be accomplished using classical bits alone, as such information is defined relative to Alice's frame, which is inaccessible to Bob. To convey this information, Alice must transmit a physical system that encodes features of her reference frame, specifically one that transforms nontrivially under spatial rotations. In the quantum setting, this corresponds to preparing and sending a quantum state that breaks the $SO(3)$ symmetry.

This observation clarifies the need for a resource-theoretic perspective on symmetries and their breaking~\cite{bartlett_reference_2007,gour_resource_2008,gour_measuring_2009,marvian_how_2016} in the context of quantum reference frame. While the above discussion focused on spatial orientation and $SO(3)$ symmetry, the same reasoning applies more broadly. Certain information---such as time, phase, and spatial direction---requires a shared reference frame in order to be meaningfully communicated using classical bits~\cite{peres_unspeakable_2002,bartlett_reference_2007,marvian_how_2016}. In the absence of such a frame, this information must instead be conveyed through physical systems whose states break the corresponding symmetries: a quantum clock breaks time-translation symmetry, a phase reference breaks $U(1)$ symmetry, and a quantum gyroscope breaks $SO(3)$ symmetry.

This resource-theoretic perspective also extends beyond communication scenarios, finding broader applications in quantum technologies. For instance, quantum modules within a quantum computer, or quantum computers within a quantum network, must communicate asymmetries to calibrate their reference frames in order to coordinate and perform quantum tasks with high precision. Importantly, efficient transmission of such resources requires their systematic transformation. For example, communication through a quantum channel is generally constrained by cost and capacity, making it desirable to maximize the amount of asymmetry per channel use via asymmetry distillation. In contrast, when the goal is to distribute a reference frame to multiple locations, asymmetry must be divided into many resources, each containing a smaller amount of asymmetry, through asymmetry dilution. These considerations highlight the importance of developing a general systematic theory of converting asymmetry resources.

The main result of this study, Theorem~\ref{thm:conversion_rate_projective_finite_number}, establishes a comprehensive asymptotic theory for the manipulation of pure-state asymmetry in RTA. In particular, it provides an explicit formula for the optimal conversion rate, given in Eq.~\eqref{eq:rate_formula_compact_connected}, which is expressed in terms of the QGTs. 
Of particular importance is the fact that the optimal conversion protocol has been explicitly constructed in Sec.~\ref{sec:optimality_sketch}, offering a concrete procedure for practical implementations of asymmetry conversion.

\subsection{Standardized reference state}
An equally important aspect of the present theory is that it enables the introduction of a \emph{standardized reference state} that serves as a benchmark for quantifying and comparing frameness. Analogous to the role of the \emph{ebit} in entanglement theory and the \emph{coherence bit} in the resource theory of coherence, such states provide an operationally meaningful foundation for measuring the degree of asymmetry present in a quantum state.

In the context of the RTA, a standardized reference state has been proposed for the rotational group (more precisely, for its universal covering group $G = SU(2)$) in \cite{yang_units_2017}, where the conversion of a restricted class of pure states is analyzed. In the present work, we develop a comprehensive construction of standardized reference states and provide a detailed interpretation of their role for arbitrary semisimple compact Lie groups $G$.

Let $G$ be a compact Lie group with a unitary representation $U$ acting on a $d$-dimensional Hilbert space $\mathcal{H}$ such that $U(g) \propto \mathbb{I}$ only if $g = e$. We also introduce a $d$-dimensional auxiliary Hilbert space $\mathcal{H}'$, equipped with the trivial representation $U'$ of $G$ satisfying $U'(g)=\mathbb{I}$ for all $g\in G$. Let $\{\ket{\psi_i}\}_{i=1}^d$ and $ \{\ket{\phi_i}\}_{i=1}^d$ be orthonormal bases of $\mathcal{H}$ and $\mathcal{H}'$, respectively. We consider a maximally entangled state
\begin{align}
\ket{\Psi}\coloneqq \frac{1}{\sqrt{d}} \sum_{i=1}^d \ket{\psi_i} \otimes \ket{\phi_i} \in \mathcal{H} \otimes \mathcal{H}'.\label{eq:standardized_state_def}
\end{align}
Since $|\braket{\Psi|U(g)\otimes U'(g)|\Psi}|=1$ holds if and only if $g=e$, the symmetry subgroup of $\Psi\coloneqq \ket{\Psi}\bra{\Psi}$ is trivial:
\begin{align}
    \mathrm{Sym}_G(\Psi)= \{e\}. \label{eq:sym_G_for_reference}
\end{align}

Another notable feature of $\ket{\Psi}$ is that its QGT admits a particularly simple form. To make this explicit, we fix a local coordinate system $\{\lambda^\mu\}_{\mu=1}^{\dim G}$ in a neighborhood of the identity element $e \in G$, such that the generators $X_\mu\coloneqq -\ii \frac{\partial}{\partial \lambda^\mu} U(g(\lambda))\big|_{\lambda=0}$ satisfy $\mathrm{Tr}(X_\mu )=0$ and $\mathrm{Tr}(X_\mu X_\nu) = d \delta_{\mu\nu}$, where $\delta_{\mu\nu}$ denotes the Kronecker delta. Such a coordinate system always exists if $G$ is a compact semisimple Lie group~\cite{georgi_lie_2018}. In this coordinate system, the QGT of $\ket{\Psi}$ is given by
\begin{align}
    \mathcal{Q}^{\Psi}= \mathbb{I}. \label{eq:QGT_for_reference}
\end{align}
Intuitively, this implies that $\ket{\Psi}$ exhibits \emph{isotropic} asymmetry. 

These properties of $\ket{\Psi}$, namely, the triviality of its symmetry subgroup as expressed in Eq.~\eqref{eq:sym_G_for_reference} and its isotropic asymmetry in Eq.~\eqref{eq:QGT_for_reference}, guarantee that $\ket{\Psi}$ serves as a \emph{universal asymmetry resource} associated with the symmetry group $G$ in the sense that it can be asymptotically converted into any pure state at a non-zero rate.

A particularly intriguing expression arises from the conversion rate formula given in Eq.~\eqref{eq:rate_formula_D_max_QGT}, which involves the quantum max-relative entropy. The optimal conversion rate is given by
\begin{align}
    R(\Psi \to \phi) = 2^{- D_{\max}(\mathcal{Q}^{\phi} \| \mathcal{Q}^{\Psi})}. \label{eq:conversion_rate_reference_state_Dmax}
\end{align}
Substituting Eq.~\eqref{eq:QGT_for_reference}, we obtain a simple expression for the optimal conversion rate to an arbitrary pure state $\phi$, given by
\begin{align}
    R(\Psi \to \phi)= 2^{- D_{\max}(\mathcal{Q}^{\phi} \| \mathbb{I})}= 2^{H_{\min}(\mathcal{Q}^{\phi})},\label{eq:conversion_from_standardized_state}
\end{align}
where $H_{\min}$ denotes the min-entropy, defined as $H_{\min}(\rho) \coloneqq -\log_2 \lambda_{\max}(\rho)$, with $\lambda_{\max}(\rho) $ the largest eigenvalue of the matrix $ \rho $. Although min-entropy is typically applied to quantum states, it appears here as a function of the QGT, suggesting that entropy-based measures may have broader applicability beyond conventional formulations in quantum resource theories.

As an illustrative example of a standardized reference state, let us analyze the rotational group $G=SO(3)$ and its spin-$J$ unitary representation $U_J$, where $J$ is a non-negative integer. We adopt a coordinate system $ g(\theta) = e^{\ii\theta^\mu A_\mu} \in SO(3) $ such that $ J_\mu\coloneqq -\ii \frac{\partial}{\partial \theta^\mu} U_J(g(\theta))|_{\theta=0} $ satisfy the standard commutation relations $[J_\mu, J_{\mu'}]=\ii \sum_{\mu''}\epsilon_{\mu\mu'\mu''}J_{\mu''}$, where $\epsilon_{\mu\mu'\mu''}$ is the Levi-Civita symbol.

The representation space $\mathcal{H}^{(J)}$ is spanned by the orthonormal basis $\{\ket{J, m}\}_{m=-J}^J$, where $\ket{J, m}$ denotes a simultaneous eigenvector of the total angular momentum $ J^2 $ and its $ z $-component $ J_z $, with eigenvalues $J(J+1)$ and $ m $, respectively. We restrict attention to the case $J>0$, since the representation is trivial for $J=0$. We define the maximally entangled state
\begin{align}
    \ket{\Psi_J} \coloneqq \frac{1}{\sqrt{2J+1}} \sum_{m=-J}^J \ket{J, m} \otimes \ket{J, m},
\end{align}
which serves as a standardized reference state for quantum Cartesian frameness.

To align with the above arguments in a general setting, we rescale the conventional coordinates as $\lambda^\mu\coloneqq\sqrt{\frac{J(J+1)}{3}}\theta^\mu$. 
The corresponding generators
\begin{align}
    X_\mu\coloneqq -\ii \frac{\partial}{\partial \lambda^\mu} U_J(g(\lambda))\big|_{\lambda=0}= \sqrt{\frac{3}{J(J+1)}} J_\mu
\end{align}
satisfy $\mathrm{Tr}(X_\mu)=0$ and $\mathrm{Tr}(X_\mu X_\nu)=(2J+1)\delta_{\mu\nu}$ for $\mu,\nu=x,y,z$. Consequently, the QGT of $\ket{\Psi_J}$ is $\mathcal{Q}^{\Psi_J}=\mathbb{I}$, and hence the conversion rate is
\begin{align}
    R(\Psi_J \to \phi) = 2^{H_{\min}(\mathcal{Q}^{\phi})},
\end{align}
where the QGT $\mathcal{Q}^{\phi}$ is calculated in the coordinates $\{\lambda^\mu\}_{\mu=x,y,z}$.

It is important to note that not all asymmetric states can be converted into $\ket{\Psi_J}$. For example, consider the highest-weight state $\ket{\xi_J}\coloneqq\ket{J, J}\in\mathcal{H}^{(J)}$, whose QGT is given by
\begin{align}
    \mathcal{Q}^{\xi_J} =\frac{3}{J(J+1)} \frac{J}{2}
    \begin{pmatrix}
        1 & -\ii & 0 \\
        \ii & 1 & 0 \\
        0 & 0 & 0
    \end{pmatrix}
\end{align}
in the coordinates $\{\lambda^\mu\}_{\mu=x,y,z}$. Since $\mathcal{Q}^{\xi_J}$ has a zero eigenvalue, the inequality $\mathcal{Q}^{\xi_J} \geq r \mathcal{Q}^{\Psi_J}$ holds only for $r=0$, implying that $R(\xi_J \to \Psi_J)=0$. This aligns with the fact that $ \ket{\xi_J} $ is invariant under rotations about the $z$-axis and hence cannot be converted to the isotropic state $\ket{\Psi_J}$ via $G$-covariant operations.

Finally, we remark that a widely employed asymmetry measure, the relative entropy of $G$-asymmetry, fails to capture this fundamental difference between $\Psi_J$ and $\xi_J$. The relative entropy of $G$-asymmetry is defined by~\cite{gour_measuring_2009}:
\begin{align}
    A_G(\rho)\coloneqq \min_{\sigma \text{: symmetric state}}D(\rho\|\sigma),\label{eq:def_A_G}
\end{align}
where $D(\rho\|\sigma)\coloneqq \mathrm{Tr}(\rho (\log_2\rho -\log_2\sigma))$ denotes the quantum relative entropy. While two states, $\Psi_J$ and $\xi_J$, exhibit significantly different symmetry properties---most notably, the isotropy of $\Psi_J$ versus the directional invariance of $\xi_J$---they are assigned the same value of the relative entropy of $G$-asymmetry. 
Specifically, for all pure states $\phi\in\mathcal{P}(\mathcal{H}^{(J)})$, since the only symmetric state is $\frac{1}{2J+1}I$, the relative entropy of asymmetry is given by $A_G(\phi)=\log_2(2J+1)$. 

In the following subsection, we provide more discussion on the distinctions between the relative entropy of $G$-asymmetry and the QGT.

\subsection{Discussion: differential geometric and entropic perspectives}
The present study identifies the QGT as a central measure of symmetry breaking and quantum reference frameness within RTA. Our result aligns with prior studies in RTA that employ differential geometric quantities---such as the quantum Fisher information~\cite{marvian_coherence_2020,marvian_operational_2022}, the quantum Fisher information matrix~\cite{kudo_fisher_2023,gao_sufficient_2024}, the metric-adjusted skew information~\cite{zhang_detecting_2017,takagi_skew_2019}, and their extensions~\cite{yamaguchi_beyond_2023,yamaguchi_smooth_2023}---yet goes further by establishing the QGT as a complete monotone.

A complementary line of research within RTA focuses on entropic measures, such as the $G$-asymmetry~\cite{vaccaro_tradeoff_2008} and the relative entropy of $G$-asymmetry~\cite{gour_measuring_2009}, which, though defined differently, are known to yield the same value~\cite{gour_measuring_2009}. Under the name \emph{entanglement asymmetry}~\cite{ares_entanglement_2023}, this quantity has recently been studied extensively as a probe of symmetry breaking across a wide range of contexts~\cite{capizzi_entanglement_2023,yamashika_entanglement_2024,fossati_entanglement_2024,chen_renyi_2024,capizzi_universal_2024,ferro_non-equilibrium_2024,benini_entanglement_2025,kusuki_entanglement_2025}, including both condensed matter and high-energy physics.

These two approaches each capture essential features of symmetry breaking, as they are both valid measures within the framework of RTA. However, as illustrated by the example of the $SO(3)$ symmetry in the previous subsection, they appear to reflect different aspects of asymmetry. In what follows, we demonstrate that a distinction between these two measures emerges particularly in the case of non-Abelian symmetry.

A key observation is that the relative entropy of $G$-asymmetry satisfies $A_G(\rho) = A_G(\sigma)$ for any pair of states $\rho$ and $\sigma$ such that $\rho = \mathcal{U}_g(\sigma)$ for some $g \in G$. In other words, $A_G$ is blind to differences arising solely from group transformations. For example, for $G = SU(2)$, the relative entropy of $G$-asymmetry does not distinguish among the eigenvectors of Pauli operators in different directions. Intuitively, one may say that the relative entropy of $G$-asymmetry quantifies the \emph{magnitude} of symmetry breaking, while being agnostic to its \emph{direction}.

In contrast, the QGT and quantum Fisher information matrix encode information about both the magnitude and the direction of symmetry breaking. For example, for $G = SU(2)$, let $\ket{0}, \ket{1}$ denote the eigenvectors of $\sigma_z$, and define $\phi_0 \coloneqq \ket{0}\bra{0}$ and $\phi_+ \coloneqq \ket{+}\bra{+}$ with $\ket{+} \coloneqq (\ket{0} + \ket{1})/\sqrt{2}$. The QGTs of these states are given by
\begin{align}
    \mathcal{Q}^{\phi_{0}} = \frac{1}{4}
    \begin{pmatrix}
        1 & -\ii & 0 \\
        \ii & 1 & 0 \\
        0 & 0 & 0
    \end{pmatrix}, \quad
    \mathcal{Q}^{\phi_{+}} = \frac{1}{4}
    \begin{pmatrix}
        0 & 0 & 0 \\
        0 & 1 & -\ii \\
        0 & \ii & 1
    \end{pmatrix},
\end{align}
which clearly reflect the directional distinctions between the eigenstates of $\sigma_z$ and $\sigma_x$, capturing the difference induced by an $SU(2)$ transformation.

We remark that for the $U(1)$ group, the asymptotic behaviors of the relative entropy of $G$-asymmetry and the QGT (which equals the variance in this case) encode essentially the same information~\cite{gour_measuring_2009}. This agreement is consistent with the fact that the QGT is also invariant under $U(1)$ transformations. However, for a general non-Abelian group, while the relative entropy of $G$-asymmetry remains invariant under group transformations, the QGT transforms non-trivially (see Appendix~\ref{app:reduction_to_a_finite_number_of_inequalities} for the explicit transformation). This highlights that the relative entropy of $G$-asymmetry and the QGT capture fundamentally different aspects of symmetry breaking when the underlying group is non-Abelian. 

When a quantum reference frame is concerned, not only the magnitude but also the direction of asymmetry becomes important. For example, in the case of the $SO(3)$ symmetry discussed in the previous section, the state $\psi_J$ possesses asymmetry in all directions, while $\xi_J$ is invariant under rotations about the $z$-axis and therefore carries no information about rotation angles around the $z$-axis. As a consequence, it is impossible to convert $\xi_J$ into $\psi_J$. This provides an intuitive reason why the QGT serves as a complete measure of asymmetry, whereas the relative entropy of $G$-asymmetry does not.

These observations underscore that the appropriate choice of an asymmetry measure depends crucially on the physical context in which it is applied. While entropic measures such as the relative entropy of $G$-asymmetry have already been proven useful in a variety of settings---including recent applications in condensed matter and high-energy physics---there are situations where not only the amount but also the direction of symmetry breaking is of central relevance. For instance, when the physical interest lies in identifying in which direction a symmetry is broken, geometric quantities like the QGT or the quantum Fisher information, which are sensitive to directional features of asymmetry, would offer a more appropriate and informative characterization.

\section{Application to quantum thermodynamics}\label{sec:quantum_thermodynamics}
In this section, we explore the implications of our results for quantum thermodynamics. We focus on a basic and widely studied setting: thermodynamic processes induced by interactions with thermal baths under conservation laws. Using the tools developed to quantify asymmetry, we derive symmetry-induced constraints on state transformations in the thermodynamic limit. Moreover, we show that there are cases in which overcoming these constraints requires supplying a macroscopic amount of asymmetry, including energy coherence.

\subsection{Operational setting: thermal contact under conservation laws}
\label{subsec:qt-operational-setting}
We begin by modeling thermal contact by coupling a system to heat baths initially in Gibbs states and letting the joint system evolve under a global unitary that conserves total energy. Here, we consider a general scenario involving $k$ thermal baths. For $i=1,\ldots,k$, let $B_i$ be a bath with Hamiltonian $H_{B_i}$, initially prepared in the Gibbs state at inverse temperature $\beta_i$: 
\begin{align}
    \tau_{B_i, \beta_i}\coloneqq \frac{1}{Z_{i}} e^{-\beta_i H_{B_i}},\quad Z_{i}\coloneqq \mathrm{Tr}\left(e^{-\beta_i H_{B_i}}\right).
\end{align}
To model an evolution under thermal contact, let $S$ be the system of interest with Hamiltonian $H_{S}$. Imposing energy conservation, we assume that the joint evolution is described by a global unitary $V$ acting on $SB_1\cdots B_k$ such that
\begin{align}
    V\left( H_{S} + \sum_{i=1}^k H_{B_i}\right)V^\dag = H_{S'}+H_{B'}.\label{eq:energy_conservation}
\end{align}
where $S'$ and $B'$ are subsystems such that $SB_1\ldots B_k=S'B'$, and $H_{S'}$ and $H_{B'}$ denote the Hamiltonians associated with these subsystems, respectively. 

More generally, thermodynamic interactions may be constrained not only by energy conservation but also by additional conserved quantities, such as particle number or angular momentum. These additional conservation laws can be conveniently expressed in terms of a symmetry group $G$. Let $U_{S}$, $U_{B_i}$, $U_{S'}$, and $U_{B'}$ denote unitary representations of $G$ on the corresponding subsystems. The joint interaction unitary $V$ is assumed to respect the symmetry in the sense that it satisfies, for all $g\in G$,
\begin{align}
    V\left(U_{S}(g) \otimes \bigotimes_{i=1}^k U_{B_i}(g)\right)V^\dag =U_{S'}(g) \otimes U_{B'}(g).
\end{align}
Introducing generators $\{X_\mu\}_{\mu=1}^{\dim G}$ by differentiating the representations (cf. Eq.~\eqref{eq:hermitian_operators_Lie_alg}), this implies the conservation of each generator:
\begin{align}
    V \left(X_{S,\mu} + \sum_{i=1}^kX_{B_i,\mu}\right)V^\dag = X_{S',\mu}+X_{B',\mu}\label{eq:conserved_quantities_bath_system}
\end{align}
for all $\mu=1,\ldots,\dim G $. For particle number and angular momentum conservation, these generators correspond to the particle number and angular momentum operators, respectively. To ensure that the free evolution of each subsystem respects the symmetry, we impose that the generators commute with the Hamiltonian in each system. 
We note that energy conservation (i.e., Eq.~\eqref{eq:energy_conservation}) is recovered as the special case where $G=\mathbb{R}$ corresponds to time translations, with the unitary representation $U(t)=e^{-\ii Ht}$ generated by the Hamiltonian $H$ of the total system. 
Although $\mathbb{R}$ is not compact, the results in this section rely on the monotonicity of the QGT, which remains valid even when $G$ is noncompact. Below, we include the time translational symmetry in the symmetry $G$ as needed.

Having modeled the thermal contact above, the central task of characterizing thermodynamic processes thus reduces to studying the transformation of system states under the resulting dynamics.

\subsection{Resource-theoretic formalism: covariance from conservation laws}\label{subsec:qt-resource-perspective}
To investigate the setup introduced in the previous subsection, we adopt the resource-theoretic viewpoint developed in the present paper. To this end, we define a channel
\begin{align}
    \mathcal{E}(\cdot) \coloneqq \mathrm{Tr}_{B'}\left[ V\left( \cdot \otimes \bigotimes_{i=1}^k\tau_{B_i,\beta_i} \right) V^\dag\right]\label{eq:definition_evolution_thermal_contact},
\end{align}
which describes the evolution of a system interacting with thermal baths under conservation laws. 

We note that, in the resource theory of athermality, it is common to restrict to a single fixed bath temperature. This is because this framework is typically designed to quantify athermality (i.e., deviation from thermal equilibrium) relative to a fixed equilibrium background, thereby revealing constraints on work extraction; access to multiple temperatures can, in principle, enable work extraction. Indeed, in the special case where all the baths are at the same temperature ($\beta_1=\cdots=\beta_k$) with time-translation symmetry ($G=\mathbb{R}$), Eq.~\eqref{eq:definition_evolution_thermal_contact} reduces to thermal operations, which constitute a central tool in the resource theory of athermality. 

In contrast, however, many practical thermodynamic setups, such as heat engines, naturally involve contact with multiple baths at different temperatures. Therefore, we allow the baths to be possibly at different temperatures and formulate our results in a general multi-bath setting. 

Our focus is on exploring the role of asymmetry (and, in particular, coherence) in quantum thermodynamics. Given that $V$ respects the symmetry $G$, the channel described in Eq.~\eqref{eq:definition_evolution_thermal_contact} is $G$-covariant whenever the Gibbs state $\bigotimes_{i=1}^k\tau_{B_i,\beta_i}$ is symmetric.
If $G$ is connected, this condition is automatically satisfied in the common setting where the free evolutions of the baths respect the symmetry, i.e., $[H_{B_i},X_{B_i,\mu}]=0$ for all $\mu$ and $i$. In this case, any asymmetry monotone is non-increasing under thermal contact described by Eq.~\eqref{eq:definition_evolution_thermal_contact}.

To investigate thermal processes based on this observation, we introduce a one-parameter family of monotones $\mathcal{Q}_q$ for $q\in(0,1)$, whose matrix elements are defined by
\begin{align}
    &\left(\mathcal{Q}_q^\rho\right)_{\mu\nu}\coloneqq\sum_{\underset{(1-q)p_l+qp_k>0}{k,l}}\frac{f_q(0)(p_l-p_k)^2}{(1-q)p_l+qp_k}\braket{l|X_\mu|k}\braket{k|X_\nu|l},\label{eq:definition_S_q}
\end{align}
which satisfies $\gamma^\dag \mathcal{Q}_q^\rho\gamma=f_q(0)\|\ii [\gamma^\dag X,\rho]\|_{f_q,\rho}^2$ for any $\gamma\in\mathbb{C}^{\dim G}$. Due to Eqs.~\eqref{eq:monotonicity_norm_general_fq} and~\eqref{eq:relation_between_S_and_f_q}, this matrix is a valid asymmetry monotone satisfying $\lim_{q\to 0}\mathcal{Q}_q^\rho=\mathcal{Q}^\rho$. Note that Eq.~\eqref{eq:norm_formula_general_r_QGT} provides a direct relation of $\mathcal{Q}_q$ to the QGT for pure states. The following theorem is a key tool characterizing thermodynamic processes in the thermodynamic limit:
\begin{thm}\label{thm:noniid_monotonicity_q}
    Let $G$ be an arbitrary Lie group. Consider a sequence of arbitrary states $\{\rho_N\}_N$ and a sequence of i.i.d. pure states $\{\phi^{\otimes \floor{rN}}\}_N$ at rate $r>0$. If $\{\rho_N\}_N\gconv \{\phi^{\otimes \floor{rN}}\}_N$, then for each $q\in(0,1)$, the following inequality holds for all $g\in G$:
    \begin{align}
        \forall \gamma\in \mathbb{C}^{\dim G},\,  \liminf_{N\to\infty}\frac{1}{N}\gamma^\dag\mathcal{Q}_q^{\mathcal{U}_g^{(N)}(\rho_N)}\gamma\geq r\, \gamma^\dag \mathcal{Q}^{\mathcal{U}_g'(\phi)}\gamma\label{eq:noniid_monotonicity_theorem_q},
    \end{align}
    where $\mathcal{U}_g^{(N)}(\cdot)\coloneqq U^{(N)}(g)(\cdot)U^{(N)}(g)^\dag$ with unitary representations $U^{(N)}$ of $G$ on the Hilbert space of $\rho_N$.
\end{thm}
\begin{proof}
    Analogous to Eq.~\eqref{eq:monotonicity_iid_asymptotic}, the monotonicity of $\mathcal{Q}_q$ and Eq.~\eqref{eq:asymptotics_norm_asymmetric} imply that for any $\gamma\in\mathbb{C}^{\dim G}$, 
    \begin{align}
    &\frac{1}{N}\gamma^\dag\mathcal{Q}_q^{\rho_N}\gamma\geq \frac{M}{N}\gamma^\dag \mathcal{Q}^\phi\gamma-\frac{M}{N}h(\epsilon)+\frac{1}{N}o(M)
    \end{align}
    holds for all sufficiently large $N$, where $M\coloneqq \floor{rN}$. Therefore, we get
    \begin{align}
        \liminf_{N\to\infty}\frac{1}{N}\gamma^\dag\mathcal{Q}_q^{\rho_N}\gamma\geq r\gamma^\dag \mathcal{Q}^\phi\gamma-rh(\epsilon).
    \end{align}
    Since this inequality holds for any sufficiently small $\epsilon>0$ and $\lim_{\epsilon\to 0^+}h(\epsilon)=0$, we get Eq.~\eqref{eq:noniid_monotonicity_theorem_q} for $g=e$. Repeating these arguments for $\mathcal{U}_g^{(N)}(\rho_N)$ and $\mathcal{U}_g'(\phi)$, we obtain Eq.~\eqref{eq:noniid_monotonicity_theorem_q} for all $g\in G$. 
\end{proof}

Several remarks are in order: 
First, this theorem is valid for an arbitrary Lie group $G$, and thus, it can also be applied to the time-translation symmetry with $G=\mathbb{R}$. 
Second, the monotonicity of $\mathcal{Q}_q$ alone is insufficient to establish the theorem due to the asymptotic discontinuity discussed in Sec.~\ref{sec:monotonicity_QGT_sketch}.
Third, Eq.~\eqref{eq:noniid_monotonicity_theorem_q} holds for any sequence of states $\{\rho_N\}_N$, without the assumption that the states are i.i.d. This generality is relevant in quantum thermodynamics, where external systems supplying work and coherence are typically not assumed to be in an i.i.d. state. 
Fourth, the QGT $\mathcal{Q}^\phi$ and its extensions $\mathcal{Q}_q^\rho$ are extensive quantities (see Appendix~\ref{app:properties_S} for further properties), which makes them particularly natural measures in the thermodynamic limit. 
Finally, for $G = \mathbb{R}$, a similar theorem can also be established for the metric adjusted skew information rates by employing Lemma~3 of Ref.~\cite{yamaguchi_smooth_2023} in place of Eq.~\eqref{eq:asymptotics_norm_asymmetric}.

\subsection{No-go theorem for distillation}\label{subsec:qt-no-go-iid}

Let us first consider a simple scenario where we asymptotically convert i.i.d. states via a thermal process described by Eq.~\eqref{eq:definition_evolution_thermal_contact}. As noted in the previous section, in the single-temperature case with time-translation symmetry, this evolution reduces to thermal operation. For \emph{quasiclassical} states $\rho$ and $\sigma$, i.e., states commuting with the Hamiltonian (equivalently, symmetric states under time translation), the optimal asymptotic conversion rate via thermal operations, which we denote by $R_{\mathrm{TO}}(\rho\to\sigma)$, is known to be calculated via the following formula~\cite{brandao_resource_2013}:
\begin{align}
    R_{\mathrm{TO}}(\rho\to\sigma)=\frac{D(\rho\|\tau)}{D(\sigma\|\tau)},\label{eq:free_energy_ratio}
\end{align}
where $\tau$ denotes the Gibbs state and $D(\rho\|\sigma)\coloneqq \mathrm{Tr}(\rho\ln\rho)-\mathrm{Tr}(\rho\ln\sigma)$ 
denotes the quantum relative entropy. Here, we adopt the natural logarithm for convenience. Since $D(\rho\|\tau)=\beta (F_\beta(\rho)-F_\beta(\tau))$ with $F_\beta(\rho)\coloneqq \braket{H}_\rho-\beta^{-1} S_{\mathrm{vN}}(\rho)$, where $S_{\mathrm{vN}}$ is the von Neumann entropy given by $S_{\mathrm{vN}}(\rho):=-\mathrm{Tr}\left(\rho \ln\rho\right)$, Eq.~\eqref{eq:free_energy_ratio} implies that the (non-equilibrium) free energy completely characterizes the optimal conversion rate. 

However, for general, non-quasiclassical states, the necessary and sufficient conditions for state conversion under thermal operations remain unknown~\cite{faist_macroscopic_2019,sagawa_asymptotic_2021}. This is because thermal operations cannot generate a coherent superposition of energy eigenstates; therefore, such energy coherence, i.e., asymmetry under time translations, must be regarded as a separate resource. 

Building on the establishment of the QGT as the complete measure of asymmetry, we here derive a no-go theorem on the distillation of a state expressed in terms of the QGT and its extension. We denote by $R_{\mathrm{TP}}^G(\rho\to\phi)$ the asymptotic conversion rate from a general state $\rho$ to a pure state $\phi$ via thermal processes defined in Eq.~\eqref{eq:definition_evolution_thermal_contact}. Note that when only the time-translation symmetry is involved ($G=\mathbb{R}$), $R_{\mathrm{TP}}^{\mathbb{R}}(\rho\to\phi)$ provides an upper bound on $R_{\mathrm{TO}}(\rho\to\phi)$, but need not coincide with $R_{\mathrm{TO}}(\rho\to\phi)$ since the definition of $R_{\mathrm{TP}}^{\mathbb{R}}(\rho\to\phi)$ allows access to baths at different temperatures.

From the additivity of $\mathcal{Q}_q$, the left-hand side of Eq.~\eqref{eq:noniid_monotonicity_theorem_q} in Theorem~\ref{thm:noniid_monotonicity_q} for i.i.d. states can be evaluated as
\begin{align}
    \liminf_{N\to\infty}\frac{1}{N}\gamma^\dag\mathcal{Q}_q^{\mathcal{U}_g(\rho)^{\otimes N}}\gamma=\gamma\mathcal{Q}^{\mathcal{U}_g(\rho)}_q\gamma
\end{align}
for any $\gamma\in\mathbb{C}^{\dim G}$. Therefore, Theorem~\ref{thm:noniid_monotonicity_q} yields the following upper bound:
\begin{align}
    &\sup\{r\geq0 \mid\forall g\in G,\,\mathcal{Q}_q^{\mathcal{U}_g(\rho)}\geq r\mathcal{Q}^{\mathcal{U}'_g(\phi)}\}\geq R^G_{\mathrm{TP}}(\rho\to\phi)
\end{align}
for each $q\in (0,1)$. Moreover, in the limit $q\to 1^-$, we also obtain
\begin{align}
    \sup\{r\geq0 \mid\forall g\in G,\,\mathcal{Q}^{\mathcal{U}_g(\rho)}\geq r\mathcal{Q}^{\mathcal{U}'_g(\phi)}\}\geq R^G_{\mathrm{TP}}(\rho\to\phi)\label{eq:upper_bound_distillable_athermality}.
\end{align}
Consequently, we obtain a no-go theorem in quantum thermodynamics analogous to  Corollary~\ref{cor:projector_distillablle_asymmetry}.
\begin{cor}[No-go theorem for distillation via thermal operation]\label{cor:projector_distillablle_athermality}
    Let $G$ be a connected Lie group describing symmetry. For a general state $\rho$ and a pure state $\phi$, if there exists a vector $\gamma\in\mathbb{C}^{\dim G}$ satisfying $[\Pi_\rho,\gamma^\dag X]=0$ and $\gamma^\dag \mathcal{Q}^\phi\gamma\neq 0$, then $R_{\mathrm{TP}}^G(\rho\to\phi)=0$.
\end{cor}
In particular, this result implies that for a typical state with full rank, it is impossible to distill a pure state through thermal contact at a linear rate.
We remark that these arguments, when particularized to time-translation symmetry with $G=\mathbb{R}$, are consistent with Eq.~\eqref{eq:free_energy_ratio}. Indeed, the QGT vanishes for quasiclassical states, and therefore the above arguments do not impose nontrivial constraints when $\phi$ is quasiclassical. In contrast, Eq.~\eqref{eq:upper_bound_distillable_athermality} provides a quantitative bound on the distillable asymmetry, which merits attention and is expected to play an important role in extending the conversion theory via thermal operations to general states.

An important point is that Corollary~\ref{cor:projector_distillablle_athermality} holds for continuous symmetries beyond the case of energy conservation. As an example, we now focus on the convertibility of i.i.d. states when the system couples to a thermal bath through an interaction that preserves the $U(1)$ symmetry, which appears in the Heisenberg XXZ model. 
Specifically, we investigate the distillation scenario, i.e., the conversion from an i.i.d. copy of a mixed state to an i.i.d. copy of a pure state, under the conservation law for the $z$-component of spin. 
To make this scenario concrete, we consider an asymmetric pure state defined as $\psi_a\coloneqq \ket{\psi_a}\bra{\psi_a}$, where
\begin{align}
    \ket{\psi_a}\coloneqq \sqrt{a}\ket{\uparrow}+\sqrt{1-a}\ket{\downarrow},\quad a\in (0,1),
\end{align}
and $\ket{\uparrow},\ket{\downarrow}$ denote the spin-up and spin-down states. We also introduce a corresponding mixed state
\begin{align}
    \rho_{a,\epsilon}\coloneqq (1-\epsilon)\psi_a+\epsilon \frac{\mathbb{I}}{2},
\end{align}
where $\epsilon\in(0,1]$ controls its purity. 

Before discussing the implications of Corollary~\ref{cor:projector_distillablle_athermality}, we first examine the behavior of a measure of asymmetry that has been widely employed in prior research~\cite{gour_measuring_2009}, namely relative entropy of $G$-asymmetry $A_G(\rho)$. 
Using the formula in~\cite{gour_measuring_2009}, we obtain $A_G(\rho_{a,\epsilon})=H\left(a(1-\epsilon)+\frac{\epsilon}{2}\right)-H(\epsilon)$ and $A_G(\psi_{a})=H(a)$, where $H(p)$ denotes the binary entropy given by $H(p)\coloneqq -p\ln p-(1-p)\ln (1-p)$. Therefore, for $a_1>a_2>\frac{1}{2}$ and a sufficiently small $\epsilon$, we find that
\begin{align}
    A_G(\rho_{a_2,\epsilon})>A_G(\psi_{a_1}),
\end{align}
which demonstrates that $\rho_{a_2,\epsilon}$ is more asymmetric than $\psi_{a_1}$ when quantified by $A_G$. This is consistent with the geometric intuition on the Bloch sphere, where the $U(1)$ transformation corresponds to the rotation around the $z$-axis (see Fig.~\ref{fig:bloch}).

\begin{figure}[htb]
    \centering
    \includegraphics[width=0.5\linewidth]{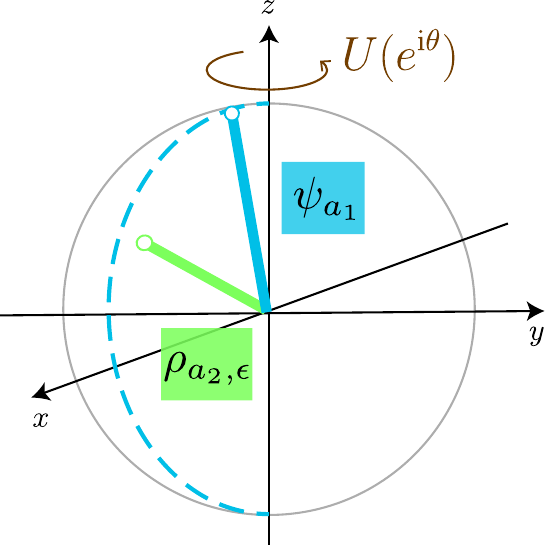}
    \caption{Geometric description of states on the Bloch sphere. 
    Since the group transformation $U(e^{\ii\theta})$ corresponds to a rotation about the $z$-axis, intuitively, a quantum state positioned farther from the $z$-axis exhibits greater asymmetry than one located closer to it. This behavior is illustrated by the relation $A_G(\rho_{a_2,\epsilon}) >A_G(\psi_{a_1})$ for $a_1>a_2>\frac{1}{2}$ and sufficiently small $\epsilon$.
}
    \label{fig:bloch}
\end{figure}

At first glance, this observation might suggest that a non-zero conversion rate from $\rho_{a_2,\epsilon}$ to $\psi_{a_1}$ is feasible. However, the newly introduced monotone $\mathcal{Q}$ reveals that this is not true. Specifically, since $\mathcal{Q}^{\rho_{a_2,\epsilon}} = 0$ while $\mathcal{Q}^{\psi_{a_1}} \neq 0$ for any $a_1, a_2 \in (0,1)$, Corollary~\ref{cor:projector_distillablle_athermality} implies that it is impossible to distill $\psi_{a_1}$ from $\rho_{a_2,\epsilon}$ at a non-zero rate via thermal operations. In other words, $\mathcal{Q}$ demonstrates that pure asymmetric states possess inherent value and cannot, in general, be distilled from mixed asymmetric states---an insight not necessarily captured by the conventional measure $A_G$. This result underscores the effectiveness of $\mathcal{Q}$ as a monotone in the context of quantum thermodynamics. Although Corollary~\ref{cor:projector_distillablle_athermality} is applied here in the specific case of $U(1)$ symmetry for illustrative purposes, the result holds for any connected Lie group.

We conclude by remarking on a more general scenario, in which a generalized bath, denoted by $B$, is prepared in a generalized Gibbs ensemble (GGE), given by
\begin{align}
    \tau_{B,\mathrm{GGE}}&\coloneqq \frac{1}{Z_{\mathrm{GGE}}}e^{-\beta (H_B-\alpha^\mu X_{B,\mu})},
\end{align}
where $\beta>0$ denotes the inverse temperature, $\{\alpha^\mu\}_{\mu=1}^{\dim G}$ are generalized ``chemical potentials", $\{X_{B,\mu}\}_{\mu=1}^{\dim G}$ are conserved charges (which corresponds to the operators in Eq.~\eqref{eq:conserved_quantities_bath_system}) and $Z_{\mathrm{GGE}}\coloneqq \mathrm{Tr}(e^{-\beta (H_B-\alpha^\mu X_{B,\mu}}))$ is the generalized partition function~\cite{guryanova_thermodynamics_2016,guryanova_thermodynamics_2016,lostaglio_thermodynamic_2017,mitsuhashi_characterizing_2022}.
For example, when the particle number is conserved under $G=U(1)$ symmetry, $\alpha_1$ and $X_{B,1}$ are the standard chemical potential and the number operator, which yields the grand canonical ensemble.
In this case, a (generalized) thermal operation realized by symmetry-respecting interaction with such a bath is not necessarily a $G$-covariant operation because $\tau_{B,\mathrm{GGE}}$ is not generally symmetric. This implies that not all asymmetry monotones qualify as athermality monotones. Nevertheless, $\mathcal{Q}^\rho$ remains non-increasing under any thermal process described in Eq.~\eqref{eq:definition_evolution_thermal_contact}, as $\tau_{B,\mathrm{GGE}}$ is full rank and hence $\mathcal{Q}^{\tau_{B,\mathrm{GGE}}}=0$. Consequently, the upper bound on distillable athermality given in Eq.~\eqref{eq:upper_bound_distillable_athermality} and Corollary~\ref{cor:projector_distillablle_athermality} remain valid. 

\subsection{Macroscopic coherence required in quantum thermodynamics}\label{subsec:qt-resource-requirement}

In the previous subsection, we identified a constraint on distillation via thermodynamic processes that originates from a state’s nontrivial asymmetry. A physically important question is whether this constraint is robust, or whether it can be circumvented once additional resources are made available. If only a subextensive (i.e., thermodynamically negligible) amount of asymmetry suffices to overcome the restriction, then the constraint would have limited physical significance; if an extensive resource is unavoidable, it remains essential in the thermodynamic limit. 

In this subsection, we show that there are infinitely many instances in which a macroscopic amount of coherence (and, more generally, asymmetry) must be supplied by an external system to enable distillation via interactions with thermal baths. In particular, this also implies the same conclusion within the resource theory of athermality, where thermal operations are free. At first glance, this may seem in tension with a common view in the literature that only a small coherence resource is needed in quantum thermodynamics~\cite{brandao_resource_2013,faist_macroscopic_2019,sagawa_asymptotic_2021}. We resolve this apparent discrepancy by noting that it stems from the choice of asymmetry quantifier: here we adopt the QGT and its extension, motivated by our result establishing the QGT as a complete asymmetry measure, whereas the literature often employs coarser, state-independent measures.

To this end, we consider the scenario in which the initial state is $\rho_N=\rho^{\otimes N}\otimes \xi_N$, where $\rho$ is the state to be distilled and $\xi_N$ is a state of the external resource system. Suppose that it is possible to asymptotically convert $\rho$ to $\phi$ at rate $r>0$ using thermal processes described by Eq.~\eqref{eq:definition_evolution_thermal_contact}, when assisted by $\xi_N$. Then, since the additivity of $\mathcal{Q}_q$ implies
\begin{align}
    \frac{1}{N}\mathcal{Q}_q^{\rho_N}=\mathcal{Q}_q^{\rho}+\frac{1}{N}\mathcal{Q}^{\xi_N}_q,
\end{align}
Eq.~\eqref{eq:noniid_monotonicity_theorem_q} in Theorem~\ref{thm:noniid_monotonicity_q} yields
\begin{align}
     \liminf_{N\to\infty}\frac{1}{N}\gamma^\dag\mathcal{Q}_q^{\xi_N}\gamma\geq  \gamma^\dag \left(r\,\mathcal{Q}^{\phi}-\mathcal{Q}_q^{\rho}\right)\gamma.\label{eq:required_asymmetry_general}
\end{align}
To avoid unnecessary complications, we have here applied Eq.~\eqref{eq:noniid_monotonicity_theorem_q} only for $g=e$, but we remark that the arguments can be straightforwardly extended to each $g\in G$.  
Thus, if there are $\gamma\in\mathbb{C}^{\dim G}$ and $q\in (0,1)$ such that the right-hand side of Eq.~\eqref{eq:required_asymmetry_general} is positive, then $\gamma^\dag\mathcal{Q}_q^{\xi_N}\gamma$ must be at least $O(N)$. In other words, distillation necessarily requires a macroscopic amount of asymmetry supplied by the external resource system.

To make the discussion more concrete, we examine the particular case in which $\xi_N$ are pure states. 
In this case, Eq.~\eqref{eq:norm_formula_general_r_QGT} yields
\begin{align}
        &\gamma^\dag\mathcal{Q}_q^{\xi_N}\gamma=\gamma^\dag \mathcal{Q}^{\xi_N}\gamma +\frac{1-q}{q}\gamma^\dag (\mathcal{Q}^{\xi_N})^*\gamma.
\end{align}
If the QGT rate $\lim_{N\to\infty}\frac{1}{N}\mathcal{Q}^{\xi_N}$ converges, we denote it by
\begin{align}
    \mathcal{Q}^{\infty,\xi}\coloneqq \lim_{N\to\infty}\frac{1}{N}\mathcal{Q}^{\xi_N}.
\end{align}
Under this assumption, Eq.~\eqref{eq:required_asymmetry_general} implies
\begin{align}
    \mathcal{Q}^{\infty,\xi} +\frac{1-q}{q} (\mathcal{Q}^{\infty,\xi})^*\geq  r\,\mathcal{Q}^{\phi}-\mathcal{Q}_q^{\rho}.
\end{align}
Taking the limit $q\to 1^-$, we obtain
\begin{align}
    \mathcal{Q}^{\infty,\xi}\geq r\,\mathcal{Q}^{\phi}-\mathcal{Q}^{\rho},\label{eq:monotonicity_required_QGT_general}
\end{align}
where we have used $\lim_{q\to 1^-}\mathcal{Q}_q^{\rho}=\mathcal{Q}^{\rho}$. For example, suppose that $\rho$ is full rank and $\phi$ is asymmetric. Then, $\mathcal{Q}^\rho=0$ and there exists a vector $\gamma\in\mathbb{C}^{\dim G}$ such that $\gamma^\dag \mathcal{Q}^{\phi}\gamma>0$. Therefore, $\gamma^\dag\mathcal{Q}^{\infty,\xi}\gamma>0$, or equivalently, $\gamma^\dag \mathcal{Q}^{\xi_N}\gamma$ must be at least $O(N)$. 

The result so far is applicable to any continuous symmetry as long as the Gibbs state is symmetric. However, the study of quantum thermodynamics under multiple conservation laws is still underdeveloped. Thus, we return to the standard setup involving only energy conservation. This allows us to compare with prior studies and to clarify the implications of the above arguments. Since the QGT is equal to the energy variance for pure states for time translation symmetry, Eq.~\eqref{eq:monotonicity_required_QGT_general} yields
\begin{align}
    &\lim_{N\to\infty}\frac{1}{N}V(\xi_N,H_N^{\mathrm{(ext)}})\geq rV(\phi,H')-\mathcal{Q}^\rho,\label{eq:variance_rate_S_0}
\end{align}
where $H_N^{\mathrm{(ext)}}$ denotes the Hamiltonian for the external resource system and $\mathcal{Q}^\rho=\mathrm{Tr}\left(\rho H (I-\Pi_\rho)H\right)$. We note that the SLD quantum Fisher information, widely used in RTA for time translation symmetry, equals four times the variance for pure states. We remark that if $\mathcal{Q}^\rho \neq 0$, then Eq.~\eqref{eq:variance_rate_S_0} yields only a trivial bound whenever $r \leq \mathcal{Q}^\rho/V(\phi,H')$.

If $\mathcal{Q}^\rho=0$, the above inequality yields 
\begin{align}
    &\lim_{N\to\infty}\frac{1}{N}V(\xi_N,H_N^{\mathrm{(ext)}})\geq rV(\phi,H').\label{eq:variance_rate_bound}
\end{align}
For example, when $\rho$ is a typical state with full rank, this inequality holds. 
Equation~\eqref{eq:variance_rate_bound} implies that for a conversion at nonvanishing rate $r>0$, the energy coherence of the external resource system, quantified by $V(\xi_N,H_N^{\mathrm{(ext)}})$, must scale at least $O(N)$, as long as $\phi$ is a non-quasiclassical state. 
In particular, if energy coherence is supplied via coherence bits---qubit systems in the state $\tfrac{1}{\sqrt{2}}(\ket{0}+\ket{1})$ with Hamiltonian $H=h\ket{1}\bra{1}$ for $h>0$---then the number of required coherence bits scales as $O(N)$. Note that Eq.~\eqref{eq:variance_rate_bound} becomes trivial when $\phi$ is a symmetric state (i.e., $[\phi,H']=0$), since $V(\phi,H')=0$ in that case. In particular, Eq.~\eqref{eq:variance_rate_bound} does not imply that macroscopic coherence is required for dilution into the ground state (or, more generally, into an energy eigenstate).

We note an advantage of adopting the QGT and its extension over the metric adjusted skew information. As mentioned earlier, Theorem~\ref{thm:noniid_monotonicity_q} can be straightforwardly extended to skew information rates by employing Lemma~3 of Ref.~\cite{yamaguchi_smooth_2023}. When $\xi_N$ are pure states, we obtain the following inequality in place of Eq.~\eqref{eq:variance_rate_S_0}:
\begin{align}
    &\lim_{N\to\infty}\frac{1}{N}V(\xi_N,H_N^{\mathrm{(ext)}})\geq rV(\phi,H')-I^{f^{(\mathrm{s})}}(\rho,H),\label{eq:variance_rate_masi}
\end{align}
where the metric adjusted skew information is defined by $I^{f^{(\mathrm{s})}}(\rho,H)\coloneqq \frac{f^{(\mathrm{s})}(0)}{2}\|\ii[\rho,H]\|_{f^{(\mathrm{s})},\rho}^2$ for a monotone function $f^{(\mathrm{s})}$ satisfying $f^{(\mathrm{s})}(t)=tf^{(\mathrm{s})}(1/t)$ and $f^{(\mathrm{s})}(0)>0$. This inequality reduces to the bound in Eq.~\eqref{eq:variance_rate_bound} only when $[\rho,H]=0$ since the metric adjusted skew information is a faithful measure of coherence. In contrast, $\mathcal{Q}^\rho=\mathrm{Tr}\left(\rho H (I-\Pi_\rho)H\right)=0$ holds either if $[\rho,H]=0$ or if $[\Pi_\rho,H]=0$. Thus, it is precisely the non-faithful property of $\mathcal{Q}^\rho$ that makes it possible to conclude that a macroscopic amount of energy coherence must be prepared in the external resource system for distillation of a wide class of states, including in particular all full-rank states $\rho$.

While our analysis does not rule out the existence of catalysts, it indicates that, in general, such catalysts should exhibit $O(N)$ asymmetry.
Indeed, our analysis remains valid even if the external state $\xi_N$ contains catalytic degrees of freedom. Moreover, any residual correlations between the system of interest and the catalyst do not affect the conclusion, since the proof does not rely on any explicit assumption about such correlations.
Therefore, even if a catalyst is employed, the total amount of asymmetry that must be supplied from an external system, including the catalyst, is $O(N)$ in general.
We also remark that coherence amplification protocols~\cite{takagi_correlation_2022,shiraishi_arbitrary_2024} do not trivialize the preparation of coherence in the external system: while these protocols can increase the sum of local coherences, they do not increase the total amount of coherence.

Finally, we explain the relation between our result and earlier analyses in the resource theory of athermality~\cite{brandao_resource_2013,faist_macroscopic_2019,sagawa_asymptotic_2021}. For time-translation symmetry ($G=\mathbb{R}$), the bound in Eq.~\eqref{eq:variance_rate_bound} applies to distillation under thermal operations, since thermal operations arise as a special case of Eq.~\eqref{eq:definition_evolution_thermal_contact} in the single-temperature setting. Consequently, there exist infinitely many instances in which a macroscopic amount of coherence must be supplied in order to distill a state into a pure state with nonzero energy coherence (see Appendix~\ref{app:quantum_thermo_setup_review_and_relation} for more details, including a brief review of the setup in Refs.~\cite{faist_macroscopic_2019,sagawa_asymptotic_2021}).

On the other hand, in the literature on the resource theory of athermality, the coherence cost of state conversion is often quantified via the energy range of an external system supplying the coherence~\cite{faist_macroscopic_2019,sagawa_asymptotic_2021}, defined by $\|H_N^{\mathrm{(ext)}}\|_\infty$. In the thermodynamic limit, the required energy range is shown to be negligible compared to the total work cost, i.e., $\|H_N^{\mathrm{(ext)}}\|_\infty=o(N)$, which prevents any non-negligible embezzlement of work from the external system~\cite{brandao_resource_2013,faist_macroscopic_2019,sagawa_asymptotic_2021}. 
This has contributed to the view that the required coherence overhead can be small when it is assessed solely in terms of the external system's energy range.

This apparent discrepancy originates from the quantifiers adopted in the analysis. As Ref.~\cite{faist_macroscopic_2019} describes, $\|H_N^{\mathrm{(ext)}}\|_\infty$ is ``a very rudimentary way'' to quantify energy coherence: it depends only on the Hamiltonian and is independent of the state of the external system. By contrast, the energy variance, which serves as a natural measure of coherence, is state-dependent. Since the variance involves second moments of the Hamiltonian, it scales quadratically with the Hamiltonian and thus can be macroscopic even when $\|H_N^{\mathrm{(ext)}}\|_\infty=o(N)$. In Ref.~\cite{brandao_resource_2013}, an example is given in which $\|H_N^{\mathrm{(ext)}}\|_\infty$ is estimated to scale as $\sqrt{N}$. In such a case, achieving a variance of order $O(N)$ requires the external system to be in a state with a large superposition of distinct energy levels. 

Taken together, these considerations underscore the critical importance of the choice of asymmetry measure; it can fundamentally alter our understanding of resource requirements in thermodynamic processes.

\section{Conclusion}\label{sec:conclusions}
In this paper, we established the asymptotic conversion theory between i.i.d. pure states in RTA for a compact Lie group. Our main result, namely the conversion rate formula, significantly extends a previous study on the $U(1)$ group to arbitrary compact Lie groups, thereby providing a unified framework for quantifying symmetry breaking across quantum systems with diverse symmetries. For example, it now covers cases involving multiple non-commutative conserved quantities.  Notably, the derived conversion rate formula has led to a resolution of the long-standing Marvian–Spekkens conjecture, which had remained open for over a decade. 
We remark that the asymptotic conversion rate in RTA cannot be derived from the approach in general resource theory for reversible conversion \cite{brandao_reversible_2015} using the generalized Stein's lemma \cite{brandao_generalization_2010,berta_gap_2023,hayashi_GeneralizedQuantumSteinslemmasecond_2025,lami_solution_2024}. This is because this framework does not apply to RTA since the regularized relative entropy vanishes in RTA \cite{gour_measuring_2009}, as emphasized in \cite{brandao_reversible_2015, hayashi_GeneralizedQuantumSteinslemmasecond_2025}.
We also explored the applicability of our approach to mixed-state conversion and derived upper bounds on asymmetry cost and distillable asymmetry. 
While our analysis in the mixed-state case is not comprehensive, it would pave the way for future investigations.

Our result establishes the QGT as the complete asymmetry measure within RTA, analogous to the role of entanglement entropy in entanglement theory. This finding provides an operational interpretation of the QGT, originally introduced as a metric on quantum state space and commonly employed as a topological indicator. Thus, it reveals a fundamental connection among symmetry breaking, the geometry of quantum states, and condensed matter physics from the viewpoint of quantum resource theories. Given the broad success of entanglement entropy as a theoretical tool for investigating quantum correlations \cite{vidal_entanglement_2003,calabrese_entanglement_2004,ryu_holographic_2006}, adopting the QGT for quantifying asymmetry breaking will offer a powerful approach to exploring deeper structures governed by symmetry and its breaking.

\begin{acknowledgments}
The authors would like to thank Yui Kuramochi for valuable comments on an early draft of this manuscript, and also Iman Marvian for the valuable discussion. 
The authors would also like to thank Mark M. Wilde and Takahiro Sagawa for separately pointing out that the conversion formula can also be rewritten using the max-relative entropy. 
K.Y. acknowledges support from JSPS KAKENHI Grant No. JP24KJ0085.
Y.M. is supported by JSPS KAKENHI Grant No. JP23KJ0421.
T.S. acknowledges support from JST Moonshot R\&D Grant No. JPMJMS2061, JST CREST Grant No. JPMJCR23I4, and MEXT Q-LEAP Grant No. JPMXS0120319794, Japan.
H.T. was supported by MEXT KAKENHI Grant-in-Aid for Transformative
Research Areas B ``Quantum Energy Innovation” Grant Numbers 24H00830 and 24H00831, and JST PRESTO No. JPMJPR2014, JST MOONSHOT No. JPMJMS2061.
\end{acknowledgments}

\widetext

\appendix

\clearpage
\section*{Appendix Contents}

\makeatletter
\let\orig@addcontentsline\addcontentsline
\def\addcontentsline#1#2#3{%
  \def\temp{#1}\def\tocname{toc}%
  \ifx\temp\tocname
    \orig@addcontentsline{apx}{#2}{#3}%
  \else
    \orig@addcontentsline{#1}{#2}{#3}%
  \fi
}
\makeatother

\makeatletter
\@starttoc{apx} 
\makeatother

\clearpage

\section{$G$-covariant channel and conservation laws}\label{app:Stinespring_general}
In this section, we review the connection between $G$-covariance of a channel and conservation laws, established by the covariant Stinespring theorem~\cite{keyl_optimal_1999,marvian_mashhad_symmetry_2012}. 
We denote by $\mathcal{L}(\mathcal{H})$ the set of all linear operators on a Hilbert space $\mathcal{H}$.
\begin{thm}[Theorem~25 in \cite{marvian_mashhad_symmetry_2012}]\label{thm:Stinespring_marvian}
    Let $\mathcal{H}$ be a finite-dimensional Hilbert space equipped with a projective unitary representation $U$ of a compact group $G$, and let $\mathcal{E}:\mathcal{L}(\mathcal{H})\to \mathcal{L}(\mathcal{H})$ be a $G$-covariant channel.
    Then there exists a finite-dimensional Hilbert space $\mathcal{H}_{E}$ with a (non-projective) unitary representation $U_{E}$ of $G$, a $G$-symmetric pure state $\eta_E$ on the system $E$, and a unitary operator $V$ on $\mathcal{H}\otimes \mathcal{H}_{E}$ satisfying 
    \begin{align}
        V\left(U(g)\otimes U_{E} (g)\right)=\left(U(g)\otimes U_{E} (g)\right)V,\quad \forall g\in G,
    \end{align}
    such that
    \begin{align}
        \mathcal{E}(\cdot)=\mathrm{Tr}_{E}\left(V\left(\cdot\otimes \eta_E\right)V^\dag\right).
    \end{align}
\end{thm}

This theorem extends to the setting in which the input and output systems may differ. Although the proof is straightforward, we include it here for completeness:
\begin{cor}
    Let $\mathcal{H}_A,\mathcal{H}_{B}$ be finite-dimensional Hilbert spaces equipped with projective unitary representation $U_A,U_B$ of a compact group $G$, and let $\mathcal{E}:\mathcal{L}(\mathcal{H}_A)\to \mathcal{L}(\mathcal{H}_B)$ be a $G$-covariant channel.
    Then there exists finite-dimensional Hilbert spaces $\mathcal{H}_{E}$ and $\mathcal{H}_{E'}$ with projective unitary representations $U_{E}$ and $U_{E'}$ of $G$ such that $\mathcal{H}_A\otimes\mathcal{H}_E=\mathcal{H}_B\otimes\mathcal{H}_{E'}$, a $G$-symmetric state $\eta_{E}$ on the system $E$, and a unitary operator $V$ on $\mathcal{H}_A\otimes \mathcal{H}_{E}$ satisfying 
    \begin{align}
        V\left(U_A(g)\otimes U_{E}(g) \right)=\left(U_B(g)\otimes U_{E'}(g) \right)V,\quad \forall g\in G,\label{eq:Stinespring_dilation_intertw_rel}
    \end{align}
    such that
    \begin{align}
        \mathcal{E}(\cdot)=\mathrm{Tr}_{E'}\left(V\left(\cdot\otimes \eta_{E}\right)V^\dag\right).\label{eq:Stinespring_dilation_unitary_rep}
    \end{align}
\end{cor}
\begin{proof}
    It is straightforward to check the channel in the form of Eq.~\eqref{eq:Stinespring_dilation_unitary_rep} is $G$-covariant. 
    Suppose that $\mathcal{E}:\mathcal{L}(\mathcal{H}_A)\to \mathcal{L}(\mathcal{H}_B)$ is $G$-covariant. Define a channel $\mathcal{F}:\mathcal{L}(\mathcal{H}_A\otimes\mathcal{H}_B )\to \mathcal{L}(\mathcal{H}_A\otimes\mathcal{H}_B)$ by $\mathcal{F}(\cdot)\coloneqq \eta_A\otimes \left(\mathcal{E}\circ\mathrm{Tr}_B(\cdot)\right)$, where $\eta_A$ is a $G$-symmetric state. Since $\mathcal{F}$ is $G$-symmetric, from Theorem~\ref{thm:Stinespring_marvian}, there exist a finite-dimensional Hilbert space $\mathcal{H}_F$ carrying a unitary representation $U_{F}$ of $G$, a $G$-symmetric state $\eta_F$, and a unitary operator $V$ on $\mathcal{H}_A\otimes\mathcal{H}_B\otimes\mathcal{H}_F$ satisfying
    \begin{align}
        V\left(U(g)_A\otimes U_B(g)\otimes U_F (g)\right)=\left(U(g)_A\otimes U_B(g)\otimes U_F (g)\right)V,\quad \forall g\in G,\label{eq:intw_rel_ABF}
    \end{align}
    such that
    \begin{align}
        \mathcal{F}(\cdot)=\mathrm{Tr}_{F}\left(V(\cdot\otimes \eta_F)V^\dag \right).
    \end{align}
    Fix any $G$-symmetric state $\eta_B$ on the system $B$. Since $\mathcal{E}(\cdot)=\mathrm{Tr}_A\circ\mathcal{F}(\cdot\otimes \eta_B)$, we obtain
    \begin{align}
        \mathcal{E}(\cdot)=\mathrm{Tr}_{E'}\left(V(\cdot\otimes \eta_E)V^\dag\right)
    \end{align}
    where $E\coloneqq BF$, $E'\coloneqq AF$, and $\eta_E\coloneqq \eta_B\otimes \eta_F$. The systems $E$ and $E'$ carry projective unitary representations $U_E=U_B\otimes U_F$ and $U_{E'}=U_A\otimes U_F$, which satisfies Eq.~\eqref{eq:Stinespring_dilation_intertw_rel} due to Eq.~\eqref{eq:intw_rel_ABF}.
\end{proof}

The condition in Eq.~\eqref{eq:Stinespring_dilation_intertw_rel} means that the unitary evolution operator $V$ respects the symmetry described by $G$. When $G$ is a Lie group, by introducing the generators through differentiation of the projective representations (cf. Eq.~\eqref{eq:hermitian_operators_Lie_alg}), this condition implies
\begin{align}
    V\left(X_{A,\mu}+X_{E,\mu}\right)V^\dag=X_{B,\mu}+X_{E',\mu},\quad \forall \mu=1,\cdots,\dim G, 
\end{align}
i.e., additive conservation laws for the generators. In particular, if $G$ is connected, the additive conservation laws are equivalent to the intertwining condition Eq.~\eqref{eq:Stinespring_dilation_intertw_rel}. In this sense, a $G$-covariant channel characterizes the most general evolution of an open system subject to additive conservation laws, provided that no asymmetry is supplied from external systems.

\clearpage
\section{Proof for the main theorem for any projective unitary representations}

In the main text, we provided the sketch of proof of Eq.~\eqref{eq:conversion_rate_formula_forall_g} for (non-projective) unitary representations of a compact Lie group, while the technical proofs are provided in later sections, in Appendices~\ref{app:section_for_converse_part} and~\ref{app:section_for_direct_part}.

In this section, we show Theorem~\ref{thm:conversion_rate_projective_finite_number} from Eq.~\eqref{eq:conversion_rate_formula_forall_g}. 
In Appendix~\ref{app:differentiability_of_reprensetation}, we first review the result in \cite{shitara_iid_2024} that relates the conversion rates for projective unitary representations to that for (non-projective) unitary representations and explain a formula for the conversion rate for projective unitary representations that are continuous but not differentiable. In Appendix~\ref{app:reduction_to_a_finite_number_of_inequalities}, we show that the QGTs at different points of a Lie group are interrelated by a congruence transformation, which plays an essential role in proving Theorem~\ref{thm:conversion_rate_projective_finite_number} from Eq.~\eqref{eq:conversion_rate_formula_forall_g}. 
In Appendix~\ref{app:formula_projective_unitary_rep}, we complete the proof for Eq.~\eqref{eq:rate_formula_compact_connected} in Theorem~\ref{thm:conversion_rate_projective_finite_number}
and Eq.~\eqref{eq:conversion_rate_formula_forall_g} for any projective unitary representations that are differentiable. Furthermore, Eq.~\eqref{eq:conversion_rate_formula_subgrop} is proven Sec.~\ref{sec:symmetry_subgroup_zero_rate}.

\subsection{Differentiability of representation}\label{app:differentiability_of_reprensetation}
A (non-projective) unitary representation is differentiable as long as it is continuous \cite{hall_lie_2015}, whereas this is not always the case for a projective unitary representation. However, by adopting the following result in \cite{shitara_iid_2024}, the conversion rate can be calculated even when projective unitary representation is not differentiable. 
\begin{lem}\label{lem:converison_rate_reduction}
    Let $U,U'$ be projective unitary representations of $G$ on finite-dimensional Hilbert spaces $\mathcal{H}$ and $\mathcal{H}'$. Define maps $\tilde{U}$ and $\tilde{U}'$ by
    \begin{align}
        \tilde{U}(g)\coloneqq \frac{U(g)^{\otimes d}}{\det (U(g))},\quad \tilde{U}'(g)\coloneqq \frac{U(g)^{\prime \otimes d'}}{\det (U(g)')},\label{eq:reduction_projective_to_nonprojective}
    \end{align}
    where $d\coloneqq \dim \mathcal{H}$ and $d'\coloneqq \dim \mathcal{H}'$. Then, $\tilde{U}$ and $\tilde{U}'$ are a (non-projective) unitary representations of $G$ on $\mathcal{H}^{\otimes d}$ and $\mathcal{H}'^{\otimes d'}$, respectively. Furthermore, it holds
    \begin{align}
        \rap(\rho\to \sigma)=\frac{d'}{d}\rap(\rho^{\otimes d}\to \sigma^{\otimes d'}),\label{eq:reduction_to_nonproj_rate}
    \end{align}
    where on the right-hand side, the convertibility is defined with respect to the unitary representations $\tilde{U}$ and $\tilde{U}'$. 
\end{lem}
The proof can be found in \cite{shitara_iid_2024}. Note that this lemma is valid for any groups, including both finite groups and Lie groups, as long as the dimensions of the representation spaces are finite. 

For projective unitary representations $U$ and $U'$ that are continuous, (non-projective) unitary representations $\tilde{U}$ and $\tilde{U}'$ in Eq.~\eqref{eq:reduction_projective_to_nonprojective} are continuous and hence differentiable. Therefore, the conversion rate between pure states is obtained from Eq.~\eqref{eq:reduction_to_nonproj_rate} since its right-hand side can be calculated from Theorem~\ref{thm:conversion_rate_projective_finite_number}. 

Except for this subsection, we only consider representations that are differentiable. In the following subsections, by using Eq.~\eqref{eq:reduction_to_nonproj_rate}, we show that the formula for conversion rate described by QGTs in Theorem~\ref{thm:conversion_rate_projective_finite_number} is valid for projective unitary representations $U$ and $U'$ that are differentiable.

\subsection{Relation among matrix inequalities between QGTs at different group elements for unitary representations}\label{app:reduction_to_a_finite_number_of_inequalities}

Let us prove the relation of the QGT at different points of a compact Lie group. 
\begin{lem}\label{lem:QGT_congruence_tfm}
    Let  $U$ be a unitary representation of a compact Lie group $G$. For each $g\in G$, there exists an invertible real matrix $V(g)$ independent of the representation $U$ such that $\mathcal{Q}^{\mathcal{U}_g(\psi)}=V(g)^\top \mathcal{Q}^{\psi} V(g)$.
\end{lem}

We remark the same argument for the congruence transformation of quantum Fisher information matrices can be found in \cite{gao_sufficient_2024} for connected compact Lie groups.
\begin{proof}[Proof of Lemma~\ref{lem:QGT_congruence_tfm}]
    Let us first review several useful facts on a compact Lie group $G$ and its Lie algebra $\mathfrak{g}$. Let $GL(n,\mathbb{C})$ denote the general linear group. A closed subgroup of $GL(n,\mathbb{C})$ is referred to as a closed linear group. Any compact Lie group $G$ is isomorphic to a closed linear group (see, e.g., Corollary~4.22 in \cite{knapp_lie_2002}). For $g\in G$, we define a linear map $\mathrm{Ad}_g:\mathfrak{g} \to \mathfrak{g}$ by $\mathrm{Ad}_g(X)\coloneqq g X g^{-1}$. The map $\mathrm{Ad}_g$ is an invertible linear transformation of $\mathfrak{g}$ (see, e.g., Proposition~3.33 in~\cite{hall_lie_2015}).

    Now, let us fix an arbitrary basis $\{B_\mu\}_{\mu=1}^{\dim G}$ of $\mathfrak{g}$. For any $g\in G$, the set $\{g^{-1} B_\mu g\}_{\mu=1}^{\dim G}$ is also a basis of $\mathfrak{g}$. Therefore, we can introduce an invertible real matrix $V(g)$ by 
    \begin{align}
     \mathrm{Ad}_{g^{-1}}(B_\mu) \eqqcolon\sum_{\nu=1}^{\dim G}V(g)_{\nu\mu }B_\nu. \label{eq:vector_tfm}
    \end{align}
    Since $\mathrm{Ad}_g$ is a linear invertible transformation, the matrix $V(g)$ is invertible, is determined by the Lie algebraic structure, and is independent of the representation of $G$. 

Let us fix a local coordinate of $G$ in a neighborhood of the identity $e \in G$, which parametrizes the elements as $ g(\lambda)$ such that $g(0)=e$. In this coordinate system, the QGT for a pure state $\psi=\ket{\psi}\bra{\psi}$ is given by
\begin{align}
    &\mathcal{Q}^{\psi}_{\mu\nu}=\braket{\psi|X_\mu X_\nu |\psi}-\braket{\psi|X_\mu|\psi}\braket{\psi|X_\nu|\psi},
\end{align}
where $X_\mu \coloneqq \left.-\ii\frac{\partial}{\partial \lambda^\mu}U(g(\lambda))\right|_{\lambda=0}$. Similarly, the QGT for $\mathcal{U}_g(\psi)=\ket{\psi(g)}\bra{\psi(g)}$, where $\ket{\psi(g)}\coloneqq U(g)\ket{\psi}$, is given by
\begin{align}
    \mathcal{Q}^{\mathcal{U}_g(\psi)}_{\mu\nu}&=\braket{\psi(g)|X_\mu(I-\ket{\psi(g)}\bra{\psi(g)})X_\nu |\psi(g)}\\
    &=\braket{\psi| X_\mu (g)X_\nu (g)|\psi }-\braket{\psi|X_\mu (g)|\psi }\braket{\psi|X_\nu (g)|\psi },
\end{align}
where $X_\mu(g)\coloneqq U(g)^\dag X_\mu U(g)$.
By using the matrix defined in Eq.~\eqref{eq:vector_tfm}, we get
\begin{align}
    U(g)^\dag X_\mu U(g)=\sum_{\nu=1}^{\dim G}V(g)_{\nu\mu }X_\nu.\label{eq:basis_change_Lie_alg}
\end{align}
Therefore, $\mathcal{Q}^{\mathcal{U}_g(\psi)}=V(g)^\top \mathcal{Q}^{\psi} V(g)$.
\end{proof}

As an immediate corollary, we find that a matrix inequality between QGTs holds for any point of a compact Lie group $G$ if it holds at a point in $G$. 
\begin{cor}\label{cor:matric_inequality_connected}
    Let $U$ and $U'$ be a unitary representation of a compact Lie group $G$. 
    If a matrix inequality $\mathcal{Q}^{\mathcal{U}_g(\psi)}\geq r\mathcal{Q}^{\mathcal{U}'_g(\phi)}$ holds for a point $g\in G$, then it also holds for any $g\in G$. 
\end{cor}
\begin{proof}
    Since the reversible matrix $V(g)$ in Lemma~\ref{lem:QGT_congruence_tfm} is independent of unitary representation, we immediately get
    \begin{align}
    &\forall g\in G,\quad  \mathcal{Q}^{\mathcal{U}_g(\psi)}\geq r\mathcal{Q}^{\mathcal{U}'_g(\phi)}\nonumber\\
    &\Longleftrightarrow \exists g\in G, \quad  \mathcal{Q}^{\mathcal{U}_g(\psi)}\geq r\mathcal{Q}^{\mathcal{U}'_g(\phi)}.\label{eq:equivalence_matrix_ineq_different_point}
    \end{align}
\end{proof}
In particular, in Eq.~\eqref{eq:equivalence_matrix_ineq_different_point}, we can always take $e\in G$ as a representative point. Thus, we get
\begin{align}
    \forall g\in G,\quad  \mathcal{Q}^{\mathcal{U}_g(\psi)}\geq r\mathcal{Q}^{\mathcal{U}'_g(\phi)}
    &\Longleftrightarrow  \mathcal{Q}^{\psi}\geq r\mathcal{Q}^{\phi}. \label{eq:equivalence_matrix_ineq_different_point_identity_as_representative}
\end{align}

\subsection{Conversion rates for projective unitary representations and (non-projective) unitary representations}\label{app:formula_projective_unitary_rep}

Using a projective unitary representation that is differentiable, the QGT is obtained as Eq.~\eqref{eq:qgt_covariance_matrix}. This QGT is proportional to the QGT defined with the unitary representation in Eq.~\eqref{eq:reduction_projective_to_nonprojective}: 
\begin{lem}\label{lem:QGT_relation}
    Let $U$ be a projective unitary representation of a Lie group $G$ on a $d$-dimensional Hilbert space $\mathcal{H}$. Define a (non-projective) unitary representation $\tilde{U}$ by Eq.~\eqref{eq:reduction_projective_to_nonprojective}. Then it holds
    \begin{align}
        \forall g\in G,\quad \mathcal{Q}^{\mathcal{U}_g(\psi)}=\frac{1}{d}\mathcal{Q}^{\tilde{\mathcal{U}}_g(\psi)^{\otimes d}}\label{eq:QGT_projective_nonproj}
    \end{align}
\end{lem}
\begin{proof}
    Let us fix a local coordinate $g(\lambda)$ in the neighborhood of the identity $e\in G$ such that $g(0)=e$. From the definition of $\tilde{U}$, we get
    \begin{align}
        &\partial_\mu \tilde{U}(g(\lambda))\biggl|_{\lambda=0}\\
        &=\frac{1}{\det(U(g(\lambda))}\sum_{i=1}^dU(g(\lambda))^{\otimes i-1}\otimes \partial_\mu U(g(\lambda))\otimes U(g(\lambda))^{\otimes d-i}\biggl|_{\lambda=0}+U(g(\lambda))^{\otimes d}\partial_\mu \frac{1}{\det(U(g(\lambda))}\biggl|_{\lambda=0}\\
        &=e^{-\ii \theta}\sum_{i=1}^dI^{\otimes i-1}\otimes \partial_\mu U(g(\lambda))\biggl|_{\lambda=0}\otimes I^{\otimes d-i}+c_\mu e^{\ii \theta d}I^{\otimes d}\qquad \left(U(e)\eqqcolon e^{\ii \theta}I,\, c_\mu\coloneqq \partial_\mu \frac{1}{\det(U(g(\lambda))}\biggl|_{\lambda=0}\right).
    \end{align}
    Note that the second term, which is proportional to the identity matrix, does not contribute to the QGT. Thus, we get
    \begin{align}
        \frac{1}{d}\mathcal{Q}_{\mu\nu}^{\psi^{\otimes d}}
        &=\frac{1}{d}\sum_{i=1}^d \left(\braket{\partial_\mu  U(g(\lambda))\psi|\partial_\nu U(g(\lambda))\psi}-\braket{\partial_\mu U(g(\lambda))\psi|U(g(\lambda))\psi}\braket{U(g(\lambda))\psi|\partial_\nu U(g(\lambda))\psi}\right)\biggl|_{\lambda=0}\nonumber\\
        &=\mathcal{Q}_{\mu\nu}^{\psi}.
    \end{align}
    Since this equality holds for an arbitrary pure state $\psi$, we get
    \begin{align}
        \forall g\in G,\quad \mathcal{Q}^{\mathcal{U}_g(\psi)}=\frac{1}{d}\mathcal{Q}^{\tilde{\mathcal{U}}_g(\psi^{\otimes d})}.
    \end{align}
\end{proof}

Now we are ready to prove Eq.~\eqref{eq:conversion_rate_formula_forall_g} for any projective unitary representations. 
\begin{proof}[Proof of Eq.~\eqref{eq:conversion_rate_formula_forall_g}]
    We consider the case where $\mathrm{Sym}_{G}(\psi) \subset \mathrm{Sym}_G(\phi)$. Note that the case where $\mathrm{Sym}_{G}(\psi)\not \subset \mathrm{Sym}_G(\phi)$ is separately treated in Appendix~\ref{sec:symmetry_subgroup_zero_rate}. 
    
    We now extend Eq.~\eqref{eq:conversion_rate_formula_forall_g} to any unitary representations, given that Eq.~\eqref{eq:conversion_rate_formula_forall_g} is proven for any (non-projective) unitary representations.
    For any projective unitary representations $U$ and $U'$ on $\mathcal{H}$ and $\mathcal{H}'$, we define unitary representations $\tilde{U}$ and $\tilde{U}'$ of $G$ on $\mathcal{H}^{\otimes d}$ and $\mathcal{H}^{\prime \otimes d'}$ by Eq.~\eqref{eq:reduction_projective_to_nonprojective}. From Eq.~\eqref{eq:reduction_to_nonproj_rate}, we have
    \begin{align}
        \rap(\psi\to \phi)=\frac{d'}{d}\rap(\psi^{\otimes d}\to \phi^{\otimes d'}).
    \end{align}
    Since Eq.~\eqref{eq:conversion_rate_formula_forall_g} is already proven for any (non-projective) unitary representations, we get
    \begin{align}
        \rap(\psi^{\otimes d}\to \phi^{\otimes d'})= \sup\{r\geq0 \mid\forall g\in G,\,\mathcal{Q}^{\tilde{\mathcal{U}}_g(\psi^{\otimes d})}\geq r\mathcal{Q}^{\tilde{\mathcal{U}}'_g(\phi^{\otimes d'})}\}.
    \end{align}
    By using Eq.~\eqref{eq:QGT_projective_nonproj}, we get
    \begin{align}
        \sup\{r\geq0 \mid\forall g\in G,\,\mathcal{Q}^{\tilde{\mathcal{U}}_g(\psi^{\otimes d})}\geq r\mathcal{Q}^{\tilde{\mathcal{U}}'_g(\phi^{\otimes d'})}\}=\frac{d}{d'}\sup\{r\geq0 \mid\forall g\in G,\,\mathcal{Q}^{\mathcal{U}_g(\psi)}\geq r\mathcal{Q}^{\mathcal{U}'_g(\phi)}\}.
    \end{align}
    Therefore,
    \begin{align}
        \rap(\psi\to \phi)=\sup\{r\geq0 \mid\forall g\in G,\,\mathcal{Q}^{\mathcal{U}_g(\psi)}\geq r\mathcal{Q}^{\mathcal{U}'_g(\phi)}\}.
    \end{align}
\end{proof}

By using the result in Appendix~\ref{app:reduction_to_a_finite_number_of_inequalities}, we prove Eq.~\eqref{eq:rate_formula_compact_connected} in Theorem~\ref{thm:conversion_rate_projective_finite_number}. 
\begin{proof}[Proof of Eq.~\eqref{eq:rate_formula_compact_connected} in Theorem~\ref{thm:conversion_rate_projective_finite_number}]
    For given projective unitary representations $U$ and $U'$ of $G$, we introduce unitary representations by Eq.~\eqref{eq:reduction_projective_to_nonprojective}. Then it holds, 
    \begin{align}
        \sup\{r\geq0 \mid \forall g\in G,\, \mathcal{Q}^{\mathcal{U}_g(\psi)}\geq r\mathcal{Q}^{\mathcal{U}_g(\phi)}\}
        &\overset{\mathrm{Eq.}\eqref{eq:QGT_projective_nonproj}}{=}\frac{d'}{d}\sup\{r\geq0 \mid \forall g\in G,\, \mathcal{Q}^{\tilde{\mathcal{U}}_g(\psi^{\otimes d})}\geq r\mathcal{Q}^{\tilde{\mathcal{U}}'_g(\phi^{\otimes d'})}\}\\
        &\overset{\mathrm{Eq.}\eqref{eq:equivalence_matrix_ineq_different_point_identity_as_representative}}{=}\frac{d'}{d}\sup\{r\geq0 \mid \mathcal{Q}^{\tilde{\psi^{\otimes d}}}\geq r\mathcal{Q}^{\phi^{\otimes d'}}\}\\
        &\overset{\mathrm{Eq.}\eqref{eq:QGT_projective_nonproj}}{=}\sup\{r\geq0 \mid  \mathcal{Q}^{\psi}\geq r\mathcal{Q}^{\phi}\}.
    \end{align}
    Applying Eq.~\eqref{eq:conversion_rate_formula_forall_g}, we complete the proof of Eq.~\eqref{eq:rate_formula_compact_connected} in Theorem~\ref{thm:conversion_rate_projective_finite_number}.
\end{proof}
\subsection{Symmetry subgroup and the conversion rate}\label{sec:symmetry_subgroup_zero_rate}
We here prove that $\rap(\psi\to\phi)=0$ unless $\mathrm{Sym}_{G}(\psi)\subset \mathrm{Sym}_G(\phi)$. In fact, we can prove this fact not only for pure states but also for any states:
\begin{prop}
   If $\mathrm{Sym}_{G}(\rho)\not \subset \mathrm{Sym}_G(\sigma)$, then $\rap(\rho\to \sigma)=0$. 
\end{prop}
\begin{proof}
     Suppose that $\mathrm{Sym}_{G}(\rho)\not\subset \mathrm{Sym}_G(\sigma)$. Then there exists an element $g_*\in G$ such that $g_*\in \mathrm{Sym}_{G}(\rho)$ and $g_*\notin \mathrm{Sym}_G(\sigma)$. Define $\Delta\coloneqq \mathrm{Fid}(\mathcal{U}_{g_*}(\sigma),\sigma)$. Since $g_*\notin \mathrm{Sym}_G(\sigma)$, we have $\mathrm{Fid}(\mathcal{U}_{g_*}(\sigma),\sigma)<1$ and hence $\mathrm{Fid}(\mathcal{U}_{g_*}(\sigma)^{\otimes N},\sigma^{\otimes N})=\mathrm{Fid}(\mathcal{U}_{g_*}(\sigma),\sigma)^N=\Delta^N\leq \Delta$, implying that $T(\mathcal{U}_{g_*}(\sigma)^{\otimes N},\sigma^{\otimes N})\geq 1-\sqrt{\mathrm{Fid}(\mathcal{U}_{g_*}(\sigma),\sigma)^N}\geq  1-\sqrt{\Delta}$ for any positive integer $N$. 

    Assume that $\{\rho^{\otimes N}\}_N\gconv \{\sigma^{\otimes \floor{rN}}\}_N$. Then for any $\epsilon>0$, there exists a sequence of $G$-covariant channels $\{\mathcal{E}_N\}_N$ such that
    \begin{align}
         T\left(\mathcal{E}_N\left(\rho^{\otimes N}\right),\sigma^{\otimes \floor{rN}}\right)\leq \epsilon
    \end{align}
    for all sufficiently large $N$. Note that this also implies
    \begin{align}
        \forall g\in G,\quad T\left(\mathcal{E}_N\left( \mathcal{U}_g(\rho)^{\otimes N}\right),\mathcal{U}_g'(\sigma)^{\otimes \floor{rN}}\right)\leq \epsilon
    \end{align}
    since $\mathcal{E}_N$ is a $G$-covariant channel. 
If $r>0$, we have $\floor{rN}\geq 1$ for $N\geq 1/r$. However, from the triangle inequality of trace distance, 
\begin{align}
    &T\left(\mathcal{U}_{g_*}(\sigma)^{\otimes \floor{rN}},\sigma^{\otimes \floor{rN}}\right)\\
    &\leq T\left(\mathcal{U}_{g_*}(\sigma)^{\otimes \floor{rN}},\mathcal{E}_N\left( \mathcal{U}_{g_*}(\rho)^{\otimes N}\right)\right)+T\left(\mathcal{E}_N\left( \mathcal{U}_{g_*}(\rho)^{\otimes N}\right),\mathcal{E}_N\left( \rho^{\otimes N}\right)\right)+T\left(\mathcal{E}_N\left( \rho^{\otimes N}\right),\sigma^{\otimes \floor{rN}}\right)\\
    &=T\left(\mathcal{U}_{g_*}(\sigma)^{\otimes \floor{rN}},\mathcal{E}_N\left( \mathcal{U}_{g_*}(\rho)^{\otimes N}\right)\right) +T\left(\mathcal{E}_N\left( \rho^{\otimes N}\right),\sigma^{\otimes \floor{rN}}\right)\\
    &\leq 2\epsilon
\end{align}
holds for all sufficiently large $N$, which contradicts $T(\mathcal{U}_{g_*}(\sigma)^{\otimes N},\sigma^{\otimes N})\geq  1-\sqrt{\Delta}$. Therefore, $r=0$. 
\end{proof}

\newpage
\section{Proof of the converse part}\label{app:section_for_converse_part}
In this section, we provide the detailed proof of the converse part, i.e., Eq.~\eqref{eq:converse_part_statement}. Section~\ref{app:symmetric_monotone_function} provides the discussion of the relation between the metric adjusted skew information and the QGT, and their asymptotic discontinuities, which provides complementary information to Sec.~\ref{sec:monotonicity_QGT_sketch}. Technical details of the proofs of the lemmas are provided in Sec.~\ref{app:sec_proof_all_lemmas_converse}.

\subsection{Metric adjusted skew information and the QGT}\label{app:symmetric_monotone_function}

We here first explain the reason why we adopted operator monotone functions that does not satisfy the symmetry condition in the main text. In the discussion of Petz's monotone function $\braket{\cdot,\cdot}_{f,\rho}$, a symmetry condition $f(t)=t f(t^{-1})$ is often imposed to the monotone function~\cite{petz_monotone_1996}, which is equivalent to $\braket{A,B}_{f,\rho}=\braket{B^\dag ,A^\dag}_{f,\rho}$ for any linear operators $A$ and $B$. However, as we shall see below, such a symmetric monotone metric does not contain the information of the anti-symmetric part of the QGT, which motivated us to introduce a family of operator monotone functions in $f_q$ in Eq.~\eqref{eq:definition_of_f_q}, which does not satisfy the symmetry condition if $q\neq 1/2$. Note that the special case for $q=1/2$ corresponds to $f_{\mathrm{SLD}}(x)= (1+x)/2$, which is associated with the symmetric logarithmic derivative (SLD) metric, playing a central role in quantum estimation theory.

Let $f^{(\mathrm{s})}$ be an operator monotone function satisfying $f^{(\mathrm{s})}(0)>0$ and the symmetry condition $f^{(\mathrm{s})}(t)=tf^{(\mathrm{s})}(t^{-1})$. From the definition of the monotone metric in Eq.~\eqref{eq:monotone_metric_definition}, we find
\begin{align}
    &\frac{f^{(\mathrm{s})}(0)}{2}\|\ii[\psi,O]\|_{f^{(\mathrm{s})},\psi}^2=\gamma^\dag\left(\frac{1}{2}\left(\mathcal{Q}^\psi+\left(\mathcal{Q}^\psi\right)^{*}\right)\right)\gamma,\label{eq:masi_generalization}
\end{align}
for $O\coloneqq \gamma^\dag X$. This means that it captures only the symmetric part of $\mathcal{Q}^\psi$. 

Note that the quantity $\frac{f^{(\mathrm{s})}(0)}{2}\|\ii[\rho,H]\|_{f^{(\mathrm{s})},\rho}^2$ for a Hermitian operator $H$ is called the metric adjusted skew information \cite{hansen_metric_2008}. It is known that the metric adjusted skew information for a pure state $\psi$ is equal to the variance, i.e.,
\begin{align}
    &\frac{f^{(\mathrm{s})}(0)}{2}\|\ii[\psi,H]\|_{f^{(\mathrm{s})},\psi}^2
    =V(\psi,H)\label{eq:masi_pure_states}.
\end{align}
This equation can be viewed as a special case of Eq.~\eqref{eq:masi_generalization}, where $\gamma$ is a real vector. 

Since the metric adjusted skew information is quadratic in $H$, for the same reason as the QGT explained in the main text, it has an asymptotic discontinuity~\cite{gour_measuring_2009,marvian_coherence_2020,marvian_operational_2022}, which requires careful treatment when analyzing its asymptotic rates. In Ref.~\cite{yamaguchi_smooth_2023}, the asymptotic behavior of metric adjusted skew information is studied, which enables to define valid asymmetry monotones obtained from asymptotic rates of the metric adjusted skew information. We here review the results in~\cite{yamaguchi_smooth_2023} since the proof of monotonicity of the QGT under asymptotic conversion are their variants.

For an arbitrary Hermitian operator $H$ we denote its i.i.d. extension by $H_N$, defined as $H_N\coloneqq \sum_{n=1}^NI^{\otimes n-1}\otimes H\otimes I^{\otimes N-n}$. For a sequence of states $\{\sigma_N\}_N$ such that $\limsup_{N\to\infty}T(\sigma_N,\phi^{\otimes N})\leq \epsilon$ for $\epsilon>0$, and an operator monotone function $ f^{(\mathrm{s})}$ satisfying $f^{(\mathrm{s})}(0)>0$ and the symmetry condition $f^{(\mathrm{s})}(t)=t f^{(\mathrm{s})}(t^{-1})$, Lemma~3 in \cite{yamaguchi_smooth_2023} shows that under a certain regularity condition, if $\epsilon>0$ is sufficiently small, it holds
\begin{align}
    &\frac{f^{(\mathrm{s})}(0)}{2}\|\ii [\sigma_N,H_N]\|_{f^{(\mathrm{s})},\sigma_N}^2\geq NV(\phi,H)- Nh(\epsilon)+o(N)\label{eq:asymptotics_norm_symmetric}
\end{align}
for all sufficiently large $N$, where $h$ is a real-valued function independent of $N$, satisfying $\lim_{\epsilon\to 0}h(\epsilon)=0$. The central limit theorem provides an intuitive explanation of this lemma. In the i.i.d. setting, the probability distribution of $\phi^{\otimes N}$ with respect to the eigenbasis of $H_N$ approaches to the normal distribution after properly redefining the random variable. It turns out that the left-hand side of Eq.~\eqref{eq:asymptotics_norm_symmetric} is approximated from below by the variance for a probability distribution modified due the the error between $\sigma_N$ and $\phi^{\otimes N}$. In general, modifying the distribution at values far from the mean can result in a considerable change in the variance. However, since the normal distribution has an exponentially small tail, it is impossible to significantly reduce the variance. In other words, the metric adjusted skew information approximately takes a local minimum around the i.i.d. pure states \cite{yamaguchi_smooth_2023}, expressed by Eq.~\eqref{eq:asymptotics_norm_symmetric}. Note that although the definition of convergence in the state conversion in \cite{yamaguchi_smooth_2023} is slightly different from ours, both yield the equivalent convertibility condition, as shown in Appendix~\ref{app:equivalent_definition_convertibility}.

By modifying Eq.~\eqref{eq:asymptotics_norm_symmetric} to make it applicable to an arbitrary linear operator $O$ instead of Hermitian operator $H$ without assuming $f$ is symmetric, we prove the following lemma, which yields a lower bound on the asymptotic behavior of the right-hand side of Eq.~\eqref{eq:monotonicity_fr_RTA_iid}:
\begin{lem}[Restatement of Eq.~\eqref{eq:asymptotics_norm_asymmetric}]\label{lem:asymptotics_norm_asymmetric}
    Let $f$ be an arbitrary operator monotone function such that $f(0)>0$ and $f(0)= \lim_{\epsilon\to 0^+}f(\epsilon)$. For a linear operator $O$ on a finite-dimensional Hilbert space, we denote its i.i.d. extension by $O_N$. Let $\{\sigma_N\}_N$ be an arbitrary sequence of states such that $\lim_{N\to\infty} T(\sigma_N,\phi^{\otimes N})=0$ for a pure state $\phi$. Then there exists a real-valued function $h$ independent of $N$ satisfying $\lim_{\epsilon\to 0}h(\epsilon)=0$ such that for any sufficiently small parameter $\epsilon>0$, it holds
    \begin{align}
        &f(0)\|\ii [\sigma_N,O_N]\|_{f,\sigma_N}^2\geq NV(\phi,O)-Nh(\epsilon)+o(N)\label{eq:asymptotics_norm_asymmetric_app}
    \end{align}
    for all sufficiently large $N$.
\end{lem}
Notice that the normalization of the left-hand sides of Eqs.~\eqref{eq:asymptotics_norm_symmetric} and \eqref{eq:asymptotics_norm_asymmetric_app} differs factor by $2$. This difference also appeared in Eqs.~\eqref{eq:masi_generalization} and \eqref{eq:norm_limit_QGT}, where the right-hand side of the former is symmetrized while that of the latter is not symmetrized. 

Lemma~\ref{lem:asymptotics_norm_asymmetric} is proven by carefully modifying the proof of Lemma~3 in \cite{yamaguchi_smooth_2023}. Since the proof is involved, the details are presented in Appendix~\ref{sec:asymptotics_norm_asymmetric}.

\subsection{Proof of Lemmas in the converse part}\label{app:sec_proof_all_lemmas_converse}

\subsubsection{Proof of Lemma~\ref{lem:monotonicity_fr}}\label{app:monotonicity_fr_proof}
In the literature, the monotonicity of a monotone metric is often proven under the assumption that $\rho$ and $\mathcal{E}(\rho)$ are invertible. We here explicitly prove the monotonicity of a monotone metric for an operator monotone function $f_q(x)=(1-q)+q x$ for $q\in(0,1)$ without assuming states are invertible for completeness. For this purpose, we here introduce several notations in \cite{hayashi_quantum_2017}. 
Let $\mathcal{M}(\mathcal{H})$ be the set of all linear operators on a Hilbert space $\mathcal{H}$. For an arbitrary probability distribution $p$ on $[0,1]$, we define a linear map $E_{p,\rho}:\mathcal{M}(\mathcal{H})\to \mathcal{M}(\mathcal{H})$ by
\begin{align}
     E_{p,\rho}(A)&\coloneqq\int_0^1\rho^\lambda A \rho^{1-\lambda} p(\dd\lambda).
\end{align}

Note that when $\rho$ is not invertible, the map $E_{p,\rho}$ has a non-trivial kernel. We denote the image of $E_{p,\rho}$ by $\mathcal{M}_{p,\rho}^{(m)}(\mathcal{H})$. For operators $A,B\in\mathcal{M}_{p,\rho}^{(m)}(\mathcal{H})$, we define $\braket{A,B}_{p,\rho}^{(m)}=\mathrm{Tr}\left(C^\dag B\right)$, where $C\in \mathcal{M}(\mathcal{H})$ is an arbitrary operator satisfying $E_{p,\rho}(C)=A$. This inner product is well defined since for operators $C,C'\in \mathcal{M}(\mathcal{H})$ satisfying $E_{p,\rho}(C)=E_{p,\rho}(C')=A$, we get $\mathrm{Tr}(C^\dag B)=\mathrm{Tr}(C^\dag E_{p,\rho}(D))=\mathrm{Tr}((E_{p,\rho}(C))^\dag D)=\mathrm{Tr}((E_{p,\rho}(C'))^\dag D)=\mathrm{Tr}(C^{\prime\dag }B)$, where $D\in \mathcal{M}(\mathcal{H})$ is an operator satisfying $E_{p,\rho}(D)=B$.

As shown in Theorem~6.1 in \cite{hayashi_quantum_2017}, the induced norm $\|A\|^{(m)}_{p,\rho}\coloneqq \sqrt{\braket{A,A}_{p,\rho}^{(m)}}$ is monotonically non-increasing under any CPTP map $\kappa$. That is, 
\begin{align}
   \|A\|_{p,\rho}^{(m)}\geq \|\kappa(A)\|_{p,\kappa(\rho)}^{(m)}\label{eq:monotonicity_m_inner_product}
\end{align}
holds for any operator $A$ such that $A\in \mathcal{M}_{p,\rho}^{(m)}(\mathcal{H})$ and $\kappa(A)\in \mathcal{M}_{p,\kappa(\rho)}^{(m)}(\mathcal{H})$.

In order to relate the above notation with ours, let us now consider a particular probability distribution $p_q\coloneqq (1-q)\delta_1+q \delta_0$ characterized by $q\in (0,1)$, where $\delta_0$ and $\delta_1$ denote the Dirac measure concentrated on $0$ and $1$, respectively. In this case, we have
\begin{align}
    E_{p_{q},\rho}(A)=(1-q)\rho A+q A\rho=\sum_{k,l}m_{f_q}(p_k,p_l)\ket{l}\braket{l|A|k}\bra{k}
\end{align}
where $\rho=\sum_{i=1}^dp_i\ket{i}\bra{i}$ is the eigenvalue decomposition of $\rho$, $m_f(x,y)\coloneqq yf(x/y)$ and $f_q$ is an operator monotone function defined in Eq.~\eqref{eq:definition_of_f_q}.
Then, we have the following:
\begin{lem}\label{lem:domain_of_inner_product}
    For any state $\rho$ and linear operator $B$, $\ii[\rho,B]\in \mathcal{M}_{p_{q},\rho}^{(m)}(\mathcal{H})$. 
\end{lem}
\begin{proof}
    Let $\rho=\sum_{i=1}^dp_i\ket{i}\bra{i}$ be the eigenvalue decomposition of $\rho$. For an operator $A$ defined by
    \begin{align}
        A\coloneqq \sum_{k,l;m_{f_q}(p_k,p_l)>0}\frac{1}{m_{f_q}(p_k,p_l)}\ket{l}\braket{l|\ii [\rho,B]|k }\bra{k},
    \end{align}
    we get
    \begin{align}
        E_{p_{q},\rho}(A)&=\sum_{k,l;(1-q)p_l+qp_k>0}\frac{(1-q)p_l}{(1-q)p_l+qp_k}\ket{l}\braket{l|\ii [\rho,B]|k }\bra{k}+\sum_{k,l;(1-q)p_l+qp_k>0}\frac{qp_k}{(1-q)p_l+qp_k}\ket{l}\braket{l|\ii [\rho,B]|k }\bra{k}\\
        &=\sum_{k,l;(1-q)p_l+qp_k>0}\ket{l}\braket{l|\ii [\rho,B]|k }\bra{k}\\
        &=\ii [\rho,B]-\sum_{k,l;p_l=p_k=0}\ket{l}\braket{l|\ii [\rho,B]|k }\bra{k}\\
        &=\ii [\rho,B],
    \end{align}
    where we have used the fact that $(1-q)p_l+qp_k>0$ unless $p_k=p_l=0$.
\end{proof}

Therefore, for any state $\rho$ and linear operator $B$, we get
\begin{align}
    \left(\|\ii [\rho,B]\|_{p_{q},\rho}^{(m)}\right)^2&=\mathrm{Tr}\left(\left(\sum_{k,l;m_{f_q}(p_k,p_l)>0}\frac{1}{m_{f_q}(p_k,p_l)}\ket{l}\braket{l|\ii [\rho,B]|k }\bra{k}\right)^\dag(\ii [\rho,B])\right)\\
    &=\braket{\ii [\rho,B],\ii[\rho,B]}_{f_q,\rho},
\end{align}
where we used Eq.~\eqref{eq:monotone_metric_definition} in the last equality.
Using this relation, the monotonicity in Eq.~\eqref{eq:monotonicity_m_inner_product} completes the proof of Lemma~\ref{lem:monotonicity_fr}.

\subsubsection{Proof of the equivalence between two alternative definitions of asymptotic convertibility}\label{app:equivalent_definition_convertibility}
\begin{prop}\label{prop:equivalent_definition_convertibility}
    For any sequences of states $\{\rho_N\}_N$ and $\{\sigma_N\}_N$, the following two conditions are equivalent:
    \begin{enumerate}[(i)]
        \item For any $\epsilon>0$, there exists $N_0$ such that for any positive integer $N\geq N_0$, there exists a $G$-covariant channel $\mathcal{E}_N$ satisfying $T(\mathcal{E}_N(\rho_N),\sigma_N)\leq \epsilon$.
        \item There exists a sequence of $G$-covariant channels $\{\mathcal{E}_N\}_N$ such that $\lim_{N\rightarrow \infty}T({\mathcal E}_N(\rho_N),\sigma_{N})=0 $.
    \end{enumerate}
\end{prop}
\begin{proof}
    (ii)$\implies$(i): Any sequence of $G$-covariant channels $\{\mathcal{E}_N\}_N$ such that $ \lim_{N\to\infty}T({\mathcal E}_N(\rho_N),\sigma_{N})=0$ provides an example of the channels satisfying condition (i). 

    (i)$\implies$(ii): Fix a sequence of positive numbers $\{\epsilon_k\}_k$ such that $\lim_{k\to\infty}\epsilon_k=0$. Assume the condition (i) is true. Then for each $\epsilon_k>0$, $\exists N_0(\epsilon_k)>0$ such that $\forall N>N_0(\epsilon_k)$, $\exists {\mathcal E}_{N,\epsilon_k}:G\text{-cov.} \ {\rm s.t.}\  T({\mathcal E}_N(\rho_N),\sigma_{N})\leq \epsilon_k$. Without loss of generality, we can assume $N_0(\epsilon_{k+1})>N_{0}(\epsilon_k)$. Let us define $G$-covariant channels $\{\mathcal{E}_N\}_N$ by
    \begin{align}
        \mathcal{E}_N\coloneqq \mathcal{E}_{N,\epsilon_k}\text{ if }N_0(\epsilon_k)\leq N<N_0(\epsilon_{k+1})
    \end{align}
    for $N\geq N(\epsilon_1)$. For $N<N(\epsilon_1)$, we arbitrarily fix a $G$-covariant channel $\mathcal{E}_N$. Then, for any $N\geq N_0(\epsilon_k)$, we have $T(\mathcal{E}_N(\rho_N),\sigma_N)\leq \epsilon_k$. Thus, we get $\lim_{N\to\infty}T(\mathcal{E}_N(\rho_N),\sigma_N)\leq \epsilon_k$. In the limit of $k\to\infty$, we find $\lim_{N\to\infty}T(\mathcal{E}_N(\rho_N),\sigma_N)=0$ since $\{\epsilon_k\}_k$ is assumed to satisfy $\lim_{k\to\infty}\epsilon_k=0$.
\end{proof}

For example, in studies on convertibility in RTA, condition~(i) is adopted in \cite{yamaguchi_beyond_2023,yamaguchi_smooth_2023}, while condition~(ii) is employed in \cite{marvian_coherence_2020,marvian_operational_2022,shitara_iid_2024} and the present paper. Note that such an equivalence is not specific to the convertibility via $G$-covariant channels, and a similar argument can be found, e.g., in \cite{ogawa_strong_2000}.

\subsubsection{Proof of Lemma~\ref{lem:asymptotics_norm_asymmetric}}\label{sec:asymptotics_norm_asymmetric}

Let us first prove a lemma relating an inner product on a (mixed) state to a covariance matrix of its eigenstate with the largest eigenvalue when the state is close to a pure state.  
\begin{lem}[Modification of Corollary~1 in \cite{marvian_coherence_2020} and Lemma~9 in \cite{yamaguchi_smooth_2023}]\label{lem:largest_ev}
    Let $\sigma$ be an arbitrary state and $\phi$ be a pure state. Suppose that their infidelity by $\delta\coloneqq 1-\braket{\phi|\sigma|\phi}$ satisfies $\delta<1/2$. For the eigenvector $\ket{\Phi}$ with the largest eigenvalue of $\sigma$, it holds
    \begin{align}
        |\braket{\phi|\Phi}|^2\geq 1-2\delta.
    \end{align}
    Let $f$ be an operator monotone function such that $f(0)>0$. Then for any linear operator $O$, it holds
    \begin{align}
    &\braket{\ii [\sigma ,O],\ii [\sigma,O]}_{f,\sigma}\geq \frac{1}{f\left(\frac{\delta}{1-\delta}\right)}(1-2\delta)^2\mathrm{Tr}\left(\Phi O (I-\Phi)O^\dag \right),\label{eq:O_inner_prod_lower_bound}
    \end{align}
    where $\Phi=\ket{\Phi}\bra{\Phi}$.
\end{lem}
The first half is proven in Corollary~1 in \cite{marvian_coherence_2020}. 
The second half is a modification of Corollary~1 in \cite{marvian_coherence_2020} and Lemma~9 in \cite{yamaguchi_smooth_2023}. In Corollary~1 in \cite{marvian_coherence_2020} and Lemma~9 in \cite{yamaguchi_smooth_2023}, a similar inequality was proven for Hermitian operators $O$ and monotone functions $f$ satisfying the symmetry condition $f(t)=tf(t^{-1})$. Equation~\eqref{eq:O_inner_prod_lower_bound} differs by a factor of 2 from those in \cite{marvian_coherence_2020,yamaguchi_smooth_2023}, and it is applicable to monotone functions without assuming the symmetry condition. 
\begin{proof}[Proof of Lemma~\ref{lem:largest_ev}]
    Let $p$ be the largest eigenvalue of $\sigma$. The eigenvalue decomposition of $\sigma$ is written as
    \begin{align}
        \sigma=\sum_{k=0}^{d-1}\lambda_k\ket{k}\bra{k}=p\ket{\Phi}\bra{\Phi}+\sum_{k=1}^{d-1}\lambda_k\ket{k}\bra{k},
    \end{align}
    where $\lambda_0\coloneqq p$ and $\ket{0}\coloneqq \ket{\Phi}$. Since $p$ is assumed to be the largest eigenvalue, we have
    \begin{align}
        1-\delta=p|\braket{\phi|\Phi}|^2+\sum_{k=1}^{d-1}\lambda_k|\braket{\phi|k}|^2\leq p\sum_{k=0}^{d-1}|\braket{\phi|k}|^2=p.
    \end{align}
    Therefore, we get
    \begin{align}
        1-\delta&=\braket{\phi|\sigma|\phi}\\
        &=p |\braket{\phi|\Phi}|^2+\sum_{k=1}^{d-1}\lambda_k|\braket{\phi|k}|^2\\
        &\leq p |\braket{\phi|\Phi}|^2+(1-p) \sum_{k=1}^{d-1}|\braket{\phi|k}|^2\\
        &\leq p |\braket{\phi|\Phi}|^2  +\delta (1-|\braket{\phi|\Phi}|^2),
    \end{align}
    meaning that 
    \begin{align}
        |\braket{\phi|\Phi}|^2\geq \frac{1-2\delta}{p-\delta}\geq 1-2\delta,
    \end{align}
    where we have used $0<1-2\delta\leq p-\delta\leq 1$. 
    
    The second half of this lemma is a modification of Corollary~1 in \cite{marvian_coherence_2020} and Lemma~9 in \cite{yamaguchi_smooth_2023}. For any operator monotone function $f$, we have
    \begin{align}
        \braket{\ii [\sigma,O],\ii [\sigma,O]}_{f,\sigma}
        &=\sum_{k,l;\lambda_lf(\lambda_k/\lambda_l)>0}\frac{1}{\lambda_lf(\lambda_k/\lambda_l)}\braket{k|[\sigma,O^\dag]|l}\braket{l|[O,\sigma]|k}\\
        &=\sum_{k,l;\lambda_lf(\lambda_k/\lambda_l)>0}\frac{(\lambda_l-\lambda_k)^2}{\lambda_lf(\lambda_k/\lambda_l)}|\braket{l|O|k}|^2+\sum_{k,l;\lambda_lf(\lambda_k/\lambda_l)>0,l\neq 0}\frac{(\lambda_l-\lambda_k)^2}{\lambda_lf(\lambda_k/\lambda_l)}|\braket{l|O|k}|^2\\
        &\geq \sum_{k;pf(\lambda_k/p)>0}\frac{(p-\lambda_k)^2}{pf(\lambda_k/p)}|\braket{\Phi|O|k}|^2\\
        &=
        \sum_{k=1}^{d-1}\frac{(p-\lambda_k)^2}{pf(\lambda_k/p)}|\braket{\Phi|O|k}|^2,
    \end{align}
    where we have used $pf(\lambda_k/p)\geq pf(0)>0$ for all $k=0,\cdots,d-1$ in the last equality.
    Since $p-\lambda_k\geq 1-2\delta$, we have $(p-\lambda_k)^2\geq (1-2\delta)^2$ if $\delta\leq 1/2$. Moreover, since $\lambda_k/p\leq \delta/(1-\delta)$ and $p\leq 1$, we get $p f(\lambda_k/p)\leq f(\lambda_k/p)\leq f(\delta/(1-\delta))$ for any monotonic function $f$. Therefore,
    \begin{align}
        \braket{\ii [\sigma,O],\ii [\sigma,O]}_{f,\sigma}
        &\geq  \sum_{k=1}^{d-1}\frac{(p-\lambda_k)^2}{pf(\lambda_k/p)}|\braket{\Phi|O|k}|^2\\
        &\geq \frac{1}{f\left(\frac{\delta}{1-\delta}\right)}(1-2\delta)^2\sum_{k=1}^{d-1}|\braket{\Phi|O|k}|^2\\
        &=\frac{1}{f\left(\frac{\delta}{1-\delta}\right)}(1-2\delta)^2\braket{\Phi|O(I-\Phi)O^\dag |\Phi}.
    \end{align}
\end{proof}

In the asymptotic conversion scenario, the initial i.i.d. pure state is transformed into a state close to i.i.d. copies of the target pure state. In order to further bound the right hand side of Eq.~\eqref{eq:O_inner_prod_lower_bound}, let us consider any linear operator $O$ on $\mathcal{H}$ and its i.i.d. extension $O_N$ defined by
\begin{align}
        O_N&\coloneqq \sum_{n=1}^N O^{(n)},\quad O^{(n)}\coloneqq I^{\otimes n-1}\otimes O\otimes I^{\otimes N-n}.\label{eq:iid_extension_O}
\end{align}
By decomposing $O$ into Hermitian and anti-Hermitian parts as
\begin{align}
        O=A+\ii B,\quad A\coloneqq \frac{O+O^\dag}{2},\quad B\coloneqq \frac{O-O^\dag}{2\ii},
\end{align}
we get
\begin{align}
    O_N&=A_N+\ii B_N,\\
    A_N&\coloneqq \sum_{n=1}^NA^{(n)},\quad A^{(n)}\coloneqq I^{\otimes n-1}\otimes A\otimes I^{\otimes N-n},\label{eq:iid_extension_A}\\
    B_N&\coloneqq \sum_{n=1}^NB^{(n)},\quad B^{(n)}\coloneqq I^{\otimes n-1}\otimes B\otimes I^{\otimes N-n}.\label{eq:iid_extension_B}
\end{align}
For any pure state $\Phi_N$, we find
\begin{align}
    \Braket{\Phi_N|O_N(I-\Phi_N)O_N^\dag|\Phi_N}
    &=V(\Phi_N,A_N)+V(\Phi_N,B_N)-\ii \braket{\Phi_N|[A_N,B_N]|\Phi_N},\label{eq:var_var_commutator}
\end{align}
where $V$ denotes the variance $V(\Phi_N,A_N)\coloneqq \braket{\Phi_N|A_N^2|\Phi_N}-\braket{\Phi_N|A_N|\Phi_N}^2$.
We derive a lower bound for each contribution in Eq.~\eqref{eq:var_var_commutator} under the assumption that $\Phi_N$ is sufficiently close to an i.i.d. copies of a pure state $\phi^{\otimes N}$.

A bound for the last term in Eq.~\eqref{eq:var_var_commutator} can be easily obtained as follows:
\begin{lem}\label{lem:antisym_lower_bound}
    Let $A,B$ be arbitrary Hermitian operators. We denote their i.i.d. extensions by $A_N$ and $B_N$ defined in Eqs.~\eqref{eq:iid_extension_A} and \eqref{eq:iid_extension_B}. For arbitrary state $\rho_N$ and $\sigma_N$, it holds 
    \begin{align}
        \left|\mathrm{Tr}\left(\ii [A_N,B_N]\rho_N\right)-\mathrm{Tr}\left(\ii [A_N,B_N]\sigma_N\right)\right|
        &\leq N\|\ii [A,B]\|_\infty\|\rho_N-\sigma_N\|_1,
    \end{align}
    where $\|\cdot\|_\infty$ denotes the operator norm.
\end{lem}
\begin{proof}
    Since
    \begin{align}
    \ii [A_N,B_N]=\ii \sum_{n,n'=1}^N[A^{(n)},B^{(n')}]
    &=\ii \sum_{n=1}^N[A^{(n)},B^{(n)}]
    \end{align}
    we get
    \begin{align}
    \| \ii [A_N,B_N]\|_\infty\leq \sum_{n=1}^N\|\ii [A^{(n)},B^{(n)}]\|_\infty=N\|\ii[A,B]\|_\infty.
    \end{align}
    Therefore,
    \begin{align}
    \left|\mathrm{Tr}\left(\ii [A_N,B_N]\rho_N\right)-\mathrm{Tr}\left(\ii [A_N,B_N]\sigma_N\right)\right|
    &=\left|\mathrm{Tr}\left(\ii [A_N,B_N](\rho_N-\sigma_N)\right)\right|\\
    &\leq \| \ii [A_N,B_N]\|_\infty\|\rho_N-\sigma_N\|_1\\
    &\leq N\|\ii[A,B]\|_\infty\|\rho_N-\sigma_N\|_1.
\end{align}
\end{proof}
Since $\ii[A_N,B_N]$ is Hermitian, its expectation value is a real number. Applying this lemma for pure states $\Phi_N$ and $\phi^{\otimes N}$, we get
\begin{align}
    -\ii \braket{\Phi_N|[A_N,B_N]|\Phi_N}
    &\geq -\ii \braket{\phi^{\otimes N}|[A_N,B_N]|\phi^{\otimes N}}-N\|\ii[A,B]\|_\infty\|\Phi_N-\phi^{\otimes N}\|_1.\label{eq:comm_l_b}
\end{align}

On the other hand, the variances can change more drastically in general. Nevertheless, for a pure state close to an i.i.d. pure state, we can prove the following bound:
\begin{lem}[Extension of Lemma~10 in \cite{yamaguchi_smooth_2023}]\label{lem:variance_lower_bound}
    Let $A$ be an arbitrary Hermitian operator on a finite-dimensional Hilbert space $\mathcal{H}$. We denote their i.i.d. extensions by $A_N$ defined in Eq.~\eqref{eq:iid_extension_A}. Fix a pure state $\phi\in\mathcal{P}(\mathcal{H})$ and a sufficiently small parameter $\epsilon$. 
    There exists a function $h_1(\epsilon)$ independent of $N$ satisfying $\lim_{\epsilon\to 0}h_1(\epsilon)=0$ such that
    for any sequence of pure states $\{\Phi_N\}_N$ satisfying $\limsup_{N\to\infty}T(\Phi_N,\phi^{\otimes N})<\epsilon$, it holds
    \begin{align}
        V(\Phi_N,A_N)\geq (1- h_1(\epsilon))V(\phi^{\otimes N},A_N)\label{eq:variance_lower_bound}
    \end{align}
    for all sufficiently large $N$. 
\end{lem}

As we have mentioned in Sec.~\ref{sec:monotonicity_QGT_sketch}, the variance per copy can generally change of the order of $O(N\epsilon_N)$ for a state $\Phi_N$ that is $\epsilon_N$-close to $\phi^{\otimes N}$ in trace distance. Even when $\lim_{N\to\infty}\epsilon_N=0$, the variance per copy may change significantly in the limit of $N\to\infty$, and therefore, the variance has an asymptotic discontinuity. Nevertheless, Lemma~\ref{lem:variance_lower_bound} shows that when the state deviates slightly from $\phi^{\otimes N}$, the variance per copy can decrease at most $h_1(\epsilon)V(\phi,A)$, which is independent of $N$.
The central limit theorem provides insight into this mechanism. The probability distribution of $\phi^{\otimes N}$ with respect to the observable $A_N$ is approximated by the normal distribution as $N\to\infty$ after properly redefining the random variable. When this probability distribution is modified due to the change in state, the variance cannot be reduced radically because the tail of the normal distribution is exponentially small. 
\begin{proof}[Proof of Lemma~\ref{lem:variance_lower_bound}]
    When all the eigenvalues of $A$ are degenerate, the statement trivially holds since $V(\Phi_N,A_N)=V(\phi^{\otimes N},A_N)=0$. Below, we assume the difference between the largest and smallest eigenvalues of $A$ is non-vanishing, which we denote $\Delta a>0$. Similarly, if the variance $\sigma^2\coloneqq \mathrm{Tr}(\phi A^2)-(\mathrm{Tr}(\phi A))^2$ vanishes for $\phi$, then the inequality is trivial since $V(\phi^{\otimes N},A_N)=N\sigma^2=0$. Therefore, we also assume that $\sigma>0$. 
    
    Define
    \begin{align}
        \Delta\mu_N&\coloneqq \mathrm{Tr}\left(\phi^{\otimes N}A_N\right)-\mathrm{Tr}\left(\Phi_NA_N\right),\\
        Q_N&\coloneqq \left(A_N-\mathrm{Tr}\left(\Phi_NA_N\right)I\right)^2=(A_N-\mathrm{Tr}\left(\phi^{\otimes N}A_N\right)I+\Delta\mu_NI)^2.
    \end{align}
    Since
    \begin{align}
    &\mathrm{Tr}\left(\phi^{\otimes N}(A_N-\mathrm{Tr}\left(\Phi_NA_N\right)I)^2\right)=\mathrm{Tr}\left(\phi^{\otimes N}(A_N-\mathrm{Tr}\left(\phi^{\otimes N}A_N\right)I)^2\right)+|\Delta\mu_N|^2,
\end{align}
    we have
    \begin{align}
        &V(\Phi_N,A_N)-V(\phi^{\otimes N},A_N)=|\Delta \mu_N|^2-\mathrm{Tr}\left(Q_N(\phi^{\otimes N}-\Phi_N)\right).
    \end{align}

    Here, we derive an upper bound of $\mathrm{Tr}\left(Q_N(\phi^{\otimes N}-\Phi_N)\right)$. Let us denote the eigenvalue decomposition of $A$ by $A=\sum_{j=1}^da_j\ket{j}\bra{j}$. 
    Its i.i.d. extension is written as
    \begin{align}
        A_N=\sum_{j_1=1}^{d-1}\cdots\sum_{j_N=1}^{d-1}s(\bm{j})\ket{\bm{j}}\bra{\bm{j}},\quad \ket{\bm{j}}\coloneqq \ket{j_1\otimes j_2\otimes\cdots\otimes j_N},\quad s(\bm{j})\coloneqq \sum_{n=1}^Na_{j_n},
    \end{align}
    where $\bm{j}\coloneqq (j_1,\cdots,j_N)$. In the following, we abbreviate $\sum_{j_1=1}^{d-1}\cdots\sum_{j_N=1}^{d-1}$ as $\sum_{\bm{j}}$.
    
    We decompose labels for eigenvectors of $A_N$ into two disjoint sets:
    \begin{align}
        \mathcal{A}^{(\mathrm{core})}(x)&\coloneqq \left\{\bm{j}\ \middle|\ \left|\frac{\sqrt{N}}{\sigma}\left( \frac{s(\bm{j})}{N}-\mu\right)\right|\leq x\right\}, \quad \mathcal{A}^{(\mathrm{tail})}(x)\coloneqq \left\{\bm{j}\ \middle|\  \left|\frac{\sqrt{N}}{\sigma}\left( \frac{s(\bm{j})}{N}-\mu\right)\right|> x\right\},
    \end{align}
    where $\mu\coloneqq \mathrm{Tr}(\phi A)$, $\sigma^2\coloneqq \mathrm{Tr}(\phi A^2)-(\mathrm{Tr}(\phi A))^2$, and $x>0$ is a positive parameter that will be fixed later. 
    
    Since $Q_N$ is diagonal in $\{\ket{\bm{j}}\}$ basis, we can also decompose $Q_N$ as follows:
    \begin{align}
        Q_N&=Q_N^{\mathrm{(core)}}(x)+Q_N^{\mathrm{(tail)}}(x),\\
        Q_N^{\mathrm{(core)}}(x)&\coloneqq \sum_{\bm{j}\in\mathcal{A}^{(\mathrm{core})}(x)}(s(\bm{j})-N\mu+\Delta\mu_N)^2\ket{\bm{j}}\bra{\bm{j}},\\
        Q_N^{\mathrm{(tail)}}(x)&\coloneqq \sum_{\bm{j}\in\mathcal{A}^{(\mathrm{tail})}(x)}(s(\bm{j})-N\mu+\Delta\mu_N)^2\ket{\bm{j}}\bra{\bm{j}}.
    \end{align}
    Thus, we get
    \begin{align}
        \mathrm{Tr}\left(Q_N(\phi^{\otimes N}-\Phi_N)\right)&=\mathrm{Tr}\left(Q_N^{(\mathrm{core})}(x)(\phi^{\otimes N}-\Phi_N)\right)+\mathrm{Tr}\left(Q_N^{(\mathrm{tail})}(x)(\phi^{\otimes N}-\Phi_N)\right).
    \end{align}
    Since the largest eigenvalue of $Q_N^{(\mathrm{core})}(x)$ is less than $(x\sqrt{N}\sigma+|\Delta\mu_N|)^2$, the first term is bounded as 
    \begin{align}
        \mathrm{Tr}\left(Q_N^{(\mathrm{core})}(x)(\phi^{\otimes N}-\Phi_N)\right) 
        &\leq\left|\mathrm{Tr}\left(Q_N^{(\mathrm{core})}(x)(\phi^{\otimes N}-\Phi_N)\right)\right|\\
        &\leq (x\sqrt{N}\sigma+|\Delta\mu_N|)^2\|\phi^{\otimes N}-\Phi_N\|_1\\
        &\leq 2(x\sqrt{N\sigma^2 }+|\Delta\mu_N|)^2\epsilon\label{eq:core_bound}
    \end{align}
    for all sufficiently large $N$.

    Since $Q_N^{(\mathrm{tail})}(x)$ is a positive operator, we have
    \begin{align}
        &\mathrm{Tr}\left(Q_N^{(\mathrm{tail})}(x)(\phi^{\otimes N}-\Phi_N)\right)
        \leq \mathrm{Tr}\left(Q_N^{(\mathrm{tail})}(x)\phi^{\otimes N}\right).
    \end{align}
    Let $p$ denote a probability distribution defined by $p(j)\coloneqq |\braket{j|\phi}|^2$. We denote its i.i.d. extension by
    \begin{align}
        p^{\otimes N}(\bm{j})\coloneqq p(j_1)p(j_2)\cdots p(j_N).
    \end{align}
    Then we have
    \begin{align}
        \mathrm{Tr}\left(Q_N^{(\mathrm{tail})}(x)\phi^{\otimes N}\right)&=\sum_{\bm{j}\in\mathcal{A}^{(\mathrm{tail})}(x)}(s(\bm{j})-N\mu+\Delta\mu_N)^2p^{\otimes N}(\bm{j})\\
        &=\sum_{\bm{j}\in\mathcal{A}^{(\mathrm{tail})}(x)}\left((s(\bm{j})-N\mu)^2+2\Delta\mu_N(s(\bm{j})-N\mu)+|\Delta\mu_N|^2\right)p^{\otimes N}(\bm{j})\\
        &\leq \sum_{\bm{j}\in\mathcal{A}^{(\mathrm{tail})}(x)}\left((s(\bm{j})-N\mu)^2+2|\Delta\mu_N||s(\bm{j})-N\mu|+|\Delta\mu_N|^2\right)p^{\otimes N}(\bm{j}).
    \end{align}
    From the Cauchy-Schwarz inequality, the second term is bounded as
    \begin{align}
        \sum_{\bm{j}\in\mathcal{A}^{(\mathrm{tail})}(x)}|s(\bm{j})-N\mu|p^{\otimes N}(\bm{j})\leq \sqrt{\sum_{\bm{j}\in\mathcal{A}^{(\mathrm{tail})}(x)}|s(\bm{j})-N\mu|^2p^{\otimes N}(\bm{j})}\sqrt{\sum_{\bm{j}\in\mathcal{A}^{(\mathrm{tail})}(x)}p^{\otimes N}(\bm{j})}.
    \end{align}
    Therefore, we get
    \begin{align}
        \mathrm{Tr}\left(Q_N^{(\mathrm{tail})}(x)\phi^{\otimes N}\right)\leq \left(\sqrt{\sum_{\bm{j}\in\mathcal{A}^{(\mathrm{tail})}(x)}(s(\bm{j})-N\mu)^2p^{\otimes N}(\bm{j})}+|\Delta\mu_N|\sqrt{\sum_{\bm{j}\in\mathcal{A}^{(\mathrm{tail})}(x)}p^{\otimes N}(\bm{j})}\right)^2.
    \end{align}
    
    For i.i.d. random variables $Y_1,\cdots,Y_N\sim p$, we denote
    \begin{align}
        Z_N\coloneqq \frac{\sqrt{N}}{\sigma}\left(\frac{1}{N}\sum_{i=1}^NY_i -\mu\right)=\frac{1}{\sqrt{N}\sigma}\left(\sum_{i=1}^NY_i -N\mu\right).
    \end{align}
   The tail contribution can be recast into
    \begin{align}
        \sum_{\bm{j}\in\mathcal{A}^{(\mathrm{tail})}(x)}(s(\bm{j})-N\mu)^2p^{\otimes N}(\bm{j})&=N\sigma^2-\sum_{\bm{j}\in\mathcal{A}^{(\mathrm{core})}(x)}(s(\bm{j})-N\mu)^2p^{\otimes N}(\bm{j})\\
        &=N\sigma^2(1-\mathbb{E}_N[f(Z_N)]),
    \end{align}
    where we defined
    \begin{align}
        f(z)\coloneqq 
        \begin{cases}
            z^2 \quad &(\text{if } -x\leq  z\leq x)\\
            0\quad &(\text{otherwise})
        \end{cases}.
    \end{align}
    
     Since the Hilbert space is assumed to be finite-dimensional, the variance $\sigma$ is finite. Therefore, we can apply the central limit theorem to $Z_N$ and find that $Z_N$ converges in distribution to the standard normal $\mathcal{N}(0,1)$ as $N\to\infty$. Since $f(z)$ is bounded and its points of discontinuity $z=\pm x$ have measure zero in $\mathcal{N}(0,1)$, we have
     \begin{align}
         \lim_{N\to\infty }\mathbb{E}[f(Z_N)]=\int\dd z \frac{1}{\sqrt{2\pi}}e^{-\frac{z^2}{2}}f(z)=\int_{-x}^{x}\dd z \frac{1}{\sqrt{2\pi}}e^{-\frac{z^2}{2}}z^2.
     \end{align}
     Therefore, for any $\epsilon>0$, $\mathbb{E}[f(Z_N)]\geq (1-\epsilon)\int_{-x}^{x}\dd z \frac{1}{\sqrt{2\pi}}e^{-\frac{z^2}{2}}z^2 $ holds for all sufficiently large $N$. 
    In this case, it holds 
    \begin{align}
        \sum_{\bm{j}\in\mathcal{A}^{(\mathrm{tail})}(x)}(s(\bm{j})-N\mu)^2p^{\otimes N}(\bm{j})\leq N\sigma^2 \left(1-(1-\epsilon)\int_{-x}^x \dd z \frac{1}{\sqrt{2\pi}}e^{-\frac{z^2}{2}}z^2\right)\eqqcolon N\sigma^2 \tilde{g}_2(\epsilon,x).
    \end{align}
    On the other hand, from the Hoeffding inequality we have
    \begin{align}
        \sum_{\bm{j}\in\mathcal{A}^{(\mathrm{tail})}(x)}p^{\otimes N}(\bm{j})\leq 2 \exp\left(-\frac{2 \sigma^2 x^2}{(\Delta a)^2}\right)\eqqcolon \tilde{g}_0(x),
    \end{align}
    where $\Delta a>0$ denotes the difference between the largest and smallest eigenvalue of $A$. 
    Therefore, we get
    \begin{align}
        \mathrm{Tr}\left(Q_N^{(\mathrm{tail})}(x)\phi^{\otimes N}\right)\leq \left(\sqrt{N\sigma ^2\tilde{g}_2(\epsilon,x)}+|\Delta\mu_N|\sqrt{\tilde{g}_0(x)}\right)^2\label{eq:tail_bound}.
    \end{align}

    From Eqs.~\eqref{eq:core_bound} and \eqref{eq:tail_bound}, we have proven
    \begin{align}
        V(\Phi_N,A_N)-V(\phi^{\otimes N},A_N)&=|\Delta\mu_N|^2-\mathrm{Tr}\left(Q_N(\phi^{\otimes N}-\Phi_N)\right)\\
        &\geq |\Delta\mu_N|^2-\left(2(x\sqrt{N\sigma^2}+|\Delta\mu_N|)^2\epsilon+\left(\sqrt{N\sigma ^2\tilde{g}_2(\epsilon,x)}+|\Delta\mu_N|\sqrt{\tilde{g}_0(x)}\right)^2\right)\label{eq:variance_difference}
    \end{align}
    holds for all sufficiently large $N$.

    Now, let us take the parameter to be $x\coloneqq \epsilon^{-1/3}$ and define
    \begin{align}
        g_2(\epsilon)\coloneqq \tilde{g}_2(\epsilon,\epsilon^{-1/3}),\quad g_0(\epsilon)\coloneqq \tilde{g}_0(\epsilon^{-1/3}),
    \end{align}
    which satisfy $\lim_{\epsilon\to 0} g_2(\epsilon)=\lim_{\epsilon\to 0} g_0(\epsilon)=0$. 

    For sufficiently small $\epsilon>0$ satisfying $1-2\epsilon -g_0(\epsilon)>0$, we have
    \begin{align}
        &|\Delta\mu_N|^2-\left(2(\epsilon^{-1/3}\sqrt{N\sigma^2}+|\Delta\mu_N|)^2\epsilon+\left(\sqrt{N\sigma ^2\tilde{g}_2(\epsilon,\epsilon^{-1/3})}+|\Delta\mu_N|\sqrt{\tilde{g}_0(\epsilon^{-1/3})}\right)^2\right)\\
        &=\left(1-2\epsilon -g_0(\epsilon)\right)\left(|\Delta\mu_N|-\frac{2\epsilon^{2/3} \sqrt{N\sigma^2}+\sqrt{N\sigma^2g_0(\epsilon)g_2(\epsilon)}}{1-2\epsilon -g_0(\epsilon)}\right)^2\\
        &\quad -\left(\frac{(2\epsilon^{2/3} \sqrt{N\sigma^2}+\sqrt{N\sigma^2g_0(\epsilon)g_2(\epsilon)})^2}{1-2\epsilon -g_0(\epsilon)}+2\epsilon^{1/3}N\sigma^2+N\sigma^2 g_2(\epsilon)\right).
    \end{align}
    Therefore, the right-hand side of Eq.~\eqref{eq:variance_difference} is bounded as
    \begin{align}
        |\Delta\mu_N|^2-\left(2(\epsilon^{-1/3}\sqrt{N\sigma^2}+|\Delta\mu_N|)^2\epsilon+\left(\sqrt{N\sigma ^2\tilde{g}_2(\epsilon,\epsilon^{-1/3})}+|\Delta\mu_N|\sqrt{\tilde{g}_0(\epsilon^{-1/3})}\right)^2\right)\geq -N\sigma^2h_1(\epsilon),
    \end{align}
    where
    \begin{align}
        h_1(\epsilon)\coloneqq \frac{(2\epsilon^{2/3} +\sqrt{g_0(\epsilon)g_2(\epsilon)})^2}{1-2\epsilon -g_0(\epsilon)}+2\epsilon^{1/3}+ g_2(\epsilon).
    \end{align}
    We remark that this function satisfies $\lim_{\epsilon\to 0}h_1(\epsilon)=0$ because $\lim_{\epsilon\to0}g_0(\epsilon)=\lim_{\epsilon\to0}g_2(\epsilon)=0$. Since $N\sigma^2=V(\phi^{\otimes N},A_N)$, we find 
    \begin{align}
        V(\Phi_N,A_N)\geq(1-h_1(\epsilon))V(\phi^{\otimes N},A_N)
    \end{align}
    holds for all sufficiently large $N$, as long as $\epsilon>0$ is sufficiently small so that $1-2\epsilon -g_0(\epsilon)>0$.     
\end{proof}

From Lemmas~\ref{lem:antisym_lower_bound} and \ref{lem:variance_lower_bound}, we obtain the following corollary:
\begin{lem}\label{lem:smoothing}
    Let $O$ be an arbitrary linear operator on a finite-dimensional Hilbert space $\mathcal{H}$. We denote by $O_N$ its i.i.d. extension defined in Eq.~\eqref{eq:iid_extension_O}. Fix a pure state $\phi\in \mathcal{P}(\mathcal{H})$ and a sufficiently small parameter $\epsilon$. Then there exists a function $h_2(\epsilon)$ independent of $N$ satisfying $\lim_{\epsilon\to 0}h_2(\epsilon)=0$ such that
    for any sequence of pure states $\{\Phi_N\}_N$ satisfying $\limsup_{N\to\infty}T(\Phi_N,\phi^{\otimes N})<\epsilon$, it holds
    \begin{align}
        &\mathrm{Tr}\left(\Phi_N O_N (I-\Phi_N)O_N^\dag\right)\geq \mathrm{Tr}\left(\phi^{\otimes N} O_N (I-\phi^{\otimes N})O_N^\dag\right)-Nh_2(\epsilon)
    \end{align}
   for all sufficiently large $N$.
\end{lem}
\begin{proof}
    From Lemma~\ref{lem:variance_lower_bound}, we get
    \begin{align}
        V(\Phi_N,A_N)+V(\Phi_N,B_N)&\geq (1- h_{1,A}(\epsilon))V(\phi^{\otimes N},A_N)+(1- h_{1,B}(\epsilon))V(\phi^{\otimes N},B_N))\label{eq:variance_l_b},
    \end{align}
    where $h_{1,A}$ and $h_{1,B}$ are real functions ensured to exist in Lemma~\ref{lem:variance_lower_bound}, each of which in general depends on the difference between the largest and smallest eigenvalues and the variance of $A$ and $B$, respectively. 
    Combining Eqs.~\eqref{eq:comm_l_b} and \eqref{eq:variance_l_b}, we obtain the following bound:
    \begin{align}
        &\mathrm{Tr}\left(\Phi_N O_N (I-\Phi_N)O_N^\dag\right)\geq \mathrm{Tr}\left(\phi^{\otimes N} O_N (I-\phi^{\otimes N})O_N^\dag\right)-Nh_2(\epsilon)
    \end{align}
    for all sufficiently large $N$, where the function
    \begin{align}
        h_2(\epsilon)\coloneqq h_{1,A}(\epsilon)V(\phi,A)+ h_{1,B}(\epsilon)V(\phi,B)+2\|\ii[A,B]\|_\infty\epsilon
    \end{align}
    is independent of $N$ and satisfies $\lim_{\epsilon\to0}h_2(\epsilon)=0$. 
\end{proof}

From Lemma~\ref{lem:largest_ev} and Lemma~\ref{lem:smoothing}, we now prove Lemma~\ref{lem:asymptotics_norm_asymmetric}. 
\begin{lem*}[Restatement of Lemma~\ref{lem:asymptotics_norm_asymmetric}]
    Let $f$ be an arbitrary operator monotone function such that $f(0)>0$ and $f(0)= \lim_{\epsilon\to 0^+}f(\epsilon)$. For a linear operator $O$ on a finite-dimensional Hilbert space, we denote its i.i.d. extension by $O_N$. Let $\{\sigma_N\}_N$ be an arbitrary sequence of states such that $\lim_{N\to\infty} T(\sigma_N,\phi^{\otimes N})=0$ for a pure state $\phi$. Then there exists a real-valued function $h$ independent of $N$ satisfying $\lim_{\epsilon\to 0}h(\epsilon)=0$ such that for any sufficiently small parameter $\epsilon>0$, it holds
    \begin{align}
        &f(0)\|\ii [\sigma_N,O_N]\|_{f,\sigma_N}^2\geq NV(\phi,O)-Nh(\epsilon)+o(N)
    \end{align}
    for all sufficiently large $N$.
\end{lem*}

\begin{proof}
    From Lemma~\ref{lem:largest_ev}, we get
    \begin{align}
        \braket{\ii [\sigma_N ,O],\ii [\sigma_N,O]}_{f,\sigma_N}\geq \frac{1}{f\left(\frac{\delta_N}{1-\delta_N}\right)}(1-2\delta_N)^2\mathrm{Tr}\left(\Phi_NO_N (I-\Phi_N)O_N^\dag \right),
    \end{align}
    where $\delta_N\coloneqq 1-\mathrm{Tr}(\phi^{\otimes N}\sigma_N)$ denotes the infidelity. For a sufficiently small $\epsilon>0$, from Lemma~\ref{lem:smoothing}, we further find
    \begin{align}
        \braket{\ii [\sigma_N ,O],\ii [\sigma_N,O]}_{f,\sigma_N}\geq \frac{1}{f\left(\frac{\delta_N}{1-\delta_N}\right)}(1-2\delta_N)^2\left( N\mathrm{Tr}(\phi O(I-\phi)O^\dag)-Nh_2(\epsilon)\right)
    \end{align}
    Note that from the Fuchs–van de Graaf inequalities, the infidelity satisfies $1-\sqrt{1-\delta_N}\leq \epsilon_N$, where $\epsilon_N\coloneqq T(\sigma_N,\phi^{\otimes N})$. Since $\lim_{N\to\infty}\epsilon_N =0$, $\delta_N\leq 1- (1-\epsilon)^2$ for all sufficiently large $N$.
    Since $c(\delta)\coloneqq  \frac{1}{f\left(\frac{\delta}{1-\delta}\right)}(1-2\delta)^2$ is a monotonically decreasing function of $\delta$ if $\delta\in[0,1/2]$, we get
    \begin{align}
        f(0)\braket{\ii [\sigma_N ,O],\ii [\sigma_N,O]}_{f,\sigma_N}\geq \frac{f(0)}{f\left(\frac{\delta_N}{1-\delta_N}\right)}(1-2\delta_N)^2N\mathrm{Tr}(\phi O(I-\phi)O^\dag)-Nh(\epsilon),
    \end{align}
    where
    \begin{align}
        h(\epsilon)\coloneqq f(0)c(1-(1-\epsilon)^2)h_2(\epsilon)
    \end{align}
    is a real-valued function satisfying $\lim_{\epsilon\to 0}h(\epsilon)=0$. 
    Therefore, we get
    \begin{align}
        f(0)\braket{\ii [\sigma_N ,O],\ii [\sigma_N,O]}_{f,\sigma_N}\geq N\mathrm{Tr}(\phi O(I-\phi)O^\dag)-Nh(\epsilon)+N\mathrm{Tr}(\phi O(I-\phi)O^\dag)\left(\frac{f(0)}{f\left(\frac{\delta_N}{1-\delta_N}\right)}(1-2\delta_N)^2-1\right).
    \end{align}
    The last term is $o(N)$ since
    \begin{align}
        \lim_{N\to\infty}\frac{f(0)}{f\left(\frac{\delta_N}{1-\delta_N}\right)}(1-2\delta_N)^2-1=\frac{f(0)}{f(0)}-1=0,
    \end{align}
    where we have used $\lim_{N\to\infty}\delta_N=0$ and $f(0)=\lim_{\epsilon\to 0^+}f(\epsilon)$. 
\end{proof}

\newpage
\section{Proofs of the direct part}\label{app:section_for_direct_part}
In this section, we present the proof of the direct part. In Sec.~\ref{app:direct_part}, we provide a refined and detailed version of the arguments sketched in Sec.~\ref{sec:optimality_sketch} of the main text. The proofs of the lemmas used in the direct part are given in Sec.~\ref{app:sec_proof_all_lemmas_direct}. 

\subsection{QGT determines asymptotic conversion rate}\label{app:direct_part}
In the previous section, we have proven an upper bound on the conversion rate, given in Eq.~\eqref{eq:converse_part_statement}, which follows from the monotonicity of QGT in asymptotic conversion. In this section, we show that this bound is optimal by proving that if $r>0$ satisfies $\mathcal{Q}^{\mathcal{U}_g(\psi)}\geq r\,\mathcal{Q}^{\mathcal{U}'_g(\phi)}$ for all $g\in G$, then $\{\psi^{\otimes N}\}_N\gconv \{\phi^{\otimes \floor{(r-\delta)N}}\}_N$ holds for arbitrary $\delta\in (0,r)$, which establishes the opposite inequality of Eq.~\eqref{eq:converse_part_statement}:
\begin{prop}[Direct part]\label{prop:direct_part}
    Let $U,U'$ be (non-projective) unitary representations of a compact Lie group $G$ on finite-dimensional Hilbert spaces $\mathcal{H}$ and $\mathcal{H}'$. For any pure states $\psi\in\mathcal{P}(\mathcal{H})$ and $\phi\in\mathcal{P}(\mathcal{H}')$ such that $\mathrm{Sym}_G(\psi)\subset \mathrm{Sym}_G(\phi)$, it holds
    \begin{align}
        &\rap(\psi \to \phi)\geq \sup\{r\geq0 \mid\forall g\in G,\,\mathcal{Q}^{\mathcal{U}_g(\psi)}\geq r\mathcal{Q}^{\mathcal{U}'_g(\phi)}\}.
    \end{align}
\end{prop}
Combining Proposition~\ref{prop:direct_part} with Eq.~\eqref{eq:converse_part_statement}, our main result in Eq.~\eqref{eq:conversion_rate_formula_forall_g} is proven for (non-projective) unitary representations, meaning that QGT determines the asymptotic conversion rate. As mentioned in Sec.~\ref{sec:main_result}, this result can be extended to projective unitary representations by using the method in \cite{shitara_iid_2024}, which is detailed in Appendix~\ref{app:formula_projective_unitary_rep}. 

In the proof of Proposition~\ref{prop:direct_part}, the results in QLAN on pure-state unitary models \cite{girotti_optimal_2024,lahiry_minimax_2024} play a crucial role. This can be intuitively understood through Lemma~\ref{lem:gcov_and_cptp}, which provides an alternative yet equivalent characterization of the convertibility in RTA using quantum channels that are not assumed to be $G$-covariant.

\subsubsection{Quantum local asymptotic normality and QGT}
Here, let us briefly review studies in QLAN \cite{guta_local_2006,kahn_local_2009,girotti_optimal_2024,lahiry_minimax_2024}, which investigate the asymptotic properties of statistical models by associating them with their limit models. This subsection discusses general statistical models, while the next subsection specializes them to the RTA setting.

We review particularly for QLAN on pure-state models \cite{girotti_optimal_2024,lahiry_minimax_2024}. For a pure state $\psi$ and a set of Hermitian operators $\{X_\mu\}_{\mu=1}^m$, we define a pure-state statistical model $\mathcal{U}_{\theta}(\psi)$ with $\mathcal{U}_{\theta}(\cdot)=e^{\ii\theta^\mu X_\mu}(\cdot)e^{-\ii\theta^\nu X_\nu}$ for $\theta\in\mathbb{R}^m$. 
To study its asymptotic property, we consider its i.i.d. copies of the state $(\mathcal{U}_{u/\sqrt{N}}(\psi))^{\otimes N}$, where we introduced a local parameter $u\coloneqq \sqrt{N}\theta$. In the limit of $N\to\infty$, it is shown \cite{girotti_optimal_2024,lahiry_minimax_2024} that this model is interconvertible to the so-called Gaussian shift model with asymptotically vanishing error.

For a more detailed explanation, let us introduce several notations in \cite{girotti_optimal_2024}. For a pure state $\psi$ on a qudit system, we take a unit vector $\ket{\psi}$ such that $\psi=\ket{\psi}\bra{\psi}$. We define a matrix $C:\mathbb{C}^{m}\to \mathbb{C}^{d-1}$ by $(C)_{k\mu}\coloneqq -\ii\braket{k|X_\mu|\psi}$, where $\{\ket{k}\}_{k=1}^{d-1}$ denotes an orthonormal basis for the subspace orthogonal to $\ket{\psi}$. For $z=(z_1,\cdots,z_{d-1})^\top\in\mathbb{C}^{d-1}$, we denote a $(d-1)$-mode coherent state with coherence amplitude $z$ by $\ket{z}\coloneqq \left(\bigotimes_{i=1}^{d-1}e^{z_ia_i^\dag -z_i^*a_i}\ket{0}_i\right)$, where $a_i^\dag,a_i$ and $\ket{0}_i$ denote the creation and annihilation operators and the vacuum state for the $i$th mode. It is proven \cite{lahiry_minimax_2024} that there exist sequences of quantum channels $\{\mathcal{T}_N\}_N$ and $\{\mathcal{S}_N\}_N$ such that 
\begin{align}
    &\lim_{N\to\infty}\sup_{\|u\|<N^\epsilon}T\left(\mathcal{T}_{N}\left(\mathcal{U}_{\frac{u}{\sqrt{N}}}(\psi)^{\otimes N}\right),\ket{Cu}\bra{Cu}\right)\nonumber\\
    &=\lim_{N\to\infty}\sup_{\|u\|<N^\epsilon}T\left(\mathcal{U}_{\frac{u}{\sqrt{N}}}(\psi)^{\otimes N},\mathcal{S}_{N}\left(\ket{Cu}\bra{Cu}\right)\right)\nonumber\\
    &=0\label{eq:QLAN_convertibility}
\end{align}
for $\epsilon\in(0,1/9)$.

In the studies on QLAN, the main purpose is to approximate the i.i.d. statistical models using the Gaussian shift model in the asymptotic limit. Our focus is slightly different, and we aim to investigate the convertibility among i.i.d. statistical models to study the conversion rate in RTA. Our finding here is that the QGT plays a central role in characterizing asymptotic convertibility. 

As shown in Eq.~\eqref{eq:QLAN_convertibility}, $\left(\mathcal{U}_{\frac{u}{\sqrt{N}}}(\psi)\right)^{\otimes N}$ is reversibly asymptotically convertible to the Gaussian-shift model $\ket{Cu}$, implying that the matrix $C$ characterizes the asymptotic properties of the i.i.d. statistical model $\left(\mathcal{U}_{\frac{u}{\sqrt{N}}}(\psi)\right)^{\otimes N}$. Such a matrix $C$ is not uniquely determined since there remains the freedom in the choice of an orthonormal basis $\{\ket{k}\}_{k=1}^{d-1}$ for the subspace orthogonal to $\ket{\psi}$. Adopting another orthonormal basis, the matrix $C$ is transformed to $C\mapsto \tilde{C}\coloneqq VC$, where $V$ denotes a $\mathbb{C}^{d-1}\times \mathbb{C}^{d-1}$ unitary matrix. By using the polar decomposition theorem, for matrices $C$ and $C'$ of the same size, we find that $\tilde{C}=VC$ holds for some unitary matrix $V$ if and only if $C^\dag C=\tilde{C}^\dag \tilde{C}$. This fact implies that the essential characteristic of a statistical model is $C^\dag C$, rather than the matrix $C$ itself. 

Importantly, the $(\mu,\nu)$ element of $C^\dag C$ is given by
\begin{align}
    \left(C^\dag C\right)_{\mu\nu}&=\sum_{k=1}^{d-1}\braket{\psi|X_\mu|k}\braket{k|X_\nu|\psi}\nonumber\\
    &=\braket{\psi|X_\mu(I-\psi)X_\nu|\psi}, 
\end{align}
meaning that $C^\dag C=\mathcal{Q}^{\psi}$, where $\mathcal{Q}^{\psi}$ is the QGT for $\mathcal{U}_{\theta}(\psi)$. Therefore, if two pure-state statistical models have the same QGT, they are asymptotically interconvertible to each other since both of them can be reversibly converted to the same Gaussian shift model, shown in Eq.~\eqref{eq:QLAN_convertibility}. 

As a generalization of the above observation, we prove the following:
\begin{lem}\label{lem:QLAN_asymptotic_conversion}
    For sets of Hermitian operators $ \{X_\mu\}_{\mu=1}^m$ and $\{X_\mu'\}_{\mu=1}^m$, we define $\mathcal{U}_{\theta}(\cdot)\coloneqq e^{\ii \theta^\mu X_\mu}(\cdot)e^{-\ii \theta^\nu X_\nu}$ and $\mathcal{U}'_{\theta}(\cdot)\coloneqq e^{\ii \theta^\mu X_\mu'}(\cdot)e^{-\ii\theta^\nu X_\nu'}$ for $\theta\in\mathbb{R}^m$. For given two pure states $\psi$ and $\phi$, we define pure-state statistical models by $\mathcal{U}_{\theta}(\psi)$ and $\mathcal{U}'_{\theta}(\phi)$ and denote their QGTs by $\mathcal{Q}^\psi$ and $\mathcal{Q}^\phi$, respectively. If $r>0$ satisfies $\mathcal{Q}^\psi\geq r\mathcal{Q}^\phi$, then there exists a sequence of quantum channels $\{\mathcal{E}_N\}_N$ such that 
    \begin{align}
        &\lim_{N\to\infty}\sup_{\|u\|<N^\epsilon}T\left(\mathcal{E}_{N}\left(\mathcal{U}_{\frac{u}{\sqrt{N}}}(\psi)^{\otimes N}\right),\mathcal{U}'_{\frac{u}{\sqrt{N}}}(\phi)^{\otimes \floor{rN}}\right)=0
    \end{align}
    for $\epsilon\in(0,1/9)$.
\end{lem}
In the proof of Lemma~\ref{lem:QLAN_asymptotic_conversion}, the argument is divided into two parts: one focuses on the case where the QGT is conserved, while the other addresses the case where the QGT is reduced. We first prove that the growth rate of i.i.d. copies can be set to one by adjusting the scaling of the parameters. Concretely, $\mathcal{U}'_{\frac{u}{\sqrt{N}}}(\phi)^{\otimes \floor{rN}}$ is shown to be asymptotically interconvertible to $\mathcal{U}'_{\frac{\sqrt{r}u}{\sqrt{N}}}(\phi)^{\otimes N}$. This asymptotic interconversion is achieved by combining the reversible conversion channels in QLAN, where the QGT is conserved. Then, we also prove that $\mathcal{U}_{\frac{u}{\sqrt{N}}}(\psi)^{\otimes N}$ can be asymptotically converted into $\mathcal{U}'_{\frac{\sqrt{r}u}{\sqrt{N}}}(\phi)^{\otimes N}$ if the QGT with respect to $u$ decreases in the sense of matrix inequality. The details of these lemmas and their proofs are provided in Appendix~\ref{sec:QLAN_asymptotic_conversion}.

\subsubsection{Proof of achievability of the conversion rate}\label{sec:achievability_proof}
The process of transforming $\mathcal{U}_g(\psi)^{\otimes N}$ to $\mathcal{U}_g'(\phi)^{\otimes \floor {rN}}$ with asymptotically vanishing error consists of two steps, which we call the estimation and conversion steps. In the estimation step, we obtain a rough estimate $\hat{g}$ of the true value $g$ by consuming a sublinear number of copies of $\mathcal{U}_g(\psi)$ in $N$. In the conversion step, based on the estimate $\hat{g}$, we convert the remaining copies of $\mathcal{U}_g(\psi)$ to i.i.d. copies of $\mathcal{U}_g'(\phi)$ with quantum channels constructed by using Lemma~\ref{lem:QLAN_asymptotic_conversion}. 

For a detailed explanation, let us first introduce notations. Let $U$ and $U'$ be (non-projective) unitary representations of a Lie group $G$ on the input and output systems. Elements in the neighborhood of the identity $e\in G$ can be parametrized as $g(\lambda)=e^{\ii \lambda^\mu A_\mu}$ by using a basis $\{A_\mu\}_{\mu=1}^{\dim G}$ of the Lie algebra $\mathfrak{g}$. Define operators $X_\mu\coloneqq -\ii \frac{\partial}{\partial \lambda^\mu}U(g(\lambda))|_{\lambda=0}$ and $X'_\mu\coloneqq -\ii \frac{\partial}{\partial \lambda^\mu}U'(g(\lambda))|_{\lambda=0}$, which correspond to the images of $A_\mu$ by the Lie algebra representations induced from $U$ and $U'$.
At least locally, these operators satisfy $U(g(\lambda))=e^{\ii\lambda^\mu X_\mu }$ and $U'(g(\lambda))=e^{\ii\lambda^\mu X_\mu' }$. In this case, by slightly abusing our notation of $\mathcal{U}_g$ and $\mathcal{U}'_g$, we also denote $\mathcal{U}_{\theta}(\cdot)\coloneqq e^{\ii \theta^\mu X_\mu}(\cdot)e^{-\ii\theta^\nu X_\nu}$ and $\mathcal{U}_{\theta}'(\cdot)\coloneqq e^{\ii \theta^\mu X_\mu'}(\cdot)e^{-\ii\theta^\nu X_\nu'}$ for $\theta\in\mathbb{R}^{\dim G}$, implicitly indicating $\|\theta\|$ is sufficiently small.

In the estimation step, we obtain a rough estimate $\hat{g}\in G$ of $g \in G$ by performing a measurement on $n$ copies of the system in a state $\mathcal{U}_g(\psi)^{\otimes n}$, where $n\coloneqq N^{1-\epsilon}$ for a fixed $\epsilon\in(0,1/2)$. Let $p(\hat{g}|\mathcal{U}_g(\psi)^{\otimes n})$ denote the probability of obtaining an estimate $\hat{g}$ when the system is in a state $\mathcal{U}_g(\psi)^{\otimes n}$. We say that the estimation is successful within an acceptable error of $\delta>0$ if and only if the estimate $\hat{g}\in G$ is an element of the following set: 
\begin{align}
    G_{\mathrm{succ.}}^{(g,\delta)}\coloneqq \{\hat{g}\in G\mid \exists \theta,\, \mathcal{U}_{\theta}\circ \mathcal{U}_{\hat{g}}(\psi)=\mathcal{U}_g(\psi),\,\|\theta\|<\delta\}.\label{eq:successful_estimation_set}
\end{align}
We quantify the success probability in the worst case by
\begin{align}
    p^{\mathrm{succ.}}(\delta)\coloneqq 
    &\inf_{g\in G}\int_G\dd \mu_G(\hat{g})\, \,p(\hat{g}|\mathcal{U}_g(\psi)^{\otimes n})\chi_{ G_{\mathrm{succ.}}^{(g,\delta)}}(\hat{g}),
 \end{align}
where $\mu_G$ is the Haar measure on the Lie group $G$, and $\chi_{A}$ denotes the indicator function of a set $A$. Note that the above success probability is well-defined since $G_{\mathrm{succ.}}^{(g,\delta)}$ is a measurable set, as shown in Appendix~\ref{sec:measurability}.

We show the following lemma, which guarantees the existence of an estimator that is successful within an acceptable error of $\delta=N^{-1/2+\epsilon}$ whose failure probability asymptotically vanishes:
\begin{lem}\label{lem:reasonable_estimator}
    Let $G$ be a compact Lie group and $\rho$ be an arbitrary state. Fix $\epsilon\in (0,1/2)$. Then there exists an estimator of $g\in G$, which consumes $\mathcal{U}_g(\rho)^{\otimes n}$ with $n=\ceil{N^{1-\epsilon}}$, such that its worst-case success probability satisfies
    \begin{align}
        \lim_{N\to\infty }p^{\mathrm{succ.}}(N^{-1/2+\epsilon})=1. \label{eq:reasonable_estimator}
    \end{align}
\end{lem}
We remark that $\|\theta\|\sim N^{-1/2}$ for the error is the standard scaling in parameter estimation for i.i.d. models \cite{gill_state_2000}, including QLAN studies \cite{guta_local_2006,kahn_local_2009,girotti_optimal_2024,lahiry_minimax_2024}. 
The proof is detailed in Appendix~\ref{sec:reasonable_estimator}.

In the conversion step, we convert states using quantum channels guaranteed by the following lemma:
\begin{lem}\label{lem:Lie_asymptotic_conversion}
    Let $U,U'$ be (non-projective) unitary representations of a compact Lie group $G$ on finite-dimensional Hilbert spaces $\mathcal{H}$ and $\mathcal{H}'$. Let $\psi\in\mathcal{P}(\mathcal{H})$ and $\phi\in\mathcal{P}(\mathcal{H}')$ be pure states. 
    If $r>0$ satisfies $\mathcal{Q}^{\mathcal{U}_g(\psi)}\geq r\mathcal{Q}^{\mathcal{U}_g'(\phi)}$ for all $g\in G$, then, there exists a sequence of quantum channels $\{\mathcal{E}^{(g)}_N\}_N$ such that the conversion error
    \begin{align}
        \delta_N(g,u)\coloneqq T\left(\mathcal{E}^{(g)}_N\left(\mathcal{U}_{\frac{u}{\sqrt{N}}}\left(\mathcal{U}_g(\psi)\right)^{\otimes N}\right),\mathcal{U}'_{\frac{u}{\sqrt{N}}}\left(\mathcal{U}_g'(\phi)\right)^{\otimes \floor{rN}}\right)\label{eq:trace_distance_g_u}
    \end{align}
    satisfies $\lim_{N\to\infty}\sup_{g\in G}\sup_{\|u\|<N^\epsilon}\delta_N(g,u)=0$ for $\epsilon\in(0,1/9)$.
\end{lem}
This lemma is proven from Lemma~\ref{lem:QLAN_asymptotic_conversion} by using properties of a compact Lie group. The proof is detailed in Appendix~\ref{sec:Lie_asymptotic_conversion}.

Proposition~\ref{prop:direct_part} is proven by combining Lemmas~\ref{lem:reasonable_estimator} and \ref{lem:Lie_asymptotic_conversion}. 
\begin{proof}[Proof of Proposition~\ref{prop:direct_part}]
    From the assumption, $r>0$ satisfies $\mathcal{Q}^{\mathcal{U}_g(\psi)} \geq r \mathcal{Q}^{\mathcal{U}_g'(\phi)}$ for all $g\in G$. Let $\hat{g}\in G$ be the estimator ensured to exist in Lemma~\ref{lem:reasonable_estimator} for $\epsilon\in (0,1/9)$. Given the estimate $\hat{g}$, we apply Lemma~\ref{lem:Lie_asymptotic_conversion} to pure states $(\mathcal{U}_{\hat{g}}(\psi))^{\otimes N'}$ and $(\mathcal{U}'_{\hat{g}}(\phi))^{\otimes \floor{rN'}}$ for $N'\coloneqq N-\ceil{N^{1-\epsilon}}$ and denote by $\{\mathcal{E}_{N'}^{(\hat{g})}\}_{N'}$ the sequence of conversion channels. Applying $\mathcal{E}_{N'}^{(\hat{g})}$ after the estimation step, the resulting state is given by
    \begin{align}
        &\mathcal{E}_N(\mathcal{U}_g(\psi)^{\otimes N})\coloneqq \int \dd\mu_G( \hat{g})\, \,p(\hat{g}|\mathcal{U}_g(\psi)^{\otimes \ceil{N^{1-\epsilon}}})\mathcal{E}_{N'}^{(\hat{g})}\left(\mathcal{U}_g(\psi)^{\otimes N'}\right).\label{eq:definition_full_channel}
    \end{align}

    From the assumption $\mathrm{Sym}_G(\psi)\subset \mathrm{Sym}_G(\phi)$, if parameters $\theta\in\mathbb{R}^{\dim G}$ satisfy $\mathcal{U}_{\theta}\circ \mathcal{U}_{\hat{g}}(\psi)=\mathcal{U}_g(\psi)$ for some $g,\hat{g}\in G$, then it also holds $\mathcal{U}'_{\theta}\circ \mathcal{U}'_{\hat{g}}(\phi)=\mathcal{U}_{g}'(\phi)$. Therefore, for $\hat{g}\in G_{\mathrm{succ.}}^{(g,\delta)}$ with $\delta>0$, we have
    \begin{align}
        &\mathrm{Fid}\left(\mathcal{E}_{N'}^{(\hat{g})}\left(\mathcal{U}_g(\psi)^{\otimes N'}\right),\mathcal{U}_{g}'(\phi)^{\otimes \floor{rN'}}\right)=\mathrm{Fid}\left(\mathcal{E}_{N'}^{(\hat{g})}\left(\mathcal{U}_{\theta}(\mathcal{U}_{\hat{g}}(\psi))^{\otimes N'}\right),\mathcal{U}_{\theta}'(\mathcal{U}_{\hat{g}}'(\phi))^{\otimes \floor{rN'}}\right),\label{eq:estimator_and_fidelity}
    \end{align}
    where $\mathrm{Fid}$ denotes the fidelity defined by $\mathrm{Fid}(\rho,\sigma)\coloneqq \left(\mathrm{Tr}\left(\sqrt{\sqrt{\rho}\sigma\sqrt{\rho}}\right)\right)^2$ for states $\rho$ and $\sigma$. 
    Since $N^{-1/2+\epsilon}<(N')^{-1/2+\epsilon}$, we have
    \begin{align}
        f_{N}
        &\coloneqq \inf_{g\in G}\inf_{\hat{g}\in G_{\mathrm{succ.}}^{(g,N^{-1/2+\epsilon})}}\mathrm{Fid}\left(\mathcal{E}_{N'}^{(\hat{g})}\left(\mathcal{U}_g(\psi)^{\otimes N'}\right),\mathcal{U}_{g}'(\phi)^{\otimes \floor{rN'}}\right)\label{eq:convergence_fidelity_succ_set_definition}\\
        &\geq \inf_{g\in G}\inf_{\hat{g}\in G_{\mathrm{succ.}}^{(g,(N')^{-1/2+\epsilon})}}\mathrm{Fid}\left(\mathcal{E}_{N'}^{(\hat{g})}\left(\mathcal{U}_g(\psi)^{\otimes N'}\right),\mathcal{U}_{g}'(\phi)^{\otimes \floor{rN'}}\right)\nonumber\\
        &=\inf_{\hat{g}\in G}\inf_{\|u\|<(N')^\epsilon}\mathrm{Fid}\left(\mathcal{E}_{N'}^{(\hat{g})}\left(\mathcal{U}_{\frac{u}{\sqrt{N'}}}(\mathcal{U}_{\hat{g}}(\psi))^{\otimes N'}\right),\mathcal{U}_{\frac{u}{\sqrt{N'}}}'(\mathcal{U}'_{\hat{g}}(\phi))^{\otimes \floor{rN'}}\right)\nonumber\\
        &\geq \left(1-\sup_{\hat{g}\in G}\sup_{\|u\|<(N')^\epsilon} \delta_{N'}(\hat{g},u)\right)^2 \xrightarrow{N \to \infty} 1,\label{eq:convergence_fidelity_succ_set}
    \end{align}
    where in the last line, we used the Fuchs–van de Graaf inequalities to relate the fidelity to the trance distance $\delta_N(g,u)$ defined in Eq.~\eqref{eq:trace_distance_g_u} in Lemma~\ref{lem:Lie_asymptotic_conversion}.

    From Eqs.~\eqref{eq:reasonable_estimator}, \eqref{eq:definition_full_channel}, and \eqref{eq:convergence_fidelity_succ_set}, we get
    \begin{align}
        \inf_{g\in G}\mathrm{Fid}\left(\mathcal{E}_N(\mathcal{U}_g(\psi)^{\otimes N}),\mathcal{U}_g'(\phi)^{\otimes \floor{rN'}}\right)
        &= \inf_{g\in G}\int_{ G} \dd\mu_G( \hat{g})\, \,p(\hat{g}|\mathcal{U}_g(\psi)^{\otimes \ceil{N^{1-\epsilon}}})\nonumber\\
        &\quad\times \mathrm{Fid}\left(\mathcal{E}_{N'}^{(\hat{g})}\left(\mathcal{U}_g(\psi)^{\otimes N'}\right),\mathcal{U}_g'(\phi)^{\otimes \floor{rN'}}\right)\nonumber\\
        &\geq\inf_{g\in G} \int\dd \mu_G(\hat{g})\, \,p(\hat{g}|\mathcal{U}_g(\psi)^{\otimes \ceil{N^{1-\epsilon}}})\chi_{ G_{\mathrm{succ.}}^{(g,N^{-1/2+\epsilon})}}(\hat{g})\nonumber\\
        &\quad \times \mathrm{Fid}\left(\mathcal{E}_{N'}^{(\hat{g})}\left(\mathcal{U}_g(\psi)^{\otimes N'}\right),\mathcal{U}_g'(\phi)^{\otimes \floor{rN'}}\right)\nonumber\\
        &\geq p^{\mathrm{succ.}}(N^{-1/2+\epsilon})\times f_{N}\label{eq:estimate_and_conversion_error_fidelity}\\
        &\xrightarrow{N \to \infty} 1,
    \end{align}
    implying that 
    \begin{align}
        \lim_{N\to\infty}\sup_{g\in G}T\left(\mathcal{E}_N(\mathcal{U}_g(\psi)^{\otimes N}),\mathcal{U}'_g(\phi)^{\otimes \floor{rN'}}\right)=0.
    \end{align}

    To complete the proof, we introduce $\delta\in(0,r)$. 
    Then for all sufficiently large $N$, it holds $\floor{rN'}>\floor{(r-\delta)N}$. Denoting by $\{\Lambda_N\}_N$ the channel discarding $\floor{rN'}-\floor{(r-\delta)N}$ copies of the output system, we get 
    \begin{align}
        \lim_{N\to\infty}\sup_{g\in G}T\left(\Lambda_N\circ\mathcal{E}_N(\mathcal{U}_g(\psi)^{\otimes N}),\mathcal{U}'_g(\phi)^{\otimes \floor{(r-\delta)N}}\right)=0.\label{eq:direct_part_discard}
    \end{align}
    From Lemma~\ref{lem:gcov_and_cptp}, this is equivalent to $\{\psi^{\otimes N}\}_N\gconv \{\phi^{\otimes \floor{(r-\delta)N}}\}_N$, which completes the proof of Proposition~\ref{prop:direct_part}.
\end{proof}

The above proof of the direct part provides a clear understanding of the reason why the asymptotic conversion rate diverges for finite groups. Since there is only a finite number of different elements in $\{\mathcal{U}_g(\psi)\mid g\in G\}$, it is possible to identify the state $\mathcal{U}_g(\psi)$ with an exponentially small failure probability using a state tomography protocol on $\mathcal{U}_g(\psi)^{\otimes N}$. Let $\hat{g}\in G$ denote the estimate of the true value $g$ obtained by this tomography protocol. Depending on the estimate $\hat{g}$, we can prepare an arbitrary number of copies of $\mathcal{U}'_{\hat{g}}(\phi)$. The condition $\mathrm{Sym}_G(\psi)\subset \mathrm{Sym}_G(\phi)$ implies that if $\mathcal{U}_{\hat{g}}(\psi)=\mathcal{U}_g(\psi)$ is satisfied, then $\mathcal{U}'_{\hat{g}}(\phi)=\mathcal{U}'_g(\phi)$ also holds. Therefore, the asymptotic conversion rate diverges since it is possible to create an arbitrary number of copies of $\mathcal{U}'_{g}(\phi)$ from $\mathcal{U}_g(\psi)^{\otimes N}$ with asymptotically vanishing failure probability. We emphasize that the exact asymptotic conversion rate in \cite{shitara_iid_2024} is \textit{not} covered by this analysis because there always remains an error in the tomography process.

\subsection{Proof of Lemmas in the direct part}\label{app:sec_proof_all_lemmas_direct}

\subsubsection{Proof of Lemma~\ref{lem:gcov_and_cptp}}\label{app:gcov_and_cptp}

For $\epsilon=0$, i.e., for the conversion without error, Lemma~\ref{lem:gcov_and_cptp} is shown as Lemma~7 in \cite{marvian_mashhad_symmetry_2012}. Although its generalization to the case of $\epsilon>0$ is straightforward, we here provide the proof for completeness.
\begin{proof}[Proof of Lemma~\ref{lem:gcov_and_cptp}]
    (ii) $\implies$ (i): This is a direct consequence of the definition of the $G$-covariance of a channel and the fact that the trace distance is invariant under unitary transformations. \par
    (i) $\implies $ (ii): 
    Define $\mathcal{E}'\coloneqq \int_{g\in G}\dd \mu_G(g)\, \mathcal{U}_{g^{-1}}'\circ \mathcal{E}\circ \mathcal{U}_g$, where $\mu_G$ denotes the normalized right-invariant Haar measure on $G$. 
    This map $\mathcal{E}'$ is a $G$-covariant channel since
    \begin{align}
        \forall h\in G,\quad \mathcal{E}'\circ \mathcal{U}_h&=\int_{g\in G}\dd \mu_G(g)\, \mathcal{U}'_g\circ \mathcal{E}\circ\mathcal{U}_{g^{-1}h}\\
        &=\mathcal{U}_h'\circ \int_{g\in G}\dd \mu_G(g) \, \mathcal{U}_{h^{-1}g}'\circ \mathcal{E}\circ \mathcal{U}_{(h^{-1}g)^{-1}}\\
        &=\mathcal{U}_h'\circ \mathcal{E}'
    \end{align}
    holds. 
    From the convexity of the trace distance, we get
    \begin{align}
        T\left(\mathcal{E}'(\rho),\sigma\right)&\leq \int_{g\in G} \dd \mu_G(g) \, T\left(\mathcal{U}_{g^{-1}}'\circ \mathcal{E}\circ \mathcal{U}_g(\rho),\sigma\right)\\
        &=\int_{g\in G} \dd \mu_G(g) \, T\left(\mathcal{E}\circ \mathcal{U}_g(\rho), \mathcal{U}_g'(\sigma)\right)\leq \epsilon.
    \end{align}
\end{proof}

\subsubsection{Proof of Lemma~\ref{lem:QLAN_asymptotic_conversion}}\label{sec:QLAN_asymptotic_conversion}

To prove Lemma~\ref{lem:QLAN_asymptotic_conversion}, we divide the argument into two lemmas, Lemma~\ref{lem:rate_changing_part} and Lemma~\ref{lem:covariance_decreasing_part}, which are detailed later. See Fig.~\ref{fig:QGT_relations} for the relations among these lemmas. On the one hand, Lemma~\ref{lem:rate_changing_part} shows that the change in the asymptotic rate $r$ is effectively equivalent to the change in the scale of the parameters, which is proven by using results on QLAN for pure-state models \cite{lahiry_minimax_2024,girotti_optimal_2024}. The asymptotics of QGTs are conserved in the interconversion among models in Lemma~\ref{lem:rate_changing_part}. On the other hand, Lemma~\ref{lem:covariance_decreasing_part} shows that the asymptotic conversion between two pure-state statistical models with rate one is possible when QGTs satisfy a matrix inequality. As a direct consequence of these lemmas, we complete the proof of Lemma~\ref{lem:QLAN_asymptotic_conversion} at the end of this subsection. 

\begin{figure}[htbp]
    \centering
    \includegraphics[width=11cm]{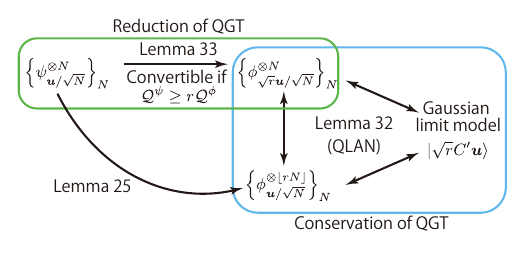}
    \caption{A schematic figure of the relations among  Lemma~\ref{lem:QLAN_asymptotic_conversion}, Lemma~\ref{lem:rate_changing_part} and Lemma~\ref{lem:covariance_decreasing_part}.
    The direction of the arrow indicates the convertibility with vanishing error. 
    In Lemma~\ref{lem:rate_changing_part}, we find the growth rate of i.i.d. copies can be set to 1 by adjusting the scaling of the parameters $u$ when analyzing the asymptotic convertibility. Lemma~\ref{lem:covariance_decreasing_part} shows that the asymptotic conversion between two pure-state models are possible if QGT is reduced during the conversion in the sense of matrix inequality. Combining these lemmas, Lemma~\ref{lem:QLAN_asymptotic_conversion} is proven.}
    \label{fig:QGT_relations}
\end{figure}

To proceed further, we here prove several lemmas. In the analysis below, we expand unitary operators to extract the asymptotic behaviors of pure-state models. The following two lemmas serve as a foundation for our analysis.

\begin{lem}\label{lem:taylor_based_bound}
    Let $\mathcal{H}$ be a finite-dimensional Hilbert space and $f$ be a function from $\mathbb{R}$ to linear operators on $\mathcal{H}$ and of class $C^3$. 
    Then, 
    \begin{align}
        \left\|f(1)-\left(f(0)+f'(0)+\frac{1}{2}f''(0)\right)\right\|_1
        \leq \frac{1}{6}\max_{\theta\in [0, 1]} \{\|f^{(3)}(\theta)\|_1\}. 
    \end{align}
\end{lem}
\begin{proof}
    By Taylor's theorem, we have 
    \begin{align}
        f(1)-\left(f(0)+f'(0)+\frac{1}{2}f''(0)\right)
        =\int_0^1 \frac{(1-\theta)^2}{2}f^{(3)}(\theta)\dd\theta. 
    \end{align}
    By the triangle inequality, we get 
    \begin{align}
        \left\|f(1)-\left(f(0)+f'(0)+\frac{1}{2}f''(0)\right)\right\|_1 
        \leq\int_0^1 \left\|\frac{(1-\theta)^2}{2}f^{(3)}(\theta)\right\|_1 \dd\theta 
        \leq\int_0^1 \frac{(1-\theta)^2}{2}M \dd\theta 
        =\frac{M}{6}, 
    \end{align}
    where $M:=\max_{\theta\in [0, 1]} \{\|f^{(3)}(\theta)\|_1\}$. 
\end{proof}

\begin{lem}\label{lem:bound_of_each_order}
    Let $l\in\mathbb{N}$, $\mathcal{H}$ and $\mathcal{H}'$ be finite-dimensional Hilbert spaces, $\rho$ be a linear operator on $\mathcal{H}$, $O$ be a Hermitian operator on $\mathcal{H}$, $O'$ be a Hermitian operator on $\mathcal{H}'$, $\mathcal{E}$ be a linear map from linear operators on $\mathcal{H}$ to those on $\mathcal{H}'$, and $F_{\rho, \mathcal{E}}$ be a function from $\mathbb{R}$ to linear operators on $\mathcal{H}'$ defined by 
    \begin{align}
        F_{\rho, \mathcal{E}}(\theta):=e^{-\ii \theta O'}\mathcal{E}(e^{\ii \theta O}\rho e^{-\ii \theta O})e^{\ii \theta O'}\ \forall\theta\in\mathbb{R}. 
    \end{align}
    Then, for any $\theta\in\mathbb{R}$, 
    \begin{align}
        F_{\rho, \mathcal{E}}^{(l)}(\theta)
        =\sum_{j=0}^l \binom{l}{j}e^{-\ii\theta O'}\underbrace{[-\ii O', ..., [-\ii O',}_{j}\mathcal{E}(e^{\ii\theta O}\underbrace{[\ii O, ..., [\ii O,}_{l-j} \rho]...]e^{-\ii\theta O})]...]e^{\ii\theta O'}, 
    \end{align}
    and 
    \begin{align}
        \left\|F_{\rho, \mathcal{E}}^{(l)}(\theta)\right\|_1
        \leq \|\rho\|_1\|\mathcal{E}\|\cdot 2^l(\|O\|_\infty+\|O'\|_\infty)^l, 
    \end{align}
    where $\|\mathcal{E}\|$ is defined by 
    \begin{align}
        \|\mathcal{E}\|:=\max_{\|L\|_1=1}\{\|\mathcal{E}(L)\|_1\}. 
    \end{align}
\end{lem}
\begin{proof}
    We note that $\frac{\dd}{\dd\theta}e^{\ii\theta A}B e^{-\ii\theta A}=e^{\ii\theta A}[\ii A, B] e^{-\ii\theta A}$ for all linear operators $A$ and $B$. 
    By using the Leibniz rule, we get the expression for $F_{\rho, \mathcal{E}}^{(l)}(\theta)$. 
    By applying the triangle inequality to the expression, we get 
    \begin{align}
        \left\|F_{\rho, \mathcal{E}}^{(l)}(\theta)\right\|_1 
        \leq\sum_{j=0}^l \binom{l}{j}\|e^{-\ii\theta O'}\underbrace{[-\ii O', ..., [-\ii O',}_{j}\mathcal{E}(e^{\ii\theta O}\underbrace{[\ii O, ..., [\ii O,}_{l-j} \rho]...]e^{-\ii\theta O})]...]e^{\ii\theta O'}\|_1. 
    \end{align}
    We note that by the triangle inequality and H\"{o}lder's inequality, we have 
    \begin{align}
        \|[A, B]\|_1
        =\|AB-BA\|_1
        \leq\|AB\|_1+\|BA\|_1
        \leq\|A\|_\infty\|B\|_1+\|B\|_1\|A\|_\infty
        =2\|A\|_\infty\|B\|_1
    \end{align}
    for all linear operators $A$ and $B$. 
    By using this relation, the definition of $\|\mathcal{E}\|$, and the unitary invariance of the trace norm, we get 
    \begin{align}
        \left\|F_{\rho, \mathcal{E}}^{(l)}(\theta)\right\|_1 
        \leq&\sum_{j=0}^l \binom{l}{j}\|\underbrace{[-\ii O', ..., [-\ii O',}_{j}\mathcal{E}(e^{\ii\theta O}\underbrace{[\ii O, ..., [\ii O,}_{l-j} \rho]...]e^{-\ii\theta O})]...]\|_1 \\
        \leq&\sum_{j=0}^l \binom{l}{j} (2\|O'\|_\infty)^j\|\mathcal{E}(e^{\ii\theta O}\underbrace{[\ii O, ..., [\ii O,}_{l-j} \rho]...]e^{-\ii\theta O})\|_1 \\
        \leq&\sum_{j=0}^l \binom{l}{j} (2\|O'\|_\infty)^j\|\mathcal{E}\|\|e^{\ii\theta O}\underbrace{[\ii O, ..., [\ii O,}_{l-j} \rho]...]e^{-\ii\theta O}\|_1 \\
        =&\sum_{j=0}^l \binom{l}{j} (2\|O'\|_\infty)^j\|\mathcal{E}\|\|\underbrace{[\ii O, ..., [\ii O,}_{l-j} \rho]...]\|_1 \\
        \leq&\sum_{j=0}^l \binom{l}{j} (2\|O'\|_\infty)^j\|\mathcal{E}\|(2\|O\|_\infty)^{l-j}\|\rho\|_1 \\
        =&\|\rho\|_1\|\mathcal{E}\|\cdot 2^l(\|O\|_\infty+\|O'\|_\infty)^l. 
    \end{align}
\end{proof}

Using Lemma~\ref{lem:taylor_based_bound} and Lemma~\ref{lem:bound_of_each_order}, we prove a lemma that relates the setup in QLAN and ours:
\begin{lem}\label{lem:interconversion_projected_generators}
    Let $\phi$ be an arbitrary pure state in a finite-dimensional Hilbert space. Fix a set of Hermitian operators $\{X_\mu'\}_{\mu=1}^m$. 
    We define another set of Hermitian operator $ \{\tilde{X}_\mu'\}_{\mu=1}^m$ by $\tilde{X}_\mu'\coloneqq (I-\phi) X_\mu' \phi +\phi X_\mu' (I-\phi)$. 
    For $\epsilon\in(0,1/6)$, it holds
    \begin{align}
        \lim_{N\to\infty }\sup_{\|u\|<N^\epsilon }T\left(\phi_{u/\sqrt{N}}^{\otimes N},\tilde{\phi}_{u/\sqrt{N}}^{\otimes N} \right)=0,
    \end{align}
    where $\phi_{\theta}\coloneqq e^{\ii\theta^\mu X_\mu'}\phi e^{-\ii\theta^\mu X_\mu'}$ and $\tilde{\phi}_{\theta}\coloneqq e^{\ii\theta^\mu \tilde{X}_\mu'}\phi e^{-\ii\theta^\mu \tilde{X}_\mu'}$ for $\theta\in\mathbb{R}^m$.
\end{lem}
\begin{proof}
    Let us define $O\coloneqq \frac{1}{\sqrt{N}}u^\mu X_\mu' $ and $\tilde{O}\coloneqq \frac{1}{\sqrt{N}}u^\mu \tilde{X}_\mu'=(I-\phi) O\phi+\phi O (I-\phi)$ for $u\in\mathbb{R}^m$ such that $\|u\|<N^\epsilon$. 
    Introducing a finite non-negative number $L$ by 
    \begin{align}
        L\coloneqq \max\left\{\max_{\|\theta\|=1}\|\theta^\mu X_\mu'\|_\infty,\max_{\|\theta\|=1}\|\theta^\mu \tilde{X}_\mu'\|_\infty\right\} ,
    \end{align}
    we have $\|O\|_\infty\leq L\frac{\|u\|}{\sqrt{N}} $ and $\|\tilde{O}\|_\infty\leq  L\frac{\|u\|}{\sqrt{N}}$. 
    Applying Lemma~\ref{lem:taylor_based_bound} and Lemma~\ref{lem:bound_of_each_order} for $f(\theta)\coloneqq e^{-\ii \theta \tilde{O}}e^{\ii \theta O}\phi e^{-\ii \theta O}e^{\ii \theta \tilde{O}}$, we find
    \begin{align}
        \left|\Braket{\phi|f(1)-\left(f(0)+f'(0)+\frac{1}{2}f''(0)\right)|\phi}\right|&\leq \left\|f(1)-\left(f(0)+f'(0)+\frac{1}{2}f''(0)\right)\right\|_1\\
        &\leq \frac{2^3}{6}(\|O\|_\infty+\|\tilde{O}\|_\infty)^3\leq \frac{2^6}{6}\left(L\frac{\|u\|}{\sqrt{N}}\right)^3.\label{eq:fid_expansion_projection}
    \end{align}

    Since
    \begin{align}
        \braket{\phi|\tilde{O}^{n}O^m|\phi}=
        \begin{cases}
            1&\quad (n,m)=(0,0)\\
            \braket{\phi|O|\phi}&\quad (n,m)=(0,1)\\
            0&\quad (n,m)=(1,0)
        \end{cases},
    \end{align}
    we get
    \begin{align}
        \braket{\phi|f(0)|\phi}&=1,\\
        \braket{\phi|f'(0)|\phi}&=\ii \braket{\phi|O|\phi}-\ii\braket{\phi|O|\phi}=0 .
    \end{align}
    Similarly, since
    \begin{align}
        \braket{\phi|\tilde{O}^{n}O^m|\phi}=
        \begin{cases}
            \braket{\phi|O^2|\phi}&\quad (n,m)=(0,2)\\
            \braket{\phi|O(I-\phi)O|\phi}&\quad (n,m)=(1,1),\,(2,0)
        \end{cases}
    \end{align}
    and hence 
    \begin{align}
        \sum_{\substack{n,m=0,1,2\\ n+m= 2}}\frac{(-\ii)^n}{n!}\frac{\ii^m}{m!}\braket{\phi|\tilde{O}^{ n}O^m|\phi}=-\frac{1}{2}\braket{\phi|O^2|\phi}+\frac{1}{2}\braket{\phi|O(I-\phi)O|\phi}=-\frac{1}{2}\braket{\phi|O|\phi}^2
    \end{align}
    we get
    \begin{align}
        \braket{\phi|f''(0)|\phi}&=\sum_{\substack{n,m,n',m'=0,1,2\\ n+m+n'+m'= 2}}\frac{(-\ii)^n}{n!}\frac{\ii^m}{m!}\frac{(-\ii)^{m'}}{m'!}\frac{\ii^{n'}}{n'!}\braket{\phi|\tilde{O}^{ n}O^m|\phi}\braket{\phi| O^{m'}\tilde{O}^{ n'}|\phi}\\
        &=-\frac{1}{2}\braket{\phi|O|\phi}^2+\braket{\phi|O|\phi}^2-\frac{1}{2}\braket{\phi|O|\phi}^2=0.
    \end{align}
    Therefore, by using $\braket{\phi|f(1)|\phi}= \mathrm{Fid}\left(\phi_{u/\sqrt{N}},\tilde{\phi}_{u/\sqrt{N}} \right)$, Eq.~\eqref{eq:fid_expansion_projection} implies
    \begin{align}
        \left|\mathrm{Fid}\left(\phi_{u/\sqrt{N}},\tilde{\phi}_{u/\sqrt{N}} \right)-1\right|\leq \frac{2^3}{6}(\|O\|_\infty+\|\tilde{O}\|_\infty)^3\leq \frac{2^6}{6}\left(L\frac{\|u\|}{\sqrt{N}}\right)^3\leq \frac{2^6}{6}(LN^{-1/2+\epsilon})^{3}\label{eq:projected_generators_approximation_error}
    \end{align}
    for any $u$ such that $\|u\|<N^\epsilon$. 

    Since $(N^{-1/2+\epsilon})^3=o(1/N)$ if $\epsilon \in (0,1/6)$, we get
    \begin{align}
        \inf_{\|u\|<N^\epsilon}\mathrm{Fid}\left(\phi_{u/\sqrt{N}}^{\otimes N},\tilde{\phi}_{u/\sqrt{N}}^{\otimes N} \right)=\inf_{\|u\|<N^\epsilon}\left(\mathrm{Fid}\left(\phi_{u/\sqrt{N}},\tilde{\phi}_{u/\sqrt{N}} \right)\right)^N\geq \left(1+o\left(\frac{1}{N}\right)\right)^N,
    \end{align}
    implying that 
    \begin{align}
        \lim_{N\to\infty}\inf_{\|u\|<N^\epsilon}\mathrm{Fid}\left(\phi_{u/\sqrt{N}}^{\otimes N},\tilde{\phi}_{u/\sqrt{N}}^{\otimes N} \right)
        &=\lim_{N\to\infty}\left(1+o\left(\frac{1}{N}\right)\right)^N=1.
    \end{align}
    From Fuchs-van de Graaf's inequalities, we also get
    \begin{align}
        \lim_{N\to\infty }\sup_{\|u\|<N^\epsilon }T\left(\phi_{u/\sqrt{N}}^{\otimes N},\tilde{\phi}_{u/\sqrt{N}}^{\otimes N} \right)=0.
    \end{align}
\end{proof}

Therefore, as long as the asymptotic properties of an i.i.d. copies of $\phi_{u/\sqrt{N}}=e^{\ii u^\mu X_\mu'/\sqrt{N}}\phi e^{-\ii u^\nu X_\nu'/\sqrt{N}}$ are concerned, it suffice to analyze $\tilde{\phi}_{u/\sqrt{N}}= e^{\ii u^\mu \tilde{X}_\mu' /\sqrt{N}}\phi e^{-\ii u^\nu\tilde{X}_\nu'/\sqrt{N}}$.

The asymptotic behavior of $\tilde{\phi}_{u/\sqrt{N}}^{\otimes N}$ can be understood by studies on QLAN. The result in \cite{lahiry_minimax_2024} is complicated since it aims to prove a general statement applicable to a wide class of statistical models. In this paper, we only focus on pure-state models. For pure-state unitary models, the result in \cite{lahiry_minimax_2024} can be stated as follows:
\begin{lem}[Theorem~3.3 in \cite{lahiry_minimax_2024}, particularized to pure-state model. See also \cite{girotti_optimal_2024}.]\label{lem:QLAN}
    Let $\phi$ be an arbitrary state in a $d$-dimensional Hilbert space. For any fixed set of Hermitian operators $\{X_\mu'\}_{\mu=1}^m$, we define a pure-state qudit model $\tilde{\phi}_{\theta}\coloneqq e^{\ii\theta^\mu \tilde{X}_\mu'}\phi e^{-\ii\theta^\nu \tilde{X}_\nu'}$ for $\theta\in\mathbb{R}^m$ with $\tilde{X}_\mu'\coloneqq (I-\phi) X_\mu '\phi +\phi X_\mu '(I-\phi)$.

    Let $\{\ket{k}\}_{k=1}^{d-1}$ denote an orthonormal basis of the subspace orthogonal to $\ket{\phi}$, where $\ket{\phi}$ is a unit vector satisfying $\phi=\ket{\phi}\bra{\phi}$.
    We introduce a $(d-1)\times m$ matrix defined by
    \begin{align}
        \ii C_{k\mu }'\coloneqq \braket{k|\tilde{X}_\mu'|\phi}\label{eq:matrix_Gaussian_shift_model_definition}
    \end{align}
    for $k=1,\ldots, d-1$ and $\mu=1,\ldots, m$.

    Let $\mathcal{H}_{d-1}$ denote the $(d-1)$-mode Fock space. For $z=(z_1,\cdots,z_{d-1})^\top\in\mathbb{C}^{d-1}$, we denote the coherent state by
    \begin{align}
        \ket{z}\coloneqq \left(\bigotimes_{i=1}^{d-1}e^{z_ia_i^\dag -z_i^*a_i}\ket{0}_i\right)\in \mathcal{H}_{d-1},\label{eq:Gaussian_shift_model_definition}
    \end{align}
    where $a_i^\dag,a_i$ and $\ket{0}_i$ denote the creation and annihilation operator and the vacuum state for the $i$th mode. 
    
    For $\epsilon\in(0,1/9)$, there exists a sequence of quantum channels $\{\mathcal{T}_N\}_{N}$ and $\{\mathcal{S}_N\}_N$ such that
    \begin{align}
        \lim_{N\to\infty}\sup_{\|u\|<N^{\epsilon}}T\left(\mathcal{T}_N(\tilde{\phi}^{\otimes N}_{u/\sqrt{N}}),\ket{C'u}\bra{C'u}\right)&=0,\\
        \lim_{N\to\infty}\sup_{\|u\|<N^{\epsilon}}T\left(\tilde{\phi}^{\otimes N}_{u/\sqrt{N}},\mathcal{S}_N(\ket{C'u}\bra{C'u})\right)&=0.
    \end{align}
\end{lem}

By using Lemma~\ref{lem:interconversion_projected_generators} and Lemma~\ref{lem:QLAN}, we prove the relation between the asymptotic rate and the scale in the parameters:  
\begin{lem}\label{lem:rate_changing_part}
    Let $\phi$ be an arbitrary pure state in a finite-dimensional Hilbert space. Fix any sets of Hermitian operators $\{X_\mu'\}_{\mu=1}^m$ and introduce a pure-state statistical model $\phi_{\theta}\coloneqq e^{\ii \theta^\mu X_\mu'}\phi e^{-\ii \theta^\nu X_\nu'}$ for $\theta\in\mathbb{R}^m$. For any $r>0$, there exist sequences of quantum channels $\{\mathcal{E}_N\}_N$ and $\{\mathcal{E}'_N\}_N$ such that
    \begin{align}
        \lim_{N\to\infty}\sup_{\|u\|<N^\epsilon}T\left(\mathcal{E}_N\left(\phi_{\sqrt{r}u/\sqrt{N}}^{\otimes N}\right),\phi_{u/\sqrt{N}}^{\otimes \floor{rN}}\right)&=0,\\
        \lim_{N\to\infty}\sup_{\|u\|<N^\epsilon}T\left(\phi_{\sqrt{r}u/\sqrt{N}}^{\otimes N},\mathcal{E}_N'\left(\phi_{u/\sqrt{N}}^{\otimes \floor{rN}}\right)\right)&=0
    \end{align}
    for $\epsilon\in(0,1/9)$.
\end{lem}
\begin{proof}
    For $\tilde{X}_\mu' \coloneqq (I-\phi)X_\mu'  \phi +\phi X_\mu'(I-\phi)$, consider a sequence of statistical models $\{(\tilde{\phi}_{\theta})^{\otimes \floor{rN}}\}_N$, where $r>0$ and $\tilde{\phi}_{\theta}\coloneqq e^{\ii \theta^\mu\tilde{X}_\mu '}\phi e^{-\ii\theta^\nu \tilde{X}_\nu'}$. From Lemma~\ref{lem:QLAN}, there exist sequences of quantum channels $\{\mathcal{T}_N\}_N$ and $\{\mathcal{S}_N\}_N$ such that
    \begin{align}
        \lim_{N\to\infty}\sup_{\|u\|<\floor{rN}^{\epsilon}}T\left(\mathcal{T}_{\floor{rN}}\left(\tilde{\phi}^{\otimes \floor{rN}}_{u/\sqrt{\floor{rN}}}\right),\ket{C'u}\bra{C'u}\right)&=0,\\
        \lim_{N\to\infty}\sup_{\|u\|<{\floor{rN}}^{\epsilon}}T\left(\tilde{\phi}^{\otimes \floor{rN}}_{u/\sqrt{\floor{rN}}},\mathcal{S}_{\floor{rN}}(\ket{C'u}\bra{C'u})\right)&=0.
    \end{align}
    for $\epsilon\in(0,1/9)$. Introducing a new parameter $u'/\sqrt{N}\coloneqq u/\sqrt{\floor{rN}}$, we get
    \begin{align}
        \lim_{N\to\infty}\sup_{\|u'\|<\floor{rN}^{\epsilon}\sqrt{\frac{\floor{rN}}{N}}}T\left(\mathcal{T}_{\floor{rN}}\left(\tilde{\phi}^{\otimes \floor{rN}}_{u'/\sqrt{N}}\right),\Ket{\sqrt{\frac{\floor{rN}}{N}}C'u'}\Bra{\sqrt{\frac{\floor{rN}}{N}}C'u'}\right)&=0,\\
        \lim_{N\to\infty}\sup_{\|u'\|<{\floor{rN}}^{\epsilon}\sqrt{\frac{\floor{rN}}{N}}}T\left(\tilde{\phi}^{\otimes \floor{rN}}_{u'/\sqrt{N}},\mathcal{S}_{\floor{rN}}\left(\Ket{\sqrt{\frac{\floor{rN}}{N}}C'u'}\Bra{\sqrt{\frac{\floor{rN}}{N}}C'u'}\right)\right)&=0.
    \end{align}
    For any $\epsilon'\in (0,\epsilon)$,  
    \begin{align}
        \lim_{N\to\infty }\frac{ N^{\epsilon'}}{{\floor{rN}}^{\epsilon}\sqrt{\frac{\floor{rN}}{N}}}=0,
    \end{align}
    meaning that ${\floor{rN}}^{\epsilon}\sqrt{\frac{\floor{rN}}{N}}\geq{ N^{\epsilon'}}$ holds for all sufficiently large $N$. In addition, it holds
    \begin{align}
        \lim_{N\to\infty}\left|\Braket{\sqrt{\frac{\floor{rN}}{N}}C'u'|\sqrt{r}C'u'}\right|=\lim_{N\to\infty}e^{-\|C'u'\|^2\left(\sqrt{r}-\sqrt{\frac{\floor{rN}}{N}}\right)^2}=1\label{eq:coherence_overlap}
    \end{align}
    for any $u'\in\mathbb{R}^m$. 
    Therefore, we get
    \begin{align}
        \lim_{N\to\infty}\sup_{\|u\|<N^{\epsilon'}}T\left(\mathcal{T}_{\floor{rN}}\left(\tilde{\phi}^{\otimes \floor{rN}}_{u/\sqrt{N}}\right),\ket{\sqrt{r}C'u}\bra{\sqrt{r}C'u}\right)&=0,\\
        \lim_{N\to\infty}\sup_{\|u\|<N^{\epsilon'}}T\left(\tilde{\phi}^{\otimes \floor{rN}}_{u/\sqrt{N}},\mathcal{S}_{\floor{rN}}\left(\ket{\sqrt{r}C'u}\bra{\sqrt{r}C'u}\right)\right)&=0.
    \end{align}
    Note that we can take an arbitrary value in $(0,1/9)$ as $\epsilon'$ by appropriately taking $\epsilon\in(0,1/9)$.

    On the other hand, from Lemma~\ref{lem:QLAN}, we have
    \begin{align}
        \lim_{N\to\infty}\sup_{\|u\|<N^{\epsilon}}T\left(\mathcal{T}_N\left(\tilde{\phi}_{\sqrt{r}u/\sqrt{N}}^{\otimes N}\right),\ket{\sqrt{r}C'u}\bra{\sqrt{r}C'u}\right)&=0,\label{eq:QLAN_error_rate_Tn}\\
        \lim_{N\to\infty}\sup_{\|u\|<N^{\epsilon}}T\left(\tilde{\phi}_{\sqrt{r}u/\sqrt{N}}^{\otimes N},\mathcal{S}_N(\ket{\sqrt{r}C'u}\bra{\sqrt{r}C'u})\right)&=0.\label{eq:QLAN_error_rate_Sn}
    \end{align}
    Therefore, we get
    \begin{align}
        \lim_{N\to\infty}\sup_{\|u\|<N^{\epsilon}}T\left(\mathcal{S}_{\floor{rN}}\circ\mathcal{T}_N\left(\tilde{\phi}_{\sqrt{r}u/\sqrt{N}}^{\otimes N}\right),\tilde{\phi}^{\otimes \floor{rN}}_{u/\sqrt{N}}\right)&=0,\label{eq:QGT_conserving_transf_channel}\\
        \lim_{N\to\infty}\sup_{\|u\|<N^{\epsilon}}T\left(\tilde{\phi}_{\sqrt{r}u/\sqrt{N}}^{\otimes N},\mathcal{S}_{N}\circ\mathcal{T}_{\floor{rN}}\left(\tilde{\phi}^{\otimes \floor{rN}}_{u/\sqrt{N}}\right)\right)&=0
    \end{align}
    for $\epsilon\in(0,1/9)$. From Lemma~\ref{lem:interconversion_projected_generators}, we also get
    \begin{align}
        \lim_{N\to\infty}\sup_{\|u\|<N^{\epsilon}}T\left(\mathcal{S}_{\floor{rN}}\circ\mathcal{T}_N\left(\phi_{\sqrt{r}u/\sqrt{N}}^{\otimes N}\right),\phi^{\otimes \floor{rN}}_{u/\sqrt{N}}\right)&=0,\\
        \lim_{N\to\infty}\sup_{\|u\|<N^{\epsilon}}T\left(\phi_{\sqrt{r}u/\sqrt{N}}^{\otimes N},\mathcal{S}_{N}\circ\mathcal{T}_{\floor{rN}}\left(\phi^{\otimes \floor{rN}}_{u/\sqrt{N}}\right)\right)&=0
    \end{align}
    for $\epsilon\in(0,1/9)$. 
\end{proof}
Importantly, for both $\phi_{\sqrt{r}u/\sqrt{N}}^{\otimes N}$ and $\phi^{\otimes \floor{rN}}_{u/\sqrt{N}}$, the QGT with respect to the parameters $u$ is equal to $r\, \cov(\phi,\{X_\mu'\})$ in the limit of $N\to\infty$, where the elements of covariance matrix are defined by $\left(\cov(\phi,\{X_\mu'\})\right)_{\mu\nu }\coloneqq \mathrm{Tr}(\phi X_\mu'(I-\phi)X_\nu')$. In other words, the QGTs are preserved under reversible asymptotic conversion between $\{\phi_{\sqrt{r}u/\sqrt{N}}^{\otimes N}\}_N$ and $\{\phi^{\otimes \floor{rN}}_{u/\sqrt{N}}\}_N$. 

Using the above lemma, the asymptotic rate can be effectively set to $r=1$ by adjusting the scale of the parameters. The following lemma claims that the asymptotic conversion with rate $r=1$ is possible only if the covariance matrix of the original state is larger than that of the target state in the sense of matrix inequality:
\begin{lem}\label{lem:covariance_decreasing_part}
    Let $\psi\in\mathcal{P}(\mathcal{H})$ and $\phi\in\mathcal{P}(\mathcal{H}')$ be arbitrary pure states in finite-dimensional Hilbert spaces $\mathcal{H},\mathcal{H}'$. Fix any sets of Hermitian operators $\{X_\mu\}_{\mu=1}^m$ and $\{X_\mu'\}_{\mu=1}^m$. If $\cov(\psi,\{X_\mu\}_{\mu=1}^m)\geq \cov(\phi,\{X_\mu'\}_{\mu=1}^m)$ holds, then there exists a quantum channel $\mathcal{E}$ such that
    \begin{align}
        \lim_{N\to\infty}\sup_{\|u\|<N^\epsilon}T\left(\mathcal{E}^{\otimes N}\left(\psi_{u/\sqrt{N}}^{\otimes N}\right),\phi_{u/\sqrt{N}}^{\otimes N}\right)=0\label{eq:covariance_decreasing_part}
    \end{align}
    for $\epsilon\in(0,1/6)$, where $\psi_{\theta}\coloneqq e^{\ii \theta^\mu X_\mu}\psi e^{-\ii \theta^\nu X_\nu}$ and $\phi_{\theta}\coloneqq e^{\ii \theta^\mu X_\mu'}\phi e^{-\ii \theta^\nu X_\nu'}$ for $\theta\in\mathbb{R}^m$. 
\end{lem}
\begin{proof}
    Following Lemma~\ref{lem:interconversion_projected_generators}, we introduce $\tilde{X}_\mu\coloneqq (I-\psi)X_\mu \psi+\psi X_\mu (I-\psi)$ and $\tilde{X}_\mu'\coloneqq (I-\phi)X_\mu'\phi+\phi X_\mu' (I-\phi)$ and denote
    \begin{align}
        \tilde{\psi}_{u/\sqrt{N}}\coloneqq e^{\ii\frac{1}{\sqrt{N}}u^\mu \tilde{X}_\mu }\psi e^{-\ii\frac{1}{\sqrt{N}}u^\nu\tilde{X}_\nu },\quad \tilde{\phi}_{u/\sqrt{N}}\coloneqq e^{\ii\frac{1}{\sqrt{N}}u^\mu\tilde{X}_\mu'}\phi e^{-\ii\frac{1}{
    \sqrt{N}}u^\nu\tilde{X}_\nu' }.
    \end{align}

    Let us fix orthonormal bases $\{\ket{k}\}_{k=1}^{d-1}$ and $\{\ket{k'}\}_{k'=1}^{d'-1}$ for the orthogonal complements to $\ket{\psi}$ and $\ket{\phi}$, respectively, where $d\coloneqq \dim\mathcal{H}$, $d'\coloneqq\dim\mathcal{H}'$, and $\ket{\psi}\in\mathcal{H},\,\ket{\phi}\in\mathcal{H}'$ are unit vectors such that $\psi=\ket{\psi}\bra{\psi}$ and $\phi=\ket{\phi}\bra{\phi}$. For matrices $C:\mathbb{C}^m\to\mathbb{C}^{d-1}$, $C':\mathbb{C}^m\to\mathbb{C}^{d'-1}$ whose elements are given by
    \begin{align}
        (C)_{ki}&\coloneqq -\ii\braket{k|X_\mu|\psi} ,\quad (C')_{l'i}\coloneqq -\ii\braket{l'|X_\mu'|\phi},
    \end{align}
    we define matrices $Z:\mathbb{C}^{d-1}\to \mathbb{C}^{d'-1}$ and $\Gamma\coloneqq \mathbb{C}^{d-1}\to \mathbb{C}^{d-1}$ by
    \begin{align}
        Z\coloneqq C'C^{+},\label{eq:definition_z}
    \end{align}
    and 
    \begin{align}
        \Gamma\coloneqq I-Z^\dag Z,\label{eq:definition_gamma}
    \end{align}
    where $C^{+}$ denotes the Moore–Penrose inverse of $C$. Note that the assumption that $\cov(\psi,\{X_\mu\}_{\mu=1}^m)\geq \cov(\phi,\{X_\mu'\}_{\mu=1}^m)$ is equivalent to
    \begin{align}
        C^\dag C\geq C^{\prime\dag}C'\label{eq:c_dag_c}.
    \end{align}
    Let $P$ be the projector to the range of $C$, given by $P =CC^{+}$. By using Eq.~\eqref{eq:c_dag_c}, it holds
    \begin{align}
        \Gamma=I-Z^\dag Z\geq P-Z^\dag Z=P^2-Z^\dag Z=(C^{+})^\dag  (C^\dag C-C^{\prime \dag}C')C^{+}\geq  0,
    \end{align}
    i.e., $\Gamma$ is positive-semidefinite. 
    By using the square root $\sqrt{\Gamma}$ of $\Gamma$, the operators $\{K_k\}_{k=0}^{d-1}$ defined by
    \begin{align}
        K_0&\coloneqq \ket{\phi}\bra{\psi}+\sum_{k'=1}^{d'-1} \sum_{k=1}^{d-1}z_{k'k}\ket{k'}\bra{k},\label{eq:K_0_definition}\\
        K_{k}&\coloneqq \sum_{l=1}^{d-1}\left(\sqrt{\Gamma}\right)_{kl}\ket{\phi}\bra{l}\qquad (k=1,2,\cdots,d-1),\label{eq:K_k_definition}
    \end{align}
    satisfy the normalization condition:
    \begin{align}
        \sum_{k=0}^{d-1}K_{k}^\dag K_{k}
     &=\ket{\psi}\bra{\psi}+\sum_{k,l=1}^{d-1}\sum_{k',l'=1}^{d'-1}z_{k'k}^*z_{l'l}\delta_{k'l'}\ket{k}\bra{l}+\sum_{k,l=1}^{d-1}\sum_{m=1}^{d-1}\left(\sqrt{\Gamma}\right)_{mk}^*\left(\sqrt{\Gamma}\right)_{ml}\ket{k}\bra{l}\\
        &=\ket{\psi}\bra{\psi}+\sum_{k,l=1}^{d-1}\ket{k}\bra{l}\left(Z^\dag Z+\Gamma\right)_{kl}=\ket{\psi}\bra{\psi}+\sum_{k=1}^{d-1}\ket{k}\bra{k}=I.
    \end{align}
    Therefore, the linear map $\mathcal{E}\left(\cdot\right)\coloneqq \sum_{k=0}^{d-1}K_{k}\left(\cdot\right)K_{k}^\dag$ is a completely positive trace-preserving map, i.e., a quantum channel. In the following, we prove that this quantum channel satisfies Eq.~\eqref{eq:covariance_decreasing_part}.

    Let us define Hermitian operators $O,O'$ by $O\coloneqq u^\mu \tilde{X}_\mu/\sqrt{N}$ and $O'\coloneqq u^\mu \tilde{X}_\mu'/\sqrt{N}$. Introducing a finite non-negative number $L$ by 
    \begin{align}
        L\coloneqq \max\left\{\max_{\|\theta\|=1}\|\theta^\mu\tilde{X}_\mu\|_\infty,\max_{\|\theta\|=1}\|\theta^\mu \tilde{X}_\mu'\|_\infty\right\},
    \end{align}
    we have $\|O\|_\infty\leq L\frac{\|u\|}{\sqrt{N}} $ and $\|O'\|_\infty\leq  L\frac{\|u\|}{\sqrt{N}}$. Applying Lemma~\ref{lem:taylor_based_bound} and Lemma~\ref{lem:bound_of_each_order} for $f(\theta)\coloneqq e^{-\ii \theta O'}\mathcal{E}(e^{\ii \theta O}\psi e^{-\ii \theta O})e^{\ii \theta O'}$, we get
    \begin{align}
        \left|\Braket{\phi|f(1)-\left(f(0)+f'(0)+\frac{1}{2}f''(0)\right)|\phi}\right|&\leq \left\|f(1)-\left(f(0)+f'(0)+\frac{1}{2}f''(0)\right)\right\|_1\\
        &\leq \frac{2^3}{6}(\|O\|_\infty+\|O'\|_\infty)^3\leq \frac{2^6}{6}\left(L\frac{\|u\|}{\sqrt{N}}\right)^3,\label{eq:fidelity_expansion}
    \end{align}
    where we have used the fact that $\|\mathcal{E}\|\leq 1$ holds for any positive and trace-preserving map $\mathcal{E}$ \cite{perez-garcia_contractivity_2006}. 

    Let us evaluate each term of $\braket{\phi|\left(f(0)+f'(0)+\frac{1}{2}f''(0)\right)|\phi}$. From the definition of the channel $\mathcal{E}$, we get $\mathcal{E}(\psi)=\phi$, implying that
    \begin{align}
        \braket{\phi|f(0)|\phi}=1.\label{eq:zeroth}
    \end{align}
    Since $\braket{\phi|O'|\phi}=0$, we have $\braket{\phi|O'\mathcal{E}(\psi)|\phi}=0$ and $\braket{\phi|\mathcal{E}(\psi)O'|\phi}=0$. In addition, since
    \begin{align}
        \mathcal{E}(O\psi)&=\frac{-\ii}{\sqrt{N}}\sum_{k'=1}^{d'-1}\left(ZCu\right)_{k'}\ket{k'}\bra{\phi},\label{eq:E_O_psi}
    \end{align}
    we also get $\braket{\phi|\mathcal{E}(O\psi)|\phi}=0$ and $\braket{\phi|\mathcal{E}(\psi O)|\phi}=0$. Therefore, 
    \begin{align}
        \braket{\phi|f'(0)|\phi}=0.\label{eq:first}
    \end{align}

    From $\mathcal{E}(\psi)=\phi$ and $\braket{\phi|O'|\phi}=0$, we get
    \begin{align}
        \frac{1}{2}\braket{\phi|(O^{\prime 2}\mathcal{E}(\psi )+\mathcal{E}(\psi)O^{\prime 2})|\phi}&=\braket{\phi|O^{\prime 2}|\phi}=\frac{1}{N}u^\mu u^\nu \braket{\phi|\tilde{X}'_\mu \tilde{X}'_\nu|\phi}=\frac{1}{N}(C^{\prime \dag}u)^\dag C'u,\\
        \braket{\phi|O'\mathcal{E}(\psi)O'|\phi}&=0.
    \end{align}
    From Eq.~\eqref{eq:E_O_psi}, we find $\bra{\phi}\mathcal{E}(O\psi)=0$ and $\mathcal{E}(\psi O)\ket{\phi}=0$. Therefore,
    \begin{align}
        \braket{\phi|\mathcal{E}(O\psi)O'|\phi}=0,\quad \braket{\phi|O'\mathcal{E}(\psi O)|\phi}=0.
    \end{align}
    From Eq.~\eqref{eq:E_O_psi}, we also get 
    \begin{align}
        O'\mathcal{E}(O\psi)=\frac{1}{N} (C'u)^\dag Z Cu\ket{\phi}\bra{\phi},\quad \mathcal{E}(\psi O)O'=\frac{1}{N} (ZCu)^\dag C'u\ket{\phi}\bra{\phi},
    \end{align}
    implying that
    \begin{align}
        \braket{\phi|O'\mathcal{E}(O\psi)|\phi}&=\frac{1}{N} (C'u)^\dag Z Cu,\quad 
         \braket{\phi|\mathcal{E}(\psi O)O'|\phi}=\frac{1}{N}(ZCu)^\dag C'u.
    \end{align}
    From the definition of the quantum channel $\mathcal{E}$, we get
    \begin{align}
        \mathcal{E}(O^2\psi)&=\frac{1}{N}\mathcal{E}((Cu)^\dag (Cu) \ket{\psi}\bra{\psi})=\frac{1}{N}(Cu)^\dag (Cu)\ket{\phi}\bra{\phi},\\
        \mathcal{E}(O\psi O)&=\frac{1}{N}\sum_{k,l=1}^{d-1}(Cu)_k(Cu)_l^*\mathcal{E}\left(\ket{k}\bra{l}\right)=\frac{1}{N}\left(\sum_{k',l'=1}^{d'-1}(ZCu)_{k'}(ZCu)_{l'}^*\ket{k'}\bra{l'}+(Cu)^\dag \Gamma Cu\ket{\phi}\bra{\phi}\right),
    \end{align}
    and hence
    \begin{align}
        \frac{1}{2}\braket{\phi|\mathcal{E}(O^2\psi+\psi O^2)|\phi}&=\frac{1}{N}(Cu)^\dag C u,\quad 
         \braket{\phi|\mathcal{E}(O\psi O)|\phi}=\frac{1}{N}(Cu)^\dag \Gamma(Cu).
    \end{align}
    By combining these results, we get
    \begin{align}
        \Braket{\phi|\frac{1}{2}f''(0)|\phi}&=-\frac{1}{2}\braket{\phi|(O^{\prime 2}\mathcal{E}(\psi )-2O'\mathcal{E}(\psi)O'+\mathcal{E}(\psi)O^{\prime 2})|\phi}+ \braket{\phi|(O'\mathcal{E}(O\psi-\psi O)-\mathcal{E}(O\psi-\psi O)O')|\phi}\nonumber \\
        &\quad -\frac{1}{2}\braket{\phi|\mathcal{E}(O^2\psi-2O\psi O+\psi O^2)|\phi}\\
        &=\frac{1}{N}\left(-(C^{\prime \dag}u)^\dag C'u+(C'u)^\dag Z Cu+(ZCu)^\dag C'u-(Cu)^\dag C u+(Cu)^\dag \Gamma(Cu)\right)\\
        &=-\frac{1}{N}u^\dag (C'-ZC)^\dag  (C'-ZC)u.
    \end{align}
    Let $Q$ be the projector $Q$ to the support of $C$, given by $Q=C^+C$. From $C^\dag C\geq C^{\prime \dag }C'$, we get
    \begin{align}
        (C'-ZC)^\dag(C'-ZC) 
        =(I-Q)(C^{\prime\dag}C')(I-Q) 
        \leq& (I-Q)C^\dag C(I-Q) 
        =(I-Q)C^\dag (C-CC^+C) 
        =0.
    \end{align}
    We thus have $C'-ZC=0$, which implies 
    \begin{align}
       \Braket{\phi|\frac{1}{2}f''(0)|\phi}=0.\label{eq:second}
    \end{align}

    Since $\mathrm{Fid}\left(\mathcal{E}(\tilde{\psi}_{u/\sqrt{N}}),\tilde{\phi}_{u/\sqrt{N}}\right)=\braket{\phi|f(1)|\phi}$, 
    from Eqs.~\eqref{eq:fidelity_expansion}, \eqref{eq:zeroth}, \eqref{eq:first}, and \eqref{eq:second}, we have proven
    \begin{align}
        \left|\mathrm{Fid}\left(\mathcal{E}(\tilde{\psi}_{u/\sqrt{N}}),\tilde{\phi}_{u/\sqrt{N}}\right)-1\right|\leq \frac{2^6}{6}\left(L\frac{\|u\|}{\sqrt{N}}\right)^3\leq \frac{2^6}{6}(LN^{-1/2+\epsilon})^{3}\label{eq:conversion_reduction_QGT_error_rate}
    \end{align}
    for any $u$ such that $\|u\|<N^\epsilon$. 

    Since $(N^{-1/2+\epsilon})^3=o(1/N)$ if $\epsilon \in (0,1/6)$, we get
    \begin{align}
        \inf_{\|u\|<N^\epsilon }\mathrm{Fid}\left(\mathcal{E}^{\otimes N}(\tilde{\psi}_{u/\sqrt{N}}^{\otimes N}),\tilde{\phi}_{u/\sqrt{N}}^{\otimes N}\right)=\inf_{\|u\|<N^\epsilon }\mathrm{Fid}\left(\mathcal{E}(\tilde{\psi}_{u/\sqrt{N}}),\tilde{\phi}_{u/\sqrt{N}}\right)^N\geq \left(1+o\left(\frac{1}{N}\right)\right)^N,\label{eq:conversion_reduction_QGT_error_rate_powerN}
    \end{align}
    and hence
    \begin{align}
        \lim_{N\to\infty }\inf_{\|u\|<N^\epsilon }\mathrm{Fid}\left(\mathcal{E}^{\otimes N}(\tilde{\psi}_{u/\sqrt{N}}^{\otimes N}),\tilde{\phi}_{u/\sqrt{N}}^{\otimes N}\right)=1
    \end{align}
    for $\epsilon\in(0,1/6)$.
    From the Fuchs-van de Graaf inequalities, we also have
    \begin{align}
        \lim_{N\to\infty }\sup_{\|u\|<N^\epsilon }T\left(\mathcal{E}^{\otimes N}(\tilde{\psi}_{u/\sqrt{N}}^{\otimes N}),\tilde{\phi}_{u/\sqrt{N}}^{\otimes N}\right)=0.\label{eq:QGT_reducing_transf_channels}
    \end{align}
    Therefore, by using Lemma~\ref{lem:interconversion_projected_generators}, we get
    \begin{align}
        \lim_{N\to\infty }\sup_{\|u\|<N^\epsilon }T\left(\mathcal{E}^{\otimes N}(\psi_{u/\sqrt{N}}^{\otimes N}),\phi_{u/\sqrt{N}}^{\otimes N}\right)=0
    \end{align}
    for $\epsilon\in(0,1/6)$.
\end{proof}

By using Lemma~\ref{lem:rate_changing_part} and Lemma~\ref{lem:covariance_decreasing_part}, we complete the proof of Lemma~\ref{lem:QLAN_asymptotic_conversion}.
\begin{lem*}[Restatement of Lemma~\ref{lem:QLAN_asymptotic_conversion}]
    For sets of Hermitian operators $ \{X_\mu\}_{\mu=1}^m$ and $\{X_\mu'\}_{\mu=1}^m$, we define $\mathcal{U}_{\theta}(\cdot)\coloneqq e^{\ii \theta^\mu X_\mu}(\cdot)e^{-\ii \theta^\nu X_\nu}$ and $\mathcal{U}'_{\theta}(\cdot)\coloneqq e^{\ii \theta^\mu X_\mu'}(\cdot)e^{-\ii\theta^\nu X_\nu'}$ for $\theta\in\mathbb{R}^m$. For given two pure states $\psi$ and $\phi$, we define pure-state statistical models by $\mathcal{U}_{\theta}(\psi)$ and $\mathcal{U}'_{\theta}(\phi)$ and denote their QGTs by $\mathcal{Q}^\psi$ and $\mathcal{Q}^\phi$, respectively. If $r>0$ satisfies $\mathcal{Q}^\psi\geq r\mathcal{Q}^\phi$, then there exists a sequence of quantum channels $\{\mathcal{E}_N\}_N$ such that 
    \begin{align}
        &\lim_{N\to\infty}\sup_{\|u\|<N^\epsilon}T\left(\mathcal{E}_{N}\left(\mathcal{U}_{\frac{u}{\sqrt{N}}}(\psi)^{\otimes N}\right),\mathcal{U}'_{\frac{u}{\sqrt{N}}}(\phi)^{\otimes \floor{rN}}\right)=0
    \end{align}
    for $\epsilon\in(0,1/9)$.
\end{lem*}
\begin{proof}
    The assumption $\mathcal{Q}^\psi\geq r\mathcal{Q}^\phi$ is equivalent to $\cov(\psi,\{X_\mu\}_{\mu=1}^m)\geq r \,\cov(\phi,\{X_\mu'\}_{\mu=1}^m)$. Lemma~\ref{lem:covariance_decreasing_part} implies that there is a channel $\mathcal{E}$ such that
    \begin{align}
        \lim_{N\to\infty}\sup_{\|u\|<N^\epsilon} T\left(\left(\mathcal{E}\left(e^{\ii \frac{u^\mu X_\mu}{\sqrt{N}}}\psi e^{-\ii \frac{u^\mu X_\mu}{\sqrt{N}}}\right)\right)^{\otimes N},\left(e^{\ii \frac{\sqrt{r}u^\mu X_\mu'}{\sqrt{N}}}\phi e^{-\ii \frac{\sqrt{r}u^\mu X_\mu'
        }{\sqrt{N}}}\right)^{\otimes N}\right)=0
    \end{align}
    for $\epsilon\in(0,1/9)$. On the other hand, Lemma~\ref{lem:rate_changing_part} implies there is a sequence of channels $\{\mathcal{E}_N'\}_N$ such that
    \begin{align}
        \lim_{N\to\infty}\sup_{\|u\|<N^\epsilon} T\left(\mathcal{E}'_N\left(e^{\ii \frac{\sqrt{r}u^\mu X_\mu'}{\sqrt{N}}}\phi e^{-\ii \frac{\sqrt{r}u^\mu X_\mu'
        }{\sqrt{N}}}\right)^{\otimes N}, \left(e^{\ii \frac{u^\mu X_\mu'}{\sqrt{N}}}\phi e^{-\ii \frac{u^\mu X_\mu'
        }{\sqrt{N}}}\right)^{\otimes \floor{rN}}\right)=0.
    \end{align}
    Therefore, for $\mathcal{E}_N'\circ \mathcal{E}^{\otimes N}$, we have
    \begin{align}
        \lim_{N\to\infty}\sup_{\|u\|<N^\epsilon} T\left(\mathcal{E}'_N\circ\mathcal{E}^{\otimes N}\left(e^{\ii \frac{u^\mu X_\mu}{\sqrt{N}}}\psi e^{-\ii \frac{u^\mu X_\mu
        }{\sqrt{N}}}\right)^{\otimes N}, \left(e^{\ii \frac{u^\mu X_\mu'}{\sqrt{N}}}\phi e^{-\ii \frac{u^\mu X_\mu
        }{\sqrt{N}}}\right)^{\otimes \floor{rN}}\right)=0.
    \end{align}
\end{proof}

\subsubsection{Proof of measurability of the set defined in Eq.~\texorpdfstring{\eqref{eq:successful_estimation_set}}{(\getrefnumber{eq:successful_estimation_set}}}
\label{sec:measurability}
For $\delta>0$, let us define
\begin{align}
    H(\delta)\coloneqq \{\hat{g}\in G\mid \exists \theta,\, \mathcal{U}_{\theta}\circ \mathcal{U}_{\hat{g}}(\psi)=\mathcal{U}_g(\psi),\,\|\theta\|\leq \delta\}
\end{align}
and prove that $H(\delta)$ is a closed set. Let $\{h_n\}_{n\in\mathbb{N}}$ be a convergent sequence in $H(\delta)$ such that $h_n\to h$ as $n\to\infty$. From the definition of $H(\delta)$, there exists a sequence $\{\theta_n\}_{n\in\mathbb{N}}$ such that $\mathcal{U}_{\theta_n}\circ\mathcal{U}_{h_n}(\psi)=\mathcal{U}_{g}(\psi)$ and $\|\theta_n\|\leq \delta$. Due to the compactness of the closed ball of radius $\delta$, we can extract a convergent subsequence $\{\theta_{n_k}\}_{k\in\mathbb{N}}$ such that $\theta_{n_k}\to\theta_0$ for some $\theta_0$ as $k\to\infty$. Then, we have
\begin{align}
    \mathcal{U}_{g}(\psi)=\mathcal{U}_{\theta_{n_k}}\circ\mathcal{U}_{h_{n_k}}(\psi)\to\mathcal{U}_{\theta_0}\circ\mathcal{U}_{h}(\psi),\qquad  \delta\geq \|\theta_{n_k}\|\to\|\theta_0\|
\end{align}
as $k\to\infty$. Therefore, $h\in H(\delta)$, implying that $H(\delta)$ is a closed set. 

From the definition of $ G_{\mathrm{succ.}}^{(g,\delta)}$, we have
\begin{align}
    G_{\mathrm{succ.}}^{(g,\delta)}=\bigcup_{n\in\mathbb{N}}H\left(\delta-\frac{1}{n}\right).
\end{align}
Therefore, $ G_{\mathrm{succ.}}^{(g,\delta)}$ is a countable union of closed sets and hence is a measurable set.

\subsubsection{Proof of Lemma~\ref{lem:reasonable_estimator}}\label{sec:reasonable_estimator}

In \cite{guta_fast_2020}, a bound on the error in a state tomography process is proven:
\begin{lem}[Theorem~1 in \cite{guta_fast_2020}]\label{lem:state_tomography}
    Fix any qudit state $\rho$. There exists an estimator $\hat{\rho}_n$ consuming $n$ samples of states, $\rho^{\otimes n}$, such that
    \begin{align}
        \mathrm{Pr}\left(\|\hat{\rho}_n-\rho\|_1\geq \delta\right)\leq d e^{-\frac{n\delta^2}{43 g(d)r^2}},\quad \delta\in[0,1],
    \end{align}
    where $r\coloneqq \min\{\mathrm{rank}(\rho),\mathrm{rank}(\hat{\rho}_n)\}$ and $g(d)$ is a constant depending only on the dimension of the Hilbert space $d$. 
\end{lem}

To prove Lemma~\ref{lem:reasonable_estimator}, we relate the magnitude of error in trace distance to that in the parameter $\theta$. For this purpose, we begin by proving several lemmas. Since any compact Lie group $G$ is isomorphic to a closed linear group, we assume $G$ is a compact linear Lie group and adopt its Schatten norm in the following arguments.
\begin{lem} \label{lem:inverse_map_continuity}
    Let $U$ be a unitary representation of a compact linear Lie group $G$. 
    Then, 
    \begin{align}
        \forall \epsilon>0\ \exists \delta>0\ \forall g, g'\in G\ \left(\|U(g')\rho U(g')^\dag-U(g)\rho U(g)^\dag\|_1<\delta\rightarrow \min_{h\in \mathrm{Sym}_G(\rho)}\|g'-gh\|_\infty<\epsilon\right). 
    \end{align}
\end{lem}

\begin{proof}
    Since
    \begin{align}
        &\|U(g')\rho U(g')^\dag-U(g)\rho U(g)^\dag\|_1 
        =\|U(g^{-1}g')\rho U(g^{-1}g')^\dag-\rho\|_1, \\
        &\|g'-gh\|_\infty
        =\|g^{-1}g'-h\|_\infty,
    \end{align}
    it is sufficient to prove that 
    \begin{align}
        \forall \epsilon>0\ \exists \delta>0\ \forall g\in G\ \left(\|U(g)\rho U(g)^\dag-\rho\|_1<\delta\rightarrow \min_{h\in \mathrm{Sym}_G(\rho)}\|g-h\|_\infty<\epsilon\right). 
    \end{align}
    We suppose that this statement does not hold. 
    Then, we can take $\epsilon>0$ and a sequence $(g_n)_{n\in\mathbb{N}}$ in $G$ such that 
    \begin{align}
        &\|U(g_n)\rho U(g_n)^\dag-\rho\|_1<\frac{1}{n}, \\
        &\min_{h\in \mathrm{Sym}_G(\rho)}\|g_n-h\|_\infty\geq\epsilon 
    \end{align}
    for all $n\in\mathbb{N}$. 
    Since $G$ is sequentially compact, we can take some subsequence $(g_{n(j)})_{j\in\mathbb{N}}$ that converges to $k\in G$. 
    This subsequence satisfies 
    \begin{align}
        &\|U(g_{n(j)})\rho U(g_{n(j)})^\dag-\rho\|_1<\frac{1}{n(j)}\leq \frac{1}{j}, \label{eq:subseq_cond1}\\
        &\min_{h\in \mathrm{Sym}_G(\rho)}\|g_{n(j)}-h\|_\infty\geq\epsilon \label{eq:subseq_cond2}
    \end{align}
    for all $j\in\mathbb{N}$. 
    By taking the limit of $j\to\infty$ in Eq.~\eqref{eq:subseq_cond1}, we get $\|U(k)\rho U(k)^\dag-\rho\|_1=0$, which implies $k\in \mathrm{Sym}_G(\rho)$. 
    Thus Eq.~\eqref{eq:subseq_cond2} implies that 
    \begin{align}
        \|g_{n(j)}-k\|_\infty\geq\min_{h\in \mathrm{Sym}_G(\rho)}\|g_{n(j)}-h\|_\infty\geq\epsilon 
    \end{align}
    for all $j\in\mathbb{N}$. 
    This contradicts with the fact that the subsequence $(g_{n(j)})_{j\in\mathbb{N}}$ converges to $k$. 
\end{proof}

\begin{lem} \label{lem:exp_2nd_order_bound}
    Let $\psi$ be a linear operator and $A$ be a Hermitian operator. 
    Then, 
    \begin{align}
        \left\|e^{\ii A}\psi e^{-\ii A}-\psi\right\|_1 
        \geq\left(1-\|A\|_\infty\right)\|[A, \psi]\|_1. \label{eq:lem:exp_2nd_order_bound1}
    \end{align}
\end{lem}

\begin{proof}
    We define $F(u):=e^{\ii uA}\psi e^{-\ii uA}$ for $u\in\mathbb{R}$. 
    Then, we have 
    \begin{align}
        e^{\ii A}\psi e^{-\ii A}-\psi-\ii [A, \psi]
        =F(1)-F(0)-F'(0) 
        =\int_0^1 \dd u\int_0^u \dd v F''(v). 
    \end{align}
    By the triangle inequality and H\"{o}lder's inequality, we get 
    \begin{align}
        \left\|e^{\ii A}\psi e^{-\ii A}-\psi-\ii [A, \psi]\right\|_1
        \leq\int_0^1 \dd u\int_0^u \dd v \|F''(v)\|_1 
        =\frac{1}{2}\|[A, [A, \psi]]\|_1 
        \leq\|A\|_\infty\|[A, \psi]\|_1. \label{eq:lem:exp_2nd_order_bound2}
    \end{align}
    By using the triangle inequality again, we have 
    \begin{align}
        \|e^{\ii A}\psi e^{-\ii A}-\psi-\ii [A, \psi]\|_1 
        \geq\|[A, \psi]\|_1-\|e^{\ii A}\psi e^{-\ii A}-\psi\|_1. \label{eq:lem:exp_2nd_order_bound3}
    \end{align}
    By Eqs.~\eqref{eq:lem:exp_2nd_order_bound2} and \eqref{eq:lem:exp_2nd_order_bound3}, we get Eq.~\eqref{eq:lem:exp_2nd_order_bound1}. 
\end{proof}

We show an upper bound of $\left\|e^{\ii (A+B)}e^{-\ii A}-I\right\|_\infty$ for general Hermitian operators $A$ and $B$. 

\begin{lem} \label{lem:exp_norm_bound}
    Let $A$ and $B$ be Hermitian operators and satisfy $\|A+B\|_\infty+2\|B\|_\infty<1$. 
    Then, 
    \begin{align}
        \left\|e^{\ii (A+B)}e^{-\ii A}-I\right\|_\infty 
        \leq\frac{\|B\|_\infty}{1-\|A+B\|_\infty-2\|B\|_\infty}. 
    \end{align}
\end{lem}

\begin{proof}
    By using the triangle inequality, we have
    \begin{align}
        \left\|e^{\ii (A+B)}e^{-\ii A}-I\right\|_\infty 
        =&\left\|e^{\ii (A+B)}-e^{\ii A}\right\|_\infty \nonumber\\
        =&\left\|\sum_{k=0}^\infty \frac{\ii^k}{k!} \left[(A+B)^k-A^k\right]\right\|_\infty \nonumber\\
        \leq&\sum_{k=0}^\infty \frac{1}{k!} \left\|\left[(A+B)^k-A^k\right]\right\|_\infty.
    \end{align}
    
    Now let us prove $\left\|\left[(A+B)^k-A^k\right]\right\|_\infty\leq(\|A\|_\infty+\|B\|_\infty)^k-\|A\|_\infty^k$ by the mathematical induction. This inequality is trivial for $k=0$. Assume that the inequality holds for some $k$. Then we get
    \begin{align}
        \|(A+B)^{k+1}-A^{k+1}\|_\infty&=\|(A+B)((A+B)^{k}-A^{k})-BA^k\|_\infty\\
        &\leq \|(A+B)\|_\infty \|(A+B)^{k}-A^{k}\|_\infty+\|BA^k\|_\infty\\
        &\leq (\|A\|_\infty+\|B\|_\infty) ((\|A\|_\infty+\|B\|_\infty)^k-\|A\|_\infty^k)+\|B\|_\infty \|A\|_\infty^k\\
        &=(\|A\|_\infty+\|B\|_\infty)^{k+1}-\|A\|_\infty^{k+1}.
    \end{align}
    Thus, the inequality also holds for $k+1$.     
    
    Therefore, we get
    \begin{align}
        \left\|e^{\ii (A+B)}e^{-\ii A}-I\right\|_\infty 
        &\leq\sum_{k=0}^\infty \frac{1}{k!} [(\|A\|_\infty+\|B\|_\infty)^k-\|A\|_\infty^k] \nonumber\\
        &= e^{\|A\|_\infty+\|B\|_\infty}-e^{\|A\|_\infty} \nonumber\\
        &\leq \|B\|_\infty e^{\|A\|_\infty+\|B\|_\infty} \nonumber\\
        &\leq\|B\|_\infty e^{\|A+B\|_\infty+2\|B\|_\infty} \nonumber\\
        &\leq\frac{\|B\|_\infty}{1-\|A+B\|_\infty-2\|B\|_\infty}, 
    \end{align}
    where we used $e^x-1\leq xe^x$ for all $x\in\mathbb{R}$ in the second inequality, the triangle inequality in the third inequality, and $e^x\leq 1/(1-x)$ for all $x\in[0, 1)$ in the fourth inequality. 
\end{proof}

When $\|A+B\|_\infty+3\|B\|_\infty<1$, we show an upper bound of $\left\|\log\left(e^{\ii(A+B)}e^{-\ii A}\right)\right\|_\infty$. Note that $\log\left(e^{\ii(A+B)}e^{-\ii A}\right)$ is guaranteed to be defined by Lemma~\ref{lem:exp_norm_bound}.

\begin{lem} \label{lem:log_norm_bound}
    Let $A$ and $B$ be Hermitian operators and satisfy $\|A+B\|_\infty+3\|B\|_\infty<1$. 
    Then, 
    \begin{align}
        \left\|\log\left(e^{\ii(A+B)}e^{-\ii A}\right)\right\|_\infty
        \leq\frac{\|B\|_\infty}{1-\|A+B\|_\infty-3\|B\|_\infty}. 
    \end{align}
\end{lem}

\begin{proof}
    By the definition of the logarithm, we have 
    \begin{align}
        \log\left(e^{\ii (A+B)}e^{-\ii A}\right) 
        =\sum_{k=1}^\infty \frac{(-1)^{k-1}}{k}\left(e^{\ii(A+B)}e^{-\ii A}-I\right)^k. 
    \end{align}
    By the triangle inequality, we have 
    \begin{align}
        \left\|\log\left(e^{\ii(A+B)}e^{-\ii A}\right)\right\|_\infty 
        \leq\sum_{k=1}^\infty \frac{1}{k}\left\|e^{\ii(A+B)}e^{-\ii A}-I\right\|_\infty^k. 
    \end{align}
    By applying Lemma~\ref{lem:exp_norm_bound} to the right-hand side of this inequality, we get 
    \begin{align}
        \left\|\log\left(e^{\ii(A+B)}e^{-\ii A}\right)\right\|_\infty 
        \leq&\sum_{k=1}^\infty \frac{1}{k}\left(\frac{\|B\|_\infty}{1-\|A+B\|_\infty-2\|B\|_\infty}\right)^k \nonumber\\
        \leq&\sum_{k=1}^\infty \left(\frac{\|B\|_\infty}{1-\|A+B\|_\infty-2\|B\|_\infty}\right)^k \nonumber\\
        =&\frac{\|B\|_\infty}{1-\|A+B\|_\infty-3\|B\|_\infty}. 
    \end{align}
\end{proof}

By using these lemmas, we prove the following:
\begin{lem}
\label{trace_distance_and_parameter_linearly_independent}
    Let $\rho$ be an operator, and $U$ be a unitary representation of a compact Lie group $G$. Then there exist positive constants $\delta_*,c>0$ such that for any $g,\hat{g}\in G$, 
    \begin{align}
         \|\mathcal{U}_g(\rho)-\mathcal{U}_{\hat{g}}(\rho)\|_1<\delta_*
         &\implies \exists \varphi\,\text{ s.t. }
         \begin{cases}
             &\mathcal{U}_{\varphi}\circ \mathcal{U}_{\hat{g}}(\rho)=\mathcal{U}_g(\rho)\\
             &\|\varphi\|\leq c\|\mathcal{U}_g(\rho)-\mathcal{U}_{\hat{g}}(\rho)\|_1
         \end{cases},\label{eq:varphi_norm}
    \end{align}
    where $\mathcal{U}_{\varphi}(\cdot)\coloneqq e^{\ii \sum_{k=1}^n\varphi^kY_k}(\cdot )e^{-\ii \sum_{k=1}^n\varphi^k Y_k}$ with $\varphi\in\mathbb{R}^{n}$ and $\{Y_k\}_{k=1}^{n}$ is a basis of the linear span of $\{X_\mu\}_{\mu=1}^{\dim G}$, where $X_\mu$ is defined in Eq.~\eqref{eq:hermitian_operators_Lie_alg}. 
\end{lem}

\begin{proof}
    We denote by $\mathcal{A}$ the linear space spanned by $\{Y_i\}_{i=1}^n$, and define its linear subspace $\mathcal{A}_0$ by 
    \begin{align}
        \mathcal{A}_0:=\{A\in\mathcal{A}\mid [A, \rho]=0\}
    \end{align}
    and its complementary subspace $\mathcal{A}_1$. 
    
    Let us first consider an exceptional case where $\mathcal{A}_1=\emptyset$, i.e., $\mathcal{A}=\mathcal{A}_0$. Let $G=\bigsqcup_{i=0}^k G_i$ be the decomposition into the connected components. Since $X_\mu$ commutes with $\rho$ for all $\mu=1,\cdots,\dim G$, we find $\mathcal{U}_g(\rho)=\mathcal{U}_{g'}(\rho)$ if $g$ and $g'$ are in the same component. Note that the number of connected components is finite since $G$ is assumed to be a compact Lie group. Therefore, there are only a finite number of distinct elements in the set $\{\mathcal{U}_g(\rho)\mid g\in G\}$. By setting
    \begin{align}
        \delta_*\coloneqq \min_{g,g'\in G;\mathcal{U}_g(\rho)\neq \mathcal{U}_{g'}(\rho)}\left\|\mathcal{U}_g(\rho)- \mathcal{U}_{g'}(\rho)\right\|_1
    \end{align}
    and $c$ to be arbitrary positive number, Eq.~\eqref{eq:varphi_norm} holds for $\varphi=0$. We remark that this essentially includes the case where $G$ is a finite group. 
    
    From now on, we consider the case where $\mathcal{A}_1\neq \emptyset$. In this case, let us define
     \begin{align}
        m\coloneqq \min_{\|D\|_\infty=1, D\in\mathcal{A}_1} \|[D, \rho]\|_1>0.
    \end{align}
    By Lemma~\ref{lem:inverse_map_continuity} and the continuity of $U$, we can take $\delta>0$ such that 
    \begin{align}
        \forall g, g'\in G\ (\|\mathcal{U}_g(\rho)-\mathcal{U}_{\hat{g}}(\rho)\|_1<\delta 
        \rightarrow \min_{h\in  \mathrm{Sym}_G(\rho)}\|U(g^{-1}\hat{g}h)-I\|_\infty<1-e^{-\frac{1}{8}}).
    \end{align}
    We define 
    \begin{align}
        l\coloneqq \min_{\|\bm{\gamma}\|=1} \left\|\sum_k \gamma^k Y_k\right\|_\infty.  \label{eq:linear_independence}
    \end{align}
    Note that $l>0$ since $\{Y_k \}$ are assumed to be linearly independent. We set $\delta_*\coloneqq \min\{\delta, m/16\}$ and $c\coloneqq 1/(lm)$ and prove that they satisfy Eq.~\eqref{eq:varphi_norm}.

    Fix arbitrary $g, \hat{g}\in G$ satisfying $\|\mathcal{U}_g(\rho)-\mathcal{U}_{\hat{g}}(\rho)\|_1<\delta$, and find appropriate $\varphi$ satisfying Eq.~\eqref{eq:varphi_norm}. 
    By the definition of $\delta$, we can take $h\in \mathrm{Sym}_{G}(\rho)$ such that $\|U(g^{-1}\hat{g}h)-I\|_\infty\leq 1-e^{-1/8}<1$. 
    Thus we can define 
    \begin{align}
        A:=-\ii\log(U(g^{-1}\hat{g}h)). \label{eq:prop:inverse_map_Lipschitz1}
    \end{align}
    By the definition of the logarithm, we have 
    \begin{align}
        \|A\|_\infty
        =&\left\|\sum_{k=1}^\infty \frac{(-1)^{k-1}}{k}\left(U(g^{-1}\hat{g}h)-I\right)^k\right\|_\infty \nonumber\\
        \leq&\sum_{k=1}^\infty \frac{1}{k}\|U(g^{-1}\hat{g}h)-I\|_\infty^k \nonumber\\
        =&-\log(1-\|U(g^{-1}\hat{g}h)-I\|_\infty) \nonumber\\
        \leq&\frac{1}{8}, \label{eq:prop:inverse_map_Lipschitz2}
    \end{align}
    where we used the triangle inequality in the first inequality and the second inequality follows from the choice of $\delta$. 
    By the definition of $A$, we have 
    \begin{align}
        \|U(\hat{g})\rho U(\hat{g})^\dag-U(g)\rho U(g)^\dag\|_\infty
        =&\|U(g^{-1}\hat{g}h)\rho U(g^{-1}\hat{g}h)^\dag-\rho\|_\infty \nonumber\\
        =&\|e^{\ii A}\rho e^{-\ii A}-\rho\|_\infty \nonumber\\
        \geq& (1-\|A\|_\infty)\|[A,\rho]\|_\infty \nonumber\\
        \geq&\frac{1}{2}\|[A, \rho]\|_\infty, \label{eq:prop:inverse_map_Lipschitz3}
    \end{align}
    where we used Lemma~\ref{lem:exp_2nd_order_bound} in the first inequality and Eq.~\eqref{eq:prop:inverse_map_Lipschitz2} in the second inequality. 
    Since $\mathcal{A}_0$ and $\mathcal{A}_1$ are complementary subspaces of $\mathcal{A}$, $A$ can be decomposed into $A=A_0+A_1$ with some $A_0\in\mathcal{A}_0$ and $A_1\in\mathcal{A}_1$. 
    We can take some $D\in\mathcal{A}_1$ such that $A_1=\|A_1\|_\infty D$ and $\|D\|_\infty=1$. 
    Then, we have 
    \begin{align}
        \|[A, \rho]\|_\infty 
        =\|[A_1, \rho]\|_\infty 
        =\|A_1\|_\infty\|[D, \rho]\|_1 
        \geq m\|A_1\|_\infty. \label{eq:prop:inverse_map_Lipschitz4}
    \end{align}
    By using Eqs.~\eqref{eq:prop:inverse_map_Lipschitz4} and \eqref{eq:prop:inverse_map_Lipschitz3}, we have 
    \begin{align}
        \|A_1\|
        \leq\frac{1}{m}\|[A, \rho]\|_\infty 
        \leq\frac{2}{m}\|U(\hat{g})\rho U(\hat{g})^\dag-U(g)\rho U(g)^\dag\|_\infty. \label{eq:prop:inverse_map_Lipschitz5}
    \end{align}
    Since $\|U(\hat{g})\rho U(\hat{g})^\dag-U(g)\rho U(g)^\dag\|_\infty$ is upper bounded by $\delta$ and $\delta\leq m/16$, we get 
    \begin{align}
        \|A_1\| 
        \leq\frac{2\delta}{m} 
        \leq\frac{1}{8}. \label{eq:prop:inverse_map_Lipschitz6}
    \end{align}
    By Eqs.~\eqref{eq:prop:inverse_map_Lipschitz2} and \eqref{eq:prop:inverse_map_Lipschitz6} and Lemma~\ref{lem:exp_norm_bound}, we get $\|e^{\ii A}e^{-\ii A_0}-I\|_\infty\leq 1/5<1$, which enables us to define 
    \begin{align}
        B:=-\ii \log(e^{\ii A}e^{-\ii A_0}). \label{eq:prop:inverse_map_Lipschitz7}
    \end{align}

    We define $\varphi$ by
    \begin{align}
        \sum_k \varphi^k Y_k=U(g)BU(g)^\dag \label{eq:prop:inverse_map_Lipschitz8}
    \end{align}
    and show that
    \begin{align}
        &\|\varphi\|\leq\frac{1}{lm}\|U(\hat{g})\rho U(\hat{g})^\dag-U(g)\rho U(g)^\dag\|_1, \label{eq:prop:inverse_map_Lipschitz3_1}\\
        &e^{-\ii\sum_k \varphi^k Y_k}U(\hat{g})\rho U(\hat{g})^\dag e^{\ii\sum_k \varphi^k Y_k}=U(g)\rho U(g)^\dag. \label{eq:prop:inverse_map_Lipschitz3_2}
    \end{align}
    As for the proof of Eq.~\eqref{eq:prop:inverse_map_Lipschitz3_2}, by using Eqs.~\eqref{eq:prop:inverse_map_Lipschitz8}, \eqref{eq:prop:inverse_map_Lipschitz7}, and \eqref{eq:prop:inverse_map_Lipschitz1} we have 
    \begin{align}
        U(\hat{g})^\dag e^{\ii\sum_k \varphi^k Y_k}
        =&U(\hat{g}^{-1})U(g)e^{\ii B}U(g)^\dag \nonumber\\
        =&U(\hat{g}^{-1}g)e^{\ii A}e^{-\ii A_0}U(g)^\dag \nonumber\\
        =&U(\hat{g}^{-1}g)U(g^{-1}\hat{g}h)e^{-\ii A_0}U(g)^\dag \nonumber\\
        =&U(h)e^{-\ii A_0}U(g)^\dag, 
    \end{align}
    which implies that 
    \begin{align}
        e^{-\ii \sum_k \varphi^k Y_k}U(\hat{g})\rho U(\hat{g})^\dag e^{\ii \sum_k \varphi^k Y_k} 
        =U(g)e^{\ii A_0}U(h)^\dag\rho U(h)e^{-\ii A_0}U(g)^\dag 
        =U(g)\rho U(g)^\dag. 
    \end{align}
    
    Finally, we prove Eq.~\eqref{eq:prop:inverse_map_Lipschitz3_1}. 
    By Eqs.~\eqref{eq:prop:inverse_map_Lipschitz2} and \eqref{eq:prop:inverse_map_Lipschitz6} and Lemma~\ref{lem:log_norm_bound}, we get 
    \begin{align}
        \|B\|_\infty
        \leq\frac{\|A_1\|_\infty}{1-\|A\|_\infty-3\|A_1\|_\infty} 
        \leq\frac{\|A_1\|_\infty}{2}. \label{eq:prop:inverse_map_Lipschitz3_3}
    \end{align}
    Since we can take $\gamma$ satisfying $\varphi=\|\varphi\|\gamma$ and $\|\gamma\|=1$, we have 
    \begin{align}
        \|B\|_\infty
        =\left\|\sum_k \varphi^k Y_k\right\|_\infty 
        =\|\varphi\|\left\|\sum_k \varphi^k Y_k\right\|_\infty 
        \geq l\|\varphi\|. \label{eq:prop:inverse_map_Lipschitz3_4}
    \end{align}
    By using Eqs.~\eqref{eq:prop:inverse_map_Lipschitz3_4},\eqref{eq:prop:inverse_map_Lipschitz3_3}, and \eqref{eq:prop:inverse_map_Lipschitz5}, we get 
    \begin{align}
        \|\varphi\|
        \leq\frac{\|B\|_\infty}{l} 
        \leq\frac{\|A_1\|_\infty}{2l} 
        \leq\frac{1}{lm}\|U(\hat{g})\psi U(\hat{g})^\dag-U(g)\psi U(g)^\dag\|_\infty\leq \frac{1}{lm}\|U(\hat{g})\psi U(\hat{g})^\dag-U(g)\psi U(g)^\dag\|_1. 
    \end{align}
\end{proof}

As a generalization of this lemma, we prove the following:
\begin{lem}\label{lem:trace_distance_and_parameter}
    Let $\rho$ be an operator, and $U$ be a unitary representation of a compact Lie group $G$. Then there exist positive constants $\delta_*,c>0$ such that for any $g,\hat{g}\in G$, 
    \begin{align}
         \|\mathcal{U}_g(\rho)-\mathcal{U}_{\hat{g}}(\rho)\|_1<\delta_*
         &\implies \exists \theta\,\text{ s.t. }
         \begin{cases}
             &\mathcal{U}_{\theta}\circ \mathcal{U}_{\hat{g}}(\rho)=\mathcal{U}_g(\rho)\\
             &\|\theta\|\leq c\|\mathcal{U}_g(\rho)-\mathcal{U}_{\hat{g}}(\rho)\|_1
         \end{cases},
       \label{eq:theta_norm}
    \end{align}
    where $\mathcal{U}_{\theta}(\cdot)\coloneqq e^{\ii \theta^\mu X_\mu}(\cdot )e^{-\ii \theta^\nu X_\nu}$ for $\theta\in\mathbb{R}^{\dim G}$, where $X_\mu$ is defined in Eq.~\eqref{eq:hermitian_operators_Lie_alg}. 
\end{lem}

\begin{proof}
    Since $\{Y_i\}_{i=1}^{n}$ in Lemma~\ref{trace_distance_and_parameter_linearly_independent} is a basis of the linear span of $\{X_\mu\}_{\mu=1}^{\dim G}$, there is an $n\times \dim G$ matrix $L$ such that $\sum_{\mu=1}^{\dim G}L_{i\mu}X_\mu=Y_i$. By using the transpose $L^\top$ of $L$, we define $\theta\in\mathbb{R}^{\dim G}$ by $\theta\coloneqq L^\top\varphi$, which satisfies $\theta^\mu X_\mu=\sum_k\varphi^kY_k$. Note that $\|\theta\|\leq \|L^\top\|_{2}\|\varphi\|$ holds, where $\|L^\top\|_{2}$ denotes the matrix norm induced from the 2-norm for vectors, which is finite. Therefore, Lemma~\ref{lem:trace_distance_and_parameter} holds by replacing the constant $c$ in Lemma~\ref{trace_distance_and_parameter_linearly_independent} to $c\|L^\top\|_{2}$. 
\end{proof}

Combining Lemma~\ref{lem:state_tomography} and Lemma~\ref{lem:trace_distance_and_parameter}, we finally prove Lemma~\ref{lem:reasonable_estimator}:
\begin{lem*}[Restatement of Lemma~\ref{lem:reasonable_estimator}]
    Let $G$ be a compact Lie group and $\rho$ be an arbitrary state. Fix $\epsilon\in (0,1/2)$. Then there exists an estimator of $g\in G$, which consumes $\mathcal{U}_g(\rho)^{\otimes n}$ with $n=\ceil{N^{1-\epsilon}}$, such that its worst-case success probability satisfies
    \begin{align}
        \lim_{N\to\infty }p^{\mathrm{succ.}}(N^{-1/2+\epsilon})=1.
    \end{align}
\end{lem*}
\begin{proof}
    Fix $\epsilon\in(0,1/2)$. Let $\delta_*,c>0$ be constants assured to exist by Lemma~\ref{lem:trace_distance_and_parameter}. 
    Applying Lemma~\ref{lem:state_tomography} to a given state $\mathcal{U}_g(\rho)$ with $\delta\coloneqq N^{-1/2+\epsilon}/(2c)$, we obtain a bound for the error probability in state tomography given by
    \begin{align}
        &\mathrm{Pr}\left(\|\hat{\rho}_n-\mathcal{U}_g(\rho)\|_1< N^{-1/2+\epsilon}/(2c)\right)\nonumber\\
        &\geq 1-d e^{-\frac{n(N^{-1/2+\epsilon}/(2c))^2}{43 g(d)r^2}}\nonumber\\
        &\geq 1-d e^{-\frac{N^\epsilon}{172 g(d)r^2 c^2}}
    \end{align}
    by performing measurements on $n=\ceil{N^{1-\epsilon}}$ copies of the system in a state $\mathcal{U}_g(\rho)^{\otimes n}$, where we have used $nN^{-1+2\epsilon}\geq N^{1-\epsilon}N^{-1+2\epsilon}=N^\epsilon$. 
    From the estimated state $\hat{\rho}_n$, we pick up any $\hat{g}\in G$ satisfying $\|\mathcal{U}_{\hat{g}}(\rho)-\hat{\rho}_n\|_1< N^{-1/2+\epsilon}/(2c)$ as an estimate of the true value $g\in G$. If there is no $\hat{g}$ satisfying $\|\mathcal{U}_{\hat{g}}(\rho)-\hat{\rho}_n\|_1< N^{-1/2+\epsilon}/(2c)$, which may happen with an exponentially small probability, we arbitrarily select $\hat{g}$ as an element in $G$, such as $e\in G$. Whenever $\|\hat{\rho}_n-\mathcal{U}_g(\rho)\|_1< N^{-1/2+\epsilon}/(2c)$ holds, the estimated value $\hat{g}\in G$ satisfies $\|\mathcal{U}_{\hat{g}}(\rho)-\mathcal{U}_g(\rho)\|_1<N^{-1/2+\epsilon}/c$ due to the triangle inequality. In this case, for all sufficiently large $N$ such that $N^{-1/2+\epsilon}/c<\delta_*$, we get $\|\theta\|<N^{-1/2+\epsilon}$. Therefore, the success probability is bounded as
    \begin{align}
        \int_G\dd \mu_G(\hat{g})\, \,p(\hat{g}|\mathcal{U}_g(\psi)^{\otimes n})\chi_{ G_{\mathrm{succ.}}^{(g,\delta)}}(\hat{g})\geq 1-c_1 e^{-c_2N^\epsilon},
    \end{align}
    where $c_1\coloneqq d$ and $c_2\coloneqq \frac{1}{172 g(d)r^2 c^2}$. In other words, the failure probability is exponentially small in $N$. Since the rank of the state, $r$, is invariant under unitary transformation and hence the constants $c_1,c_2$ are independent of $g$, this bound is uniform in $g\in G$, implying that 
    \begin{align}
        p^{\mathrm{succ.}}(N^{-1/2+\epsilon})=\inf_{g\in G}\int_G\dd \mu_G(\hat{g})\, \,p(\hat{g}|\mathcal{U}_g(\psi)^{\otimes n})\chi_{ G_{\mathrm{succ.}}^{(g)}(\delta)}(\hat{g})\geq 1-c_1 e^{-c_2N^\epsilon}\label{eq:success_prob_rate}.
    \end{align}
    Therefore, we get
    \begin{align}
        \lim_{N\to\infty }p^{\mathrm{succ.}}(N^{-1/2+\epsilon})=1.
    \end{align}
\end{proof}

\subsubsection{\texorpdfstring{Proof of Lemma~\protect\ref{lem:Lie_asymptotic_conversion}}{Proof of Lemma~\ref*{lem:Lie_asymptotic_conversion}}}
\label{sec:Lie_asymptotic_conversion}
We here prove Lemma~\ref{lem:Lie_asymptotic_conversion} by using Lemma~\ref{lem:QLAN_asymptotic_conversion} and the relation in Eq.~\eqref{eq:basis_change_Lie_alg}.
\begin{lem}[Restatement of Lemma~\ref{lem:Lie_asymptotic_conversion}]
    Let $U,U'$ be (non-projective) unitary representations of a compact Lie group $G$ on finite-dimensional Hilbert spaces $\mathcal{H}$ and $\mathcal{H}'$. Let $\psi\in\mathcal{P}(\mathcal{H})$ and $\phi\in\mathcal{P}(\mathcal{H}')$ be pure states.  
    If $r>0$ satisfies $\mathcal{Q}^{\mathcal{U}_g(\psi)}\geq r\mathcal{Q}^{\mathcal{U}_g'(\phi)}$ for all $g\in G$, then, there exists a sequence of quantum channels $\{\mathcal{E}^{(g)}_N\}_N$ such that the conversion error
    \begin{align}
        \delta_N(g,u)\coloneqq T\left(\mathcal{E}^{(g)}_N\left(\mathcal{U}_{\frac{u}{\sqrt{N}}}\left(\mathcal{U}_g(\psi)\right)^{\otimes N}\right),\mathcal{U}'_{\frac{u}{\sqrt{N}}}\left(\mathcal{U}_g'(\phi)\right)^{\otimes \floor{rN}}\right)
    \end{align}
    satisfies $\lim_{N\to\infty}\sup_{g\in G}\sup_{\|u\|<N^\epsilon}\delta_N(g,u)=0$ for $\epsilon\in(0,1/9)$.
\end{lem}
\begin{proof}
    Applying Lemma~\ref{lem:QLAN_asymptotic_conversion} to $\psi$ and $\phi$, we find that there exists a sequence of quantum channels $\{\mathcal{E}_N\}_N$ such that
        \begin{align}
        &\lim_{N\to\infty}\sup_{\|u\|<N^\epsilon}T\left(\mathcal{E}_N\left(\mathcal{U}_{\frac{u}{\sqrt{N}}}(\psi)^{\otimes N}\right),\mathcal{U}'_{\frac{u}{\sqrt{N}}}(\phi)^{\otimes \floor{rN}}\right)=0\label{eq:conversion_representatives}
    \end{align}
    for $\epsilon\in(0,1/9)$.

    For a general $g\in G$, we define a channel $\mathcal{E}_N^{(g)}$ by using $\mathcal{E}_N$ as follows: From Eq.~\eqref{eq:basis_change_Lie_alg}, there exists an invertible $\dim G\times \dim G$ real matrix $V(g)$ satisfying 
    \begin{align}
        U(g)^\dag X_\mu U(g)=\sum_{\nu=1}^{\dim G}V(g)_{\nu\mu }X_\nu.
    \end{align}
    We define $\mathcal{E}_{N}^{(g)}\coloneqq \left(\mathcal{U}'_{g}\right)^{\otimes \floor{rN}}\circ \mathcal{E}_N\circ \left(\mathcal{U}_{g}^{-1}\right)^{\otimes N}$. 
    Since $U(g)^\dag (u^\mu X_\mu)U(g)=v^\mu X_\mu$, where $v\coloneqq uV(g)^{\top}$, for $\psi_g\coloneqq \mathcal{U}_g(\psi)$ we have
    \begin{align}
        \mathcal{E}_{N}^{(g)}\left(\mathcal{U}_{\frac{u}{\sqrt{N}}}\left(\psi_g\right)^{\otimes N}\right)&=\left(\mathcal{U}_{g}\right)^{\otimes \floor{rN}}\circ \mathcal{E}_N\left(\mathcal{U}_{g}^{-1}\circ\mathcal{U}_{\frac{u}{\sqrt{N}}}\circ \mathcal{U}_{g}\left(\psi\right)\right)^{\otimes N}\\
        &=\left(\mathcal{U}'_{g}\right)^{\otimes \floor{rN}}\circ \mathcal{E}_N\left(\mathcal{U}_{\frac{v}{\sqrt{N}}}\left(\psi\right)\right).
    \end{align}
    Therefore, we get
    \begin{align}
        T\left(\mathcal{E}_N^{(g)}\left(\mathcal{U}_{\frac{u}{\sqrt{N}}}(\psi_{g})^{\otimes N}\right),\mathcal{U}'_{\frac{u}{\sqrt{N}}}(\phi_{g})^{\otimes \floor{rN}}\right)
        &=T\left(\mathcal{E}_N\left(\mathcal{U}_{\frac{v}{\sqrt{N}}}\left(\psi\right)^{\otimes N}\right),\mathcal{U}'_{\frac{v}{\sqrt{N}}}\left(\phi\right)^{\otimes \floor{rN}}\right),
    \end{align}
    where we have used $U'(h)^\dag (u^\mu X_\mu')U'(h)=v^\mu X_\mu'$. Note that the induced norm $\|V(g)\|_\infty\coloneqq \max_{u\neq 0}\|V(g)u\|/\|u\|$ is finite for any $g\in G$. Since $\|V(h)\|_\infty$ and continuous in $g\in G$ and $G$ is compact, we find $c\coloneqq \sup_{g\in G}\|V(g)\|_\infty$ is also finite. By using $\|v\|\leq c\|u\|$, for any $g\in G$ and $\|u\|<N^\epsilon$, we get
    \begin{align}
        \delta_N(g,u)&=T\left(\mathcal{E}_N^{(g)}\left(\mathcal{U}_{\frac{u}{\sqrt{N}}}(\psi_{g})^{\otimes N}\right),\mathcal{U}'_{\frac{u}{\sqrt{N}}}(\phi_{g})^{\otimes \floor{rN}}\right)\\
        &\leq \sup_{\|v\|<cN^{\epsilon}}T\left(\mathcal{E}_N\left(\mathcal{U}_{\frac{v}{\sqrt{N}}}\left(\psi\right)^{\otimes N}\right),\mathcal{U}'_{\frac{v}{\sqrt{N}}}\left(\phi\right)^{\otimes \floor{rN}}\right).
    \end{align}
    Since $c$ is independent of $N$, for any $\epsilon\in(0,1/9)$, $cN^\epsilon <N^{\epsilon'}$ holds for all sufficiently large $N$ for any $\epsilon'\in(\epsilon,1/9)$. Therefore, 
    \begin{align}
        \lim_{N\to\infty}\sup_{g\in G}\sup_{\|u\|<N^\epsilon}\delta_N(g,u)
        &\leq \lim_{N\to\infty}\sup_{\|v\|<N^{\epsilon'}}T\left(\mathcal{E}_N\left(\mathcal{U}_{\frac{v}{\sqrt{N}}}\left(\psi\right)^{\otimes N}\right),\mathcal{U}'_{\frac{v}{\sqrt{N}}}\left(\phi\right)^{\otimes \floor{rN}}\right)\\
        &=0,
    \end{align}
    where the last equality follows from Eq.~\eqref{eq:conversion_representatives}. 
\end{proof}

\subsection{Asymptotic conversion error}\label{app:conversion_error}
We study the asymptotic behavior of the conversion error in our channels. While we do not aim to determine the optimal error rate, our results will provide a basis for future higher-order analysis. In this subsection, we repeatedly use the Fuchs--van de Graaf inequalities, i.e.,
\begin{align}
    1-\sqrt{\mathrm{Fid}(\rho,\sigma)}\leq T(\rho,\sigma)\leq \sqrt{1-\mathrm{Fid}(\rho,\sigma)}.
\end{align}

In what follows, we fix some $\epsilon\in(0,1/9)$. We first analyze the errors in the conversion from $\{\psi^{\otimes N}_{u/\sqrt{N}}\}_N$ to $\{\phi_{u/\sqrt{N}}^{\otimes \floor{rN}}\}_N$. For this purpose, we investigate each of the following steps~(i)--(iv):
\begin{align}
    \psi_{u/\sqrt{N}}^{\otimes N}\overset{\mathrm{Lem.}~\ref{lem:interconversion_projected_generators}} {\underset{\mathrm{(i)}}{\approx}}\tilde{\psi}_{u/\sqrt{N}}^{\otimes N}\xrightarrow[\mathrm{(ii)}]{\mathrm{Eq.}~\eqref{eq:QGT_reducing_transf_channels}}\tilde{\phi}^{\otimes N}_{\sqrt{r}u/\sqrt{N}}\xrightarrow[\mathrm{(iii)}]{\mathrm{Eq.}\eqref{eq:QGT_conserving_transf_channel}}\tilde{\phi}^{\otimes \floor{rN}}_{u/\sqrt{N}}\overset{\mathrm{Lem.}~\ref{lem:interconversion_projected_generators}}{\underset{\mathrm{(iv)}}\approx}\phi^{\otimes \floor{rN}}_{u/\sqrt{N}}.
\end{align}
Here, we remind the readers of our notations: $\tilde{\psi}_\theta=e^{\ii \theta^\mu \tilde{X}_\mu}\psi e^{-\ii \theta^\nu \tilde{X}_\nu}$ and $\tilde{\phi}_\theta=e^{\ii \theta^\mu \tilde{X}'_\mu}\psi e^{-\ii \theta^\nu \tilde{X}'_\nu}$ for $\theta\in\mathbb{R}^{\dim G}$ with $\tilde{X}_\mu=(I-\psi)X_\mu \psi+\psi X_\mu (I-\psi)$ and $\tilde{X}'_\mu=(I-\phi)X'_\mu \phi+\phi X_\mu'(I-\phi)$.

\medskip
{\bf Steps~(i) and (iv):} 
From Lemma~\ref{lem:interconversion_projected_generators} and Eq.~\eqref{eq:projected_generators_approximation_error} in its proof, we have 
\begin{align}
        \left|\mathrm{Fid}\left(\psi_{u/\sqrt{N}},\tilde{\psi}_{u/\sqrt{N}} \right)-1\right|&\leq a_1N^{-3/2+3\epsilon}\\
        \left|\mathrm{Fid}\left(\phi_{u/\sqrt{N}},\tilde{\phi}_{u/\sqrt{N}} \right)-1\right|&\leq a_2N^{-3/2+3\epsilon}
\end{align}
with some finite non-negative constants $a_1$ and $a_2$ independent of $N$. Since the fidelity is multiplicative, and is less than or equal to one, we obtain 
\begin{align}
        \mathrm{Fid}\left(\psi_{u/\sqrt{N}}^{\otimes N},\tilde{\psi}_{u/\sqrt{N}}^{\otimes N} \right)&=\mathrm{Fid}\left(\psi_{u/\sqrt{N}},\tilde{\psi}_{u/\sqrt{N}}\right)^N\geq \left(1-a_1N^{-3/2+3\epsilon}\right)^N=1-O(N^{-1/2+3\epsilon})\\
        \mathrm{Fid}\left(\phi_{u/\sqrt{N}}^{\otimes \floor{rN}},\tilde{\phi}_{u/\sqrt{N}}^{\otimes \floor{rN}} \right)&=\mathrm{Fid}\left(\phi_{u/\sqrt{N}},\tilde{\phi}_{u/\sqrt{N}}\right)^{\floor{rN}}\geq \left(1-a_2N^{-3/2+3\epsilon}\right)^{\floor{rN}}=1-O(N^{-1/2+3\epsilon})
\end{align}
where we used $\epsilon<1/9<1/6$. Therefore, we find
\begin{align}
    \sup_{\|u\|<N^\epsilon}T\left(\psi_{u/\sqrt{N}}^{\otimes N},\tilde{\psi}_{u/\sqrt{N}}^{\otimes N} \right) &\leq \inf_{\|u\|<N^\epsilon}\sqrt{1-\mathrm{Fid}\left(\psi_{u/\sqrt{N}}^{\otimes N},\tilde{\psi}_{u/\sqrt{N}}^{\otimes N} \right) }=O(N^{-\frac{1}{4}(1-6\epsilon)}),\label{eq:error_i}\\
     \sup_{\|u\|<N^\epsilon}T\left(\phi_{u/\sqrt{N}}^{\otimes \floor{rN}},\tilde{\phi}_{u/\sqrt{N}}^{\otimes \floor{rN}} \right)&\leq \inf_{\|u\|<N^\epsilon}\sqrt{1-\mathrm{Fid}\left(\phi_{u/\sqrt{N}}^{\otimes \floor{rN}},\tilde{\phi}_{u/\sqrt{N}}^{\otimes \floor{rN}} \right) }=O(N^{-\frac{1}{4}(1-6\epsilon)}).\label{eq:error_iv}
\end{align}

\medskip
{\bf Step~(ii):} 
For $r$ satisfying $\mathcal{Q}^\psi\geq r \mathcal{Q}^\phi$, in Lemma~\ref{lem:interconversion_projected_generators} and its proof, we have constructed a quantum channel $\mathcal{E}$ satisfying Eq.~\eqref{eq:conversion_reduction_QGT_error_rate}, i.e., 
\begin{align}
        \left|\mathrm{Fid}\left(\mathcal{E}(\tilde{\psi}_{u/\sqrt{N}}),\tilde{\phi}_{\sqrt{r}u/\sqrt{N}}\right)-1\right|\leq a_3 N^{-3/2+3\epsilon}
\end{align}
for a finite non-negative constant $a_3$ independent of $N$. Therefore, as in Steps~(i) and (iv), we obtain
\begin{align}
    \sup_{\|u\|<N^\epsilon}T\left(\mathcal{E}^{\otimes N}\left(\tilde{\psi}_{u/\sqrt{N}}^{\otimes N}\right),\tilde{\phi}_{\sqrt{r}u/\sqrt{N}}^{\otimes N} \right) &\leq \inf_{\|u\|<N^\epsilon}\sqrt{1-\mathrm{Fid}\left(\mathcal{E}\left(\tilde{\psi}_{u/\sqrt{N}}\right),\tilde{\phi}_{\sqrt{r}u/\sqrt{N}} \right) ^N}=O(N^{-\frac{1}{4}(1-6\epsilon)}).\label{eq:error_ii}
\end{align}

\medskip
\medskip
{\bf Step~(iii):} 
The sequences of channels $\{\mathcal{T}_N\}_N$ and $\{\mathcal{S}_N\}_N$ in Eqs.\eqref{eq:QLAN_error_rate_Tn} and \eqref{eq:QLAN_error_rate_Sn} are provided in Theorem~3.3 of \cite{lahiry_minimax_2024}. The conversion errors are shown to be polynomial in $N$, i.e., 
\begin{align}
    \sup_{\|u\|<N^{\epsilon}}T\left(\mathcal{T}_N\left(\tilde{\phi}_{\sqrt{r}u/\sqrt{N}}^{\otimes N}\right),\ket{\sqrt{r}C'u}\bra{\sqrt{r}C'u}\right)&=O(N^{-\kappa}),\\
    \sup_{\|u\|<N^{\epsilon}}T\left(\tilde{\phi}_{\sqrt{r}u/\sqrt{N}}^{\otimes N},\mathcal{S}_N(\ket{\sqrt{r}C'u}\bra{\sqrt{r}C'u})\right)&=O(N^{-\kappa})
\end{align}
holds for some $\kappa>0$. 

Since the term arising from Eq.~\eqref{eq:coherence_overlap} is of order $O(N^{-1/2})$ and hence subdominant, we get
\begin{align}
    \sup_{\|u\|<N^{\epsilon'}}T\left(\tilde{\phi}^{\otimes \floor{rN}}_{u/\sqrt{N}},\mathcal{S}_{\floor{rN}}\left(\ket{\sqrt{r}C'u}\bra{\sqrt{r}C'u}\right)\right)=O\left(N^{-\kappa}\right).
\end{align}
Consequently, we get
\begin{align}
    \sup_{\|u\|<N^{\epsilon}}T\left(\mathcal{S}_{\floor{rN}}\circ\mathcal{T}_N\left(\tilde{\phi}_{\sqrt{r}u/\sqrt{N}}^{\otimes N}\right),\tilde{\phi}^{\otimes \floor{rN}}_{u/\sqrt{N}}\right)=O\left(N^{-\kappa}\right).\label{eq:error_iii}
\end{align}

From Eqs.~\eqref{eq:error_i}, \eqref{eq:error_iv}, \eqref{eq:error_ii}, and \eqref{eq:error_iii}, we obtain
\begin{align}
    \sup_{\|u\|<N^{\epsilon}}T\left(\mathcal{S}_{\floor{rN}}\circ\mathcal{T}_N\circ\mathcal{E}^{\otimes N}\left(\psi_{u/\sqrt{N}}^{\otimes N}\right),\phi^{\otimes \floor{rN}}_{u/\sqrt{N}}\right)=O\left(N^{-\kappa'}\right),\quad \kappa'\coloneqq \min \left\{\kappa,\frac{1}{4}(1-6\epsilon)\right\}.
\end{align}
We remark that the conversion channels $\mathcal{S}_{\floor{rN}}$ and $\mathcal{T}_N$ are constructed in the proof of Theorem~3.3 of \cite{lahiry_minimax_2024}, while the channel $\mathcal{E}$ is constructed in Lemma~\ref{lem:interconversion_projected_generators}.

\vspace{1cm}

From these sequences of channels, the sequence of channels $\{\mathcal{E}_N^{(g)}\}$ is constructed in the proof of Lemma~\ref{lem:Lie_asymptotic_conversion}, which satisfies
\begin{align}
    \sup_{g\in G}\sup_{\|u\|<N^\epsilon}\delta_N(g,u)&=\sup_{g\in G}\sup_{\|u\|<N^\epsilon}T\left(\mathcal{E}_N^{(g)}\left(\mathcal{U}_{\frac{u}{\sqrt{N}}}(\psi_{g})^{\otimes N}\right),\mathcal{U}'_{\frac{u}{\sqrt{N}}}(\phi_{g})^{\otimes \floor{rN}}\right)=O\left(N^{-\kappa'}\right),
\end{align}
where $\epsilon_1$ is an arbitrarily small positive constant. 
Therefore, for the channels $\{\mathcal{E}_N\}$ defined in Eq.~\eqref{eq:definition_full_channel}, from Eqs.~\eqref{eq:convergence_fidelity_succ_set}, \eqref{eq:estimate_and_conversion_error_fidelity} and~\eqref{eq:success_prob_rate}, we obtain
\begin{align}
    \inf_{g\in G}\mathrm{Fid}\left(\mathcal{E}_N(\mathcal{U}_g(\psi)^{\otimes N}),\mathcal{U}_g'(\phi)^{\otimes \floor{rN'}}\right)
        \geq p^{\mathrm{succ.}}(N^{-1/2+\epsilon})\times f_{N}=\left(1-O\left(N^{-\kappa'}\right)\right)^2,
\end{align}
implying that 
\begin{align}
    \sup_{g\in G}T\left(\mathcal{E}_N(\mathcal{U}_g(\psi)^{\otimes N}),\mathcal{U}_g'(\phi)^{\otimes \floor{rN'}}\right)=O\left(N^{-\frac{1}{2}\kappa'}\right).
\end{align}

Following the arguments around Eq.~\eqref{eq:direct_part_discard}, for an arbitrarily small $\delta>0$, by using the channel $\Lambda_N$ discarding $\floor{rN'}-\floor{(r-\delta)N}$ copies of the output system, we find 
\begin{align}
    \sup_{g\in G}T\left(\Lambda_N\circ \mathcal{E}_N(\mathcal{U}_g(\psi)^{\otimes N}),\mathcal{U}_g'(\phi)^{\otimes \floor{(r-\delta)N}}\right)=O\left(N^{-\frac{1}{2}\kappa'}\right).
\end{align}

Finally, we introduce
\begin{align}
    \mathcal{E}_N'\coloneqq \int \dd \mu_G(g)\,\mathcal{U}_{g^{-1}}'^{\otimes \floor{(r-\delta)N}}\circ \Lambda_N\circ \mathcal{E}_N\circ \mathcal{U}_{g}^{\otimes N}.
\end{align}
From Lemma~\ref{lem:gcov_and_cptp}, $\{\mathcal{E}_N'\}_N$ is a sequence of $G$-covariant channels such that
\begin{align}
    T\left(\mathcal{E}_N'\left(\psi^{\otimes N}\right),\phi^{\otimes \floor{(r-\delta)N}}\right)=O\left(N^{-\frac{1}{2}\kappa'}\right).
\end{align}
Therefore, with this sequence of conversion channels $\{\mathcal{E}_N'\}_N$, the conversion error in trace distance decays polynomially with $N$.

Adopting the result of \cite{lahiry_minimax_2024}, we see that the coefficient $\kappa'$ is determined by $\kappa$, which characterizes the conversion error in QLAN. This implies that an improvement of the QLAN error rate directly translates into an improvement of our conversion error. To demonstrate this, we use Lemma~C.6 of \cite{lahiry_minimax_2024}, which states that
\begin{align}
    &\sup_{\|u\|<N^{\epsilon}}T\left(\mathcal{T}_{N}\left(\tilde{\phi}^{\otimes N}_{\sqrt{r}u/\sqrt{N}}\right),\ket{\sqrt{r}C'u}\bra{\sqrt{r}C'u}\right)\nonumber \\
    &=O\left(N^{-\frac{1}{4}+\Delta}+N^{-\frac{1}{2}+\zeta}+N^{-\frac{1}{2}+\zeta+\eta}+N^{\frac{9\eta-2}{24}}+N^{-\frac{1}{4}+\frac{3\epsilon}{2}+\Delta}+N^{-\frac{1}{2}+\frac{\alpha}{2}+\frac{\eta}{2}}+N^{-\frac{1}{2}+\frac{3\eta}{2}}+N^{-\frac{\Delta}{2}}\right)
\end{align}
where $\zeta$, $\Delta$, $\alpha$, $\eta$ are constants, which can be arbitrarily chosen as long as they satisfy 
\begin{align}
    0<\zeta<\frac{1}{4},\quad 0<\Delta,\quad \frac{1}{2}+\zeta<\alpha,\quad \eta<\frac{2}{9},\quad 0<\epsilon<\frac{\eta-\Delta}{2}.
\end{align}
The same bound applies $\mathcal{S}_N$ as noted in Eq. (C.9) of \cite{lahiry_minimax_2024}. For example, let $\epsilon_0 \in (4\epsilon/7,4/63)$ and choose
\begin{align}
    \zeta=\epsilon_0,\quad \Delta =\frac{2}{21}-\frac{3}{2}\epsilon_0,\quad \alpha=\frac{1}{2}+2\zeta=\frac{1}{2}+2\epsilon_0,\quad \eta=\frac{2}{21}+2\epsilon_0,
\end{align}
which satisfy the above constraints. Substituting these parameters yields
\begin{align}
    \sup_{\|u\|<N^{\epsilon}}T\left(\mathcal{T}_{N}\left(\tilde{\phi}^{\otimes N}_{\sqrt{r}u/\sqrt{N}}\right),\ket{\sqrt{r}C'u}\bra{\sqrt{r}C'u}\right)=O\left(N^{-\frac{1}{21}+\frac{3}{4}\epsilon_0}\right).
\end{align}
Hence, $\kappa$ can be chosen arbitrarily close to $\frac{1}{21}$.

\newpage
\section{\texorpdfstring{Proof of Eq.~\protect\eqref{eq:cost_ensemble}}{Proof of Eq.~(\ref*{eq:cost_ensemble})}}
\label{app:cost_ensemble}
Following the standard argument using typical sequence \cite{hayden_asymptotic_2001,marvian_operational_2022}, we prove Eq.~\eqref{eq:cost_ensemble}, i.e., $ \sum_{i\in [k]}p_i\mathcal{A}_{\cost}(\rho_i)\geq \mathcal{A}_{\cost}(\rho)$ for $\rho=\sum_{i\in[k]}p_i\rho_i$, where $[k]\coloneqq \{1,2,\cdots,k\}$.

\begin{proof}
Without loss of generality, we can assume that $p_i>0$ for all $i\in[k]$. If $\mathcal{A}_{\cost}(\rho_i)=\infty$, the inequality is trivial. Therefore, we only consider the case where $\mathcal{A}_{\cost}(\rho_i)<\infty$ for all $i\in[k]$ below.

The i.i.d. copies of $\rho$ is given by
 \begin{align}
    \rho^{\otimes N}=\left(\sum_{i\in [k]}p_i\rho_i\right)^{\otimes N}=\sum_{\bm{i}\in[k]^N}p_{\bm{i}}\rho_{\bm{i}},
\end{align}
where $\bm{i}\coloneqq i_1i_2\cdots i_N\in [k]^N$, $p_{\bm{i}}\coloneqq \prod_{j=1}^Np_{i_j}$, and $\rho_{\bm{i}}\coloneqq \bigotimes_{j=1}^N\rho_{i_j}$. For $l\in [k]$, let $n_l(\bm{i})$ denote the number of occurrence of $l$ in $i_1i_2\cdots i_N$. For $\delta>0$, we define $\delta$-typical sequence as
\begin{align}
    \mathcal{T}_\delta\coloneqq \left\{\bm{i}=i_1i_2\cdots i_N\in[k]\,\middle|\,\forall l\in [k],\, \left|\frac{n_{l}(\bm{i})}{N}-p_l\right|\leq \delta\right\}.
\end{align}
The i.i.d. state $\rho^{\otimes N}$ is decomposed as 
\begin{align}
    \rho^{\otimes N}=\sum_{\bm{i}\in\mathcal{T}_\delta}p_{\bm{i}}\rho_{\bm{i}}+\sum_{\bm{i}\notin\mathcal{T}_\delta}p_{\bm{i}}\rho_{\bm{i}}.
\end{align}
Let us introduce a density operator 
\begin{align}
    \tilde{\rho}_N\coloneqq \sum_{\bm{i}\in\mathcal{T}_\delta}p_{\bm{i}}\rho_{\bm{i}}+p_{\mathrm{err}}\rho_{\mathrm{sym}},
\end{align}
where $p_{\mathrm{err}}\coloneqq 1- \sum_{\bm{i}\in\mathcal{T}_\delta}p_{\bm{i}}$ and $\rho_{\mathrm{sym}}$ is some symmetric state, e.g., the maximally mixed state. From the standard argument on a typical sequence, for any $\delta>0$, the probability $p_{\mathrm{err}}$ converges to zero as $N\to\infty$, which implies $\lim_{N\to\infty}T\left(\rho^{\otimes N},\tilde{\rho}_N\right)=0$. Thus, we prepare $\tilde{\rho}_N$ instead of $\rho^{\otimes N}$.

The state $\rho_{\bm{i}}$ is equivalent to $\bigotimes_{l\in[k]}\rho_l^{\otimes n_l(\bm{i})}$ up to permutation. Note that the permutation, which is composed of the SWAP operations among subsystems, is a $G$-covariant operation. For a typical sequence $\bm{i}\in \mathcal{T}_\delta$, $n_l(\bm{i})\leq N(p_l+\delta)$ holds for all $l\in [k]$. Therefore, a state $\rho_{\bm{i}}$ for a typical sequence $\bm{i}\in\mathcal{T}_\delta$ can be created by first preparing $\bigotimes_{l\in[k]}\rho_l^{\otimes \floor{N(p_l+\delta)}}$ via a $G$-covariant operation, discarding some copies if needed, and then operating a permutation. For each $l\in[k]$, for $r_l$ satisfying $r_l>(p_l+\delta)\mathcal{A}_{\cost}(\rho_l)$, it holds
\begin{align}
    \{\phi^{\otimes \ceil{r_lN}}\}_N\gconv \{\rho_l^{\otimes \floor{N(p_l+\delta)}}\}_N .
\end{align}
Therefore, introducing $r\coloneqq \sum_{i\in [k]}r_l$, there exists a sequence of $G$-covariant channels $\{\mathcal{E}_{N}^{(\bm{i})}\}_N$ for any $\bm{i}\in \mathcal{T}_\delta$ such that 
\begin{align}
    \lim_{N\to\infty}T\left(\mathcal{E}_{N}^{(\bm{i})}(\phi^{\otimes \ceil{rN}}),\rho_{\bm{i}}\right)=0 . \label{eq:typical_sequence_conversion}
\end{align}
Defining a sequence of $G$-covariant channels $\{\mathcal{E}_N\}_N$ by
\begin{align}
    \mathcal{E}_N(\cdot)\coloneqq \sum_{\bm{i}\in\mathcal{T}_\delta}p_{\bm{i}}\mathcal{E}_{N}^{(\bm{i})}(\cdot)+p_{\mathrm{err}}\mathcal{E}_{\mathrm{sym},N}(\cdot),
\end{align}
where $\mathcal{E}_{\mathrm{sym},N}(\cdot)$ denotes the $G$-covariant channel that maps arbitrary state into a symmetric state $\rho_{\mathrm{sym}}$, we get
\begin{align}
    T\left(\mathcal{E}_N(\phi^{\otimes \ceil{rN}}),\tilde{\rho}_N\right)
    &=T\left(\sum_{\bm{i}\in\mathcal{T}_\delta}p_{\bm{i}}\mathcal{E}_{N}^{(\bm{i})}(\phi^{\otimes \ceil{rN}})+p_{\mathrm{err}}\rho_{\mathrm{sym}},\sum_{\bm{i}\in\mathcal{T}_\delta}p_{\bm{i}}\rho_{\bm{i}}+p_{\mathrm{err}}\rho_{\mathrm{sym}}\right)\\
    &\leq\sum_{\bm{i}\in\mathcal{T}_\delta}p_{\bm{i}}T\left(\mathcal{E}_{N}^{(\bm{i})}(\phi^{\otimes \ceil{rN}}),\rho_{\bm{i}}\right)\\
    &\leq \max_{\bm{i}\in\mathcal{T}_\delta}T\left(\mathcal{E}_{N}^{(\bm{i})}(\phi^{\otimes \ceil{rN}}),\rho_{\bm{i}}\right),
\end{align}
where we have used the strong convexity of the trace distance in the first inequality and $\sum_{\bm{i}\in\mathcal{T}_\delta}p_{\bm{i}}\leq 1$ in the last inequality. Thus, from Eq.~\eqref{eq:typical_sequence_conversion}, we find $\lim_{N\to\infty}T\left(\mathcal{E}_N(\phi^{\otimes \ceil{rN}}),\tilde{\rho}_N\right)=0$. By using $\lim_{N\to\infty}T\left(\rho^{\otimes N},\tilde{\rho}_N\right)=0$, we get $\lim_{N\to\infty}T\left(\mathcal{E}_N(\phi^{\otimes \ceil{rN}}),\rho^{\otimes N}\right)=0$. Therefore, for $r=\sum_{l\in[k]}r_l$, we find $\{\phi^{\otimes \ceil{rN}}\}_N\gconv \{\rho^{\otimes N}\}_N$. Note that $r$ satisfies
\begin{align}
    r=\sum_{l\in[k]}r_l>\sum_{l\in[k]}(p_l+\delta)\mathcal{A}_{\cost}(\rho_l)=\sum_{l\in[k]}p_l\mathcal{A}_{\cost}(\rho_l)+\delta \times \left(\sum_{l\in[k]}\mathcal{A}_{\cost}(\rho_l)\right).
\end{align}
Since $\delta>0$ can be arbitrarily small and $\sum_{l\in[k]}\mathcal{A}_{\cost}(\rho_l)<\infty$, we find that $\{\phi^{\otimes \ceil{rN}}\}_N\gconv \{\rho^{\otimes N}\}_N$ holds for any $r>\sum_{l\in[k]}p_l\mathcal{A}_{\cost}(\rho_l)$. Therefore, we get $ \sum_{i\in [k]}p_i\mathcal{A}_{\cost}(\rho_i)\geq \mathcal{A}_{\cost}(\rho)$. 
\end{proof}

If a state $\rho$ can be written as $\rho=\sum_{i=1}^{k-1}p_i\psi_i+p_{\mathrm{s}}\rho_{\mathrm{s}}$ for symmetric state $\rho_{\mathrm{s}}$ and pure states $\{\psi_i\}_{i=1}^{k-1}$ satisfying $\mathrm{Sym}_G(\phi)\subset \mathrm{Sym}_G(\psi_i)$, then we get Eq.~\eqref{eq:cost_upper_bound} as an upper bound on the asymmetry cost of $\rho$, given by
\begin{align}
    \sum_{i=1}^{k-1}p_i\inf\left\{r\geq 0\,\middle|\, \forall g\in G,\, r\mathcal{Q}^{\mathcal{U}_g(\phi)} \geq  \mathcal{Q}^{\mathcal{U}_g'(\psi_i)}\right\}\geq \mathcal{A}_{\cost}(\rho).
\end{align}
Minimizing the left-hand side over such an ensemble yields a better bound. 

However, such a decomposition does not always exist. Let us consider an example of a two-qubit system with $G=U(4)$-symmetry. Let $\phi$ be a spin-singlet state and $\rho$ be a mixed state consisting of three spin-triplet states with equal probabilities. Note that $\mathrm{Sym}_G(\phi)=\mathrm{Sym}_G(\rho)$ since $\phi +3\rho=I$. We will now demonstrate that $\rho$ cannot be expressed as a probabilistic mixture of symmetric state $\rho_{\mathrm{s}}$ and pure states $\{\psi_i\}_{i}$ satisfying $\mathrm{Sym}_G(\phi)\subset \mathrm{Sym}_G(\psi_i)$. For any $U\in U(2)$, $U^{\otimes 2}$ leaves $\phi$ invariant, meaning that $U^{\otimes 2}\in \mathrm{Sym}_G(\phi)$. Any state $\psi_i$ for which $U^{\otimes 2} \in \mathrm{Sym}_G(\psi_i)$ can be expressed as a linear combination of the identity and SWAP operators. In particular, the unique pure state satisfying this condition is the singlet state $\phi$. Since the maximally mixed state is the unique symmetric state under $G=U(4)$ symmetry, the state $\rho$ cannot be expressed as a probabilistic mixture of a symmetric state and pure states $\{\psi_i\}_{i}$ satisfying $\mathrm{Sym}_G(\phi)\subset \mathrm{Sym}_G(\psi_i)$.

\section{Proof of an upper bound on the asymmetry cost (Proposition~\ref{prop:asymmetry_of_formation})}\label{app:asymmetry_of_formation}
In this section, we present the proof of Proposition~\ref{prop:asymmetry_of_formation}, which provides an upper bound on the asymmetry cost. It is worth emphasizing that, unlike the result in the previous section, Proposition~\ref{prop:asymmetry_of_formation} does not impose any constraint on the symmetry subgroups in the ensemble decomposition. Intuitively, this relaxation is possible because the symmetry subgroup constraint is only relevant for the estimation step in the estimation-and-conversion strategy introduced in Sec.~\ref{app:direct_part}, while the conversion step succeeds as long as the monotonicity condition of QGTs is satisfied. Below, we carefully modify the argument of typical sequences presented in the previous section, and use it to extend the proof of Proposition~\ref{prop:direct_part}, thereby establishing Proposition~\ref{prop:asymmetry_of_formation}. Following the convention in the main text, we define $\mathcal{U}_g(\cdot)\coloneqq U(g)(\cdot )U(g)^\dag$ for a unitary representation $U$ of $G$. Note that, with a slight abuse of notation, we also write $\mathcal{U}_{\theta}(\cdot)\coloneqq e^{\ii \theta^\mu X_\mu}(\cdot)e^{-\ii\theta^\nu X_\nu}$ for $\theta\in\mathbb{R}^{\dim G}$.

\begin{lem}\label{lem:conversion_ensemble}
    Let $U,U'$ be (non-projective) unitary representations of a compact Lie group $G$ on finite-dimensional Hilbert spaces $\mathcal{H}$ and $\mathcal{H}'$. Let $\rho$ be a state on $\mathcal{H}$ and $\sigma$ be a state on $\mathcal{H}'$. Consider an ensemble $\{p_i,\sigma_i\}_{i\in[k]}$ such that $\sigma=\sum_ip_i\sigma_i$, where $[k]\coloneqq \{1,\cdots, k\}$. Suppose that for each $i$, for some finite $r_i>0$, there exists a sequence of quantum channels $\mathcal{E}_N^{(i,g)}$ such that
    \begin{align}
        \delta_{N}^{(i)}(g,u)\coloneqq T\left(\mathcal{E}^{(i,g)}_N\left(\mathcal{U}_{\frac{u}{\sqrt{N}}}\left(\rho_g\right)^{\otimes \ceil{r_i N}}\right),\mathcal{U}'_{\frac{u}{\sqrt{N}}}\left(\sigma_{i,g}\right)^{\otimes N}\right)
    \end{align}
    satisfies $\lim_{N\to\infty}\sup_{g\in G}\sup_{\|u\|<N^\epsilon}\delta_N^{(i)}(g,u)=0$ for $\epsilon'\in(0,\epsilon)$, where $\rho_g\coloneqq \mathcal{U}_g(\rho)$ and $\sigma_{i,g}\coloneqq \mathcal{U}_{g}'(\sigma_i)$. 
    Then there exists a sequence of quantum channels $\mathcal{E}_N^{(g)}$ and $\epsilon'\in(0,1)$ such that
    \begin{align}
        \lim_{N\to\infty}\sup_{g\in G}\sup_{\|u\|<N^{\epsilon'}}T\left(\mathcal{E}^{(g)}_N\left(\mathcal{U}_{\frac{u}{\sqrt{N}}}\left(\rho_g\right)^{\otimes\ceil{rN}}\right),\mathcal{U}'_{\frac{u}{\sqrt{N}}}\left(\sigma_{g}\right)^{\otimes N}\right)
    \end{align}
    for $r\coloneqq \sum_ip_ir_i+\Delta$ and $\sigma_{g}\coloneqq  \mathcal{U}_{g}'(\sigma)$, where $\Delta$ is an arbitrary positive number.   
\end{lem}
\begin{proof}
    If $p_i=1$ for some $i$, then the statement is trivial. In the following, we only consider the case where $p_i<1$ for all $i$. We take $\delta>0$ such that $p_i+\delta<1$ for all $i$. 
    
    The i.i.d. copies of $\sigma$ are expressed as
    \begin{align}
        \sigma^{\otimes N}=\left(\sum_{i\in [k]}p_i\sigma_i\right)^{\otimes N}=\sum_{\bm{i}\in[k]^N}p_{\bm{i}}\rho_{\bm{i}},
    \end{align}
    where $\bm{i}\coloneqq( i_1,\cdots ,i_N)\in[k]^N$, $p_{\bm{i}}\coloneqq \prod_{j=1}^Np_{i_j}$, and $\sigma_{\bm{i}}\coloneqq \otimes_{j=1}^N\rho_{i_j}$. For $l\in[k]$, let $n_l(\bm{i})$ denote the number of occurrence of $l$ in $(i_1,\cdots,i_N)$. For $\delta>0$, we define $\delta$-typical sequence as
    \begin{align}
        \mathcal{T}_\delta\coloneqq \left\{\bm{i}=(i_1,i_2,\cdots ,i_N)\in[k]\,\middle|\,\forall l\in [k],\, \left|\frac{n_{l}(\bm{i})}{N}-p_l\right|\leq \delta\right\}.
    \end{align}

    Introducing 
    \begin{align}
        \tilde{\sigma}_N\coloneqq \sum_{\bm{i}\in\mathcal{T}_\delta}p_{\bm{i}}\sigma_{\bm{i}}+p_{\mathrm{err}}\sigma_{\mathrm{sym}},
    \end{align}
    where $p_{\mathrm{err}}\coloneqq 1- \sum_{\bm{i}\in\mathcal{T}_\delta}p_{\bm{i}}$ and $\sigma_{\mathrm{sym}}$ is some symmetric state (,e.g., the maximally mixed state), the standard argument on a typical sequence implies
    \begin{align}
        \lim_{N\to\infty}T\left(\sigma^{\otimes N},\tilde{\sigma}_N\right)=0. 
    \end{align}
    From the invariance of the trace distance under unitary operations, it also implies
    \begin{align}
        \lim_{N\to\infty}\sup_{g\in G}\sup_{\|u\|<N^\epsilon}T\left(\left(\mathcal{U}_{\frac{u}{\sqrt{N}}}(\sigma_g)\right)^{\otimes N},\left(\mathcal{U}_{\frac{u}{\sqrt{N}}}^{\prime \otimes N}\circ \mathcal{U}_g^{\prime\otimes N}\right)\left(\tilde{\sigma}_N\right)\right)=0\label{eq:typical_sequence_approximation}.
    \end{align}

    Since the state $\sigma_{\bm{i}}$ is equivalent to $\bigotimes_{l\in[k]}\sigma_l^{\otimes n_l(\bm{i})}$ up to permutation, let us first consider a channel that approximately creates $\bigotimes_{l\in[k]}\sigma_l^{\otimes n_l(\bm{i})}$ for a typical sequence $\bm{i}\in\mathcal{T}_\delta$. Define a sequence of quantum channels $\mathcal{F}_{N}$ by
    \begin{align}
        \mathcal{F}_{N}^{(g)}\coloneqq \bigotimes_{l\in [k]}\left(\mathcal{E}^{(l,g)}_{N_l}\right),
    \end{align}
    where we have introduced $N_l\coloneqq \ceil{N(p_l+\delta)}$. 
    From the subadditivity of the trace distance, we get
    \begin{align}
        &T\left(\mathcal{F}_{N}^{(g)}\left(\bigotimes_{l\in[k]}\left(\mathcal{U}_{\frac{u}{\sqrt{N}}}(\rho_g)\right)^{\otimes  \ceil{r_lN_l}}\right),\bigotimes_{l\in[k]}\left(\mathcal{U}'_{\frac{u}{\sqrt{N}}}(\sigma_{l,g})\right)^{ \otimes N_l}\right)\nonumber \\
        &\leq \sum_{l\in[k]}T\left(\mathcal{E}^{(l,g)}_{N_l}\left(\mathcal{U}_{\frac{u}{\sqrt{N}}}(\rho_g)^{\otimes  \ceil{r_lN_l}}\right),\mathcal{U}'_{\frac{u}{\sqrt{N}}}(\sigma_{l,g})^{ \otimes N_l}\right)\nonumber \\
        &=
        \sum_{l\in [k]} \delta_{N_l}^{(l)}\left(g,u\frac{\sqrt{N_l}}{\sqrt{N}}\right).
    \end{align}
    Since $\delta$ satisfies $p_l+\delta<1$, we have $N_l\leq N$, implying that $\frac{\sqrt{N_l}}{\sqrt{N}}\leq 1$. For all sufficiently large $N$, $N^{\epsilon'}<N_l^\epsilon$ holds for $\epsilon'\in (0,\epsilon)$, meaning that
    \begin{align}
        \|u\|<N^\epsilon\implies \left\|u\frac{\sqrt{N_l}}{\sqrt{N}}\right\|< N_l^\epsilon.
    \end{align}
    Therefore, we get
    \begin{align}
        &\sup_{g\in G}\sup_{\|u\|<N^{\epsilon'}}T\left(\mathcal{F}_{N}^{(g)}\left(\bigotimes_{l\in[k]}\left(\mathcal{U}_{\frac{u}{\sqrt{N}}}(\rho_g)\right)^{\otimes  \ceil{r_lN_l}}\right),\bigotimes_{l\in[k]}\left(\mathcal{U}_{\frac{u}{\sqrt{N}}}(\sigma_{l,g})\right)^{\otimes N_l}\right)\nonumber \\
        &\leq \sum_{l\in [k]}\sup_{g\in G}\sup_{\|u\|<N_l^\epsilon}\delta_{N_l}^{(l)}\left(g,u\right)
    \end{align}
    holds for all sufficiently large $N$. 
    Since $N_l\geq  n_l(\bm{i})$ holds for a typical sequence $\bm{i}\in \mathcal{T}_\delta$, we can introduce a channel $\mathcal{D}_{N_l\to n_l(\bm{i)}}$ that discards $N_l- n_l(\bm{i})$ subsystems if $N_l>n_l(\bm{i})$ and does nothing if $N_l=n_l(\bm{i})$. For $\mathcal{D}_N^{\bm{i}}\coloneqq \bigotimes_{l\in [k]}\mathcal{D}_{N_l\to n_l(\bm{i)}}$, we get
    \begin{align}
        &\sup_{g\in G}\sup_{\|u\|<N^{\epsilon'}}T\left(\mathcal{D}_N^{(\bm{i})}\circ \mathcal{F}_{N}^{(g)}\left(\bigotimes_{l\in[k]}\left(\mathcal{U}_{\frac{u}{\sqrt{N}}}(\rho_g)\right)^{\otimes  \ceil{r_lN_l}}\right),\bigotimes_{l\in[k]}\left(\mathcal{U}_{\frac{u}{\sqrt{N}}}(\sigma_{l,g})\right)^{\otimes n_l(\bm{i})}\right)\nonumber \\
        &\leq \sum_{l\in [k]}\sup_{g\in G}\sup_{\|u\|<N_l^\epsilon}\delta_{N_l}^{(l)}\left(g,u\right).
    \end{align}
    Denoting by $\mathcal{P}_{\bm{i}}$ a channel that swaps subsystems so that $\mathcal{P}^{(\bm{i})}(\bigotimes_{l\in[k]}\sigma_l^{\otimes n_l(\bm{i})})=\sigma_{\bm{i}}$, which commutes with $\mathcal{U}^{\otimes N}_{u}\circ \mathcal{U}^{\otimes N}_{g}$, we get
    \begin{align}
        &\sup_{g\in G}\sup_{\|u\|<N^{\epsilon'}}T\left(\mathcal{E}^{(\bm{i},g)}_N\left(\bigotimes_{l\in[k]}\left(\mathcal{U}_{\frac{u}{\sqrt{N}}}(\rho_g)\right)^{\otimes  \ceil{r_lN_l}}\right),\mathcal{U}_{\frac{u}{\sqrt{N}}}^{\otimes N}(\sigma_{\bm{i},g})\right)\nonumber \\
        &\leq \sum_{l\in [k]}\sup_{g\in G}\sup_{\|u\|<N_l^\epsilon}\delta_{N_l}^{(l)}\left(g,u\right)\label{eq:sum_trace_distance_l}
    \end{align}
    for $\mathcal{E}^{(\bm{i},g)}_N\coloneqq \mathcal{P}^{(\bm{i})}\circ \mathcal{D}_N^{\bm{i}}\circ \mathcal{F}_{N}^{(g)}$. 

    We now fix $\Delta>0$ so that $\sum_l \ceil{r_lN_l}<\ceil{rN}$ for $r\coloneqq \sum_lp_lr_l+\Delta$ for all sufficiently large $N$. Since $\delta>0$ can be chosen arbitrarily small and $r_l$ is finite from the assumption, $\Delta$ can also be made arbitrarily small. Defining a sequence of channels $\{\mathcal{E}_N\}_N$ by
    \begin{align}
        \mathcal{E}_N^{(g)}(\cdot)\coloneqq \sum_{\bm{i}\in\mathcal{T}_\delta}p_{\bm{i}}\mathcal{E}_{N}^{(\bm{i},g)}\circ \mathcal{D}_{\ceil{rN}\to \sum_l\ceil{r_lN_l }}(\cdot)+p_{\mathrm{err}}\mathcal{E}_{\mathrm{sym},N}(\cdot),
    \end{align}
    where $\mathcal{E}_{\mathrm{sym},N}(\cdot)$ denotes a channel that maps an arbitrary state into a symmetric state $\rho_{\mathrm{sym}}$, we get
    \begin{align}
         &T\left(\mathcal{E}_N^{(g)}\left(\mathcal{U}_{\frac{u}{\sqrt{N}}}(\rho_g)\right)^{\otimes  \ceil{rN}},\mathcal{U}_{\frac{u}{\sqrt{N}}}^{\prime \otimes N}(\tilde{\sigma}_{N,g})\right)\\
        &=T\left(\sum_{\bm{i}\in\mathcal{T}_\delta}p_{\bm{i}}\mathcal{E}^{(\bm{i},g)}_N\left(\bigotimes_{l\in[k]}\left(\mathcal{U}_{\frac{u}{\sqrt{N}}}(\rho_g)\right)^{\otimes  \ceil{r_lN_l}}\right)+p_{\mathrm{err}}\rho_{\mathrm{sym}},\sum_{\bm{i}\in\mathcal{T}_\delta}p_{\bm{i}}\left(\mathcal{U}_{\frac{u}{\sqrt{N}}}'(\sigma_{\bm{i},g})\right)^{\otimes N}+p_{\mathrm{err}}\rho_{\mathrm{sym}}\right)\\
        &\leq\sum_{\bm{i}\in\mathcal{T}_\delta}p_{\bm{i}}T\left(\mathcal{E}^{(\bm{i},g)}_N\left(\bigotimes_{l\in[k]}\left(\mathcal{U}_{\frac{u}{\sqrt{N}}}(\rho_g)\right)^{\otimes  \ceil{r_lN_l}}\right),\mathcal{U}_{\frac{u}{\sqrt{N}}}^{\otimes N}(\sigma_{\bm{i},g})\right)\\
        &\leq \max_{\bm{i}\in\mathcal{T}_\delta}T\left(\mathcal{E}^{(\bm{i},g)}_N\left(\bigotimes_{l\in[k]}\left(\mathcal{U}_{\frac{u}{\sqrt{N}}}(\rho_g)\right)^{\otimes  \ceil{r_lN_l}}\right),\mathcal{U}_{\frac{u}{\sqrt{N}}}^{\otimes N}(\sigma_{\bm{i},g})\right),\label{eq:max_trace_evaluation}
    \end{align}
    where we have used the joint convexity of the trace distance and $\sum_{\bm{i}\in\mathcal{T}_\delta}p_{\bm{i}}\leq 1$. From Eqs.~\eqref{eq:typical_sequence_approximation}, \eqref{eq:sum_trace_distance_l}, \eqref{eq:max_trace_evaluation}, we get
    \begin{align}
        \lim_{N\to\infty}\sup_{g\in G}\sup_{\|u\|<N^{\epsilon'}}T\left(\mathcal{E}^{(g)}_N\left(\mathcal{U}_{\frac{u}{\sqrt{N}}}\left(\rho_g\right)^{\otimes\ceil{rN}}\right),\mathcal{U}'_{\frac{u}{\sqrt{N}}}\left(\sigma_{g}\right)^{\otimes N}\right)=0. 
    \end{align}

\end{proof}

\section{Convex roof of QGT}\label{app:convex_roof_QGT}
We show that the minimum of $\sum_{i}p_i \mathcal{Q}^{\psi_i}$, in the sense of matrix inequality, over the set of all finite ensembles $\{\{p_i\},\{\psi_i\},p_{\mathrm{s}},\rho_{\mathrm{s}}\}$ such that $\rho=\sum_{i=1}^{k-1}p_i\psi_i+p_{\mathrm{s}}\rho_{\mathrm{s}}$ does not exist in general. Here, $\rho_{\mathrm{s}}$ is a symmetric (possibly mixed) state, $\{\psi_i\}_{i=1}^{k-1}$ are pure states, $\{p_{\mathrm{s}},\{p_i\}_i\}$ is a probability distribution, i.e., $p_{\mathrm{s}}, p_i\geq 0$ and $\sum_{i}p_i+p_{\mathrm{s}}=1$. 

Here we provide an example where the minimum of $\sum_{i}p_i \mathcal{Q}^{\psi_i}$ does not exist. Consider a qubit system with a unitary representation $e^{\ii \sum_{i=x,y,z}\theta^i\sigma_i}$ of $G=SU(2)$ symmetry, where $ \sigma_x,\sigma_y,\sigma_z$ are the Pauli operators. Consider a mixed state $\rho=\frac{1}{2}(I+\epsilon\sigma_z)$ for some $\epsilon\in (0,1)$. We will analyze two different decompositions $\rho=\epsilon \ket{0}\bra{0}+(1-\epsilon)\frac{1}{2}I$ and $\rho=\frac{1}{2}(\ket{\psi_+}\bra{\psi_+}+\ket{\psi_+}\bra{\psi_+})$, where $\ket{\psi_{\pm}}\coloneqq \cos\frac{\varphi}{2}\ket{0}\pm \sin\frac{\varphi}{2}\ket{1}$ for $\varphi\in\mathbb{R}$ such that $\cos\varphi=\epsilon$. The QGT for each pure state is calculated as follows:
\begin{align}
    \mathcal{Q}^{\ket{0}}=\mathcal{Q}^{\ket{1}}=
    \begin{pmatrix}
        1&\ii &0\\
        -\ii& 1& 0\\
        0&0&0
    \end{pmatrix},\quad \mathcal{Q}^{\psi_\pm}=
    \begin{pmatrix}
        \cos^2\varphi& \ii \cos\varphi&\mp \cos\varphi\\
        -\ii\cos\varphi &1 &\pm \ii\sin\varphi\\
        \pm \cos\varphi\sin\varphi & \mp \ii\sin\varphi & \sin^2\varphi
    \end{pmatrix}. 
\end{align}
For the first decomposition $\rho=\epsilon \ket{0}\bra{0}+(1-\epsilon)\frac{1}{2}I$, the average QGT is given by
\begin{align}
    \epsilon \mathcal{Q}^{\ket{0}}=
    \begin{pmatrix}
        \cos\varphi&\ii\cos\varphi &0\\
        -\ii\cos\varphi& \cos\varphi& 0\\
        0&0&0
    \end{pmatrix},\label{eq:aveQGT_decompsition_1}
\end{align}
while the average QGT for the second decomposition $\rho=\frac{1}{2}(\ket{\psi_+}\bra{\psi_+}+\ket{\psi_+}\bra{\psi_+})$ is given by 
\begin{align}
    \frac{1}{2}\mathcal{Q}^{\psi_+}+\frac{1}{2}\mathcal{Q}^{\psi_-}=
    \begin{pmatrix}
        \cos^2\varphi& \ii \cos\varphi& 0\\
        -\ii\cos\varphi &1 & 0\\
        0 & 0 & \sin^2\varphi
    \end{pmatrix}.\label{eq:aveQGT_decompsition_2}
\end{align}

Now, suppose that there is an optimal decomposition of the state $\rho$ such that $\rho=\sum_ip_i\psi_i+p_{\mathrm{s}}\rho_{\mathrm{s}}$, where $\{\psi_i\}$ are pure states and $\rho_{\mathrm{s}}$ is a symmetric state and denote the minimum of $\sum_ip_i\mathcal{Q}^{\psi_i}$ by $\tilde{Q}^{\rho}$. On the one hand, from Eq.~\eqref{eq:aveQGT_decompsition_1}, $\tilde{Q}^{\rho}$ must satisfy $(\tilde{Q}^{\rho})_{33}=0$. On the other hand, Eq.~\eqref{eq:aveQGT_decompsition_2} implies that $(\tilde{Q}^{\rho})_{11}<\cos^2\varphi$. The condition $(\tilde{Q}^{\rho})_{33}=0$ implies that all the pure states in the optimal decomposition are the eigenvectors of $\sigma_z$. For such a decomposition, however, $(\tilde{Q}^{\rho})_{11}=\cos\varphi$, which contradicts the condition $(\tilde{Q}^{\rho})_{11}<\cos^2\varphi$ since $\cos^2\varphi<\cos\varphi$. Therefore, the minimum of $\sum_ip_i\mathcal{Q}^{\psi_i}$ does not exist for $\rho=\frac{1}{2}(I+\epsilon\sigma_z)$.

\section{QGT of purification}\label{app:QGT_of_purification}
For a mixed state $\rho$ on a system of interest $S$, we consider minimizing the QGT over the set of all purifications $\ket{\Phi_\rho}$ of $\rho$, by extending the arguments in \cite{marvian_operational_2022}. Let $A$ denote the ancillary system added to purify $\rho$. For sets of Hermitian operator $\{X_{S,\mu}\}_{\mu=1}^{\dim G}$ and $\{X_{A,\mu}\}_{\mu=1}^{\dim G}$, we define
\begin{align}
    O_{tot}\coloneqq O_S\otimes I_A+I_S\otimes O_A,\quad O_S=\gamma^\dag X_S\coloneqq \gamma^{\mu*}X_{S,\mu},\quad O_A=\gamma^\dag X_A\coloneqq \gamma^{\mu*}X_{A,\mu}
\end{align}
for $\gamma\in\mathbb{C}^{\dim G}$. We introduce $\Phi_\rho\coloneqq \ket{\Phi_\rho}\bra{\Phi_\rho}$. 
We will minimize
\begin{align}
    V\left(\Phi_\rho,O_{tot}\right)\coloneqq \braket{\Phi_\rho|O_{tot}(I-\Phi_\rho)O_{tot}^\dag |\Phi_\rho}=\braket{\Phi_\rho|O_{tot}O_{tot}^\dag |\Phi_\rho}-\braket{\Phi_\rho |O_{tot}|\Phi_\rho}\braket{\Phi_\rho |O_{tot}^\dag |\Phi_\rho}
\end{align}
over purifications $\Phi_\rho$ of $\rho$ and the sets of Hermitian operators $\{X_{A,\mu}\}_{\mu=1}^{\dim G}$ on $A$. We remark that the QGT $\mathcal{Q}^{\Phi_\rho}$ of purification $\Phi_\rho$ satisfies $V\left(\Phi_\rho,O_{tot}\right)=\gamma^\dag \mathcal{Q}^{\Phi_\rho} \gamma$. 

For two different purifications $\ket{\Phi_\rho}$ and $\ket{\Phi_\rho'}$ of $\rho$, there exists a unitary operator $U_A$ on $A$ such that $\ket{\Phi_\rho'}=I_S\otimes U_A\ket{\Phi_\rho}$. Since
\begin{align}
    \braket{\Phi_\rho'|O_{tot}(I-\Phi_\rho')O_{tot}^\dag |\Phi_\rho'}=  \braket{\Phi_\rho|O_{tot}'(I-\Phi_\rho)O_{tot}^{\prime \dag} |\Phi_\rho},  
\end{align}
where $O_{tot}'\coloneqq O_S\otimes I_A+I_S\otimes \gamma^\dag X_A'$ and $X_{A,\mu}'\coloneqq U_A^\dag X_{A,\mu}U_A$, it suffices to minimize over possible $\{X_{A,\mu}\}_{\mu=1}^{\dim G}$ for a fixed purification. In the following, we denote by $\ket{\Phi_\rho}$ a purification of $\rho$ given by
\begin{align}
    \ket{\Phi_\rho}=\sum_i\sqrt{p_i}\ket{\phi_i}\ket{\phi_i}=(\sqrt{\rho}\otimes I)\sum_i\ket{\phi_i}\ket{\phi_i},
\end{align}
where $\rho=\sum_{i}p_i\ket{\phi_i}\bra{\phi_i}$ is the eigenvalue decomposition of $\rho$. 

By redefining
\begin{align}
    X_{A,\mu}\to \tilde{X}_{A,\mu}\coloneqq X_{A,\mu}-\braket{\Phi_\rho|X_{S,\mu}\otimes I_A+I_S\otimes X_{A,\mu}|\Phi_\rho}I_A,
\end{align}
we get
\begin{align}
    \braket{\Phi_\rho|X_{S,\mu}\otimes I_A+I_S\otimes \tilde{X}_{A,\mu}|\Phi_\rho}=0
\end{align}
for all $\mu$, implying that $\braket{\Phi_\rho|\gamma^\dag X_{S}\otimes I_A+I_S\otimes \gamma^\dag \tilde{X}_{A}|\Phi_\rho}=0$. Since adding the multiple of $I_A$ to $X_{A,\mu}$ does not change $V\left(\Phi_\rho,O_{tot}\right)$, we assume, without loss of generality, that $\braket{\Phi_\rho|O_{tot}|\Phi_\rho}=\braket{\Phi_\rho|O_{tot}^\dag |\Phi_\rho}=0$ in the following. In this case, we have
\begin{align}
    V\left(\Phi_\rho,O_{tot}\right)
    &=\sum_{i}p_i\braket{\phi_i|O_SO_S^\dag |\phi_i}+\sum_{i}p_i\braket{\phi_i|O_AO_A^\dag|\phi_i}+\sum_{i,j}\sqrt{p_ip_j}\braket{\phi_i\otimes \phi_i|(O_S\otimes O_A^\dag +O_S^\dag \otimes O_A )|\phi_j\otimes\phi_j}.
\end{align}

Under a small variation of $X_{A,\mu}$ to $X_{A,\mu}+\delta X_{A,\mu}$ for $\mu=1,\ldots,\dim G$ with Hermitian operators $\delta X_{A,\mu}$, the operator $O_A$ changes $O_A\to O_A+\delta O_A$, where $\delta O_A\coloneqq \gamma^{\mu*} \delta X_{A,\mu}$, and hence we get
\begin{align}
    \delta V\left(\Phi_\rho,O_{tot}\right)&=\sum_ip_i\braket{\phi_i|O_A(\delta O_A)^\dag |\phi_i}+\sum_ip_i\braket{\phi_i|(\delta O_A) O_A^\dag |\phi_i}\nonumber\\
    &\quad +\sum_{i,j}\sqrt{p_ip_j}\braket{\phi_i\otimes \phi_i|(O_S\otimes( \delta O_A)^\dag +O_S^\dag \otimes(\delta O_A))|\phi_j\otimes\phi_j}+O((\delta O_A)^2)\\
    &=\mathcal{K}_{\mu\nu}\gamma^{\mu*}\gamma^\nu+O((\delta O_A)^2),
\end{align}
where
\begin{align}
    \mathcal{K}_{\mu\nu}&\coloneqq \sum_ip_i\braket{\phi_i|X_{A,\mu}(\delta X_{A,\nu})|\phi_i}+\sum_{i,j}\sqrt{p_ip_j}\braket{\phi_i\otimes \phi_i|X_{S,\mu}\otimes( \delta X_{A,\nu}) |\phi_j\otimes\phi_j}\nonumber\\
    &\quad +\sum_ip_i\braket{\phi_i|(\delta X_{A,\mu}) X_{A,\nu}|\phi_i}+\sum_{i,j}\sqrt{p_ip_j}\braket{\phi_i\otimes \phi_i|X_{S,\nu} \otimes( \delta X_{A,\mu}) |\phi_j\otimes\phi_j}.
\end{align}
Since $\mathcal{K}$ is Hermitian, $\delta \left\|\ii [\Phi_\rho,O_{tot} ]\right\|_{f_q,\Phi_\rho}^2$ vanishes up to the first order with respect to $\delta O_A$ for all $\gamma$ if and only if $\mathcal{K}=0$.

$\mathcal{K}_{\mu\mu}=0$ holds for all $\{\delta X_{A,\mu}\}_{\mu=1}^{\dim G}$ only if
\begin{align}
    \braket{\phi_i|X_{A,\mu}|\phi_j}=-\frac{2\sqrt{p_ip_j}}{p_i+p_j}\braket{\phi_j|X_{S,\mu}|\phi_i},
\end{align}
or equivalently,
\begin{align}
    X_{A,\mu}=-\sum_{i,j}\frac{2\sqrt{p_ip_j}}{p_i+p_j}\braket{\phi_j|X_{S,\mu}|\phi_i}\ket{\phi_i}\bra{\phi_j}.
\end{align}
However, for $\dim G>1$, $\mathcal{M}_{\mu\nu}\neq 0$ for $\mu\neq \nu$ in general. For example, for a variation such that $\delta X_{A,\nu}=0$ for $\nu>1$, we get 
\begin{align}
    \mathcal{K}_{2,1}&=\sum_ip_i\braket{\phi_i|X_{A,2}(\delta X_{A,1})|\phi_i}+\sum_{i,j}\sqrt{p_ip_j}\braket{\phi_i\otimes \phi_i|X_{S,2}\otimes( \delta X_{A,1}) |\phi_j\otimes\phi_j},
\end{align}
which does not vanish for $X_{A,\mu}=-\sum_{i,j}\frac{2\sqrt{p_ip_j}}{p_i+p_j}\braket{\phi_j|X_{S,\mu}|\phi_i}\ket{\phi_i}\bra{\phi_j}$. Therefore, for $\dim G>1$, the minimum of QGT of purification does not exist in general.

To explore further, let us minimize
\begin{align}
    \left\|\ii [\Phi_\rho,O_{tot} ]\right\|_{f_q,\Phi_\rho}^2=\frac{1}{1-q}V\left(\Phi_\rho,O_{tot}\right)+\frac{1}{q}V\left(\Phi_\rho,O_{tot}^\dag \right)
\end{align}
over purification $\Phi_\rho$ of $\rho$ and the set of Hermitian operators $\{X_{A,\mu}\}_{\mu=1}^{\dim G}$ on $A$. Note that $f_q$ for $q=1/2$ corresponds to the SLD norm. In this case, we find
\begin{align}
    \delta \left\|\ii [\Phi_\rho,O_{tot} ]\right\|_{f_q,\Phi_\rho}^2
    &=\frac{1}{1-q}\delta V\left(\Phi_\rho,O_{tot}\right)+\frac{1}{q}\delta V\left(\Phi_\rho,O_{tot}^\dag \right)\\
    &=\sum_i\left(\frac{1}{1-q}p_i\braket{\phi_i|O_A(\delta O_A)^\dag |\phi_i}+\frac{1}{q}p_i\braket{\phi_i|(\delta O_A)^\dag  O_A |\phi_i}\right)\nonumber\\
    &\quad +\sum_{i,j}\sqrt{p_ip_j}\left(\frac{1}{1-q}+\frac{1}{q}\right)\braket{\phi_i\otimes \phi_i|O_S\otimes( \delta O_A)^\dag |\phi_j\otimes\phi_j}\nonumber\\
    &\quad +\sum_i\left(\frac{1}{1-q}p_i\braket{\phi_i|(\delta O_A) O_A^\dag |\phi_i}+\frac{1}{q}p_i\braket{\phi_i|O_A^\dag  (\delta O_A )|\phi_i}\right)\nonumber\\
    &\quad +\sum_{i,j}\sqrt{p_ip_j}\left(\frac{1}{1-q}+\frac{1}{q}\right)\braket{\phi_i\otimes \phi_i|O_S^\dag \otimes( \delta O_A) |\phi_j\otimes\phi_j}+O((\delta O_A)^2)\\
    &=\mathcal{M}_{\mu\nu}\gamma^{\mu*}\gamma^\nu+O((\delta O_A)^2),
\end{align}
where
\begin{align}
    \mathcal{M}_{\mu\nu}&\coloneqq \sum_i\left(\frac{1}{1-q}p_i\braket{\phi_i|X_{A,\mu}(\delta X_{A,\nu})|\phi_i}+\frac{1}{q}p_i\braket{\phi_i|(\delta X_{A,\nu})  X_{A,\mu} |\phi_i}\right)\nonumber\\
    &\quad +\sum_{i,j}\sqrt{p_ip_j}\left(\frac{1}{1-q}+\frac{1}{q}\right)\braket{\phi_i\otimes \phi_i|X_{S,\mu}\otimes( \delta X_{A,\nu}) |\phi_j\otimes\phi_j}\nonumber\\
    &\quad +\sum_i\left(\frac{1}{1-q}p_i\braket{\phi_i|(\delta X_{A,\mu}) X_{A,\nu}|\phi_i}+\frac{1}{q}p_i\braket{\phi_i|X_{A,\nu}  (\delta X_{A,\mu} )|\phi_i}\right)\nonumber\\
    &\quad +\sum_{i,j}\sqrt{p_ip_j}\left(\frac{1}{1-q}+\frac{1}{q}\right)\braket{\phi_i\otimes \phi_i|X_{S,\nu} \otimes( \delta X_{A,\mu}) |\phi_j\otimes\phi_j}.
\end{align}
Since $\mathcal{M}$ is Hermitian, $\delta \left\|\ii [\Phi_\rho,O_{tot} ]\right\|_{f_q,\Phi_\rho}^2$ vanishes up to the first order with respect to $\delta O_A$ for all $\gamma$ if and only if $\mathcal{M}=0$. 

From $\mathcal{M}_{\mu\mu}=0$, we have
\begin{align}
    \left(\frac{1}{1-q}+\frac{1}{q}\right) (p_i+p_j)\braket{\phi_i|X_{A,\mu}|\phi_j}=-2\sqrt{p_ip_j}\left(\frac{1}{1-q}+\frac{1}{q}\right)\braket{\phi_j|X_{S,\mu}|\phi_i},
\end{align}
i.e.,
\begin{align}
    \braket{\phi_i|X_{A,\mu}|\phi_j}=-\frac{2\sqrt{p_ip_j}}{p_i+p_j}\braket{\phi_j|X_{S,\mu}|\phi_i},
\end{align}
or equivalently,
\begin{align}
    X_{A,\mu}=-\sum_{i,j}\frac{2\sqrt{p_ip_j}}{p_i+p_j}\braket{\phi_j|X_{S,\mu}|\phi_i}\ket{\phi_i}\bra{\phi_j}.\label{eq:variation_vanish_solution}
\end{align}
However, for $\dim G>1$, $\mathcal{M}_{\mu\nu}\neq 0$ for $\mu\neq \nu$ in general. For example, for a variation such that $\delta X_{A,\nu}=0$ for $\nu>1$, we get 
\begin{align}
    \mathcal{M}_{2,1}&=\sum_i\left(\frac{1}{1-q}p_i\braket{\phi_i|X_{A,2}(\delta X_{A,1})|\phi_i}+\frac{1}{q}p_i\braket{\phi_i|(\delta X_{A,1})  X_{A,2} |\phi_i}\right)\nonumber\\
    &\quad +\sum_{i,j}\sqrt{p_ip_j}\left(\frac{1}{1-q}+\frac{1}{q}\right)\braket{\phi_i\otimes \phi_i|X_{S,2}\otimes( \delta X_{A,1}) |\phi_j\otimes\phi_j},
\end{align}
which vanishes for $X_{A,\mu}=-\sum_{i,j}\frac{2\sqrt{p_ip_j}}{p_i+p_j}\braket{\phi_j|X_{S,\mu}|\phi_i}\ket{\phi_i}\bra{\phi_j}$ only for $q=1/2$. We remark that for $q=1/2$, $\mathcal{M}=0$ for the operators $X_{A,\mu}$ in Eq.~\eqref{eq:variation_vanish_solution}, which corresponds to the result in \cite{marvian_operational_2022} stating that the SLD quantum Fisher information satisfies $\mathcal{F}_H(\rho)=\min_{\Phi_\rho,H_A} 4V(H+H_A,\Phi_\rho)$, where the minimization is taken over all the purifications and the Hermitian operators $H_A$.

\section{Properties of QGT and its extension}\label{app:properties_S}
We here summarize the fundamental properties of the tensor $\mathcal{Q}_q$ defined in Eq.~\eqref{eq:definition_S_q}. Taking the limit $q\to 1^-$, they also hold for $\mathcal{Q}$, defined in Eq.~\eqref{eq:definition_S_matrix}, which coincides with the QGT $\mathcal{Q}$ for pure states. These properties are inherited from those of Petz's monotone metric:
\begin{prop}[Properties of $\mathcal{Q}_q$]
    \leavevmode\par
    \begin{enumerate}[(i)]
        \item Positivity: For any state $\rho$, $\mathcal{Q}_q^{\rho}\geq 0$. 
        \item Additivity: For any pair of states $\rho$ and $\sigma$, $\mathcal{Q}_q^{\rho\otimes \sigma}=\mathcal{Q}_q^{\rho}+\mathcal{Q}_q^{\sigma}$.
        \item Convexity: For any probability distribution $\{p_k\}_k$ and set of states $\{\rho_k\}_k$, $\sum_kp_k \mathcal{Q}_q^{\rho_k}\geq  \mathcal{Q}_q^{\sum_{k}p_k\rho_k}$
        \item Monotonicity: For any $G$-covariant quantum channel $\mathcal{E}$ and state $\rho$, $\mathcal{Q}_q^{\rho}\geq \mathcal{Q}_q^{\mathcal{E}(\rho)}$. 
        \item Strong monotonicity (also referred to as selective monotonicity): Let $\{\mathcal{E}_k\}_k$ be any $G$-covariant instrument, that is, a set of $G$-covariant completely-positive and trace-nonincreasing maps such that $\sum_k\mathcal{E}_k$ is trace-preserving. For $p_k\coloneqq \mathrm{Tr}(\mathcal{E}_k(\rho))$ and $\rho_k\coloneqq \frac{\mathcal{E}_k(\rho)}{p_k}$, 
        \begin{align}
            \mathcal{Q}_q^\rho\geq \sum_kp_k\mathcal{Q}_q^{\rho_k}.
        \end{align}
    \end{enumerate}
\end{prop}
\begin{proof}
    The positivity and the additivity of the norm $\|\cdot\|_{f,\rho}$ imply (i) and (ii), respectively. The monotonicity (iv) has already been proved in Eq.~\eqref{eq:monotonicity_norm_general_fq}. Below, we show (iii) and (v). 

    Let $\{p_k\}_{k}$ be a probability distribution and $\{\rho_k\}_k$ be a set of states. We define $\sigma\coloneqq \sum_kp_k\rho_k\otimes \ket{k}\bra{k}$, where $\{\ket{k}\}_k$ is an orthonormal basis of an ancillary system. For any linear operator $O$, we show
    \begin{align}
        \sum_kp_k \|\ii[\rho_k,O]\|_{f,\rho_k}^2=\left\|\ii\left[\sigma,O\otimes \mathbb{I}\right]\right\|_{f,\sigma}^2 \label{eq:flag}
    \end{align}
    for an operator monotone function $f$.
    To prove this equality, let $\rho_k=\sum_{i=1}^d\lambda_i^{(k)}\ket{i^{(k)}}\bra{i^{(k)}}$ be the eigenvalue decomposition of $\rho_k$. Since 
    \begin{align}
        \sigma=\sum_k \sum_i p_k\lambda_i^{(k)}\ket{i^{(k)}}\bra{i^{(k)}}\otimes \ket{k}\bra{k}
    \end{align}
    is the eigenvalue decomposition of $\sigma$, we get
    \begin{align}
        &\left\|\ii\left[\sigma,O\otimes \mathbb{I}\right]\right\|_{f,\sigma}^2\nonumber \\
        &=\sum_{i,j,k,l;m_f(p_k\lambda^{(k)}_i,p_l\lambda^{(l)}_j)>0}\frac{1}{m_f(p_k\lambda^{(k)}_i,p_l\lambda^{(l)}_j)}\braket{i^{(k)}\otimes k|\left(\ii [\sigma,O\otimes \mathbb{I}]\right)^\dag |j^{(l)}\otimes l}\braket{j^{(l)}\otimes l|\ii [\sigma,O\otimes \mathbb{I}] |i^{(k)}\otimes k}\nonumber \\
        &=\sum_{i,j,k;m_f(p_k\lambda^{(k)}_i,p_k\lambda^{(k)}_j)>0}\frac{1}{m_f(p_k\lambda^{(k)}_i,p_k\lambda^{(k)}_j)}\braket{i^{(k)}|\left(\ii [p_k\rho_k,O]\right)^\dag |j^{(k)}}\braket{j^{(k)}|\ii [p_k\rho_k,O] |i^{(k)}}\nonumber \\
        &=\sum_{k;p_k>0}p_k\sum_{i,j;m_f(\lambda^{(k)}_i,\lambda^{(k)}_j)>0}\frac{1}{m_f(\lambda^{(k)}_i,\lambda^{(k)}_j)}\braket{i^{(k)}|\left(\ii [\rho_k,O]\right)^\dag |j^{(k)}}\braket{j^{(k)}|\ii [\rho_k,O] |i^{(k)}}\nonumber \\
        &=\sum_{k;p_k>0}p_k \|\ii[\rho_k,O]\|_{f,\rho_k}^2=\sum_kp_k \|\ii[\rho_k,O]\|_{f,\rho_k}^2,
    \end{align}
    where in the third equality, we have used the homogeneity of $m_f$, i.e., $m_f(ay,ay)=a m_f(x,y)$ for $a,x,y\in\mathbb{R}$.

    For the ancillary system, we adopt the trivial unitary representation of $G$. Consequently, the projective unitary representation $V$ of $G$ on the total system satisfies $V(g)=U(g)\otimes \mathbb{I}$ for $g\in G$, where $U$ denotes the representation on the quantum system on which $\rho_k$ is defined. By differentiating $V$ as defined in Eq.~\eqref{eq:hermitian_operators_Lie_alg}, we get
    \begin{align}
        Y_\mu\coloneqq \frac{\partial}{\partial \lambda^\mu} V(g(\lambda))\biggl|_{\lambda=0}=X_\mu\otimes \mathbb{I},
    \end{align}
    where $X_\mu\coloneqq\frac{\partial}{\partial \lambda^\mu} U(g(\lambda))|_{\lambda=0}$. By applying Eq.~\eqref{eq:flag} for $O\coloneqq \gamma^\dag X$ with $\gamma\in\mathbb{C}^{\dim G}$, we get
    \begin{align}
        \sum_kp_k \|\ii[\rho_k,\gamma^\dag X]\|_{f,\rho_k}^2=\left\|\ii\left[\sigma,\gamma^\dag Y\right]\right\|_{f,\sigma}^2.\label{eq:flag_diff}
    \end{align}
    
    We now prove the properties~(iii) and (v) from this relation and the monotonicity for $f_q(x)=(1-q)+q x$. From the monotonicity under partial trace, which is a $G$-covariant operation, we obtain
    \begin{align}
        \left\|\ii\left[\sigma,\gamma^\dag Y\right]\right\|_{f_q,\sigma}^2\geq \left\|\ii\left[\sum_{k}p_k\rho_k,\gamma^\dag X\right]\right\|_{f_q,\sum_{k}p_k\rho_k}^2.
    \end{align}
    Thus, from Eq.~\eqref{eq:flag_diff}, we get
    \begin{align}
        \sum_kp_k \|\ii[\rho_k,\gamma^\dag X]\|_{f_q,\rho_k}^2\geq \left\|\ii\left[\sum_{k}p_k\rho_k,\gamma^\dag X\right]\right\|_{f_q,\sum_{k}p_k\rho_k}^2.
    \end{align}
    Multiplying both sides by $f_q(0)$, we find
    \begin{align}
        \sum_kp_k\gamma^\dag \mathcal{Q}_q^{\rho_k}\gamma\geq \gamma^\dag \mathcal{Q}_q^{\sum_kp_k\rho_k}\gamma
    \end{align}
    holds for any $\gamma\in\mathbb{C}^{\dim G}$, which completes the proof of the property~(iii).

    In order to prove the property~(v), we define a quantum channel $\tilde{\mathcal{E}}(\cdot)\coloneqq \sum_{k}\mathcal{E}_k(\cdot)\otimes \ket{k}\bra{k}$, where $\{\mathcal{E}_k\}_k$ is a $G$-covariant instrument. Since the group transformation $G$ trivially acts on the ancillary system, $\tilde{\mathcal{E}}(\cdot)$ is a $G$-covariant operation. For a state $\rho$, we define
    \begin{align}
        p_k\coloneqq \mathrm{Tr}\left(\mathcal{E}_k(\rho)\right), \quad \rho_k\coloneqq \frac{1}{p_k}\mathcal{E}_k(\rho)
    \end{align}
    so that $\sigma=\tilde{\mathcal{E}}(\rho)$. From the monotonicity under $\tilde{\mathcal{E}}$, we obtain
    \begin{align}
         \left\|\ii\left[\rho,\gamma^\dag X\right]\right\|_{f_q,\rho}^2\geq \left\|\ii\left[\sigma,\gamma^\dag Y\right]\right\|_{f_q,\sigma}^2.
    \end{align}
    Therefore, from Eq.~\eqref{eq:flag_diff}, we get
    \begin{align}
        \left\|\ii\left[\rho,\gamma^\dag X\right]\right\|_{f_q,\rho}^2\geq \sum_kp_k \|\ii[\rho_k,\gamma^\dag X]\|_{f_q,\rho_k}^2.
    \end{align}
    Again, multiplying both sides by $f_q(0)$, we find
    \begin{align}
        \gamma^\dag \mathcal{Q}_q^{\rho}\gamma\geq \sum_{k}p_k\gamma^\dag \mathcal{Q}_q^{\rho_k}\gamma
    \end{align}
     holds for any $\gamma\in\mathbb{C}^{\dim G}$, which completes the proof of the property~(v). 
\end{proof}

\section{Conversion condition in quantum thermodynamics and its relation to RTA}\label{app:quantum_thermo_setup_review_and_relation}
We begin by briefly reviewing the setup of state conversion in quantum thermodynamics in Refs.~\cite{faist_macroscopic_2019,sagawa_asymptotic_2021}, where auxiliary systems that provide work and coherence are explicitly taken into account. A $(w,\eta)$-work/coherence-assisted thermal operation is a quantum process that can be achieved through a thermal operation with the help of ancillary systems that provide work $w$ and coherence whose energy range is bounded by $\eta$. Precisely, a completely positive and trace-nonincreasing map $\mathcal{E}_{S\to S'}$ is referred to as a $(w,\eta)$-work/coherence-assisted thermal operation if there exist 
\begin{itemize}
    \item quantum systems $W,C,W',C'$ with respective Hamiltonians $H_W,H_C,H_{W'},H_{C'}$ satisfying $\|H_C\|_\infty<\eta$ and $\|H_{C'}\|_\infty<\eta$,
    \item energy eigenstates $\ket{E}_W$ and $\ket{E'}_{W'}$ of $H_{W}$ and $H_{W'}$ satisfying $E-E'=w$,
    \item pure states $\ket{\zeta}_{C}$ and $\ket{\zeta'}_{C'}$,
    \item a thermal operation $\tilde{\mathcal{E}}_{SWC\to S'W'C'}$
\end{itemize}
such that
\begin{align}
    &\mathcal{E}_{S\to S'}(\rho_S)=\mathrm{Tr}_{W'C'}\left(\ket{E'}\bra{E'}_{W'}\otimes \ket{\zeta'}\bra{\zeta'}_{C'}\tilde{\mathcal{E}}_{SWC\to S'W'C'}\left( \rho_{S}\otimes \ket{E}\bra{E}_{W}\otimes \ket{\zeta}\bra{\zeta}_{C}\right)\right).\label{eq:CPTNI_TO}
\end{align}
We say that a state $\rho_S$ is $(w,\eta,\epsilon)$-transformable into $\sigma_{S'}$ by a thermal operation if there exists $(w,\eta)$-work/coherence-assisted thermal operation such that $T(\mathcal{E}_{S\to S'}(\rho_S),\sigma_{S'})\leq \epsilon$. In this case, we write $(\rho,H_S)\xrightarrow[\mathrm{TO}]{(w,\eta,\epsilon)}(\sigma_{S'},H_{S'})$, where $H_{S}$ and $H_{S'}$ denote the Hamiltonians of the systems $S$ and $S'$, respectively. Note that for subnormalized states $\rho$ and $\sigma$, the generalized trace distance~\cite{tomamichel_duality_2010,tomamichel_framework_2012} is defined by $T(\rho,\sigma)\coloneqq \frac{1}{2}\|\rho-\sigma\|_1+\frac{1}{2}|\mathrm{Tr}(\rho)-\mathrm{Tr}(\sigma)|$. Moreover, for sequences of states $\{\rho_N\}_N$ and $\{\sigma_N\}_N$ and sequences of Hamiltonians $\{H_N\}_N$ and $\{H'_N\}_N$, we say that $\{\rho_N\}_N$ is asymptotically convertible to $\{\sigma_N\}_N$ at a work rate $w$ if there exist sequences of $w_N,\eta_N,\epsilon_N$ such that $(\rho_N,H_N)\xrightarrow[\mathrm{TO}]{w_N,\eta_N,\epsilon_N}(\sigma_N,H_N')$ and 
\begin{align}
    \lim_{N\to\infty}\frac{w_N}{N}=w,\quad \lim_{N\to\infty}\frac{\eta_N}{N}=0,\quad \lim_{N\to\infty}\epsilon_N=0.\label{eq:conversion_condition_TO}
\end{align}
In this case, we write $\{\rho_N,H_N\}_N\xrightarrow[\mathrm{TO}]{w} \{\sigma_N,H_N'\}_N$.

In Sec.~\ref{sec:quantum_thermodynamics}, we aimed to quantify the coherence in $\ket{E}\bra{E}_W\otimes \ket{\zeta}\bra{\zeta}_{C}$ that must be supplied from the external systems. To formalize this argument, we first prove the following lemma:
\begin{lem}\label{lem:generalized_TD_and_gentle_measurement}
    For quantum systems $A$ and $B$, let $\tau_{AB}$ and $\sigma_A$ be states on $AB$ and $A$, respectively. For any pure state $\psi_B\coloneqq \ket{\psi}\bra{\psi}_B$ on $B$, it holds
    \begin{align}
        T(\tau_{AB},\sigma_A\otimes \psi_B)\leq 2T\left(\chi_A,\sigma_A\right)+\sqrt{2T(\chi_A,\sigma_A)},\label{eq:trace_distance_inequalities}
    \end{align}
    where $\chi_{A}\coloneqq \mathrm{Tr}_B((\mathbb{I}_A\otimes \psi_{B})\tau_{AB})$ is a subnormalized state. 
\end{lem}

\begin{proof}   
    Define the projector $P_{AB}\coloneqq \mathbb{I}_A\otimes \psi_B$ and introduce $p\coloneqq\mathrm{Tr}(P_{AB}\tau_{AB})=\mathrm{Tr}_A (\chi_A)$. If $p=0$, $T(\tau_{AB},\sigma_A\otimes \psi_B)=T(\chi_A,\sigma_A)=0$, implying that Eq.~\eqref{eq:trace_distance_inequalities} holds trivially. When $p>0$, we introduce a normalized state
    \begin{align}
        \gamma_{AB}\coloneqq \frac{P_{AB}\tau_{AB}P_{AB}}{p}= \frac{\chi_A}{p}\otimes \psi_B.
    \end{align}
    From the triangle inequality for the trace distance, we have
    \begin{align}
        T(\tau_{AB},\sigma_A\otimes \psi_B)\leq T(\tau_{AB},\gamma_{AB})+T(\gamma_{AB},\sigma_A\otimes \psi_B).\label{eq:triangle_ineq}
    \end{align}
    We derive an upper bound for each term in the right-hand side. 

    For $\delta\coloneqq T(\chi_A,\sigma_A)$, since
    \begin{align}
        \delta=\frac{1}{2}\|\chi_A-\sigma_A\|_1+\frac{1}{2}|\mathrm{Tr}(\chi_A)-\mathrm{Tr}(\sigma_A)|\geq \frac{1}{2}|\mathrm{Tr}(\chi_A)-\mathrm{Tr}(\sigma_A)|=\frac{1}{2}(1-p),
    \end{align}
    we have
    \begin{align}
        \mathrm{Tr}\left(P_{AB}\tau_{AB}\right)=p\geq 1-2\delta. 
    \end{align}
    From the gentle measurement lemma~\cite{wilde_quantum_2017}, we have $\|\tau_{AB}-\gamma_{AB}\|_1\leq 2\sqrt{2\delta}$, implying that 
    \begin{align}
        T(\tau_{AB},\gamma_{AB})\leq \sqrt{2\delta}.\label{eq:gentle_measurement}
    \end{align}
    
    To derive an upper bound for the second term, we use
    \begin{align}
        \left\|\frac{\chi_A}{p}-\sigma_A\right\|_1&=\left\|\chi_A-\sigma_A+\left(1-\frac{1}{p}\right)\chi_A\right\|_1\\
        &\leq \left\|\chi_A-\sigma_A\right\|_1+\left\|\left(1-\frac{1}{p}\right)\chi_A\right\|_1\\
        &=2\delta+\left|\left(1-\frac{1}{p}\right)\right|\mathrm{Tr}(\chi_A)\\
        &=2\delta+(1-p)\leq 4\delta.
    \end{align}
    Therefore, we get
    \begin{align}
        T(\gamma_{AB},\sigma_A\otimes\psi_B)&=\frac{1}{2}\left\| \frac{\chi_A}{p}\otimes \psi_B-\sigma_A\otimes\psi_B\right\|_1\\
        &=\frac{1}{2}\left\|\frac{\chi_A}{p}-\sigma_A\right\|_1\\
        &\leq 2\delta.\label{eq:proj_state}
    \end{align}

    From Eqs.~\eqref{eq:triangle_ineq}, \eqref{eq:gentle_measurement}, and \eqref{eq:proj_state}, we obtain
    \begin{align}
        T(\tau_{AB},\sigma_A\otimes\psi_B)\leq 2\delta+\sqrt{2\delta}=2T(\chi_A,\sigma_A)+\sqrt{2T(\chi_A,\sigma_A)}.
    \end{align}
\end{proof}

As an immediate consequence of this lemma, we find a relation between conversion in quantum thermodynamics to that in RTA under time translation symmetry:
\begin{prop}
    Suppose that $\{\rho_N,H_N\}_N\xrightarrow[\mathrm{TO}]{w} \{\sigma_N,H_N'\}_N$. Then there exists a sequence of covariant quantum channels $\{\Lambda_N\}_N$ such that $\lim_{N\to\infty}T(\Lambda_N(\rho_N\otimes \xi_N),\sigma_N)=0$, where $\{\xi_N\}_N$ is a sequence of the initial pure states of ancillary systems. 
\end{prop}
\begin{proof}
    By definition, $(\rho_N,H_N)\xrightarrow[\mathrm{TO}]{w_N,\eta_N,\epsilon_N}(\sigma_N,H_N')$ implies that there exist systems $W_N,C_N,W_N',C_N'$ with Hamiltonians $H_{W_N},H_{C_N},H_{W_N'},H_{C_N'}$, eigenstates $\ket{E_N}_{W_N}$ and $\ket{E_N'}_{W_N'}$, pure states $\ket{\zeta_N}_{C_N}$ and $\ket{\zeta_N'}_{C_N'}$, and a thermal operation $\tilde{\mathcal{E}}_N$ such that
    \begin{align}
       T\left( \mathrm{Tr}_{W_N'C_N'}\left(\ket{E_N'}\bra{E_N'}_{W_N'}\otimes \ket{\zeta_N'}\bra{\zeta_N'}_{C_N'}\tilde{\mathcal{E}}_N\left(\rho_N\otimes \xi_N\right)\right),\sigma_N\right)\leq \epsilon_N,
    \end{align}
    where $\xi_N\coloneqq \ket{E_N}\bra{E_N}_{W_N}\otimes\ket{\zeta_N}\bra{\zeta_N}_{C_N}$. From Lemma~\ref{lem:generalized_TD_and_gentle_measurement}, 
    \begin{align}
        T\left( \tilde{\mathcal{E}}_N\left(\rho_N\otimes \eta_N\right),\ket{E_N'}\bra{E_N'}_{W_N'}\otimes \ket{\zeta_N'}\bra{\zeta_N'}_{C_N'}\otimes \sigma_N\right)\leq 2\epsilon_N+\sqrt{2\epsilon_N}.
    \end{align}
    Therefore, from the data-processing inequality, we obtain
    \begin{align}
        T\left(\Lambda_N(\rho_N\otimes \xi_N),\sigma_N\right) \leq 2\epsilon_N+\sqrt{2\epsilon_N}, \label{eq:conversion_TI_operation}
    \end{align}
    where $\Lambda_N\coloneqq  \mathrm{Tr}_{W_N'C_N'}\circ \tilde{\mathcal{E}}_N$ is a covariant channel under time translation symmetry. 
    
    If $\{\rho_N,H_N\}_N\xrightarrow[\mathrm{TO}]{w} \{\sigma_N,H_N'\}_N$, then $(\rho_N,H_N)\xrightarrow[\mathrm{TO}]{w_N,\eta_N,\epsilon_N} (\sigma_N,H_N')$ such that $\lim_{N\to\infty}\epsilon_N=0$. In this case, Eq.~\eqref{eq:conversion_TI_operation} leads to $\lim_{N\to\infty}T(\Lambda_N(\rho_N\otimes \xi_N),\sigma_N)=0$.
\end{proof}

\section{Circumventing constraint imposed by symmetry subgroup}
\label{sec:sym_subgroup_conversion_rate}

When converting a pure state $\psi$ to another pure state $\phi$ in the asymptotic setup, the monotonicity of the QGTs yields an upper bound on possible conversion rate as shown in Sec.~\ref{sec:monotonicity_QGT_sketch}. However, the symmetry subgroups further restrict achievable conversion rate: the conversion rate always vanishes if $\mathrm{Sym}_G(\psi)\not \subset\mathrm{Sym}_G(\phi)$. In this section, we investigate ways to circumvent this restriction due to the symmetry subgroups by modifying the setup. 

\subsection{Synchronization of asymmetry}
Let us consider a scenario in which one aims to convert copies of pure states, $\psi$ and $\phi$, into copies of one of them---say, $\phi$. For instance, such a process achieves the synchronization of quantum clocks in RTA for $G = U(1)$, where copies of states $\psi$ and $\phi$, possibly having different periods, are converted into copies of $\phi$. 
See Fig.~\ref{fig:sync} for a schematic figure of this setup. 
We here first analyze the conversion rate in this synchronization scenario for $G=U(1)$ and later generalize it for an arbitrary compact Lie group. In such a synchronization scenario, as a corollary of Theorem~\ref{thm:conversion_rate_projective_finite_number}, we find that the restriction arising from the symmetry subgroups can be bypassed. 

\begin{figure}[tb]
    \centering
    \includegraphics[width=8.3cm]{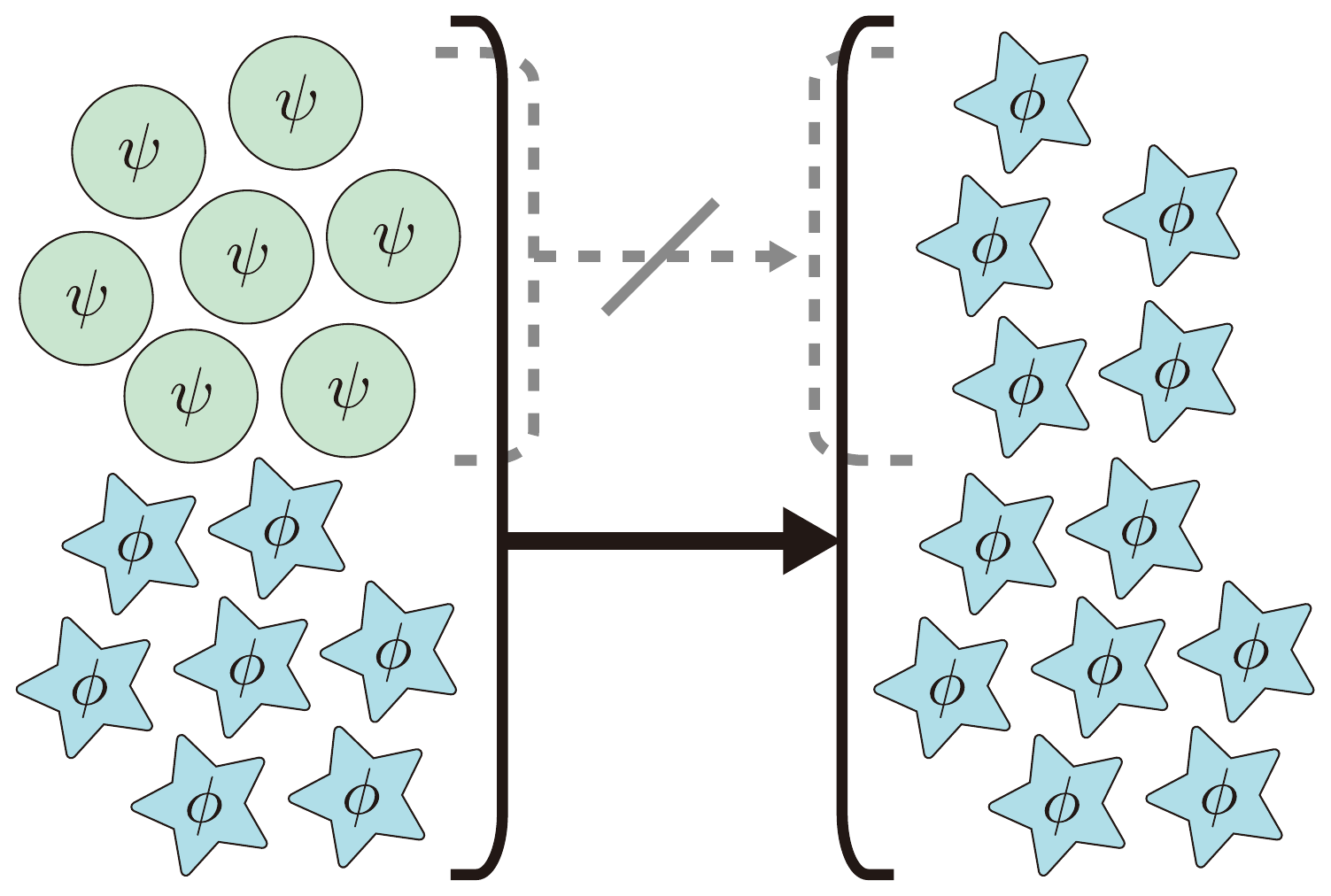}
    \caption{Schematic picture of the synchronization scenario. When $\mathrm{Sym}_G(\psi)\not\subset \mathrm{Sym}_G(\phi)$, the conversion rate from $\psi$ to $\phi$ vanishes, as shown by the dashed arrow in the figure. However, this restriction is circumvented when i.i.d. copies of $\psi\otimes \phi$ are converted into i.i.d. copies of $\phi$, as indicated by the solid arrow in the figure. }
    \label{fig:sync}
\end{figure}

Consider a pair of systems in pure states $\psi$ and $\phi$ with unitary representations $U(e^{\ii \theta})=e^{\ii H\theta}$ and $U'(e^{\ii \theta})=e^{\ii H'\theta}$ of $G=U(1)$ for $\theta\in[0,2\pi)$, where $H$ and $H'$ are Hermitian operators with integer eigenvalues. As mentioned in Sec.~\ref{sec:U(1)_same_periods} and Sec.~\ref{sec:irreversiblity}, when converting $\psi$ into $\phi$, Theorem~\ref{thm:conversion_rate_projective_finite_number} implies
\begin{align}
    \rap(\psi\to\phi)=
    \begin{cases}
        \frac{V(\psi,H)}{V(\phi,H')}\, &(\text{if }\exists k\in \mathbb{N},\,\tau=k\tau')\\
        0&(\text{otherwise})
    \end{cases},
\end{align}
where $\tau$ and $\tau'$ denotes the periods of $\psi$ and $\phi$, respectively. Now, note that in a synchronization scenario where $\psi\otimes \phi$ is converted into $\phi$, the period of $\psi\otimes \phi$ is always a positive integer multiple of the period of $\phi$. Therefore, from Theorem~\ref{thm:conversion_rate_projective_finite_number}, we immediately get
\begin{align}
    \rap(\psi\otimes \phi\to\phi)&=\frac{V(\psi\otimes \phi,H+H')}{V(\phi,H')}\nonumber \\
    &=\frac{V(\psi,H)}{V(\phi,H')}+1,
\end{align} 
regardless of the periods of $\psi$ and $\phi$. Importantly, even when $R(\psi\to\phi)=0$, $\{\psi^{\otimes N}\otimes \phi^{\otimes N}\}_N$ is asymptotically convertible to $\{\phi^{\otimes \floor{Nr^*}}\otimes \phi^{\otimes N}\}_N$ for any $r^*\in[0,V(\psi,H)/V(\phi,H'))$. In this case, therefore, $\phi^{\otimes N}$ acts like a catalyst, enhancing the conversion rate from zero to $V(\psi, H)/V(\phi, H')$ by bypassing the restriction on state convertibility imposed by their periods. 

Let us generalize the above observation to a scenario where copies of $\psi\otimes \phi$ are asymptotically converted to copies $\phi$ in RTA for a general compact Lie group $G$, which we shall refer to as the synchronization of asymmetry. Since $\mathrm{Sym}_G(\psi\otimes \phi)=\mathrm{Sym}_G(\psi)\cap \mathrm{Sym}_G( \phi)\subset \mathrm{Sym}_G(\phi)$, Eq.~\eqref{eq:conversion_rate_formula_forall_g} implies
\begin{align}
    &\rap (\psi\otimes \phi\to \phi)\nonumber \\
    &=\sup\{r\geq0 \mid\forall g\in G,\,\mathcal{Q}^{\mathcal{U}_g(\psi)\otimes \mathcal{U}_g'(\phi)}\geq r\mathcal{Q}^{\mathcal{U}'_g(\phi)}\}\\
    &=\sup\{r\geq0 \mid\forall g\in G,\,\mathcal{Q}^{\mathcal{U}_g(\psi)}\geq r\mathcal{Q}^{\mathcal{U}'_g(\phi)}\}+1,
\end{align}
where in the last equality, we used the additivity of the QGT, i.e., $\mathcal{Q}^{\mathcal{U}_g(\psi)\otimes \mathcal{U}'_g(\phi)}=\mathcal{Q}^{\mathcal{U}_g(\psi)}+\mathcal{Q}^{\mathcal{U}'_g(\phi)}$. Therefore, even when $\mathrm{Sym}_G(\psi)\not \subset \mathrm{Sym}_G(\phi)$ and hence $R(\psi\to\phi)=0$, it is possible to asymptotically convert $\{\psi^{\otimes N}\otimes \phi^{\otimes N}\}_N$ into $\{\phi^{\otimes \floor{r^*N}}\otimes \phi^{\otimes N}\}_N$ with vanishing error if $r^*<\sup\{r\geq0 \mid\forall g\in G,\,\mathcal{Q}^{\mathcal{U}_g(\psi)}\geq r\mathcal{Q}^{\mathcal{U}'_g(\phi)}\}$. 

So far, we have analyzed the case where copies of $\psi$ and $\phi$ are initially available at a ratio of 1:1 for simplicity. We here remark that the same argument applies to a more general mixture ratio. Specifically, let us consider the case where the mixture ratio is $m_1:m_2$ for positive integers $m_1$ and $m_2$, i.e., the conversion from copies of $\psi^{\otimes m_1} \otimes \phi^{\otimes m_2}$ into copies of $\phi$. Note that a general ratio $1:s$ for $s>0$ falls into this setup when $s$ is approximated by a rational ratio $m_2/m_1$. Since $\mathrm{Sym}_G(\psi^{\otimes m_1}\otimes \phi^{\otimes m_2})=\mathrm{Sym}_G(\psi)\cap \mathrm{Sym}_G(\phi)$, Eq.~\eqref{eq:conversion_rate_formula_forall_g} implies $\rap(\psi^{\otimes m_1}\otimes \phi^{\otimes m_2}\to\phi)=m_1\sup\{r\geq0 \mid\forall g\in G,\,\mathcal{Q}^{\mathcal{U}_g(\psi)}\geq r\mathcal{Q}^{\mathcal{U}'_g(\phi)}\}+m_2$. Therefore, $\sup\{r\geq0 \mid\forall g\in G,\,\mathcal{Q}^{\mathcal{U}_g(\psi)}\geq r\mathcal{Q}^{\mathcal{U}'_g(\phi)}\}$ generally provides the conversion rate from $\psi$ to $\phi$ in the synchronization scenario, regardless of the inclusion relation between their symmetry subgroups. 

\subsection{Sublinear additional resource}
In the previous subsection, we demonstrated that the conversion rate in the synchronization scenario is given by the expression in Eq.~\eqref{eq:conversion_rate_formula_forall_g}, regardless of the inclusion relation of symmetry subgroups. 
In this subsection, we investigate a generalized setup, showing that the restriction on the asymptotic conversion rate arising from the symmetry subgroup can be removed when a sublinear number of i.i.d. copies of a resource state is additionally available. Specifically, we prove the following:
\begin{thm}\label{thm:subliner}
    Let $\psi$ and $\phi$ be pure states. Fix any state $\chi$ on a finite-dimensional Hilbert space such that $\mathrm{Sym}_G(\chi)\subset \mathrm{Sym}_G(\phi)$. Consider a scenario where a sublinear number of $\chi$, $\chi^{\otimes \ceil{N^{1-\epsilon}}}$ with $\epsilon\in(0,1/9)$, is additionally available as a resource when converting $\psi^{\otimes N}$ into $\phi^{\otimes \floor{rN}}$. The optimal achievable rate in this setup, defined by $R^{*}(\psi\to\phi\, ;\chi)\coloneqq \sup\{r\geq 0\mid \{\psi^{\otimes N}\otimes \chi^{\otimes N^{1-\epsilon}}\}_N\gconv\{\phi^{\otimes \floor{rN}}\}_N\}$
    is given by
    \begin{align}
        R^{*}(\psi\to\phi\, ;\chi)=\sup\{r\geq 0\mid  \forall g\in G,\, \mathcal{Q}^{\mathcal{U}_g(\psi)}\geq r\mathcal{Q}^{\mathcal{U}_g'(\phi)}\}.
    \end{align}
\end{thm}
See Fig.~\ref{fig:sublinear} for a schematic picture of the setup of Theorem~\ref{thm:subliner}.
\begin{figure}[tb]
    \centering
    \includegraphics[width=8.3cm]{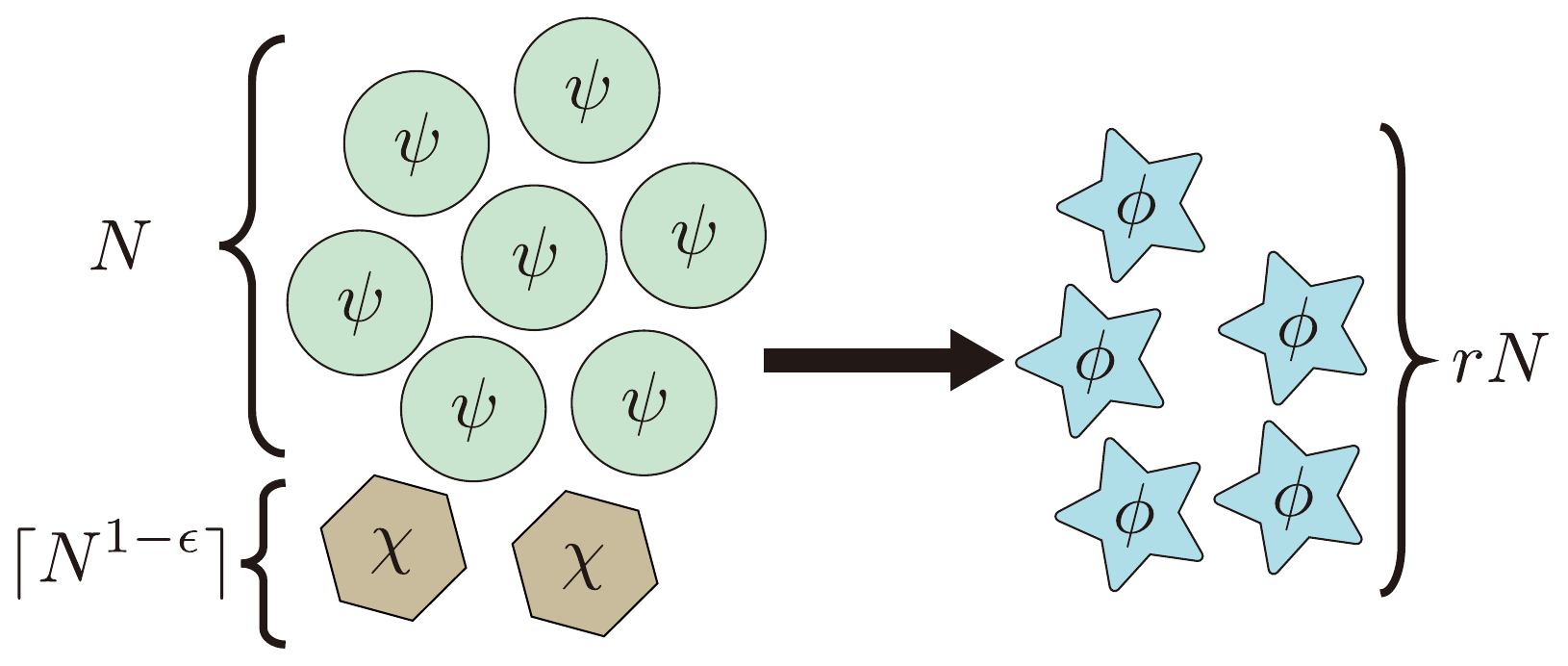}
    \caption{Schematic picture of the setup of Theorem~\ref{thm:subliner}, where a sublinear additional resource, $\chi^{\otimes \ceil{N^{1-\epsilon}}}$, is available.}
    \label{fig:sublinear}
\end{figure}

It should be emphasized that this theorem is valid regardless of the inclusion relation between the symmetry subgroups of pure states $\psi$ and $\phi$. The special case of $\chi=\phi$ reproduces the results in the previous subsection, in which the use of $\phi^{\otimes \ceil{N^{1-\epsilon}}}$ was not obvious because we had focused only on the linear conversion rate. 
As will be explained in the proof of Theorem~\ref{thm:subliner} below, $\chi^{\otimes \ceil{N^{1-\epsilon}}}$ is used when identifying the underlying group transformations in the estimation step of the direct part. 
A state $\chi$ such that $\mathrm{Sym}_G(\chi)=\{e\}$, which satisfies $\mathrm{Sym}_G(\chi)\subset \mathrm{Sym}_G(\phi)$ for any $\phi$, is of particular importance, since such a state $\chi$ ensures that no ambiguity originating from the symmetry subgroup arises in the estimation. 
For $G=U(1)$, an example of such a state is the coherence bit, given by a superposition of the ground state and the first excited state. Given that Lemma~\ref{lem:reasonable_estimator} is valid for any state $\rho$, we emphasize that Theorem~\ref{thm:subliner} applies for a general state $\chi$, including mixed states.

\begin{proof}[Proof of Theorem~\ref{thm:subliner}]
    Converse part: The converse part remains valid since adding a sublinear number of resource states does not change the optimal possible conversion rate. More precisely, instead of Eq.~\eqref{eq:fq_norm_monotonicity}, we get
    \begin{align}
        &f_q(0)\|\ii[\psi,O]\|_{f_q,\psi}^2+\frac{\ceil{N^{1-\epsilon}}}{N} \|\ii[\chi,O'']\|_{f_q,\chi}^2\nonumber \\
        &\quad \geq  \frac{M}{N}\gamma^\dag \mathcal{Q}^\phi\gamma-\frac{M}{N}h(\epsilon)+\frac{1}{N}o\left(M\right),
    \end{align}
    where $O''$ is defined by $O''\coloneqq \gamma^\dag X''$ and $X''_{\mu}\coloneqq -\ii \frac{\partial}{\partial\lambda^\mu}U''(g(\lambda))|_{\lambda=0}$ by using the projective unitary representation of $G$ on the Hilbert space on which $\chi$ is defined. Since $ \|\ii[\chi,O'']\|_{f_q,\chi}^2<\infty$ for any finite-dimensional system, the second term in the left-hand side vanishes in the limit of $N\to\infty$. By following the arguments from Eq.~\eqref{eq:fq_norm_monotonicity} to Eq.~\eqref{eq:monotonicity_QGT_converse_part}, we therefore get
    \begin{align}
        R^{*}(\psi\to\phi\,;\chi)\leq \sup\{r\geq 0\mid  \forall g\in G,\, \mathcal{Q}^{\mathcal{U}_g(\psi)}\geq r\mathcal{Q}^{\mathcal{U}_g'(\phi)}\}. 
    \end{align}

    Direct part: In the estimation step, we estimate $g\in G$ by consuming $(\mathcal{U}_g(\psi)\otimes \mathcal{U}_g''(\chi))^{\otimes \ceil{N^{1-\epsilon}}}$. The estimator $\hat{g}$ ensured to exist by Lemma~\ref{lem:reasonable_estimator} satisfies $((\mathcal{U}_{\theta}\otimes \mathcal{U}''_{\theta})\circ (\mathcal{U}_{\hat{g}}\otimes \mathcal{U}_{\hat{g}}''))(\psi\otimes \chi)=(\mathcal{U}_g\otimes \mathcal{U}''_g)(\psi\otimes \chi)$, i.e., 
    \begin{align}
        \mathcal{U}_{\theta}\circ\mathcal{U}_{\hat{g}}(\psi) =\mathcal{U}_g(\psi)\, \land\,  \mathcal{U}''_{\theta}\circ\mathcal{U}''_{\hat{g}}(\chi) =\mathcal{U}''_g(\chi)
    \end{align}
    with a failure probability that asymptotically vanishes as $N\to\infty$. From the assumption $\mathrm{Sym}_G(\chi)\subset \mathrm{Sym}_G(\phi)$, if $\theta$ satisfies $\mathcal{U}''_{\theta}\circ\mathcal{U}''_{\hat{g}}(\chi) =\mathcal{U}''_g(\chi)$, then $\mathcal{U}'_{\theta}\circ\mathcal{U}'_{\hat{g}}(\phi) =\mathcal{U}'_g(\phi)$ also holds, implying that Eq.~\eqref{eq:estimator_and_fidelity} remains valid in this case. 

    We slightly modify the channel in Eq.\eqref{eq:definition_full_channel} as
    \begin{align}
        &\mathcal{E}_N(\mathcal{U}_g(\psi)^{\otimes N}\otimes (\mathcal{U}_g''(\chi))^{\otimes \ceil{N^{1-\epsilon}}})\nonumber\\
        &\coloneqq \int \dd\mu_G( \hat{g})\, \,p(\hat{g}|(\mathcal{U}_g(\psi)\otimes \mathcal{U}''_g(\chi))^{\otimes \ceil{N^{1-\epsilon}}})\mathcal{E}_{N'}^{(\hat{g})}\left(\mathcal{U}_g(\psi)^{\otimes N'}\right).
    \end{align}
    For this sequence of the channels, by following the arguments from Eq.~\eqref{eq:estimator_and_fidelity} to Eq.~\eqref{eq:direct_part_discard}, we complete the proof of the direct part, i.e., 
    \begin{align}
        R^{*}(\psi\to\phi\,;\chi)\geq \sup\{r\geq 0\mid  \forall g\in G,\, \mathcal{Q}^{\mathcal{U}_g(\psi)}\geq r\mathcal{Q}^{\mathcal{U}_g'(\phi)}\}. 
    \end{align}
\end{proof}

\makeatletter
\let\addcontentsline\orig@addcontentsline
\makeatother

\end{document}